\documentclass[onecolumn,letterpaper,left=1in, right=1in, bottom=1in, top=0.75in, 10pt]{IEEEtran}

\usepackage[utf8]{inputenc} 
\usepackage[T1]{fontenc}
\usepackage{url}
\usepackage{ifthen}
\usepackage{cite}
\usepackage[cmex10]{amsmath}
\usepackage[includefoot,left=13mm,right=13mm,top = 13mm, bottom=13mm]{geometry}
\usepackage{mleftright}       
\mleftright                   

\usepackage{graphicx}         
\usepackage{booktabs}         
\usepackage{amsthm}
 \usepackage{graphicx}
\usepackage{amsmath}
\usepackage{graphics,graphicx,psfrag,color,float, hyperref}
\usepackage{epsfig,bbold,bbm}
\usepackage{bm}
\usepackage{mathtools}
\usepackage{amsmath}
\usepackage{amssymb}
\usepackage{amsfonts}
\usepackage{amsfonts}
\usepackage {times}
\usepackage{bbm}
\usepackage{soul}
\usepackage{cleveref}
\usepackage{enumitem}
\usepackage{xcolor}
\usepackage{array,multirow}
\newtheorem{theorem}{\bf \emph{Theorem}}
\newtheorem{fact}{Fact}

\newcommand{\beq}{\begin{equation}}
\newcommand{\enq}{\end{equation}}
\newcommand{\bel}{\begin{lemma}}
\newcommand{\enl}{\end{lemma}}
\newcommand{\bet}{\begin{theorem}}
\newcommand{\ent}{\end{theorem}}

\newcommand{\tr}{\mathrm{Tr}}

\newcommand{\nn}{\nonumber}

\newcommand{\Tr}{\mathsf{Tr\,}}
\newcommand{\ketbra}[1]{|#1\rangle\langle#1|}

\DeclareMathOperator*{\bigplus}{\text{\Large $+$}}
\newcommand{\RNum}[1]{\uppercase\expandafter{\romannumeral #1\relax}}
\newcommand{\eps}{\varepsilon}

\newcommand*{\bbR}{\mathbb{R}}

\newcommand*{\bbI}{\mathbb{I}}

\newcommand*{\bbN}{\mathbb{N}}

\newcommand*{\bbT}{\mathbb{T}}

\newcommand*{\bbE}{\mathbb{E}}

\newcommand*{\cP}{\mathcal{P}}
\newcommand*{\cC}{\mathcal{C}_{n}}
\newcommand*{\cA}{\mathcal{A}}
\newcommand*{\cR}{\mathcal{R}}
\newcommand*{\cH}{\mathcal{H}}
\newcommand*{\cM}{\mathcal{M}}

\newcommand*{\cB}{\mathcal{B}}
\newcommand*{\cD}{\mathcal{D}}

\newcommand*{\cI}{\mathcal{I}}

\newcommand*{\cL}{\mathcal{L}}
\newcommand*{\cN}{\mathcal{N}}
\newcommand*{\cO}{\mathcal{O}}
\newcommand*{\cS}{\mathcal{S}}

\newcommand*{\cT}{\mathcal{T}}
\newcommand*{\cX}{\mathcal{X}}
\newcommand*{\cW}{\mathcal{W}}
\newcommand*{\cZ}{\mathcal{Z}}
\newcommand*{\cE}{\mathcal{E}}
\newcommand*{\cU}{\mathcal{U}}

\newcommand*{\cV}{\mathcal{V}}
\newcommand*{\cY}{\mathcal{Y}}

\newcommand{\bra}[1]{\langle #1|}
\newcommand{\ket}[1]{|#1 \rangle}

\mathchardef\mhyphen="2D

\newcommand{\ketbrasys}[2]{\overset{#2}{|#1\rangle\langle#1|}}

\newcommand{\norm}[2]{{\left\lVert #1 \right\rVert}_{#2}}
\newcommand{\abs}[1]{\left\vert#1\right\vert}

\makeatletter
\newcommand*{\rom}[1]{\expandafter\@slowromancap\romannumeral #1@}
\makeatother
\mathchardef\mhyphen="2D
\newlist{steps}{enumerate}{1}
\setlist[steps, 1]{leftmargin = 1.1cm, label = Step \arabic*}
\newtheorem{remark}{Remark}
\newtheorem{definition}{Definition}

\newtheorem{lemma}{Lemma}
\newtheorem{corollary}{Corollary}
\newtheorem{proposition}{Proposition}
\newtheorem{assumption}{Assumption}
\usepackage[utf8]{inputenc} 
\usepackage[T1]{fontenc}
\usepackage{url}
\usepackage{ifthen}
\usepackage{cite}
\usepackage[cmex10]{amsmath}
\usepackage{amsfonts} 

\def\QED{\mbox{\rule[0pt]{1.5ex}{1.5ex}}}
\def\endproof{\hspace*{\fill}~\QED\par\endtrivlist\unskip}

\begin{document}
\title{Universal tester for multiple independence testing and
classical-quantum arbitrarily varying multiple access channel} 

 \author{Ayanava Dasgupta, Naqueeb Ahmad Warsi,
and Masahito Hayashi \IEEEmembership{Fellow, IEEE}
\thanks{The work of MH was supported 
in part by the National Natural Science Foundation of China (Grant No. 62171212).}
\thanks{
Ayanava Dasgupta and Naqueeb Ahmad Warsi are with
Electronics and Communication Sciences Unit,
Indian Statistical Institute, Kolkata, Kolkata 700108, India.
(e-mail: [ayanavadasgupta\_r, naqueebwarsi]@isical.ac.in).
Masahito Hayashi is with 
School of Data Science, The Chinese University of Hong Kong, Shenzhen, Longgang District, Shenzhen, 518172, China,
International Quantum Academy (IQA), Futian District, Shenzhen 518048, China,
and
Graduate School of Mathematics, Nagoya University, Chikusa-ku, Nagoya 464-8602, Japan.
(e-mail: hmasahito@cuhk.edu.cn).}}
                    
\markboth{A. Dasgupta, N. Warsi, and M. Hayashi:
Universal tester for multiple independence testing \&
CQ-VAMAC}{}

\maketitle

\begin{abstract}
We study two kinds of different problems.
One is the multiple independence testing, which can be considered as a kind of generalization of quantum Stein's lemma.
We test whether the quantum system is correlated to the classical system or is independent of it. 
Here, the null hypothesis is composed of 
states having the quantum system is correlated to the classical system in an arbitrarily varying form.
The second problem is 
the problem of reliable communication over classical-quantum arbitrarily varying multiple access channels (CQ-AVMAC) and establishing its capacity region by giving multiple achievability techniques. 
We prove that each of these techniques is optimal by proving a converse. Further, for both these techniques, the decoder designed is a \emph{universal} decoder and can achieve any rate pair in the capacity region without time sharing and also these decoders do not depend on the channel and therefore they are universal. 
Our result covers the case when the channel parameter is continuous, which has not been studied even in the classical case.
Further, both these techniques can be easily generalized to the case when there are $T (T>2)$ senders.
The design of each of these decoders is based on the study of multiple independence testing.
This approach allows us to study the problem of reliable communication over CQ-AVMAC from the point of view of hypothesis testing. Further, we also give a necessary and sufficient condition for the deterministic code capacity region of CQ-AVMAC to be non-empty.  
\end{abstract}


\section{Introduction}
\subsection{Arbitrarily varying channel and its quantum extension}
Most of the results in traditional information theory are based on the assumption that the source and the channel are independent and identically distributed. However, in the real-world scenario, these assumptions may not be true. To relax these assumptions \cite{Blackwell1960,Berger1975} introduced the notion of arbitrarily varying source and channel and studied the two fundamental problems of information theory i.e. the problem of source compression for arbitrarily varying source (AVS) and the problem of reliable communication over an arbitrarily varying channel (AVC). 
In arbitrarily varying channel settings, 
the channel changes arbitrarily and is 
chosen randomly from a finite collection of channels
$\{W_s\}_{s \in \cS}$, where $s \in \cS$ is the channel parameter. 
In this setting, the sender and the receiver do not know 
the value of the parameter $s$ nor 
the probability distribution for the parameter $s$.
When we have a transmission with length $n$,
the parameter is chosen differently as
$s_1, \ldots, s_n$.
Therefore, we focus on the worst value of 
the decoding error probability under the choice of $n$ parameters
$s_1, \ldots, s_n$.
In contrast, in the compound channel model \cite{Blackwell1959,Wolfowitz1959,Dobrusin1959},
the sender and the receiver do not know the channel parameter
$s$, but the parameter $s$ is fixed during the communication with length $n$.
This kind of indeterminateness can appear even in the quantum channel setting
because it is extensively difficult to control quantum devices.
Therefore, the channel coding over AVC with the quantum setting
is a more challenging problem.

For the classical setting,
in \cite{Ahlswede1970} it was shown that the capacity of an AVC may be zero under a certain necessary and sufficient condition, termed as ``Symmetrizability''. 
However, when the sender and the decoder share randomness to decide 
the choice of code, the capacity is always non-zero (positive), which is called the randomized code capacity. 
The problem of reliable communication over AVC with the classical setting
has been studied in great detail, see for example, \cite{Kiefer1962,Stiglitz66,AHLSWEDE1969457,Ahlswede1973,Ahlswede1978,Ahlswede1980,Csiszar1981,Ericcon1985,Csiszar1988}. 
When the channel has a classical input and a quantum output,
the channel is called a classical-quantum channel.
In this setting,
Ahlswede \cite{Ahlswede07} studied the problem of reliable communication over classical-quantum arbitrarily varying channels (CQ-AVC) both for fixed codes and random codes and established its capacity.
The paper \cite{Bjelakovic2013} made further progress.
However, these papers assume the finiteness of the set $\cS$ of 
channel parameters.

\subsection{Arbitrarily varying multiple-access channel and its quantum extension}
Establishing quantum network communication is another key topic 
in the area of quantum information.
As one of its key approaches,
we focus on multiple-access classical-quantum channel (CQ-MAC),
where we have multiple senders and one receiver.
Since CQ-MAC has a more complicated structure than 
the conventional channel setting, it is very difficult to fix the channel parameters.
Therefore, classical-quantum arbitrarily varying multiple-access channel (CQ-AVMAC) is a more important problem setting.
This paper focuses on CQ-AVMAC, and 
design two different kinds of \textit{universal decoders} for CQ-AVMAC
whose method is explained later.
We use these decoders to first establish the randomized code capacity region of CQ-AVMAC, where
each sender is allowed to have shared randomness with the receiver.
We then, derandomize our code to establish the deterministic code capacity region. 
We show that this derandomization is not possible under certain necessary and sufficient conditions which we term as ``$\cX$-symmetrizable'', ``$\cY$-symmerizable'' and ``$\cX\cY$-symmetrizable''. 
These conditions are similar in spirit to their classical counterparts which were used in studying the problem of reliable communication over classical arbitrarily varying multiple-access channel (AVMAC)\cite{Jahn1981,gubner1990deterministic,Ahlswede1999}. 
We also prove a matching \textit{converse} for this problem. To the best of our knowledge,  the capacity region for CQ-AVMAC has not been established in the literature prior to this work.

\subsection{Approach of asymmetric hypothesis testing}
To accomplish the above achievements,
we focus on an approach of asymmetric hypothesis testing.
In \cite{HN2003,AJW2019,wang-renner-prl}, it has been observed that the problem of reliable communication over a noisy channel is closely related to asymmetric hypothesis testing. 
Recently, quantum Stein's lemma \cite{HP1991,ON2000,Hayashi2002} and its various generalizations \cite{BP2010,Berta2023gapinproofof,HH2024,Lami2024,HT,TH2018,BD2005,Notzel_2014,BBH2021} attracts much attention in the area of quantum information.
Our approach is more related to the independence test formulated in Hayashi and Tomamichel \cite{HT}, which imposes a constant constraint for the error probability of incorrectly accepting the alternative hypothesis $H_1$ in a way similar to quantum Stein's lemma.
To see this relation, we focus on compound channels in the CQ channels.


In the classical-quantum setting, Datta and Dorlas in \cite{Datta_2007}, came up with a construction of 
code for a compound channel.
However, this code only works for a finite collection of possible classical-quantum channels. Later, \cite{Hayashi2009,BB09} proposed a universal coding strategy for classical-quantum channels, which works 
for compound channels composed of an infinite collection of possible classical-quantum channels. 
In particular, in \cite{Hayashi2009}, the construction of the decoder for this universal code uses Schur-Weyl duality \cite{fulton1991representation,Group2} and the method of types and the encoder does not depend on the dimension of the output quantum system. 
The independence test \cite{HT} also employs the same structure.
Later, the approach \cite{Hayashi2009}
has been extended to CQMAC in \cite{HC}.

To extend this approach to CQ-AVC and CQ-AVMAC,
we extend the result in \cite{HT} to 
various settings of asymmetric hypothesis testing between arbitrarily varying sources. We design two kinds of tests to study this problem. In particular, the first test consists of a collection of \textit{non-simultaneous} tests, where each test identifies a certain type of independence between the classical and quantum subsystems. However, the second test which we design is a \textit{simultaneous} test (based on the ideas of \cite{Sen2021}) and unlike the first test, it consists of only one test, which simultaneously identifies whether the quantum subsystem is correlated with each of the classical subsystems with high probability.

Then, employing the standard relation between 
cq channel coding and asymmetric hypothesis testing,
we design \textit{universal decoders} 
for CQ-AVC and for CQ-AVMAC.
Furthermore, our method for CQ-AVC covers the case when
the set of channel parameters $\cS$ is not a finite set.
This case was not covered by \cite{Ahlswede07,Bjelakovic2013}.
In addition,
our analysis covers the case when the set $\cS$ forms a continuous set.
Table \ref{tab} summarizes
the relation between our target problems and existing studies. 

\begin{table}[ht]
\caption{Relationship between our channel model and related studies}\label{tab}
\begin{center}
\begin{tabular}{|c|c|c|c|c|c|}
\hline
\multicolumn{2}{|c|}{} & \multicolumn{2}{|c|}{Compound channel} & \multicolumn{2}{|c|}{Arbitrary varying channel} \\
\hline
\multicolumn{2}{|c|}{Behavior of unknown channel parameter}& 
\multicolumn{2}{|c|}{fixed} & \multicolumn{2}{|c|}{changing} \\
\hline
\multicolumn{2}{|c|}{Parameter set}& 
finite & general & finite & general \\ 
\hline
\multirow{2}{*}{point-to-point channel} 
& classical & & \cite{Blackwell1959,Wolfowitz1959,Dobrusin1959,csiszar-korner-book} & \cite{Blackwell1960,Kiefer1962,Stiglitz66,AHLSWEDE1969457,Ahlswede1973,Ahlswede1978,Ahlswede1980,Csiszar1981,Ericcon1985,Csiszar1988}&\\
\cline{2-6}
& classical-quantum & \cite{Datta_2007,Bjelakovic2013} &
\cite{Hayashi2009,BB09}
  & \cite{Ahlswede07,Bjelakovic2013} 
  &  Theorem \ref{lemma_rand_cap_cqavc} of this paper\\
\hline
\multirow{2}{*}{multiple access channel} 
& classical && \cite{Ahlswede1974,PW85,YH96} &\cite{Jahn1981,gubner1990deterministic,Ahlswede1999}& \\
\cline{2-6}
& classical-quantum & & \cite{HMW16,BJS19,HC}
&& Theorems \ref{lemma_rand_capacity_avmac} and \ref{lemma_converse_det} of
this paper \\
\hline
\end{tabular}
\end{center}
\end{table}

The rest of the paper is organized as follows. 
Section \ref{SIT} formulates our key approach, 
the independence test, and briefly summarizes our approach
for CQ-AVC and for CQ-AVMAC.
In Section \ref{sec:bg} we mention the notations and facts used in this manuscript, in Section \ref{sec:IT} we give a proof for Theorem \ref{theorem_generalised_independence_test_arbitrary_varying}. 
In Section \ref{sec:section_CQAVC}, we use Theorem \ref{theorem_generalised_independence_test_arbitrary_varying} to give an achievability proof for the problem of reliable communication over CQ-AVC. 
In Section \ref{sec:GIT}, 
we generalize the independence test to the case with three subsystems, which is formulated as 
the problem of ``multiple independence testing''. 
We study the ``universal test'' to solve this problem.
In Section \ref{section_CQAVMAC} we study CQ-AVMAC by giving an achievability and converse. 
Thus, deriving its randomized code capacity region and in section \ref{sec:determinsitic_cap_section} we derandomize the code obtained in section \ref{sec:section_CQAVC} and give a necessary and sufficient condition for it to be possible.

\section{Independence test and brief summary of our approach}\label{SIT}
We build on the ideas developed in \cite{Hayashi2009} to study the problem of independence testing defined below. \textit{Throughout the manuscript we use $\cB^n$ to denote the quantum system that contains the Hilbert space $\cH_\cB^{\otimes n}$.}

\begin{definition}(Independence Testing)\label{def_cq_avht_null}
We consider a classical subsystem $X$ of dimension 
$\abs{\cX} < \infty$, spanned by an orthonormal basis $\left\{\ket{x}^{X}\right\}_{x \in \cX}$.
We choose a probability distribution $P$ defined over $\cX$, which introduces 
the state ${\rho^{X}_{P}} := \sum_{x \in \cX}P(x)\ketbrasys{x}{X}$
and the distribution $P^{n}(x^n) := \prod_{i=1}^{n}P(x_i)$.
Consider a collection of density operators $\left\{\rho_{x,s}^{\cB}\right\}_{\substack{ x \in \cX, s \in \cS}}  \subset \cD(\cH_\cB)$, where 
$\cS$ is a non-empty set and $\cH_\cB$ is a 
finite-dimensional
Hilbert space. Now consider two collections 
$H_0:=\left\{\rho^{X^n\cB^n}_{P,s^n}\right\}_{s^n \in \cS^n}$ and $H_1:=\left\{{\rho^{X}_{P}}^{\otimes n} \otimes \sigma^{\cB^n}\right\}_{\sigma^{\cB^n} \in \cD(\cH_\cB^{\otimes n})}$ of classical-quantum states, such that $\forall s^n \in \cS^n$,
    \begin{align}
        \rho^{X^n\cB^n}_{P,s^n} &:= \sum_{x^n \in \cX^n}P^{n}(x^n)\ketbrasys{x^n}{X^n}\otimes \rho_{x^n,s^n}^{\cB^n},
\label{CMZP}
    \end{align}
where the state $\rho_{x^n,s^n}^{\cB^n}$ is defined to be $\bigotimes_{i=1}^{n}\rho_{x_i,s_i}^{\cB}$ for $x^n:=(x_1,\cdots,x_n)$. 
The aim here is to design a hypothesis test (which is a measurement in the quantum setting and for the sake of simplicity we assume it to be projective) that accepts $H_0$ (Null Hypothesis) with a very large probability and accepts $H_1$ (Alternative Hypothesis) with an arbitrarily small probability. 
\end{definition}

To handle a general case for the set $\cS$ of parameters,
we introduce the following assumption.
\begin{assumption}\label{NM9}
Given such a real number $\eps  > 0$, 
there exists a finite subset $\cS_{\eps} \subset \cS$ 
to satisfy the following condition.
For any $s \in \cS$, there exists an element $\bar{s} \in \cS_{\eps}$ 
such that any $x \in \cX$ satisfies 
    \begin{equation}
        \rho^{\cB}_{x,s} \preceq (1 
 + \eps)\rho^{\cB}_{x,\bar{s}}.\label{epsilon_net_state_density_eq}
    \end{equation}
\end{assumption}

Subsection \ref{S4-AB} discusses what case satisfies this assumption.
In this problem, similar to quantum Stein's lemma,
we impose a constant constraint for the error probability of incorrectly accepting the alternative hypothesis $H_1$.
Using Schur-Weyl duality and the method of types, we obtain the following theorem:

\begin{theorem}
\label{theorem_generalised_independence_test_arbitrary_varying}
Assume that a collection of density operators $\left\{\rho_{x,s}^{\cB}\right\}_{\substack{ x \in \cX, s \in \cS}}$
satisfies Assumption \ref{NM9}.
Consider the two hypotheses $H_0, H_1$ mentioned in the statement of Definition \ref{def_cq_avht_null} and a non-empty set $\cS$. For $\forall \eps \in (0,1)$, 
we define
    \begin{equation}
\beta_{n}(\eps) := \min_{\substack{0 \preceq \bbT_{n} \preceq \bbI}} \left\{\max_{\sigma^{\cB^n} \in \cD(\cH_\cB^{\otimes n})}{\tr\left[\bbT_{n}
    \left({\rho^{X}_{P}}^{\otimes n} \otimes \sigma^{\cB^n}\right)
    \right]} \quad \bigg| \quad \max_{s^n \in \cS^n} \tr[(\bbI^{X^n\cB^n} - \bbT_{n})\rho^{X^n\cB^n}_{P,s^n}] \leq \eps\right\}.
    \end{equation}
Then, we have
    \begin{equation}
        \lim_{n \to \infty}-\frac{1}{n}\log\beta_{n}(\eps) = 
        \inf_{Q\in \cP(\cS)} I[X;\cB]_{P,Q}, \label{Stein}
    \end{equation}
        where $\cP(\cS)$ is the set of distributions on $\cS$, and 
    see \eqref{sibsons_qmi_convergence_eq} in Fact \ref{sibsons_quantum_mutual_info} for the definition of $I[X;\cB]_{P,Q}$.
        \end{theorem}

In this paper, we use only the achievability of Theorem \ref{theorem_generalised_independence_test_arbitrary_varying}, i.e., 
$\ge$ in \eqref{Stein},
to design a universal decoder for CQ-AVC and establish its randomized code capacity. 
We show its achievability in Section \ref{sec:IT}, and show its converse part in Appendix \ref{NMA8}.

Our technique to CQ-AVC
differs from the approach used in \cite{Ahlswede07} in two ways. Firstly, our technique is based on hypothesis testing and secondly, our decoder is a universal decoder which was not the case in \cite{Ahlswede07}.
Here, the projective measurement, which is used as a hypothesis test, does not depend on $\rho^{\cA\cB}$ and needs to work 
for any arbitrary bipartite correlated state in $H_0$ over $\cA\otimes \cB$. 
Hence, our difficulty relies on 
the universal choice of our test.
That's why the authors term it as ``independence test''.
We note here that in Definition \ref{def_cq_avht_null}, both $H_0$ and $H_1$ are composite. 

Here, we overview the relation with related problem settings.
Hayashi and Tomamichel studied a similar problem in \cite{HT} wherein, 
only the hypothesis $H_1$ is composite. 
In particular, they considered an arbitrary bipartite correlated state $\rho^{\cA\cB}$ and studied an asymmetric hypothesis testing between the following two hypotheses:
 \begin{align}
     H_0 := {\rho^{\cA\cB}}^{\otimes n}, & & H_1 := \left\{{\rho^{\cA}}^{\otimes n} \otimes \sigma^{\cB^n}\right\}_{\sigma^{\cB^n} \in \cD(\cH_\cB^{\otimes n})},\label{MT}
 \end{align}
 where $\rho^{\cA} = \tr_{\cB}\left[\rho^{\cA\cB}\right]$. 
That is, $H_0$ is the $n$-copy state ${\rho^{\cA\cB}}^{\otimes n}
$ of $\rho^{\cA\cB}$.
The aim here is to design a measurement-based test that accepts $H_0$ with very high probability and accepts $H_1$ with arbitrarily small probability. 
The paper \cite{TH2018} generalized these results under the classical setting by relaxing the condition for $H_1$, 
which includes the testing of Markovian property.
The quantum generalized Stein's lemma
 \cite{BP2010,Berta2023gapinproofof,HH2024,Lami2024}
 considers the more general case when 
 $H_0$ is composed of a single $n$-copy state and 
 $H_1$ satisfies several conditions like convexity etc.
In contrast, quantum Sanov theorem \cite{BD2005,Notzel_2014} 
discusses the case 
when $H_0$ is composed of various $n$-copy state and
 $H_1$ is composed of a single $n$-copy state.
The paper \cite{BBH2021} extended the above setting to 
the case when $H_0$ and $H_1$ are composed of 
multiple $n$-copy states.
However, in our problem, 
$H_0$ and $H_1$ are composed of more complicated states, which reflects 
the arbitrarily varying channel parameter. 
Both hypotheses $H_0$ and $H_1$ contain non $n$-copy states.
In particular, the main difference from \cite{HT} is the richer structure 
of $H_1$ as \eqref{CMZP}.

Also, in our case $\abs{\cS}$ may not be finite.
To address CQ-AVMAC, we extend the above setting to the case when 
the classical system $X$ is composed of two subsystems.
We study testing for whether the first or second classical system
is independent of other systems, which is formulated as
the problem of ``Multiple independence testing'' and is discussed in detail in Section \ref{sec:GIT}. 
Table \ref{tab2} summarizes
the relation between two target hypothesis testing problems 
of ours and existing studies, which include 
Theorem \ref{theorem_generalised_independence_test_arbitrary_varying}
and ``Multiple independence testing''.
In fact, we give two different techniques to study this problem. 
Both these techniques yield a universal POVM. In particular, in the first technique, we use Schur-Weyl duality \cite{fulton1991representation,Group2}, the method of types and sandwiching of projectors to come up with a universal \textit{non-simultaneous} test. 
On the other hand in the second approach, we first use the test obtained in Theorem \ref{theorem_generalised_independence_test_arbitrary_varying} and use techniques like ``tilting'' and ``augmentation'' (developed by Sen in \cite{Sen2021}) to arrive at a universal \textit{simultaneous} tester. We make a comparison between both these techniques in Table \ref{tab:table1}.

\begin{table}[ht]
\caption{Hypothesis testing problem settings of 
this paper and related studies}\label{tab2}
\begin{center}
\begin{tabular}{|c|c|c|c|}
\hline
& \multirow{2}{*}{$H_0$} & \multirow{2}{*}{$H_1$} & form of \\
&&&optimal rate \\
\hline
Independence test in  & 
State on the second subsystem is arbitrary 
& \multirow{2}{*}{Independent states for two systems} & 
\multirow{3}{*}{single-letterized}\\
Theorem \ref{theorem_generalised_independence_test_arbitrary_varying} (Sections \ref{SIT} and \ref{sec:IT})
&
chosen depending on the first subsystem &&\\
\cline{2-3}
of this paper
& \multicolumn{2}{|c|}{State on the first subsystem is fixed} &\\
\hline
\multirow{2}{*}{Multiple independence test in}  & 
State on the third subsystem is arbitrarily chosen
& 
States where first or second subsystem &
single-letterized\\
\multirow{2}{*}{Section \ref{sec:GIT} of this paper} &
depending on the first and second subsystems &
is independent of other systems
&
achievable rate\\
\cline{2-3}
&\multicolumn{2}{|c|}{Joint state on the first and second subsystems is fixed}
&  is obtained\\
\hline
\multirow{2}{*}{Independence test in \cite{HT}} &
\multirow{2}{*}{Single $n$-copy state} & Independent states for two systems 
&\multirow{2}{*}{single-letterized}\\
&&State on the first subsystem is fixed & \\
\hline
\multirow{4}{*}{Reference \cite{TH2018}} &
\multirow{4}{*}{Single $n$-copy state} & Several general conditions
&\multirow{4}{*}{single-letterized}\\
&& except for convexity &\\
&& Independent case or &\\
&& Markovian case are included &\\
\hline
Quantum generalized 
&
\multirow{2}{*}{Single $n$-copy state}& 
Convexity and 
&\multirow{2}{*}{limiting expression}\\
Stein's lemma \cite{BP2010,Berta2023gapinproofof,HH2024,Lami2024} 
&&other general conditions&\\
\hline
Quantum Sanov theorem \cite{BD2005,Notzel_2014} 
&Multiple $n$-copy states& Single $n$-copy state
&{single-letterized}\\
\hline
Reference \cite{BBH2021} &Multiple $n$-copy states& Multiple $n$-copy states
&{single-letterized}\\
\hline
\end{tabular}
\end{center}
\end{table}

The above two techniques allow us to construct a \textit{universal simultaneous decoder} for CQ-AVMAC. Even though we study the problem of CQ-AVMAC for the case of two senders in this manuscript, we show that both the above techniques can easily be generalized to the case when there are $T$-senders (where $T>2$). We note here that the first technique does not require the Hilbert space to be enlarged to obtain a simultaneous decoder unlike the second technique (based on the techniques developed in \cite{Sen2021}). However, the first technique requires sandwiching of $O(2^{T-1})$ projectors to construct the decoding POVM for the case of $T (T>2)$ senders. Further, the error analysis for this decoder is arguably simpler in comparison to the decoder obtained using the second technique.

\section{Notations and Facts}\label{sec:bg}
We use $\cH$ to denote a finite-dimensional Hilbert space, $\mathcal{D}(\cH)$ to represent the set of all state density operators acting on $\cH$ and $\cL(\cH)$ represents the set of all linear operators over $\cH$ ($\cD(\cH)\subset\cL(\cH)$). For any quantum system $A$ over the Hilbert space $\cH_{A}$ $\abs{\cA}$ represents the dimension of the Hilbert space $\cH_{A}$. Let $\cX$ and $\cS$ be two finite sets denoting the source and state alphabets. For any integer $b>0$, $[b]:=[1,2,\cdots, b]$. $\cP(\cX)$ is the set of probability distribution over $\cX$. $\cT_{n}(\cX)$ is the set of empirical distributions with denominator $n$.  For any quantum state $\rho$, we define $\mbox{supp}(\rho):=\mbox{Span}\{\ket{i}: \lambda_{i}> 0\},$ where $\left\{\lambda_i\right\}$ represent non-zero eigenvalues of $\rho$. Similarly, for any distribution $P$ defined over a set $\mathcal{X}$, $\mbox{supp}(P):=\{x\in \mathcal{X}: P(x)>0\}.$ For any sequence $x^n : (x_1,\cdots,x_n) \in \cX^n$, a permutation $\pi_n:\cX^n\to\cX^n,$ of length $n$ is a map such that $\pi_n\left(x^n\right):=(x_{\pi_{n}^{-1}(1)},\cdots,x_{\pi_{n}^{-1}(n)}).$ Let $S_{n}$ be the permutation group containing all permutations of length $n$ and it is also known as $n$-th symmetric group. For a sequence $x^n \in \cX^n$, we define the subgroup $S_{n,x^n} := \left\{\pi_n \in S_n : \pi(x^n) = x^n\right\}$
Given two distribution $P \text{  and } Q$ over the set $\mathcal{X}$ such that $\mbox{supp}(P) \subseteq \mbox{supp}(Q)$, the relative entropy between $P$ and $Q$ is defined as,
\begin{equation*}
    D(P\|Q) := \sum_{x \in \mathcal{X}}P(x)\log\left(\frac{P(x)}{Q(x)}\right).
\end{equation*}
Given two quantum states $\rho, \sigma \in \mathcal{D}(\cH)$ such that $\mbox{supp}(\rho) \subseteq \mbox{supp}(\sigma)$, the quantum relative entropy (also known to be \emph{quantum KL-divergence}) between $\rho$ and $\sigma$ is defined as,
\begin{equation*}
    D(\rho\| \sigma):= \tr[\rho(\log(\rho)-\log(\sigma))].
\end{equation*}
 A distance measure on $\cP(\cX)$ is defined by $\left\lVert{ P-Q }\right\rVert =\sum_{x\in \cX}|P(x)-Q(x)|$ where $P, Q \in \cP(\cX)$. For an operator $O \in\cL(H)$, $\norm{O}{p}$ is the operator Schatten-$p$ norm \cite[Definition $2.6$]{khatri2024principlesquantumcommunicationtheory} of $O$. 
 The entropy of $P \in \cP(\cX)$ defined as $H(P)=-\sum_{x \in \cX}P(x)\log P(x)$. For a set $A$, $conv(A)$ denotes the closed convex hull of the set $A$. For two normal operators $A$ and $B$ we will use the notation $\Pi :=\{A \succeq B\}$ to denote a projection operator on the positive eigenspace of the operator $(A-B).$
\begin{definition}(Dirac delta function)\label{deltadis}
    Given a non-empty set $\cS$, for any $s \in \cS$, we define the Dirac delta function or the delta distribution $\delta_{s}$ on $s$ as follows:
    \begin{equation*}
       \forall s' \in \cS, \delta_{s}(s') := \begin{cases}
            1 \quad \textnormal{if} \quad s' = s,\\
            0  \quad \textnormal{otherwise}.
        \end{cases}
    \end{equation*}
\end{definition} \begin{definition}\label{smallo}
 Given two functions $f,g : \bbN \to \bbR$, for any $n \in \bbN$, $f(n)$ is said to be $o(g(n))$ if,
\begin{equation*}
    \lim_{n \to \infty}\frac{f(n)}{g(n)} = 0
\end{equation*}
 \end{definition}
 \begin{definition}\label{smallomega}
 Given two functions $f,g : \bbN \to \bbR$, for any $n \in \bbN$, $f(n)$ is said to be $\omega(g(n))$ if,
\begin{equation*}
    \lim_{n \to \infty}\frac{g(n)}{f(n)} = 0
\end{equation*}
 \end{definition}
 \begin{definition}[Topological interior]\label{def_top_interior}
     The topological interior $\mathbf{Int}\{A\}$ of a subset $A$ of a topological space $X$ is defined as the largest open subset of $X$ contained in $A$.
 \end{definition}
 \begin{definition}(Type and Set of all types)\label{set_of_types}
     A vector of integers $n^d := \left(n_1,n_2,\cdots,n_d\right)$ (where $d > 0$) is called \textbf{a type of size $n$ and length $d$} if $\forall i \in [d], n_i \geq 0$ and $\sum_{i} n_i = n.$ $T^{n}_{d}$ is denoted as the set of all distinct types of size $n$ and length $d$.
 \end{definition}
 \begin{definition}(Young diagram and Set of all Young diagrams)\label{set_of_tableux}
     A vector of integers $n^d := \left(n_1,n_2,\cdots,n_d\right)$ (where $d > 0$) is called \textbf{a Young diagram of size $n$ and depth $d$ }if $n_1\geq n_2\geq\cdots\geq n_d \geq 0$ and $\sum_{i} n_i = n.$ $\Lambda^{n}_{d}$ is denoted as the set of all distinct Young diagrams of size $n$ and depth $d$.
 \end{definition}
 \begin{definition}[Frequency type classes]\label{def_freq_type}
    For a non-empty finite set $\cX$ and an integer $n > 0$, given a frequency (or, type) function $f : \cX \to \mathbb{N} $ ( such that $\sum_{i \in \cX}f(i)=n$), we define a set $T_{\bar{f}} \subset \cX^n$ to be the frequency type class or type set corresponding to $\bar{f}$ where $\forall i \in \cX,  \bar{f}(i) = \frac{f(i)}{n}$ as follows,
    \begin{equation*}
        T_{\bar{f}} := \left\{(x_1,\cdots,x_n) \in \cX^n : \frac{\left\lvert\left\{k : x_k = i\right\}\right\rvert}{n} = \bar{f}(i) , \forall i \in \cX \right\}.
    \end{equation*}
\end{definition}
The above set $ T_{\bar{f}}$ is a permutation invariant set i.e. for every sequence $x^n : (x_1,\cdots,x_n) \in \cX^n$, under any permutation $\pi_n \in S_n,$ $\pi_n\left(x^n\right) \in  T_{\bar{f}}$ where $\pi_n\left(x^n\right):=(x_{\pi_{n}^{-1}(1)},\cdots,x_{\pi_{n}^{-1}(n)}).$

\begin{definition}[Permutation operators] \cite[Definition 8]{Mele2024}\label{fact_perm_op}
    Given a permutation $\pi \in S_n$ and a Hilbert space $\cB$, the permutation operator corresponding to $\pi$, denoted as $V^{\cB^n}(\pi)$, is an unitary matrix acting on $\cH_\cB^{\otimes n}$ which satisfies the following :
    \begin{equation}
        V^{\cB^n}(\pi) \ket{\psi_1} \otimes \ket{\psi_2} \otimes \cdots \otimes \ket{\psi_n} = \ket{\psi_{\pi^{-1}(1)}} \otimes \ket{\psi_{\pi^{-1}(2)}} \otimes \cdots \otimes \ket{\psi_{\pi^{-1}(n)}},\label{perm_op_closed_form}
    \end{equation}
    for all $\ket{\psi_1},\ket{\psi_2},\cdots,\ket{\psi_n} \in \cB$.
    Equivalently, for any orthonormal eigenbasis $\left\{\ket{j}\right\}_{j \in [\abs{\cB}]}$ of the Hilbert space $\cB$, $V^{\cB^n}(\pi)$ can be given as following:
    \begin{equation*}
        V^{\cB^n}(\pi) := \sum_{i_1,i_2,\cdots,i_n \in [\abs{\cB}]^{n}} \ket{i_{\pi^{-1}(1)},i_{\pi^{-1
        }(2)}, \cdots, i_{\pi^{-1}(n)}}\bra{i_1,i_2,\cdots,i_n},
    \end{equation*}
    Thus, $V^{\cB^n}(\pi)$ satisfies the following property:
    \begin{equation}
        V^{\cB^n}(\pi) (\cO_1 \otimes \cO_2 \otimes \cdots \otimes \cO_n){V^{\cB^n}}^{\dagger}(\pi) = \cO_{\pi^{-1}(1)} \otimes \cO_{\pi^{-1}(2)} \otimes \cdots\otimes\cO_{\pi^{-1}(n)}\label{perm_op_eq},   \end{equation}
    for all $\cO_1, \cO_2, \cdots, \cO_n \in \cL(\cB)$. From \eqref{perm_op_closed_form}, it directly follows that
    \begin{equation}
        V^{\cB^n}(\pi) {V^{\cB^n}}^{\dagger}(\pi) = {V^{\cB^n}}^{\dagger}(\pi)V^{\cB^n}(\pi) = \bbI^{\cB^n}.\label{perm_op_prop_1}
    \end{equation}
\end{definition}
\begin{definition}[Iterative Matrix Multiplication]\label{def_iter_mul}
    Given $n$ matrices ( for the sake of simplicity we assume them to be square matrices) $M_1,M_2,\cdots,M_n$, $\bullet_{i \in [n]}M_i$ can be defined as follows,
    \begin{equation*}
        \bullet_{i \in [n]}M_i := M_1\cdot M_2\cdot M_3 \cdots M_n.
    \end{equation*}
\end{definition}
\begin{fact}[Direct Sum of vector spaces]\label{fact_direct_sum}
    Let $V_1, V_2,\dots, V_n$ be $n$  orthogonal subspaces of a vector space $V$. Then, the n-fold direct sum (or direct sum) of the vector spaces $V_1, V_2,\dots, V_n$ is defined as
    \begin{equation*}
        \bigoplus_{i=1}^{n} V_i  := \left\{(v_1, v_2,\dots, v_n) \hspace{5pt} | v_i \in V_i (1\leq i \leq n) \right\},
    \end{equation*}
    and, 
    \begin{equation*}
        \dim{\left\{\bigoplus_{i=1}^{n} V_i \right\}} = \sum_{i=1}^{n} \dim\{V_i\}.
    \end{equation*}
  
   For example, consider two vector spaces $V_1 = \mathbb{R}^2$ and $V_2 = \mathbb{R}^2$ over $\mathbb{R}$. Then, $V_1 \oplus V_2$ is a 4-dimensional vector space over $R$. If we consider a vector $\ket{a} := \begin{bmatrix}
       a_1 \\
       a_2 \\
   \end{bmatrix}
   $ in $V_1$ and another vector $\ket{b} := \begin{bmatrix}
       b_1 \\
       b_2 \\
   \end{bmatrix}
   $ in $V_2$, then $\ket{a}$ in the 4-dimensional vector space $V_1 \oplus V_2$ can be defined with padding $|\dim\{V_2\}|$ zeros below vector $\ket{a}$, making $\ket{a} = \begin{bmatrix}
       a_1 \\
       a_2 \\
       0 \\
       0 \\
   \end{bmatrix}$ and similarly $\ket{b}$ in $V_1 \oplus V_2$ would be $\begin{bmatrix}
       0 \\
       0 \\
       b_1 \\
       b_2 \\
   \end{bmatrix}$. Thus, $\ket{a} + \ket{b} = \begin{bmatrix}
       a_1 \\
       a_2 \\
       b_1 \\
       b_2 \\
   \end{bmatrix}$.
\end{fact}
\begin{fact}\cite{Group2}\label{fact_type_size_ub}
    $T^{n}_{d}$ (See Definition \ref{set_of_types}) and $\Lambda^{n}_{d}$ (See Definition \ref{set_of_tableux}) satisfies the following upper-bound :
    \begin{equation}
        \abs{\Lambda^{n}_{d}} \leq \abs{T^{n}_{d}} \leq (n+1)^{d-1}.\label{set_Youg_Type_UB}
    \end{equation}
\end{fact}

\begin{fact}[Petz Quantum Re\`nyi Divergence \cite{marcobook}]
    Consider $\rho,\sigma \in \cD(\cH)$  and $\alpha \in (0,1) \cup (1,+\infty)$. Then, Petz quantum Re\`nyi divergence between $\rho$ and $\sigma$ is defined as follows,
    \begin{align*}
        D_{\alpha} (\rho \| \sigma) := \begin{cases}
            \frac{1}{\alpha - 1}\log\tr\left[\rho^{\alpha}\sigma^{1 - \alpha}\right] & {\mbox{if }} (\alpha < 1 \cap \rho \not\perp \sigma) \cup (\rho \ll \sigma),\\
            +\infty & \mbox{else}.
        \end{cases}
    \end{align*}
\end{fact}
\begin{fact}[$\alpha$-Quantum Mutual Information \cite{Sharma_Warsi2013,Hayashi2015}]\label{sibsons_quantum_mutual_info}
   Assuming the problem setup mentioned in Definition \ref{def_cq_avht_null} and given a probability distribution $Q$ over $\cS$, we consider the following states:
   \begin{align}
       \rho^{\cB}_{Q,x} &:= \sum_{s \in \cS}Q(s)\rho^{\cB}_{x,s}, \quad \forall x \in \cX,\nn\\
       \rho^{X\cB}_{P,Q} &:= \sum_{s \in \cS}Q(s)\sum_{x \in \cX}P(x)\ketbrasys{x}{X}\otimes\rho^{\cB}_{x,s} = \sum_{x \in \cX}P(x)\ketbrasys{x}{X}\otimes \rho^{\cB}_{Q,x},\label{sibson_QMI_eq1}\\
       \rho^{\cB}_{P,Q} &:= \tr_{X}\left[\rho^{X\cB}_{P,Q}\right] = \sum_{x \in \cX}P(x)\rho^{\cB}_{Q,x}.\nn
   \end{align}
   
   Then, for any $\alpha \in (0,1)$ we define $I_{\alpha}[X;\cB]_{P,Q}$ as follows,
   \begin{align}
       I_{\alpha}[X;\cB]_{P,Q} &:= \min_{\mu^{\cB} \in \cD(\cB)}D_{\alpha}(\rho^{X\cB}_{P,Q} || \rho^{X}_{P} \otimes \mu^{\cB}) \overset{a}{=} -\frac{\alpha}{1 - \alpha}\log\tr\left[\left(\sum_{x \in \cX}P(x)(\rho^{\cB}_{Q,x})^{\alpha}\right)^{\frac{1}{\alpha}}\right],\label{fact_sibson_qmi_expression}
   \end{align}
   where $ \rho^{X}_{P} := \sum_{x \in \cX}P(x) \ketbrasys{x}{\cX}$ and $a$ follows from follows from Sibson's identity.
See \cite[Lemma $3$ in the supplementary material]{Sharma_Warsi2013} or
\cite[Lemma $2$]{Hayashi2015}. Further, $I_{\alpha}[X;\cB]_{P,Q}$ satisfies the following:
\begin{align}
    \lim_{\alpha \to 1} I_{\alpha}[X;\cB]_{P,Q} = I[X;\cB]_{P,Q} &:= D(\rho^{X\cB}_{P,Q} || \rho^{X}_{P} \otimes \rho^{\cB}_{P,Q}).\label{sibsons_qmi_convergence_eq}
\end{align}
\end{fact}

\begin{fact}[Schur Weyl Duality \cite{fulton1991representation,Group2}]\label{fact_schur}
    Given a Young diagram $\lambda \in \Lambda^{n}_{d}$, we denote $\cU_{\lambda}$ as the irreducible representation of $SU(d)$ and $\cV_{\lambda}$ as the irreducible representation of $S_n$ characterized with respect to $\lambda$. Then, for an Hilbert space $\cH$ of dimension $d$, Schur duality can be given as
    \begin{equation}
        \cH^{\otimes n} = \bigoplus_{\lambda \in \Lambda^{n}_{d}} \cU_{\lambda} \otimes \cV_{\lambda},\label{schur_dulaity_eq}
    \end{equation}
    where $\forall \lambda \in \Lambda^{n}_{d} : \abs{U_\lambda} \leq (n+1)^{\frac{r(r-1)}{2}}$.
\end{fact}
\begin{fact}[H\"older's Inequality for operators\cite{khatri2024principlesquantumcommunicationtheory}]\label{trace_ineq}
     For two operators $A,B \in \cL(H)$ we have the following:
 \begin{equation*}
     \tr[AB] \leq \norm{AB}{1} \leq \min\{\norm{A}{\infty}\norm{B}{1},\norm{A}{1}\norm{B}{\infty}\}.
\end{equation*}
     Further, we state a more general statement of H\"older's Inequality for operators. Given two positive semidefinite operator $A,B \in \cL(\cH)$ and two real numbers $p,q \geq 1$ such that $\frac{1}{p}+\frac{1}{q} = 1$, we have the following:
    \begin{equation}
        \abs{\tr[AB]} \leq \left( \tr[A^p]\right)^{\frac{1}{p}}\left( \tr[B^q]\right)^{\frac{1}{q}}.\label{Holder_fact}
    \end{equation}
 \end{fact}
\begin{fact}\label{Holder_corllary}
    Given a positive semi-definite operator $A$ over a Hilbert space $\cH$ and a density matrix $\rho \in \cD(\cH)$, we have,
    \begin{equation}
        \tr[A \rho^t] \leq \left(\tr[A^{\frac{1}{1-t}}]\right)^{1-t},\label{Holder_corllary_eq}
    \end{equation}
    
    for all $t \in [0,1]$.
\end{fact}
\begin{proof}
    Given a $t \in [0,1]$ assume $r = \frac{1}{1-t}$ and $s = \frac{1}{t}$ then from  \eqref{Holder_fact}, we have the following:
    \begin{align*}
        \tr[A \rho^t] &\leq \abs{\tr[A \rho^t]}\\
        &\leq \left(\tr[A^{r}]\right)^{\frac{1}{r}}\left(\tr[\rho^{t s}]\right)^{\frac{1}{s}}\\
        &\leq \left(\tr[A^{\frac{1}{1-t}}]\right)^{1-t} (\tr[\rho])^{t}\\
        &= \left(\tr[A^{\frac{1}{1-t}}]\right)^{1-t}.
    \end{align*}
\end{proof}
 
\begin{fact}[Gao's union bound\cite{Gao_2015,RyanRamgopal2021}]
\label{Gao}
Let $\Pi_1,\Pi_2, \cdots, \Pi_n$ be projectors over $\cH$ and $\rho \in \cD(\cH)$. Then,  
\begin{equation*}
\tr(\Pi_n \cdots \Pi_2\Pi_1 \rho \Pi_1\Pi_2 \cdots \Pi_n) \geq 1-4 \sum_{i=1}^n \tr[\Pi^c_i\rho],
\end{equation*}
where $\Pi^c_i= \mathbb{I} - \Pi_i$.
\end{fact}
\begin{fact}[Gentle Measurement lemma under Ensemble of States\cite{Winter2001}]\label{gent_measurement}
    Let $\{p(x),\rho_x\}$ be an ensemble and let $\bar{\rho}:= \sum_{x}p(x)\rho_x$. If an operator $E$, where $0\preceq E \preceq I$, has high overlap with the expected state $\bar{\rho}$ i.e. $\tr[E\bar{\rho}] \geq 1 - \eps$, where $\eps \in (0,1)$. Then,
    \begin{equation*}
        \mathbb{E}_{X}\left[\norm{\sqrt{E}\rho_{X}\sqrt{E} - \rho_{X}}{1}\right]\leq 2\eps^{1/2}.
    \end{equation*}
\end{fact}

\begin{fact}\cite[Theorem 6]{Ahlswede_Hyp_1980}\label{fact_robust}
    If $\alpha = (\alpha_1,\cdots,\alpha_L)$ and $\beta = (\beta_1,\cdots,\beta_L)$ are two $L$-length sequences where $\forall i \in [L], \alpha_i,\beta_i \in [0, 1]$ such that for some $\eps \in (0,1)$
    \begin{equation*}
        \frac{1}{L}\sum_{i=1}^{L}\alpha_i \geq 1 - \eps, \frac{1}{L}\sum_{i=1}^{L}\beta_i \geq 1 - \eps,
        \end{equation*}
        then
        \begin{equation*}
            \frac{1}{L}\sum_{i=1}^{L}\alpha_i\beta_i \geq 1 - 2\eps.
        \end{equation*}
\end{fact}
\begin{fact}\label{KL-prop}
     For the state $\rho^{\cA\cB} \in \cD(\cA \otimes \cB)$, we define  $I[\cA;\cB]_{\rho}$ as follows,
\begin{equation}
    I[\cA;\cB]_{\rho} = \min_{\tau^{\cB} \in \cD(\cB)} D(\rho^{\cA\cB}||\rho^{\cA}\otimes\tau^{\cB}),\label{KL_div_prop}
\end{equation}
where $ \rho^{\cA} := \tr_{\cB}[\rho^{\cA\cB}]$ and $I[\cA;\cB]_{\rho}$ is minimized when $\tau^{\cB} = \tr_{\cA}[\rho^{\cA\cB}]$.
\end{fact}

\begin{fact}[Hayashi Nagaoka inequality \cite{HN2003}]\label{Hayashi_nagaoka}
    If $A,B \in \cL(\cH)$ are two operators such that $B \succeq 0$ and $0 \preceq A \preceq \bbI^{\cH}$, then we have,
    \begin{equation*}
        \bbI^{\cH} - (A + B)^{-\frac{1}{2}} A (A + B)^{-\frac{1}{2}} \preceq 2(\bbI^{\cH} - A) + 4B.
    \end{equation*}
\end{fact}
\begin{fact}(\cite[Proof of Lemma 2.3]{csiszar-korner-book})\label{type_prob_lb}
    For any type set $T_{P}$ of sequences in $\cX^n$, assuming \eqref{set_Youg_Type_UB} Fact \ref{fact_type_size_ub}, we have the following inequality:
    \begin{equation*}
        P^{n}(T_{P}) \geq (n+1)^{-(\abs{\cX} - 1)} \geq (2n)^{-\abs{\cX}}.
    \end{equation*}
\end{fact}
\begin{fact}\cite{khatri2024principlesquantumcommunicationtheory}\label{trace_norm2}
    Consider a positive operator $0 \preceq \Pi \preceq \bbI$ in $\cH$ and two states density operators $\rho_1,\rho_2 \in \mathcal{D}(\mathcal{H})$. Then,
    \begin{equation*}
        \tr[\Pi(\rho_1 - \rho_2)] \leq \frac{1}{2}\norm{\rho_1 - \rho_2}{1}.
    \end{equation*}
\end{fact}
\begin{fact}[Fano's inequality]\cite{covertom}[Theorem 2.10.1]\label{fano}
    Let $\hat{W}$ and $W$ be two random variables taking values over the set $\cW$ ($0<\abs{\cW}<\infty$) such that $\Pr\{W \neq \hat{W}\} \leq \eps.$ Then,
    \begin{equation*}
        H_{b}(\eps) + \eps\log\abs{\cW} \geq H(W|\hat{W}),
    \end{equation*}
    where $H_{b}(\eps) = -\eps \log \eps - (1 - \eps)\log (1 - \eps)$ and $\eps \in (0,1)$.
\end{fact}

\begin{fact}\cite[Corllary $1$]{Sen2021}\label{corol_sen2}
    Let \(|h\rangle \in \) \(\cH\) be a unit vector. Let \( W_0, W_1, \cdots,  W_l\) be subspaces of \(\cH\). \(\forall j \in [l]\cup 0 \), define $\epsilon_j := \norm{\Pi_{W_j}|h\rangle}{2}^2$, where $\Pi_{W_j}$ is the orthogonal projector on $W_j$ and let $\alpha \in (0,1)$. we define two subspaces $\mathbf{W}$ and $\mathbf{W_\alpha}$ on $\cH\oplus\bigoplus_{j\in[l]}\cH_j$,  as follows,
    \begin{equation*}
        \mathbf{W}:= W_0 + {\bigplus_{j=1}^l \mathcal{E}_{j,\alpha}(W_j)},
    \end{equation*}
    where $\forall j \in [l], \cE_{j,\alpha}:\cH \to \cH\oplus\cH_j$ (where $\forall j \in [l]$, $\cH_j $ is isomorphic to $\cH$ and $\abs{\cH_j} = \abs{\cH}$) is an isometric map defined as follows,
    \begin{equation*}
        \mathcal{E}_{j,\alpha} := \frac{1}{\sqrt{{1+\alpha^2}}}\mathcal{I} + \frac{\alpha}{\sqrt{{1+\alpha^2}}}\mathcal{E}_j,\label{eq2.1}
    \end{equation*}
    where $\cE_{j}$ is an isometry from $\cH$ to $\cH_j$ and $\cI$ is an identity embedding from $\cH$ to $\cH\oplus\{\bigoplus_{j\in[l]}\cH_j\}$.
   Then,\\
 \begin{align}
     \max\{\epsilon_0,(1-\alpha)(\max_{j:j\in[l]}\epsilon_j)\} &\leq \norm{\Pi_{{\mathbf{W_\alpha}}}|h\rangle}{2}^{2} \leq \frac{3l}{\alpha^2} \sum_{j = 0}^{l} \epsilon_{j}\label{corol_sen2_eq1}.\\
 \end{align}
\end{fact}
\begin{fact}[Chernoff Bound \cite{Chernoff1952}]\label{chernoff}
    Given a random variable $X$, $p \in \bbR$ and some $\alpha > 0$, we have
    \begin{equation*}
        \Pr\{X \geq p \} = \Pr\{e^{\alpha X} \geq e^{\alpha p}\} \leq e^{-\alpha p} \bbE\left[e^{\alpha X}\right].
    \end{equation*}
\end{fact}
\section{Construction of independence test}\label{sec:IT}
In this section, we construct independence test. That is,
we prove Theorem \ref{theorem_generalised_independence_test_arbitrary_varying}. 
We first construct it for a finite discrete set $\cS$, 
and later we construct it for a continuous set $\cS$. 
\subsection{Case with finite set $\cS$}
To prove Theorem \ref{theorem_generalised_independence_test_arbitrary_varying} given a discrete finite $\cS$, we show that for any $\widehat{R} > 0$ and $t \in (0,1)$, there exists a projective-measurement based test $\bbT_{n}$ which satisfies the following:
    \begin{align}
     \tr\left[\left( \bbI^{X^n\cB^n} - \bbT_{n}\right)\rho^{X^n\cB^n}_{P,s^n}\right] &\leq 
(n+1)^{\abs{\cS}-1 + \frac{t\abs{\cX}(\abs{\cB}-1)\abs{\cB}}{2} + (\abs{\cB} - 1)\abs{\cX}t}2^{nt\left(\widehat{R} - \min_{Q \in \cP(\cS)}I_{1-t}[X;\cB]_{P,Q}\right)}, \quad \forall s^n \in \cS^n,\label{lemma_achievability_theor_gen_eq1}\\
    \tr\left[\bbT_{n}\left({\rho^{X}_{P}}^{\otimes n} \otimes\sigma^{\cB^n}\right)\right]
&\leq n^{\frac{\abs{\cB}(\abs{\cB}-1)}{2}}\abs{\Lambda^{n}_{\abs{\cB}}} 2^{-n\widehat{R}}, \hspace{140pt}\forall \sigma^{\cB^n} \in \cD(\cH_\cB^{\otimes n}),\label{lemma_achievability_theor_gen_eq2}
    \end{align}
where $\bbT_{n}$ does not depend on the probability distribution $P$. 
Since we can choose $t>0$ such that       
$\widehat{R} - \min_{Q \in \cP(\cS)}I_{1-t}[X;\cB]_{P,Q}<0$,
        the pair of \eqref{lemma_achievability_theor_gen_eq1} 
        and \eqref{lemma_achievability_theor_gen_eq2} implies 
        Theorem \ref{theorem_generalised_independence_test_arbitrary_varying}. 
        
        Here, we employ the same method that has been used in \cite{Hayashi2009}. From Fact \ref{fact_schur} we have,
\begin{equation*}
    \cH_\cB^{\otimes n} = \bigoplus_{\lambda \in \Lambda^{n}_{\abs{\cB}}} \cU_{\lambda} \otimes \cV_{\lambda}.
\end{equation*}

Here, for each $\lambda \in \Lambda^{n}_{d}$, we denote $\Pi^{\cB^n}_{\lambda}$ as the projection onto $ \cU_{\lambda} \otimes \cV_{\lambda}$ and consider the following states:
\begin{align}
    \rho^{\cB^n}_{\lambda} &:= \frac{1}{\abs{\cU_{\lambda} \otimes \cV_{\lambda}}} \Pi^{\cB^n}_{\lambda},
    \nn\\
    \rho^{\cB^n}_{U,n} &:= \sum_{\lambda \in \Lambda^{n}_{\abs{\cB}}}\frac{1}{\abs{\Lambda^{n}_{\abs{\cB}}}}\rho^{\cB^n}_{\lambda}.\label{uniform_n_state}
\end{align}

Note that $\rho^{\cB^n}_{U,n}$ is permutation invariant and it follows from the properties of Schur-Weyl duality (Fact \ref{fact_schur}). Before proceeding to the proof of Theorem \ref{theorem_generalised_independence_test_arbitrary_varying}, we require Lemma \ref{lemma_perm_Inv_universal} mentioned below.
\begin{lemma}\cite[Equation (15)]{Hayashi2009}\label{lemma_perm_Inv_universal}
For any permutation invariant state $\sigma^{\cB^n} \in \cD(\cH_\cB^{\otimes n})$, we have
    \begin{equation*}
        \sigma^{\cB^n} \preceq n^{\frac{\abs{\cB}(\abs{\cB}-1)}{2}}\abs{\Lambda^{n}_{\abs{\cB}}}\rho^{\cB^n}_{U,n}.
    \end{equation*}
\end{lemma}

We consider an $\bar{x}^n := (\bar{x}_1,\bar{x}_2,\cdots,\bar{x}_n)\in \cX^n$, which has a form $\bar{x}^n := \left(\underbrace{1,\cdots,1}_{m_1},\underbrace{2,\cdots,2}_{m_2},\cdots,\underbrace{\abs{\cX},\cdots,\abs{\cX}}_{m_{\abs{\cX}}}\right)$, where $\forall i \in [\abs{\cX}], m_i \geq 0$ and $\sum_{i \in [\abs{\cX}]}m_i = n$. Given $\bar{x}^n$, consider the following state:
\begin{equation}
    \widehat{\rho}^{\cB^n}_{\bar{x}^n} := \rho^{\cB^{m_1}}_{U,m_1} \otimes \rho^{\cB^{m_2}}_{U,m_2} \otimes \cdots \otimes \rho^{\cB^{m_{\abs{\cX}}}}
    _{U,m_{\abs{\cX}}}\label{rhohatxbarn}.
\end{equation}

For any general $x^n \in \cX^n$, which is a permutation of $\bar{x}^n$ i.e. $x^n  = \pi(\bar{x}^n) = (\bar{x}_{\pi^{-1}(1)},\bar{x}_{\pi^{-1}(2)},\cdots,\bar{x}_{\pi^{-1}(n)})$, for some $\pi \in S_n$, we define $\hat{\rho}^{\cB^n}_{x^n}$ as follows
\begin{equation}
    \widehat{\rho}^{\cB^n}_{x^n} := V^{\cB^n}(\pi) \left(\widehat{\rho}^{\cB^n}_{\bar{x}^n}\right){V^{\cB^n}}^{\dagger}(\pi),\label{rhohatxn}
\end{equation}
where $V^{\cB^n}(\pi)$ is a permutation operator (see Definition \ref{fact_perm_op} for formal definition of $V^{\cB^n}(\pi)$ ). We observe the Lemma \ref{lemma_commutativity} mentioned below, which will be required in the proof of Theorem \ref{theorem_generalised_independence_test_arbitrary_varying}.

\begin{lemma}\label{lemma_commutativity}
    For any $x^n \in \cX^n$, we have $\rho^{\cB^n}_{U,n}\widehat{\rho}^{\cB^n}_{x^n} = \widehat{\rho}^{\cB^n}_{x^n}\rho^{\cB^n}_{U,n}$, where $\rho^{\cB^n}_{U,n}$ and $\widehat{\rho}^{\cB^n}_{x^n}$ are defined in \cref{uniform_n_state,rhohatxn}.
\end{lemma}
\begin{proof}
    See Appendix \ref{proof_lemma_commutativity} for the proof.
\end{proof}
 Given a probability distribution $P \in \cP(\cX)$, we define the following state:
\begin{equation*}
    \widehat{\rho}^{X^n\cB^n}_{U,P} := \sum_{x^n \in \cX^n}P^{n}(x^n) \ketbrasys{x^n}{X^n} \otimes \widehat{\rho}^{\cB^n}_{x^n}.
\end{equation*}

For $\widehat{R}>0$, we define a projective measurement-based test $\bbT_{n}$ as follows,
\begin{equation}
    \bbT_{n} := \sum_{x^n \in \cX^n} \ketbrasys{x^n}{X^n} \otimes \left\{\widehat{\rho}^{\cB^n}_{x^n} \succeq 2^{n\widehat{R}}\rho^{\cB^n}_{U,n}\right\} \triangleq \left\{\widehat{\rho}^{X^n\cB^n}_{U,P} \succeq 2^{n\widehat{R}} \left({\rho^{X}_{P}}^{\otimes n} \otimes \rho^{\cB^n}_{U,n}\right)\right\}\label{universal_test},
\end{equation}
where ${\rho^{X}_{P}}^{\otimes n} := \sum_{x^n \in \cX^n}P^{n}(x^n)\ketbrasys{x^n}{X^n}$. It is important to note that in \eqref{universal_test}, the first expression of $\bbT_{n}$ is independent of the probability distribution $P$. Observe that $\widehat{\rho}^{X^n\cB^n}_{U,P}$ and ${\rho^{X}_{P}}^{\otimes n} \otimes \rho^{\cB^n}_{U,n}$ both are permutation invariant. Thus, the test $\bbT_{n}$ is also permutation invariant which means $\forall s^n \in \cS^n, \pi \in S_n$ 
\begin{equation}
    \tr\left[\left( \bbI^{X^n\cB^n} - \bbT_{n}\right)\rho^{X^n\cB^n}_{P,s^n}\right] = \tr\left[\left( \bbI^{X^n\cB^n} - \bbT_{n}\right)\rho_{P,\pi(s^n)}^{X^n\cB^n}\right]\label{universal_test_perm_inv}.
\end{equation}

Given an element $s^n \in \cS^n$, we denote its corresponding empirical distribution (the empirical distribution which $s^n$ follows) by $Q_{s^n}$ and we denote the set $T_{Q_{s^n}}$ as set of all sequences in $\cS^n$ which has the same type as $s^n$ i.e. $T_{Q_{s^n}} = \left\{\pi(s^n)\right\}_{\pi \in S_n}$. For a $\bar{s}^n \in T_{Q_{s^n}}$, from \cite[Equation ($11.15$)]{covertom} we have 
\begin{equation}
    Q^{n}_{s^n}(\bar{s}^n) = 2^{-n H(Q_{s^n})},\label{correct_type_elem_prob}
\end{equation}
where $H(Q_{s^n})$ is the Shannon entropy with respect to the distribution $Q_{s^n}$ and we have the following,
\begin{align}
     Q^{n}_{s^n}(T_{Q_{s^n}}) &= \sum_{t^{n} \in T_{Q_{s^n}}}Q_{s^n}(t^n)
     \overset{a}{=} \abs{T_{Q_{s^n}}}Q^{n}_{s^n}(\Tilde{s}^n),\label{type_prob}
\end{align}
where $\Tilde{s}^n$ is any element in $T_{Q_{s^n}}$ and $a$ follows from \eqref{correct_type_elem_prob}. Thus, for any $\widehat{s}^n \in T_{Q_{s^n}}$, we have
\begin{align}
    \frac{1}{\abs{T_{Q_{s^n}}}} \overset{a}{=} \frac{1}{ Q^{n}_{s^n}(T_{Q_{s^n}})} Q^{n}_{s^n}(\widehat{s}^n)
    \overset{b}{\leq} (n+1)^{\abs{\cS}-1}Q^{n}_{s^n}(\widehat{s}^n),\label{inv_type_set_size_ub}
\end{align}
where $a$ follows from \eqref{type_prob} and $b$ follows from Fact \ref{type_prob_lb}. Thus, for any $s^n \in \cS^n$, we upper-bound the quantity $ \tr[\left( \bbI^{X^n\cB^n} - \bbT_{n}\right)\rho^{X^n\cB^n}_{P,s^n}]$ as follows,

\begin{align}
     &\tr\left[\left( \bbI^{X^n\cB^n} - \bbT_{n}\right)\rho^{X^n\cB^n}_{P,s^n}\right] \overset{a}{=}  \tr\left[\left( \bbI^{X^n\cB^n} - \bbT_{n}\right)\frac{1}{\abs{T_{Q_{s^n}}}} \sum_{\pi \in S_n}\rho_{P,\pi(s^n)}^{X^n\cB^n}\right]\nn\\
     &= \tr\left[\left( \bbI^{X^n\cB^n} - \bbT_{n}\right)\frac{1}{\abs{T_{Q_{s^n}}}}\sum_{\widehat{s}^{n} \in T_{Q_{s^n}}} \rho^{X^n\cB^n}_{P,\widehat{s}^{n}}\right]\nn\\
     &\overset{b}{\leq} (n+1)^{\abs{\cS}-1}\tr\left[\left( \bbI^{X^n\cB^n} - \bbT_{n}\right)\sum_{\widehat{s}^{n} \in T_{Q_{s^n}}} Q^{n}_{s^n}(\widehat{s}^n)\rho^{X^n\cB^n}_{P,\widehat{s}^{n}}\right]\nn\\
     &\leq (n+1)^{\abs{\cS}-1}\tr\left[\left( \bbI^{X^n\cB^n} - \bbT_{n}\right)\sum_{\widehat{s}^{n} \in \cS^n} Q^{n}_{s^n}(\widehat{s}^n)\rho^{X^n\cB^n}_{P,\widehat{s}^{n}}\right]\nn\\
     &\overset{c}{=} (n+1)^{\abs{\cS}-1}\tr\left[\left( \bbI^{X^n\cB^n} - \bbT_{n}\right)\left(\rho^{{X\cB}}_{P,Q_{s^n}}\right)^{\otimes n}
     \right]\nn\\
     &= (n+1)^{\abs{\cS}-1}\tr\left[\left( \bbI^{X^n\cB^n} - \bbT_{n}\right)\sum_{x^n \in \cX^n}P^{n}(x^n)\ketbrasys{x^n}{X^n} \otimes \left(\bigotimes_{i=1}^{n}\rho^{\cB}_{Q_{s^n},x_i}\right)
     \right]\nn\\
     &\overset{d}{=}(n+1)^{\abs{\cS}-1}\sum_{x^n \in \cX^n}P^{n}(x^n)\tr\left[ \left(\bigotimes_{i=1}^{n}\rho^{\cB}_{Q_{s^n},x_i}\right)\left\{\widehat{\rho}^{\cB^n}_{x^n} \prec 2^{n\widehat{R}}\rho^{\cB^n}_{U,n}\right\}
     \right]\nn\\
     &\overset{e}{\leq}(n+1)^{\abs{\cS}-1}2^{nt \widehat{R} }\sum_{x^n \in \cX^n}P^{n}(x^n)\tr\left[ \left(\bigotimes_{i=1}^{n}\rho^{\cB}_{Q_{s^n},x_i}\right)(\widehat{\rho}^{\cB^n}_{x^n})^{-t}(\rho^{\cB^n}_{U,n})^{t}
     \right]\nn\\
     &\overset{f}{\leq}(n+1)^{\abs{\cS}-1}2^{nt \widehat{R} }n^{\frac{t\abs{\cX}(\abs{\cB}-1)\abs{\cB}}{2}}\abs{\Lambda^{n}_{\abs\cB}}^{t \abs{\cX}}\sum_{x^n \in \cX^n}P^{n}(x^n)\tr\left[ \left(\bigotimes_{i=1}^{n}\rho^{\cB}_{Q_{s^n},x_i}\right)^{1-t}(\rho^{\cB^n}_{U,n})^{t}
     \right]\nn\\
     &\leq (n+1)^{\abs{\cS}-1 + \frac{t\abs{\cX}(\abs{\cB}-1)\abs{\cB}}{2}}\abs{\Lambda^{n}_{\abs\cB}}^{t \abs{\cX}}2^{nt \widehat{R} } \tr\left[ \left(\sum_{x^n \in \cX^n}P^{n}(x^n)\left(\bigotimes_{i=1}^{n}\rho^{\cB}_{Q_{s^n},x_i}\right)^{1-t}\right)(\rho^{\cB^n}_{U,n})^{t}
     \right]\nn\\
     &= (n+1)^{\abs{\cS}-1 + \frac{t\abs{\cX}(\abs{\cB}-1)\abs{\cB}}{2}}\abs{\Lambda^{n}_{\abs\cB}}^{t \abs{\cX}}2^{nt \widehat{R} } \tr\left[ \left(\sum_{x \in \cX}P(x)\left(\rho^{\cB}_{Q_{s^n},x}\right)^{1-t}\right)^{\otimes n}(\rho^{\cB^n}_{U,n})^{t}
     \right]\nn\\
     &\overset{g}{\leq} (n+1)^{\abs{\cS}-1 + \frac{t\abs{\cX}(\abs{\cB}-1)\abs{\cB}}{2}}\abs{\Lambda^{n}_{\abs\cB}}^{t \abs{\cX}}2^{nt \widehat{R} } \left(\tr\left[ \left(\left(\sum_{x \in \cX}P(x)\left(\rho^{\cB}_{Q_{s^n},x}\right)^{1-t}\right)^{\otimes n}\right)^{\frac{1}{1 - t}}
     \right]\right)^{1 -t}\nn\\
     &= (n+1)^{\abs{\cS}-1 + \frac{t\abs{\cX}(\abs{\cB}-1)\abs{\cB}}{2}}\abs{\Lambda^{n}_{\abs\cB}}^{t \abs{\cX}}2^{nt \widehat{R} } \left(\tr\left[ \left(\sum_{x \in \cX}P(x)\left(\rho^{\cB}_{Q_{s^n},x}\right)^{1-t}\right)^{\frac{1}{1 - t}}
     \right]\right)^{n(1 -t)}\nn\\
     &\overset{h}{=} (n+1)^{\abs{\cS}-1 + \frac{t\abs{\cX}(\abs{\cB}-1)\abs{\cB}}{2}}\abs{\Lambda^{n}_{\abs\cB}}^{t \abs{\cX}} 2^{nt\left(\widehat{R} - I_{1-t}[X;\cB]_{P,Q_{s^n}}\right)}\nn\\
     &\overset{i}{\leq}(n+1)^{\abs{\cS}-1 + \frac{t\abs{\cX}(\abs{\cB}-1)\abs{\cB}}{2} + (\abs{\cB} - 1)\abs{\cX}t}2^{nt\left(\widehat{R} - \min_{Q \in \cP(\cS)}I_{1-t}[X;\cB]_{P,Q}\right)}\label{BNB1},
\end{align}
where $a$ follows from the fact that $\bbT_{n}$ is permutation invariant, $b$ follows from \eqref{inv_type_set_size_ub}, $c$ follows from \eqref{sibson_QMI_eq1} of Fact \ref{sibsons_quantum_mutual_info}, $d$ follows from the definition of $\bbT_{n}$ in \eqref{universal_test}, $e$ follows from the fact that since $ \widehat{\rho}^{\cB^n}_{x^n}$ commutes with $\rho^{\cB^n}_{U,n}$ (which follows from Lemma \ref{lemma_commutativity}), we have $\left\{\widehat{\rho}^{\cB^n}_{x^n} - 2^{n\widehat{R}}\rho^{\cB^n}_{U,n} \prec 0\right\} \leq 2^{nt\widehat{R}}\left(\widehat{\rho}^{\cB^n}_{x^n}\right)^{-t}\left(\rho^{\cB^n}_{U,n}\right)^{t}$ for all $t \in (0,1)$, $f$ follows from Lemma \ref{claim_f} mentioned below this paragraph, $g$ follows from \eqref{Holder_corllary_eq} of Fact \ref{Holder_corllary}, $h$ follows from \eqref{fact_sibson_qmi_expression} of Fact \ref{sibsons_quantum_mutual_info} and $i$ follows from Fact \ref{fact_type_size_ub}.
Therefore, we obtain \eqref{lemma_achievability_theor_gen_eq1}.

\begin{lemma}\label{claim_f}
For all $t \in (0,1)$, for any $x^n := (x_1,\cdots,x_n) \in \cX^n$, we have the following operator inequality:
\begin{equation*}
    \left(\bigotimes_{i=1}^{n}\rho^{\cB}_{Q_{s^n},x_i}\right)^{t}(\widehat{\rho}^{\cB^n}_{x^n})^{-t} \preceq n^{\frac{t\abs{\cX}(\abs{\cB}-1)\abs{\cB}}{2}}\abs{\Lambda^{n}_{\abs\cB}}^{t \abs{\cX}} \bbI^{\cB^n},
\end{equation*}
where the above notations follow from the above discussions.
\end{lemma}
\begin{proof}
    See Appendix \ref{proof_claim_f} for the proof.
\end{proof}

For any state $\sigma^{\cB^n} \in \cD(\cH_\cB^{\otimes n})$, we now upper-bound the quantity $\tr\left[\bbT_{n}\left({\rho^{X}_{P}}^{\otimes n} \otimes\sigma^{\cB^n}\right)\right]$ as follows,
\begin{align}
    \tr\left[\bbT_{n}\left({\rho^{X}_{P}}^{\otimes n} \otimes\sigma^{\cB^n}\right)\right] &\overset{a}{=} \tr\left[\bbT_{n}\left({\rho^{X}_{P}}^{\otimes n} \otimes\left(\frac{1}{n!}\sum_{\pi \in S_n}V^{\cB^n}(\pi) \sigma^{\cB^n}{V^{\cB^n}}^{\dagger}(\pi)\right)\right)\right] \nn\\
    &\overset{b}{\leq} n^{\frac{\abs{\cB}(\abs{\cB}-1)}{2}}\abs{\Lambda^{n}_{\abs{\cB}}}\tr\left[\bbT_{n}\left({\rho^{X}_{P}}^{\otimes n} \otimes\rho^{\cB^n}_{U} \right)\right] \nn\\
    &\overset{c}{\leq} n^{\frac{\abs{\cB}(\abs{\cB}-1)}{2}}\abs{\Lambda^{n}_{\abs{\cB}}} 2^{-n\widehat{R}} \tr\left[\bbT_{n}\widehat{\rho}^{X^n\cB^n}_{U,P}\right] \nn\\
    &\leq n^{\frac{\abs{\cB}(\abs{\cB}-1)}{2}}\abs{\Lambda^{n}_{\abs{\cB}}} 2^{-n\widehat{R}},\label{lemma_achievability_theor_gen_eq2_proof}
\end{align}
where $a$ follows from the fact that $\bbT_{n}$ is permutation invariant. Thus, for any $\pi \in S_n$, the following holds,
$$\tr\left[\bbT_{n}\left({\rho^{X}_{P}}^{\otimes n} \otimes\sigma^{\cB^n}\right)\right] {=} \tr\left[\bbT_{n}\left({\rho^{X}_{P}}^{\otimes n} \otimes\left(V^{\cB^n}(\pi) \sigma^{\cB^n}{V^{\cB^n}}^{\dagger}(\pi)\right)\right)\right].
$$
In addition, step $b$ follows from Lemma \ref{lemma_perm_Inv_universal} and the fact that $\frac{1}{n!}\sum_{\pi \in S_n}V^{\cB^n}(\pi)\sigma^{\cB^n}{V^{\cB^n}}^{\dagger}(\pi)$ is permutation invariant and $c$ follows from the definition of $\bbT_{n}$, given in \eqref{universal_test}. 
Therefore, we obtain \eqref{lemma_achievability_theor_gen_eq2}.

\subsection{Analysis on Assumption \ref{NM9}}\label{S4-AB}
We consider when Assumption \ref{NM9} holds.
Clearly, when $\cS$ is a finite set, Assumption \ref{NM9} holds.
Also, we have the following lemma.
\begin{lemma}\label{LJGF}
When the set $D_x:=\left\{\rho_{x,s}^{\cB}\right\}_{\substack{s \in \cS}}$
is a compact subset of $\cD(\cH_{\cB}) $ 
for any $x \in \cX$
and any element $\rho_{x,s}^{\cB}$ is a full rank state.
Then, the set $\left\{\rho_{x,s}^{\cB}\right\}_{\substack{x \in \cX,s \in \cS}}$
satisfies Assumption \ref{NM9}.
\end{lemma}

\begin{proof}
We fix any real number $\varepsilon >0$.
For any element $x \in \cX,s \in \cS$,
the subset $U_{x,s}:=
\{ \rho^{\cB}_{x,s'}: s' \in \cS,
\rho^{\cB}_{x,s'}\preceq (1+\varepsilon) \rho^{\cB}_{x,s} \}
\subset \cD(\cH_{\cB})$
is an open subset of $\cD(\cH_{\cB})$
because the element $\rho_{x,s}^{\cB}$ is a full rank state.
Since $ \cup_{s \in \cS} U_{x,s}$ equals $D_x$ and is a compact set,
there is a finite subset $\cS_{x,\varepsilon} \subset \cS$ such that
 $ \cup_{s \in \cS_{x,\varepsilon}} U_{x,s}=D_x$.
 Since the subset 
 $\cup_{x \in \cX} \cS_{x,\varepsilon}$ is a finite set,
 the subset 
 $\cup_{x \in \cX} \cS_{x,\varepsilon}$ satisfies the 
desired condition.
 \end{proof}
 
However, when  the set $D_x$ contains a pure state even when
$\cS$ is a finite set,
the assumption of the above lemma does not hold.
Therefore, 
the assumption of Lemma \ref{LJGF} is stronger than 
Assumption \ref{NM9}.

\subsection{Case with general set $\cS$}\label{S4-B}
We assume Assumption \ref{NM9}.
We choose an arbitrary real number 
$\widehat{R} $ to satisfy
\begin{align}
\widehat{R} < \inf_{Q \in \cP(\cS)}I[X;\cB]_{P,Q}. \label{BN85}
\end{align} 
We choose $t>0,\varepsilon>0$ such that
\begin{align}
t\left(\widehat{R} - \inf_{Q}I_{1-t}[X;\cB]_{P,Q}\right) <  - \log(1+\eps)\label{BNC7}.
\end{align} 

From Assumption \ref{NM9},
for any $s^n \in \cS^n$, there exists a $\bar{s}^n \in \cS^{n}_{\eps}$ 
such that any $x^n \in \cX^n$ satisfies 
\begin{equation}
\rho^{\cB^n}_{x^n,s^n} \preceq (1 + \eps)^{n}\rho^{\cB^n}_{x^n,\bar{s}^n}.\label{eps_covering_density_operators_eq}
\end{equation}
Applying Theorem \ref{theorem_generalised_independence_test_arbitrary_varying} on the set $\cS_{\eps}$,  for any $s^n \in \cS^n$, we upper-bound the quantity $ \tr[\left( \bbI^{X^n\cB^n} - \bbT_{n}\right)\rho^{X^n\cB^n}_{P,s^n}]$ as follows
\begin{align}
    \tr[\left( \bbI^{X^n\cB^n} - \bbT_{n}\right)\rho^{X^n\cB^n}_{P,s^n}] &\overset{a}{\leq} (1 + \eps)^{n}\tr[\left( \bbI^{X^n\cB^n} - \bbT_{n}\right)\rho^{X^n\cB^n}_{P,\bar{s}^n}]\nn\\
    &\overset{b}{\leq} (n+1)^{\abs{\cS_{\eps}}-1 + \frac{t\abs{\cX}(\abs{\cB}-1)\abs{\cB}}{2} + (\abs{\cB} - 1)\abs{\cX}t}\left((1 + \eps)2^{t\left(\widehat{R} - I_{1-t}[X;\cB]_{P,Q_{\bar{s}^n}}\right)}\right)^{n}\nn\\
    &= (n+1)^{\abs{\cS_{\eps}}-1 + \frac{t\abs{\cX}(\abs{\cB}-1)\abs{\cB}}{2} + (\abs{\cB} - 1)\abs{\cX}t}\left(2^{\left(t\left(\widehat{R} - I_{1-t}[X;\cB]_{P,Q_{\bar{s}^n}}\right) + \log(1 + \eps)\right)}\right)^{n},\label{continuous_eps_covering_density_operators_eq}
\end{align}
where $a$ follows from \eqref{eps_covering_density_operators_eq}, 
$b$ follows from \eqref{lemma_achievability_theor_gen_eq1}. 
From \eqref{lemma_achievability_theor_gen_eq2}, 
for any $\sigma^{\cB^n} \in \cD(\cH_\cB^{\otimes n})$, we directly have
\begin{equation}
    \tr\left[\bbT_{n}\left({\rho^{X}_{P}}^{\otimes n} \otimes\sigma^{\cB^n}\right)\right] \leq n^{\frac{\abs{\cB}(\abs{\cB}-1)}{2}}\abs{\Lambda^{n}_{\abs{\cB}}} 2^{-n\widehat{R}}.
\label{NSL9}
\end{equation}
Eq. \eqref{BNC7} guarantees that \eqref{continuous_eps_covering_density_operators_eq} goes to zero. Since
$\widehat{R} $ is 
an arbitrary real number to satisfy \eqref{BN85},
\eqref{NSL9} implies 
\begin{equation}
     \lim_{n \to \infty}-\frac{1}{n}\log\beta_{n}(\eps) \geq \inf_{Q\in \cP(\cS)}I[X;\cB]_{P,Q}. \label{continuous_achievability_lb}
\end{equation}

This gives a proof for Theorem \ref{theorem_generalised_independence_test_arbitrary_varying}, 
when $\cS$ is continuous. 
As an application of Theorem \ref{theorem_generalised_independence_test_arbitrary_varying}, 
we study the problem of reliable communication over CQ-AVC in Section \ref{sec:section_CQAVC}. 

\section{Reliable Communication over Classical Quantum Arbitrarily Varying Channel (CQ-AVC)}\label{sec:section_CQAVC}
\subsection{Formulation}
In this section, we explain how a code for 
classical-quantum arbitrarily varying point-to-point channel (CQ-AVC)
can be constructed by using the independence test given in Theorem \ref{theorem_generalised_independence_test_arbitrary_varying}.
For this aim, we formulate a CQ-AVC as follows.
\begin{definition} (channel) We model a CQ-AVC  between parties
Alice and Bob as a map
$$\cN^{X^n \to B^n}_{s^n} : x^n \to \rho^{B^n}_{x^n,s^n},$$
where $x^n \in \cX^n$ and $s^n \in \cS^n$, (where $\abs{\cX}< \infty$ and $\cS$ is a non-empty set). For each $x^n := (x_1\cdots,x_n) \in \cX^n$ and $s^n := (s_1\cdots,s_n) \in \cS^n$, 
we define the state $\rho^{B^n}_{x^n,s^n}$ to be $\bigotimes_{i=1}^{n}\rho^{B}_{x_i,s_i}$.
Further, both Alice (the sender) and Bob (the receiver) have no information about the sequence $s^n.$
\end{definition}

We aim to use this channel and enable Alice to transmit message $m \in [2^{nR}]$ such that Bob can recover $m$ with high probability for every $s^n$.
\begin{definition}(Deterministic code)
An $(n,2^{nR})$-code $\cC$ for communication over a CQ-AVC  consists of 
\begin{itemize}
\item an encoding function $\cE^{(n)}: \cM_n \to \cX^n$, where $\cM_n := [2^{nR}]$,
\item a decoding POVM $\{ \cD^{(n)}_{m}: m \in \cM_n\}$.
  \end{itemize}
An $(n,2^{nR})$-code $\cC$ for communication over a CQ-AVC  is called
an $(n,2^{nR},\beta)$-code 
when for any $s^n \in \cS^n$, 
the average probability $\bar{e}(\cC,s^n)$ of error satisfies 
\begin{equation*}
      \bar{e}(\cC,s^n):= \frac{1}{2^{nR}}\sum_{m=1}^{2^{nR}}e(m,\cC,s^n) < \beta,
  \end{equation*}
  where $e(m,\cC,s^n):= 1 - \tr[D^{(n)}_{m} \cN^{X^n \to B^n}_{s^n}(\cE^{(n)}(m))]$.
  \end{definition}
\begin{definition}\label{defcqavc_det}
      A number $R>0$ is called an \textit{achievable deterministic code rate} for a given CQ-AVC , if for any $\beta>0,\delta>0$ and sufficiently large $n$, there exists a $(2^{nR},n,\beta)$ deterministic code $\cC$ such that
      \begin{equation*}
          \begin{tabular}{c  c}
            $\frac{1}{n}\log \abs{\cM_n} > R - \delta,$\hspace{10pt} &  $\sup_{s^n \in \cS^n}\bar{e}(\cC,s^n) < \beta$.
      \end{tabular}
      \end{equation*}  
  \end{definition}
  The supremum over all such achievable rates is called the \textit{deterministic code capacity} $C_d$ of a CQ-AVC. From \cite[Lemma 2]{Ahlswede07} we know that if $C_d > 0$ then 
  \begin{equation}
     C_{d} = \max_{P}\inf_{Q}I[X;\cB]_{P,Q},\label{Cd_ahlswede}
  \end{equation}
  where $I[X;\cB]_{P,Q}$ is defined in \eqref{sibsons_qmi_convergence_eq} of Fact \ref{sibsons_quantum_mutual_info}. 
  \begin{definition}(Random code)
An $(n,2^{nR})$ random code $\cC$ for communication over a CQ-AVC  consists of 
\begin{itemize}
\item a random encoding function $\cE^{(n),\gamma}: \cM_n \to \cX^n$, parameterized by a random variable $\gamma$, where $\cM_n := [2^{nR}]$,
\item a decoding POVM $\{ \cD^{(n),\gamma}_{m}: m \in \cM_n\}$, parameterized by a random variable $\gamma$.
  \end{itemize}
In the above, $\gamma \sim G$ is a random variable, shared between Alice and Bob and takes values over a finite set. 

An $(n,2^{nR})$ random code $\cC$ for communication over a CQ-AVC  is called
an $(n,2^{nR},\beta)$ random code 
when for any $s^n \in \cS^n$, 
the average probability $\bbE_{\cC}\left[\bar{e}(\cC,s^n)\right]$ of error under the expectation of choices of random code (choices of $\gamma$) satisfies 
\begin{equation*}
      \bbE_{\cC}\left[\bar{e}(\cC,s^n)\right]:= \frac{1}{2^{nR}}\sum_{m=1}^{2^{nR}}\bbE_{\cC}\left[e(m,\cC,s^n)\right] < \beta,
  \end{equation*}
  where $\bbE_{\cC}\left[e(m,\cC,s^n)\right]:= 1 - \bbE_{\gamma \sim G}\left[\tr[D^{(n),\gamma}_{m} \cN^{X^n \to B^n}_{s^n}(\cE^{(n),\gamma}(m))]\right]$.
  \end{definition}
  
When the sender and the receiver share the randomness $\gamma$, 
the random code $\cC$
can be realized.
  
  \begin{definition}\label{defcqavc_rand}
      A number $R>0$ is called an \textit{achievable random code rate} for a given CQ-AVC , if for any $\beta>0,\delta>0$ and sufficiently large $n$, there exists a $(2^{nR},n,\beta)$ random code $\cC$ such that
      \begin{equation*}
          \begin{tabular}{c  c}
            $\frac{1}{n}\log \abs{\cM_n} > R - \delta,$\hspace{10pt} &  $\sup_{s^n \in \cS^n}\bbE_{\cC}\left[\bar{e}(\cC,s^n)\right] < \beta$.
      \end{tabular}
      \end{equation*}  
  \end{definition}
  The supremum over all such achievable rates is called the \textit{random code capacity} $C_{r}$ of a CQ-AVC. 
    \begin{theorem}(\textbf{Random Code Capacity of CQ-AVC})\label{lemma_rand_cap_cqavc}
Assume that a collection of density operators $\left\{\rho_{x,s}^{\cB}\right\}_{\substack{ x \in \cX, s \in \cS}}$ satisfies Assumption \ref{NM9}.
      The random code capacity $C_{r}$ of a CQ-AVC is equal to $\max_{P}\inf_{Q}I[X;\cB]_{P,Q}$ (where $I[X;\cB]_{P,Q}$ is defined in \eqref{sibsons_qmi_convergence_eq} of Fact \ref{sibsons_quantum_mutual_info}).
  \end{theorem} 

When $\cS$ is a finite set,
the random code capacity $C_{r}$ is known in \cite[Lemma 1]{Ahlswede07}.

For derandomization of our code, see Section \ref{sec:determinsitic_cap_section}.

\begin{proof}
The original proof of achievability for \cite[Lemma 1]{Ahlswede07} required an approximation lemma (see \cite[equation $10$]{Ahlswede07}), min-max theorem and typicality-based arguments. Our proof follows straightforwardly only from Theorem \ref{theorem_generalised_independence_test_arbitrary_varying}. Further, unlike \cite{Ahlswede07} our decoder is universal and it does not depend on the statistics of the channel. We now proceed with the proof.

(\textbf{Achievability}) We here show that for every rate $R < \min_{Q}I[X;\cB]_{P,Q}$ there exists a $(n,2^{nR},\eps_n)$ random code $\cC$ for which the average error probability, averaged over the choice of code $\cC$ satisfies the following:
\begin{align}
\varepsilon_n:=
\sup_{s^n \in \cS^n}\bbE_{\cC}\left[\bar{e}(\cC,s^n)\right]\to 
0.
\label{NBO}
\end{align}
\subsubsection{Randomized Encoder Construction}
Since
$R < \inf_{Q \in \cP(\cS)} I[X;\cB]_{P,Q}$,
we choose $\widehat{R}$ as
\begin{align}
R <\widehat{R}< \inf_{Q \in \cP(\cS)} I[X;\cB]_{P,Q}.
\label{M77}
\end{align}
Hence, we can choose 
$t>0$ and $\epsilon>0$ such that
\begin{align}
t\Big(\widehat{R} - \inf_{Q \in \cP(\cS)} I_{1-t}[X;\cB]_{P,Q}\Big)
< -\log(1+\epsilon).
\label{M78}
\end{align}
Alice randomly generates $2^{nR}$ independent input sequences of random 
variables 
$\{X^n(m) \in \cX^n : m \in \cM_n\}$ according to $P^{n}$, where
$P \in \cP(\cX)$.
That is, for any element $\forall m \in \cM_n$,
the random variable $X^n(m)$ obeys the distribution 
$\Pr\bigl\{X^n(m)\bigr\} := \prod_{i=1}^{n}P(X_i(m))$, where $X^n(m):=\left(X_1(m),\cdots,X_n(m)\right)$. 
Then, we randomly choose the encoder $\cE^{(n)}$ subject to the following distribution;
\begin{equation*}
    \Pr( \cE^{(n)}(m)=X^n(m) \hbox{ for } m=1, \ldots,  2^{nR} ) := \prod_{m=1}^{2^{nR}}\Pr\left\{X^n(m)\right\}.
\end{equation*}
If Alice has to send a message $m \in \cM_n$, she encodes the message to the input sequence $X^n(m)$ and sends a state $\ket{X^n(m)}^{X^n}$ over the channel.
\subsubsection{Decoding Strategy (Universal Decoder)}
Upon receiving a quantum state $\rho_{\substack{X^n, s^n}}^{\cB^n} \in \cD(\cH_{\cB}^{\otimes n})$, Bob tries to guess the message sent by the sender with the help of the POVM $\left\{\Lambda_{m}\right\}_{m \in \cM_n}$, where for each $m \in \cM_n$,
\begin{equation*}
    \Lambda_{m}:= \left(\sum_{j \in \cM_n} T_{n,X^n(j)}\right)^{-\frac{1}{2}}T_{n,X^n(m)}\left(\sum_{ j \in \cM_n}T_{n,X^n(j)}\right)^{-\frac{1}{2}},
\end{equation*}
where for any $x^n \in \cX^n$, $T_{n,x^n} := \left\{\widehat{\rho}^{\cB^n}_{x^n} \succeq 2^{nR}\rho^{\cB^n}_{U,n}\right\}$ (mentioned in \eqref{universal_test}). Observe that the above POVM is universal as $T_{n,x^n}$ is independent of the channel. If Bob gets the measurement outcome corresponding to $\Lambda_{\hat{m}}$, Bob then declares the message to be $\hat{m}$.
In the following, we denote the pair $(\cE^{(n)},
\left\{\Lambda_{m}\right\}_{m \in \cM_n})$ by $\cC$.

\subsubsection{Error Analysis}
We now calculate the average error probability for the random codes $\cC$ for any particular $s^n \in \cS^n$:
\begin{align}
    &\hspace{10pt}\bbE_{\cC}\left[\bar{e}(\cC,s^n)\right] = \frac{1}{2^{nR}}\sum_{m=1}^{2^{nR}}\bbE_{\cC}\left[e(m,\cC,s^n)\right] = \frac{1}{2^{nR}}\sum_{m=1}^{2^{nR}}\bbE_{\cC}\left[ \tr\left[\left(\bbI^{\cB^n}-\Lambda_{m}\right)\left(\rho_{\substack{X^n(m), s^n}}^{\cB^n}\right)\right]\right]\notag\\
    &\overset{a}{\leq} \frac{1}{2^{nR}}\sum_{m=1}^{2^{nR}}\Biggl(2\bbE_{\cC}\left[\tr\left[\left(\bbI^{\cB^n}-T_{n,X^n(m)}\right)\rho_{\substack{X^n(m), s^n}}^{\cB'}\right]\right] + 4\sum_{ k \in \cM_n\setminus{m}}\bbE_{\cC}\left[\tr\left[T_{n,X^n(k)}\rho_{\substack{x^n(m), s^n}}^{\cB^n}\right]\right]\Biggr)\notag\\
    &= \frac{1}{2^{nR}}\sum_{m=1}^{2^{nR}}\Biggl(2\sum_{X^n(m)}P^{n}(X^n(m))\tr\left[\left(\bbI^{\cB^n}-T_{n,X^n(m)}\right)\rho_{\substack{X^n(m), s^n}}^{\cB^n}\right]\notag\\
    &\hspace{25pt}+ 4\sum_{ k \in \cM_n\setminus{m}}\sum_{X^n(k)}P^{n}(X^n(k))\tr\left[T_{n,X^n(k)}\sum_{X^n(m)}P^{n}(X^n(m))\left(\rho_{\substack{X^n(m), s^n}}^{\cB'}\right)\right]\Biggr)\notag\\
    &= \frac{1}{2^{nR}}\sum_{m=1}^{2^{nR}}\Biggl(4\sum_{ k \in \cM_n\setminus{m}}\sum_{X^n(k)}P^{n}(X^n(k))\tr\left[T_{n,X^n(k)}\left(\rho_{\substack{s^n}}^{\cB^n}\right)\right] + 2\tr\left[\left( \bbI^{X^n\cB^n} - \bbT_{n}\right)\rho^{X^n\cB^n}_{P,s^n}\right]\Biggr)\notag\\
    &\overset{b}{\leq} \frac{1}{2^{nR}}\sum_{m=1}^{2^{nR}}\left(2(n+1)^{\abs{\cS_\varepsilon}-1 + \frac{t\abs{\cX}(\abs{\cB}-1)\abs{\cB}}{2} + (\abs{\cB} - 1)\abs{\cX}t}2^{n\left(t\left(\widehat{R} - \min_{Q \in \cP(\cS)} I_{1-t}[X;\cB]_{P,Q}\right)+\log (1+\varepsilon)\right)} + 4\sum_{ k \in \cM_n\setminus{m}}\tr[\bbT_{n}(\rho^{\cH_X^{\otimes n}}_{P} \otimes \rho_{\substack{s^n}}^{\cB^n})]\right)\notag\\
    &\overset{c}{\leq}\frac{1}{2^{nR}}\sum_{m=1}^{2^{nR}}\left(2(n+1)^{\abs{\cS_\varepsilon}-1 + \frac{t\abs{\cX}(\abs{\cB}-1)\abs{\cB}}{2} + (\abs{\cB} - 1)\abs{\cX}t}2^{n\left(t\left(\widehat{R} - \min_{Q \in \cP(\cS)}I_{1-t}[X;\cB]_{P,Q}\right)+\log (1+\varepsilon)\right)} + 4\cdot 2^{nR}\cdot n^{\frac{\abs{\cB}(\abs{\cB}-1)}{2}}\abs{\Lambda^{n}_{\abs{\cB}}} 2^{-n\widehat{R}}\right)\notag\\
    &= 2(n+1)^{\abs{\cS_\varepsilon}-1 + \frac{t\abs{\cX}(\abs{\cB}-1)\abs{\cB}}{2} + (\abs{\cB} - 1)\abs{\cX}t}2^{n\left(t\left(\widehat{R} - \min_{Q \in \cP(\cS)} I_{1-t}[X;\cB]_{P,Q}\right)+\log (1+\varepsilon)\right)} + 4\cdot 2^{n(R - \widehat{R})}\cdot n^{\frac{\abs{\cB}(\abs{\cB}-1)}{2}}\abs{\Lambda^{n}_{\abs{\cB}}},
    \label{NVY5}
\end{align}
where $a$ follows from Fact \ref{Hayashi_nagaoka}, $b$ follows from \eqref{continuous_eps_covering_density_operators_eq} and $c$ follows from \eqref{NSL9}. 
The combination of \eqref{M77}, \eqref{M78}, and \eqref{NVY5} yields
\eqref{NBO}.
This proves the achievability of Theorem \ref{lemma_rand_cap_cqavc} using the hypothesis testing approach. 
This completes the proof.

(\textbf{Converse}) 
We assume that a sequence of  codes $\{\cC\}$ satisfy 
\begin{align}
\varepsilon_n:=
\sup_{s^n \in \cS^n}\bbE_{\cC}\left[\bar{e}(\cC,s^n)\right]
=
\sup_{Q_{S^n} \in \cP(\cS^n)}
\sum_{s^n \in \cS^n}Q_{S^n}(s^n)
\bbE_{\cC}\left[\bar{e}(\cC,s^n)\right]\to 0.
\end{align}
For an arbitrary element $Q \in \cP(\cS) $, 
using Fano inequality (Fact \ref{fano}), we have
\begin{align}
R- \frac{1}{n}H_b(\varepsilon_n)- \varepsilon_n R
\le \frac{1}{n}I[M_n;\cB^n|\cC]_{Q^n}
\overset{a}{\leq} \frac{1}{n}I[X^n;\cB^n|\cC]_{P_{X^n},Q^n},
\end{align}
where $a$ follows from the Markov chain $M_n- X^n-\cB^n$. 
We denote the marginal distribution of $ P_{X^n}$
for $X_j$ by $P_{X_j}$. 
Then, we have
\begin{align}
& R- \frac{1}{n}H_b(\varepsilon_n)- \varepsilon_n R
\le 
\frac{1}{n}\inf_{Q \in \cP(\cS)} I[X^n;\cB^n|\cC]_{P_{X^n},Q^n} \nn
\\
=&
\frac{1}{n}\inf_{Q \in \cP(\cS)} 
\sum_{j=1}^n
I[X_j;\cB^n|X_1,\ldots, X_{j-1}|\cC]_{P_{X^n},Q^n} \nn\\
\overset{b}{\leq} &
\frac{1}{n}\inf_{Q \in \cP(\cS)} 
\sum_{j=1}^n
I[X_j;\cB^n, X_1,\ldots, X_{j-1}|\cC]_{P_{X^n},Q^n} \nn\\
\overset{c}{=}&
\frac{1}{n}\inf_{Q \in \cP(\cS)} 
\sum_{j=1}^n
I[X_j;\cB_j|\cC ]_{P_{X^n},Q^n}\label{NGH1}\\
=&
\frac{1}{n}\inf_{Q \in \cP(\cS)} 
\sum_{j=1}^n
I[X_j;\cB_j |\cC]_{P_{X_j},Q}
\le
\max_{P \in \cP(\cX)} \inf_{Q \in \cP(\cS)} 
I[X; \cB]_{P,Q},
\end{align}
where $b$ follows from $I[X_j;X_1,\ldots, X_{j-1}|\cC]_{P_{X^n},Q^n}\ge 0$.
Since the random variable $S$ subject to $Q^n$, 
the channel $X^n\to \cB^n$ forms a discrete memoryless channel, which implies
the Markov chain $(X_1,\ldots, X_{j-1}, \cB_1,\ldots, \cB_{j-1},\cB_{j+1},\ldots,
\cB_n) - X_j-\cB_j$. 
This relation yields Step $c$.
Taking the limit, we obtain the desired inequality
\begin{align}
R \le \max_{P \in \cP(\cX)} \inf_{Q \in \cP(\cS)} 
I[X; \cB]_{P,Q}.
\end{align}
\end{proof}

\section{Universal Test for Multiple Independence Testing}\label{sec:GIT}
\subsection{Problem formulation}
We generalize independence testing studied in Section \ref{sec:IT} to the case when
the classical system is composed of two classical systems. 
    Consider a collection of density operators $\left\{\rho_{x,y,s}^{\cB} : x \in \cX, y \in \cY, s \in \cS\right\}  \subset \cD(\cB)$, where $\cX,\cY$ are some non-empty finite sets and $\cB$ is a Hilbert space. 
    This collection of density operators satisfies Assumption \ref{NM9}.
    Now for each $s^n \in \cS^n$ we consider the following state:
    \begin{align}
        \rho^{X^nY^n\cB^n}_{P_{X},P_{Y},s^n} &:= \sum_{\substack{ x^n\in \cX^n\\y^n \in \cY^n}}P^{n}_{X}(x^n)\ketbrasys{x^n}{X^n}\otimes P^{n}_{Y}(y^n)\ketbrasys{y^n}{Y^n}\otimes \rho^{\cB^n}_{x^n,y^n,s^n},\label{Omega1states_sn}
    \end{align}
    where $\forall x^n:=(x_1,\cdots,x_n) \in \cX^n, y^n := (y_1,\cdots,y_n) \in \cY^n, \rho^{\cB^n}_{x^n,y^n,s^n}:=\bigotimes_{i=1}^{n}\rho_{x_i,y_i,s_i}^{\cB}$, $P_{X}$ and $P_{Y}$ are probability distributions over $\cX$ and $\cY$,  we here denote $X,Y$ as the classical subsystems of respective dimensions $\abs{\cX},\abs{\cY}$, corresponding to the random variables $X$ and $Y$ (i.e. through out this section, $X$ and $Y$ are used both as random variables and classical subsystems in appropriate contexts) spanned by orthonormal basis $\left\{\ket{x}^{X}\right\}_{x \in \cX}$ and $\left\{\ket{y}^{Y}\right\}_{y \in \cY}$ respectively. Now consider the following hypotheses
     \begin{align}
        H_{0} &: \left\{\rho^{X^nY^n\cB^n}_{P_{X},P_{Y},s^n}\right\}_{s^n \in \cS^n},\label{GIT_hypotheses}\\
        H_{1,X} &: \left\{\sum_{\substack{ x^n\in \cX^n\\y^n \in \cY^n}}P^{n}_{X}(x^n)\ketbrasys{x^n}{X^n}\otimes P^{n}_{Y}(y^n)\ketbrasys{y^n}{Y^n} \otimes \sigma^{\cB^n}_{y^n}\right\}_{\substack{\substack{\{\sigma^{\cB^n}_{y^n}\} \subset \cD(\cH_\cB^{\otimes n})}}}\label{Omega2states_arbit_X},\\
        H_{1,Y} &: \left\{\sum_{\substack{ x^n\in \cX^n\\y^n \in \cY^n}}P^{n}_{X}(x^n)\ketbrasys{x^n}{X^n}\otimes P^{n}_{Y}(y^n)\ketbrasys{y^n}{Y^n} \otimes \sigma^{\cB^n}_{x^n}\right\}_{\substack{\substack{\{\sigma^{\cB^n}_{x^n}\} \subset \cD(\cH_\cB^{\otimes n})}}}\label{Omega2states_arbit_Y},\\
        H_{1} &: \left\{\sum_{\substack{ x^n\in \cX^n\\y^n \in \cY^n}}P^{n}_{X}(x^n)\ketbrasys{x^n}{X^n}\otimes P^{n}_{Y}(y^n)\ketbrasys{y^n}{Y^n} \otimes \sigma^{\cB^n}\right\}_{\substack{\substack{\sigma^{\cB^n} \in \cD(\cH_\cB^{\otimes n})}}}\label{Omega2states_arbit_None},
    \end{align}
where $\sigma^{\cB^n}_{x^n},\sigma^{\cB^n}_{y^n}$ and $\sigma^{\cB^n}$ come from the respective marginals of the following state:
\begin{align}
    \sigma^{X^nY^n\cB^n}_{P_{X},P_{Y}, \{\sigma^{\cB^n}_{x^n,y^n}\}} 
 &:= \sum_{\substack{x^n\in \cX^n\\y^n \in \cY^n}} P^{n}_{X}(x^n) \ketbrasys{x^n}{X^n}\otimes P^{n}_{Y}(y^n)\ketbrasys{y^n}{Y^n} \otimes \sigma^{\cB^n}_{x^n,y^n}\label{Omega2states_arbit},
\end{align}
and $\{\sigma_{x^n,y^n}\}$ is an arbitrary collection of states indexed by $x^n \in \cX^n$ and 
$y^n \in \cY^n$ (i.e., $\{\sigma_{x^n,y^n}\}_{x^n \in \cX^n, y^n \in \cY^n}$).

The aim here is to design a measurement-based test $\bbT^{\star}_{n}$, which accepts $H_0$ with very high probability and rejects $H_{1,X},H_{1,X}$ and $H_1$ with exponentially large probability. However, before solving this problem, we will study a weaker version of the above problem in the following subsection.
In addition, the converse part can be shown in the same way as Appendix
\ref{NMA8}, we focus only on the direct part, i.e., the construction of 
asymptotically good codes.

\subsection{Universal Non-simultaneous Test for Multiple Independence Testing}\label{cq_mac_avht_null_statement_non_sim}
Here, our objective is to find two tests $\hat{\bbT}_{n}^{X}$ and  $\hat{\bbT}_{n}^{Y}$  which accept $H_0$ (mentioned in \eqref{GIT_hypotheses}) with very high probability and accepts $H_1^X$ and $H_1^Y$ with arbitrarily small probability respectively, where $H_1^X$ and $H_1^Y$ is  defined as follows,
\begin{equation*}
    H_1^X := H_{1,X} \cup H_1 \quad \text{ and }  H_1^Y := H_{1,Y} \cup H_1.
\end{equation*}
In the following, we construct the tests $\hat{\bbT}_{n}^{X}$ and  $\hat{\bbT}_{n}^{Y}$ which satisfies the objective mentioned above.

For an $\left(\bar{x}^n,\bar{y}^n\right) := ((\bar{x}_1,\bar{y}_1),(\bar{x}_2, \bar{y}_2),\cdots,(\bar{x}_n,\bar{y}_n))\in \cX^n\times\cY^n$ which has a form $\left(\bar{x}^n,\bar{y}^n\right) := \Big(\underbrace{(1,1),\cdots,(1,1)}_{m_{1,1}},\underbrace{(1,2),\cdots,(1,2)}_{m_{1,2}},$\\$\cdots,\underbrace{(\abs{\cX},\abs{\cY}),\cdots,(\abs{\cX},\abs{\cY})}_{m_{\abs{\cX},\abs{\cY}}}\Big)$, 
where $\forall i \in [\abs{\cX}], j \in [\abs{\cY}], m_{i,j} \geq 0$ and $\sum_{\substack{i \in [\abs{\cX}] \\ j \in[ \abs{\cY}]}}m_{i,j} = n$. 
Given $\left(\bar{x}^n,\bar{y}^n\right)$, we consider the following state:
\begin{equation}
    \widehat{\rho}^{\cB^n}_{\bar{x}^n,\bar{y}^n} := \rho^{\cB^{m_{1,1}}}_{U,m_{1,1}} \otimes \rho^{\cB^{m_{1,2}}}_{U,m_{1,2}} \otimes \cdots \otimes \rho^{\cB^{m_{\abs{\cX},\abs{\cY}}}}
    _{U,m_{\abs{\cX},\abs{\cY}}},\label{rhohatxbarnybarn1}
\end{equation}
where 
$\rho^{\cB^{m_{i,j}}}_{U,m_{i,j}}$ is defined in \eqref{uniform_n_state}
    for any $i \in [\abs{\cX}], j \in [\abs{\cY}]$.
Now for any general pair $(x^n,y^n) \in \cX^n\times\cY^n$, which is a permutation of $(\bar{x}^n,\bar{y}^n)$ i.e. $x^n  = \pi(\bar{x}^n) = (\bar{x}_{\pi^{-1}(1)},\bar{x}_{\pi^{-1}(2)},\cdots,\bar{x}_{\pi^{-1}(n)})$ and $y^n  = \pi(\bar{y}^n) = (\bar{y}_{\pi^{-1}(1)},\bar{y}_{\pi^{-1}(2)},\cdots,\bar{y}_{\pi^{-1}(n)})$, for some $\pi \in S_n$, we define $\hat{\rho}^{\cB^n}_{x^n,y^n}$ as
\begin{equation}
    \widehat{\rho}^{\cB^n}_{x^n,y^n} := V^{\cB^n}(\pi) \left(\widehat{\rho}^{\cB^n}_{\bar{x}^n,\bar{y}^n}\right){V^{\cB^n}}^{\dagger}(\pi).\label{rhohatxnyn1}
\end{equation}
We define $\widehat{\rho}^{\cB^n}_{y^n}$ for any $y^n \in \cY^n$ in an analogous way, $\widehat{\rho}^{\cB^n}_{x^n}$ is defined for any $x^n \in \cX^n$ in \eqref{rhohatxn}. We now define the following states:
\begin{align*}
    \widehat{\rho}^{X^nY^n\cB^n}_{U,P_X,P_Y} &:= \sum_{\substack{x^n \in \cX^n \\ y^n \in \cY^n}}P^{n}_X(x^n) \ketbrasys{x^n}{X^n} \otimes P^{n}_Y(y^n) \ketbrasys{y^n}{Y^n} \otimes \widehat{\rho}^{\cB^n}_{x^n,y^n},\\
    \widehat{\rho}^{X^n\cB^n}_{U,P_X} &:= \sum_{\substack{x^n \in \cX^n}}P^{n}_X(x^n) \ketbrasys{x^n}{X^n} \otimes \widehat{\rho}^{\cB^n}_{x^n},\\
    \widehat{\rho}^{Y^n\cB^n}_{U,P_Y} &:= \sum_{\substack{y^n \in \cY^n}}P^{n}_Y(y^n) \ketbrasys{y^n}{Y^n} \otimes \widehat{\rho}^{\cB^n}_{y^n}.
\end{align*}

For $\widehat{R}_1,\widehat{R}_2>0$, we define three projective measurement-based tests $\hat{\bbT}_{n,X},\hat{\bbT}_{n,Y}$ and $\hat{\bbT}_{n}$ as follows,
\begin{align}
    \hat{\bbT}_{n,X} &:= \sum_{\substack{x^n \in \cX^n \\ y^n \in \cY^n}} \ketbrasys{x^n}{X^n} \otimes \ketbrasys{y^n}{Y^n} \otimes \left\{\widehat{\rho}^{\cB^n}_{x^n,y^n} \succeq 2^{n\widehat{R}_2}\widehat{\rho}^{\cB^n}_{x^n}\right\} \triangleq \left\{\widehat{\rho}^{X^nY^n\cB^n}_{U,P_X,P_Y} \succeq 2^{n\widehat{R}_2} \left({\rho^{\cH_Y^{\otimes n}}_{P_Y}} \otimes \rho^{X^n\cB^n}_{U,P_X}\right)\right\}\label{universal_test_X_TS},\\
    \hat{\bbT}_{n,Y} &:= \sum_{\substack{x^n \in \cX^n \\ y^n \in \cY^n}} \ketbrasys{x^n}{X^n} \otimes \ketbrasys{y^n}{Y^n} \otimes \left\{\widehat{\rho}^{\cB^n}_{x^n,y^n} \succeq 2^{n\widehat{R}_1}\widehat{\rho}^{\cB^n}_{y^n}\right\} \triangleq \left\{\widehat{\rho}^{X^nY^n\cB^n}_{U,P_X,P_Y} \succeq 2^{n\widehat{R}_1} \left({\rho^{\cH_X^{\otimes n}}_{P_X}} \otimes \rho^{Y^n\cB^n}_{U,P_Y}\right)\right\}\label{universal_test_Y_TS},\\
    \hat{\bbT}_{n} &:= \sum_{\substack{x^n \in \cX^n \\ y^n \in \cY^n}} \ketbrasys{x^n}{X^n} \otimes \ketbrasys{y^n}{Y^n} \otimes \left\{\widehat{\rho}^{\cB^n}_{x^n,y^n} \succeq 2^{n(\widehat{R}_1+\widehat{R}_2)}\rho^{\cB^n}_{U,n}\right\} \triangleq \left\{\widehat{\rho}^{X^nY^n\cB^n}_{U,P_X,P_Y} \succeq 2^{n(\widehat{R}_1+\widehat{R}_2)} \left({\rho^{\cH_X^{\otimes n}}_{P_X}} \otimes {\rho^{\cH_Y^{\otimes n}}_{P_Y}} \otimes \rho^{\cB^n}_{U,n}\right)\right\}\label{universal_test_None_TS}.
\end{align}

     Observe that like $\bbT_{n}$ defined in \eqref{universal_test}, the expressions of 
     $\hat{\bbT}_{n,X}, \hat{\bbT}_{n,Y}$, and $\hat{\bbT}_{n}$ are independent of 
     the forms of the distributions $P_X$ and $P_Y$. We now define our tests $\hat{\bbT}_{n}^{X}$ and $\hat{\bbT}_{n}^{Y}$ as follows,
\begin{equation}
    \hat{\bbT}_{n}^{X} := \hat{\bbT}_{n}\hat{\bbT}_{n,Y} \quad \text{ and } \quad \hat{\bbT}_{n}^{Y} := \hat{\bbT}_{n}\hat{\bbT}_{n,X}. \label{tests_hayashi_generalised}
\end{equation}
Note that the constructions of $\hat{\bbT}_{n}^{X}$ and $\hat{\bbT}_{n}^{Y}$ are independent of the forms 
of the distributions $P_X$ and $P_Y$. 
To evaluate the performance of the tests $\hat{\bbT}_{n}^{X}$ and $\hat{\bbT}_{n}^{Y}$,
for any probability distribution $Q \in \cP(\cS)$, we introduce the state as
\begin{align}
\rho^{XY\cB}_{P_X,P_Y,Q} :=  \sum_{\substack{x \in \cX \\ y \in \cY}} P_X(x)\ketbrasys{x}{X} \otimes P_Y(y)\ketbrasys{y}{Y} \otimes \rho_{Q,x,y}^{\cB}.\label{rhoPXPYQ_1}
\end{align}
where $\rho_{Q,x,y}^{\cB} = \sum_{s \in \cS} Q(s)\rho^{\cB}_{x,y,s}$. Then,
we introduce the mutual-information-type quantities;
\begin{align}
    I[Y;\cB|X]_{P_X,P_Y,Q}&:=  D(\rho^{XY\cB}_{P_X,P_Y,Q} || \rho^{Y}_{P_Y} \otimes \rho^{X\cB}_{P_X,Q}),\label{iy;b|x}\\
    I[X;\cB|Y]_{P_{X},P_{Y},Q}&:= D(\rho^{XY\cB}_{P_X,P_Y,Q} || \rho^{X}_{P_X} \otimes \rho^{Y\cB}_{P_Y,Q}),\label{ix;b|y}\\
    I[XY;\cB]_{P_{X},P_{Y},Q}&:=  D(\rho^{XY\cB}_{P_X,P_Y,Q} || \rho^{X}_{P_X} \otimes \rho^{Y}_{P_Y} \otimes \rho^{\cB}_{Q}),\label{ixy;b}
\end{align}
where 
\begin{equation*}
   \rho^{X\cB}_{P_X,Q} := \tr_{Y}[\rho^{XY\cB}_{P_X,P_Y,Q}], \rho^{Y\cB}_{P_Y,Q} := \tr_{X}[\rho^{XY\cB}_{P_X,P_Y,Q}] \text{ and } \rho^{\cB}_{Q} := \tr_{XY}[\rho^{XY\cB}_{P_X,P_Y,Q}].
\end{equation*}
To characterize the above quantities, we introduce 
R\'{e}nyi-type quantities;
\begin{align}    
&I_{1-t}[Y;\cB|X]_{P_{X},P_{Y},Q} :=  -\frac{1 - t}{t} \log\left(\sum_{x \in \cX}P_{X}(x)\tr\left[ \left(\sum_{y \in \cY}P_{Y}(y) \left({\rho_{Q,x,y}^{\cB}}\right)^{1 - t}\right)^{\frac{1}{1-t}}\right]\right), \label{sibsonqmiY;XB_1}\\
    &I_{1-t}[X;\cB|Y]_{P_{X},P_{Y},Q} :=  -\frac{1 - t}{t} \log \left(\sum_{y \in \cY}P_{Y}(y)\tr\left[ \left(\sum_{x \in \cX}P_{X}(x) \left({\rho_{Q,x,y}^{\cB}}\right)^{1 - t}\right)^{\frac{1}{1-t}}\right]\right),\label{sibsonqmiX;YB_1}\\
    &I_{1-t}[XY;\cB]_{P_{X},P_{Y},Q} :=  -\frac{1 - t}{t} \log \tr\left[ \left(\sum_{\substack{x \in \cX \\ y \in \cY}}P_{X}(x)P_{Y}(y) \left({\rho_{Q,x,y}^{\cB}}\right)^{1 - t}\right)^{\frac{1}{1-t}}\right] \label{sibsonqmiXY;B_1},
    \end{align}
where $\rho_{Q,x,y}^{\cB} := \sum_{s \in \cS} Q(s)\rho_{x,y,s}^{\cB}$.
By using the above quantities, the mutual-information-type quantities are recovered as
\begin{align}
    I[Y;\cB|X]_{P_X,P_Y,Q}&= \lim_{t\to 0} I_{1-t}[Y;\cB|X]_{P_X,P_Y,Q} ,\label{NV1}\\
    I[X;\cB|Y]_{P_{X},P_{Y},Q}& = \lim_{t\to 0} I_{1-t}[X;\cB|Y]_{P_X,P_Y,Q},\label{NV2}\\
    I[XY;\cB]_{P_{X},P_{Y},Q}&= \lim_{t\to 0} I_{1-t}[XY;\cB]_{P_X,P_Y,Q} .\label{NV3}
\end{align}
The performance of the tests $\hat{\bbT}_{n}^{X}$ and $\hat{\bbT}_{n}^{Y}$
are evaluated by the following lemma.

\begin{lemma}\label{lemma_cq_mac_avht_independent_ts}
Assume that
the subset $\left\{\rho_{x,y,s}^{\cB} : x \in \cX, y \in \cY, s \in \cS\right\}  \subset \cD(\cB)$ satisfies
 Assumption \ref{NM9}.
For any real number $\varepsilon>0$,
any element $ s^n \in \cS^n$, 
any subsets $\{\sigma^{\cB^n}_{x^n}\}$, $\{\sigma^{\cB^n}_{y^n}\}$, and $\sigma^{\cB^n}$,  
the tests $\hat{\bbT}_{n}^{X}$ and $\hat{\bbT}_{n}^{Y}$ satisfy the following,
\begin{align}
     &\tr\left[\left( \bbI^{X^nY^n\cB^n} - \hat{\bbT}_{n}^{X}\right)\rho^{X^nY^n\cB^n}_{P_{X},P_{Y},s^n}\right] \leq 
f(n,\abs{\cX},\abs{\cY}, \abs{\cB}, \abs{\cS_\varepsilon},t)
(1+\varepsilon)^n\left(2^{nt\left(\widehat{R}_1 - \min _{Q \in \cP(\cS_\varepsilon)} I_{1-t}[X;\cB|Y]_{P_{X},P_{Y},Q}\right)} \right.\nn\\
&\hspace{170pt}+ \left.2^{nt\left(\widehat{R}_1 + \widehat{R}_2 -  \min _{Q \in \cP(\cS_\varepsilon)} I_{1-t}[XY;\cB]_{P_{X},P_{Y},Q}\right)}\right)\label{lemma_achievability_mac_independent_eq1_TS},\\
     &\tr\left[\left( \bbI^{X^nY^n\cB^n} - \hat{\bbT}_{n}^{Y}\right)\rho^{X^nY^n\cB^n}_{P_{X},P_{Y},s^n}\right] \leq 
f(n,\abs{\cX},\abs{\cY}, \abs{\cB}, \abs{\cS_\varepsilon},t)
(1+\varepsilon)^n
\left(2^{nt\left(\widehat{R}_2 - \min _{Q \in \cP(\cS_\varepsilon)} I_{1-t}[Y;\cB|X]_{P_{X},P_{Y},Q}\right)} \right.\nn\\
&\hspace{170pt}+ \left.2^{nt\left(\widehat{R}_1 + \widehat{R}_2 -  \min _{Q \in \cP(\cS_\varepsilon)} I_{1-t}[XY;\cB]_{P_{X},P_{Y},Q}\right)}\right),\label{lemma_achievability_mac_independent_eq2_TS}\\
&\tr\left[\hat{\bbT}_{n}^{X}\left(\sum_{\substack{ x^n\in \cX^n\\y^n \in \cY^n}}P^{n}_{X}(x^n)\ketbrasys{x^n}{X^n}\otimes P^{n}_{Y}(y^n)\ketbrasys{y^n}{Y^n} \otimes \sigma^{\cB^n}_{y^n}\right)\right]
\leq g_2(n,\abs{\cY},\abs{\cB})2^{-n\widehat{R}_1}, \label{lemma_achievability_mac_independent_eq3_TS}\\
&\tr\left[\hat{\bbT}_{n}^{X}\left(\sum_{\substack{ x^n\in \cX^n\\y^n \in \cY^n}}P^{n}_{X}(x^n)\ketbrasys{x^n}{X^n}\otimes P^{n}_{Y}(y^n)\ketbrasys{y^n}{Y^n} \otimes \sigma^{\cB^n}\right)\right]
\leq g_3(n,\abs{\cB})2^{-n(\widehat{R}_1 + \widehat{R}_2)}, \label{lemma_achievability_mac_independent_eq4_TS}\\
&\tr\left[\hat{\bbT}_{n}^{Y}\left(\sum_{\substack{ x^n\in \cX^n\\y^n \in \cY^n}}P^{n}_{X}(x^n)\ketbrasys{x^n}{X^n}\otimes P^{n}_{Y}(y^n)\ketbrasys{y^n}{Y^n} \otimes \sigma^{\cB^n}_{x^n}\right)\right]
\leq g_1(n,\abs{\cX},\abs{\cB})2^{-n\widehat{R}_2}, \label{lemma_achievability_mac_independent_eq5_TS}\\
&\tr\left[\hat{\bbT}_{n}^{Y}\left(\sum_{\substack{ x^n\in \cX^n\\y^n \in \cY^n}}P^{n}_{X}(x^n)\ketbrasys{x^n}{X^n}\otimes P^{n}_{Y}(y^n)\ketbrasys{y^n}{Y^n} \otimes \sigma^{\cB^n}\right)\right]
\leq g_3(n,\abs{\cB})2^{-n(\widehat{R}_1 + \widehat{R}_2)},\label{lemma_achievability_mac_independent_eq6_TS}
    \end{align}
where $\widehat{R}_1,\widehat{R}_2 > 0$, $t \in (0,1)$, and the symbols used here are defined as 
\begin{align}
    &f(n,\abs{\cX},\abs{\cY}, \abs{\cB}, \abs{\cS_\varepsilon},t) := (n+1)^{\abs{\cS_\varepsilon}-1 + \frac{t\abs{\cX}\abs{\cY}(\abs{\cB}-1)\abs{\cB}}{2} + (\abs{\cB} - 1)t \abs{\cX}\abs{\cY}} ,
    \quad g_1(n,\abs{\cX},\abs{\cB}) := n^{\frac{\abs{\cX}\abs{\cB}(\abs{\cB}-1)}{2}}\abs{\Lambda^{n}_{\abs{\cB}}}^{\abs{\cX}}\label{g_1_f},\\
    &g_2(n,\abs{\cY},\abs{\cB}) := n^{\frac{\abs{\cY}\abs{\cB}(\abs{\cB}-1)}{2}}\abs{\Lambda^{n}_{\abs{\cB}}}^{\abs{\cY}} \hspace{2pt} ,
    \quad g_3(n,\abs{\cB}) := n^{\frac{\abs{\cB}(\abs{\cB}-1)}{2}}\abs{\Lambda^{n}_{\abs{\cB}}}. \hspace{48pt}\label{g_3_f}
    \end{align}
\end{lemma}

Therefore, when $\widehat{R}_1,\widehat{R}_2$ satisfy the conditions
\begin{align}
 \widehat{R}_1 &< \min _{Q \in \cP(\cS_\eps)} I[X;\cB|Y]_{P_{X},P_{Y},Q},\label{VHI1}\\
 \widehat{R}_2 &< \min _{Q \in \cP(\cS_\eps)} I[Y;\cB|X]_{P_{X},P_{Y},Q},\label{VHI2}\\
  \widehat{R}_1 + \widehat{R}_2 &< \min _{Q \in \cP(\cS_\eps)} I[XY;\cB]_{P_{X},P_{Y},Q}
 \label{VHI3},
\end{align}
the tests $\hat{\bbT}_{n}^{X}$ and $\hat{\bbT}_{n}^{Y}$ well detect
the hypotheses $H_1^X$ and $H_1^Y$, respectively for large enough $n$.
This lemma is shown in Subsection \ref{S5-C}.
A corollary of Lemma \ref{lemma_cq_mac_avht_independent_ts} is mentioned below, which will be required to give a proof of achievability for deriving the capacity region for CQ-AVMAC (see Theorem \ref{lemma_rand_capacity_avmac}).

\begin{corollary}\label{corllary_cq_mac_avht_independent_ts}
Assume that
the subset $\left\{\rho_{x,y,s}^{\cB} : x \in \cX, y \in \cY, s \in \cS\right\}  \subset \cD(\cB)$ satisfies Assumption \ref{NM9}.
For any real number $\epsilon>0$ and
any element $s^n \in \cS^n$, 
the tests $\hat{\bbT}_{n}^{X}$ and $\hat{\bbT}_{n}^{Y}$ satisfy the following,
\begin{align}
     &\tr\left[\left( \bbI^{X^nY^n\cB^n} - \hat{\bbT}_{n}^{X}\right)\rho^{X^nY^n\cB^n}_{P_{X},P_{Y},s^n}\right] \leq 
f(n,\abs{\cX},\abs{\cY}, \abs{\cB}, \abs{\cS_\varepsilon},t)
(1+\varepsilon)^n
\left(2^{nt\left(\widehat{R}_1 - \min _{Q \in \cP(\cS_\varepsilon)} I_{1-t}[X;\cB|Y]_{P_{X},P_{Y},Q}\right)} \right.\nn\\
&\hspace{170pt}+ \left.2^{nt\left(\widehat{R}_1 + \widehat{R}_2 -  \min _{Q \in \cP(\cS_\varepsilon)} I_{1-t}[XY;\cB]_{P_{X},P_{Y},Q}\right)}\right)\label{corllary_achievability_mac_independent_eq1_TS}\\
     &\tr\left[\left( \bbI^{X^nY^n\cB^n} - \hat{\bbT}_{n}^{Y}\right)\rho^{X^nY^n\cB^n}_{P_{X},P_{Y},s^n}\right] \leq 
f(n,\abs{\cX},\abs{\cY}, \abs{\cB}, \abs{\cS_\varepsilon},t)(1+\varepsilon)^n
\left(2^{nt\left(\widehat{R}_2 - \min _{Q \in \cP(\cS_\varepsilon)} I_{1-t}[Y;\cB|X]_{P_{X},P_{Y},Q}\right)} \right.\nn\\
&\hspace{170pt}+ \left.2^{nt\left(\widehat{R}_1 + \widehat{R}_2 -  \min _{Q \in \cP(\cS_\varepsilon)} I_{1-t}[XY;\cB]_{P_{X},P_{Y},Q}\right)}\right)\label{corllary_achievability_mac_independent_eq2_TS}\\
&\tr\left[\hat{\bbT}_{n}^{X}\left(\sum_{\substack{ x^n\in \cX^n\\y^n \in \cY^n}}P^{n}_{X}(x^n)\ketbrasys{x^n}{X^n}\otimes P^{n}_{Y}(y^n)\ketbrasys{y^n}{Y^n} \otimes \rho^{\cB^n}_{y^n,s^n}\right)\right]
\leq g_2(n,\abs{\cY},\abs{\cB})2^{-n\widehat{R}_1}, \label{corllary_achievability_mac_independent_eq3_TS}\\
&\tr\left[\hat{\bbT}_{n}^{X}\left(\sum_{\substack{ x^n\in \cX^n\\y^n \in \cY^n}}P^{n}_{X}(x^n)\ketbrasys{x^n}{X^n}\otimes P^{n}_{Y}(y^n)\ketbrasys{y^n}{Y^n} \otimes \rho^{\cB^n}_{s^n}\right)\right]
\leq g_3(n,\abs{\cB})2^{-n(\widehat{R}_1 + \widehat{R}_2)}, \label{corllary_achievability_mac_independent_eq4_TS}\\
&\tr\left[\hat{\bbT}_{n}^{Y}\left(\sum_{\substack{ x^n\in \cX^n\\y^n \in \cY^n}}P^{n}_{X}(x^n)\ketbrasys{x^n}{X^n}\otimes P^{n}_{Y}(y^n)\ketbrasys{y^n}{Y^n} \otimes \rho^{\cB^n}_{x^n,s^n}\right)\right]
\leq g_1(n,\abs{\cY},\abs{\cB})2^{-n\widehat{R}_2}, \label{corllary_achievability_mac_independent_eq5_TS}\\
&\tr\left[\hat{\bbT}_{n}^{Y}\left(\sum_{\substack{ x^n\in \cX^n\\y^n \in \cY^n}}P^{n}_{X}(x^n)\ketbrasys{x^n}{X^n}\otimes P^{n}_{Y}(y^n)\ketbrasys{y^n}{Y^n} \otimes \rho^{\cB^n}_{s^n}\right)\right]
\leq g_3(n,\abs{\cB})2^{-n(\widehat{R}_1 + \widehat{R}_2)},\label{corllary_achievability_mac_independent_eq6_TS}
    \end{align}
where $\widehat{R}_1,\widehat{R}_2 > 0$, $t \in (0,1)$.
For each $x^n\in \cX^n,y^n \in \cY^n,s^n \in \cS^n$, we define the following
\begin{align}
\rho^{\cB^n}_{x^n,s^n} &:= \sum_{y^n \in \cY^n}P^{n}_{Y}(y^n)\rho^{\cB^n}_{x^n,y^n,s^n},\label{marginal_X}\\
    \rho^{\cB^n}_{y^n,s^n} &:= \sum_{x^n \in \cX^n}P^{n}_{X}(x^n)\rho^{\cB^n}_{x^n,y^n,s^n},\label{marginal_Y}\\
    \rho^{\cB^n}_{s^n} &:= \sum_{\substack{x^n \in \cX^n \\ y^n \in \cY^n}}P^{n}_{X}(x^n)P^{n}_{Y}(y^n)\rho^{\cB^n}_{x^n,y^n,s^n}.\label{marginal_None}
\end{align}
\end{corollary}

\subsection{Proof of Lemma \ref{lemma_cq_mac_avht_independent_ts}}\label{S5-C}
     We note from Lemma \ref{lemma_commutativity}, for any $x^n \in \cX^n,y^n \in \cY^n$, the states $\widehat{\rho}^{\cB^n}_{x^n,y^n}, \widehat{\rho}^{\cB^n}_{y^n}$ and $\rho^{\cB^n}_{U,n}$ commute with each other. Thus, the tests $\hat{\bbT}_{n}$ and $\hat{\bbT}_{n,Y}$ commute with each other. Using similar arguments, we see that the tests $\hat{\bbT}_{n}$ and $\hat{\bbT}_{n,X}$ commute with each other as well. 

Further, the states $\widehat{\rho}^{X^nY^n\cB^n}_{U,P_X,P_Y}, {\rho^{\cH_X^{\otimes n}}_{P_X}} \otimes \rho^{Y^n\cB^n}_{U,P_Y}, {\rho^{\cH_Y^{\otimes n}}_{P_Y}} \otimes \rho^{X^n\cB^n}_{U,P_X}$ and $\rho^{\cH_X^{\otimes n}}_{P_X} \otimes \rho^{\cH_Y^{\otimes n}}_{P_Y} \otimes \rho^{\cB^n}_{U,n}$ are permutation invariant, $\hat{\bbT}_{n}^{X}$ and $\hat{\bbT}_{n}^{Y}$ are permutation invariant. Thus, for any permutation $\pi \in S_n$, the following holds,

\begin{align*}
    \tr\left[\left( \bbI^{X^nY^n\cB^n} - \hat{\bbT}_{n}^{X}\right)\rho^{X^nY^n\cB^n}_{P_{X},P_{Y},s^n}\right] &= \tr\left[\left( \bbI^{X^nY^n\cB^n} - \hat{\bbT}_{n}^{X}\right)\rho^{X^nY^n\cB^n}_{P_{X},P_{Y},\pi(s^n)}\right],\\
    \tr\left[\left( \bbI^{X^nY^n\cB^n} - \hat{\bbT}_{n}^{Y}\right)\rho^{X^nY^n\cB^n}_{P_{X},P_{Y},s^n}\right] &= \tr\left[\left( \bbI^{X^nY^n\cB^n} - \hat{\bbT}_{n}^{Y}\right)\rho^{X^nY^n\cB^n}_{P_{X},P_{Y},\pi(s^n)}\right].
\end{align*}

Thus, for any $s^n \in \cS_\varepsilon^n$, we upper-bound the quantity $ \tr\left[\left( \bbI^{X^nY^n\cB^n} - \hat{\bbT}_{n}^{X}\right)\rho^{X^nY^n\cB^n}_{P_X,P_Y,s^n}\right]$ as follows,
\begin{align}
    &\tr\left[\left( \bbI^{X^nY^n\cB^n} - \hat{\bbT}_{n}^{X}\right)\rho^{X^nY^n\cB^n}_{P_{X},P_{Y},s^n}\right]\nn\\ 
     &\overset{a}{=}  \tr\left[\left( \bbI^{X^nY^n\cB^n} - \hat{\bbT}^{X}_{n}\right)\frac{1}{\abs{T_{Q_{s^n}}}} \sum_{\pi \in S_n}\rho^{X^nY^n\cB^n}_{P_{X},P_{Y},\pi(s^n)}\right]\nn\\
     &= \tr\left[\left( \bbI^{X^nY^n\cB^n} - \hat{\bbT}_{n}^{X}\right)\frac{1}{\abs{T_{Q_{s^n}}}}\sum_{\widehat{s}^{n} \in T_{Q_{s^n}}} \rho^{X^nY^n\cB^n}_{P_{X},P_{Y},\widehat{s}^{n}}\right]\nn\\
     &\overset{b}{\leq} (n+1)^{\abs{\cS_\varepsilon}-1}\tr\left[\left( \bbI^{X^nY^n\cB^n} - \hat{\bbT}_{n}^{X}\right)\sum_{\widehat{s}^{n} \in T_{Q_{s^n}}} Q^{n}_{s^n}(\widehat{s}^n)\rho^{X^nY^n\cB^n}_{P_{X},P_{Y},\widehat{s}^n}\right]\nn\\
     &\leq (n+1)^{\abs{\cS_\varepsilon}-1}\tr\left[\left( \bbI^{X^nY^n\cB^n} - \hat{\bbT}_{n}^{X}\right)\sum_{\widehat{s}^{n} \in \cS_\varepsilon^n} Q^{n}_{s^n}(\widehat{s}^n)\rho^{X^nY^n\cB^n}_{P_{X},P_{Y},\widehat{s}^n}\right]\nn\\
     &\overset{c}{=} (n+1)^{\abs{\cS_\varepsilon}-1}\tr\left[\left( \bbI^{X^nY^n\cB^n} - \hat{\bbT}_{n}^{X}\right)\left(\rho^{{VXY\cB}}_{P_{X},P_{Y},Q_{s^n}}\right)^{\otimes n}
     \right]\nn\\
     &\leq (n+1)^{\abs{\cS_\varepsilon}-1}\tr\left[\left( \bbI^{X^nY^n\cB^n} - \hat{\bbT}_{n,Y}\right)\left(\rho^{{XY\cB}}_{P_{X},P_{Y},Q_{s^n}}\right)^{\otimes n}
     \right] + (n+1)^{\abs{\cS_\varepsilon}-1}\tr\left[\left( \bbI^{X^nY^n\cB^n} - \hat{\bbT}_{n}\right)\left(\rho^{{XY\cB}}_{P_{X},P_{Y},Q_{s^n}}\right)^{\otimes n}
     \right],\label{combined}
\end{align}
where $a$ follows from the fact that $\hat{\bbT}^{X}_{n}$ is permutation invariant, $b$ follows from \eqref{inv_type_set_size_ub}, $c$ follows from \eqref{rhoPXPYQ_1}. We upper bound the first term in \eqref{combined} as follows
\begin{align}
     &\hspace{14pt}(n+1)^{\abs{\cS_\varepsilon}-1}\tr\left[\left( \bbI^{X^nY^n\cB^n} - \hat{\bbT}_{n,Y}\right)\left(\rho^{{XY\cB}}_{P_X,P_Y,Q_{s^n}}\right)^{\otimes n}
     \right]\nn\\
     &= (n+1)^{\abs{\cS_\varepsilon}-1}\tr\left[\left( \bbI^{X^nY^n\cB^n} - \hat{\bbT}_{n,Y}\right)\sum_{\substack{x^n \in \cX^n \\ y^n \in \cY^n}}P^{n}_X(x^n)\ketbrasys{x^n}{X^n} \otimes P^{n}_Y(y^n)\ketbrasys{y^n}{Y^n} \otimes \left(\bigotimes_{i=1}^{n}\rho^{\cB}_{Q_{s^n},x_i,y_i}\right)
     \right]\nn\\
     &\overset{a}{=}(n+1)^{\abs{\cS_\varepsilon}-1}\sum_{y^n \in \cY^n}P^{n}_Y(y^n)\tr\left[ \sum_{x^n \in \cX^n}P^{n}_X(x^n)\left(\bigotimes_{i=1}^{n}\rho^{\cB}_{Q_{s^n},x_i,y_i}\right)\left\{\widehat{\rho}^{\cB^n}_{x^n,y^n} \prec 2^{n\widehat{R}_1}\widehat{\rho}^{\cB^n}_{y^n}\right\}
     \right]\nn\\
     &\overset{b}{\leq}(n+1)^{\abs{\cS_\varepsilon}-1}2^{nt \widehat{R}_1 }\sum_{y^n \in \cY^n}P^{n}_Y(y^n)\tr\left[ \sum_{x^n \in \cX^n}P^{n}_X(x^n)\left(\bigotimes_{i=1}^{n}\rho^{\cB}_{Q_{s^n},x_i,y_i}\right)(\widehat{\rho}^{\cB^n}_{x^n,y^n})^{-t}(\widehat{\rho}^{\cB^n}_{y^n})^{t}
     \right]\nn\\
     &\overset{c}{\leq}(n+1)^{\abs{\cS_\varepsilon}-1}2^{nt \widehat{R}_1}n^{\frac{t\abs{\cX}\abs{\cY}(\abs{\cB}-1)\abs{\cB}}{2}}\abs{\Lambda^{n}_{\abs\cB}}^{t \abs{\cX}\abs{\cY}}\sum_{y^n \in \cY^n}P^{n}_Y(y^n)\tr\left[ \sum_{x^n \in \cX^n}P^{n}_X(x^n) \left(\bigotimes_{i=1}^{n}\rho^{\cB}_{Q_{s^n},x_i,y_i}\right)^{1-t}(\widehat{\rho}^{\cB^n}_{y^n})^{t}
     \right]\nn\\
     &\overset{d}{\leq} (n+1)^{\abs{\cS_\varepsilon}-1 + \frac{t\abs{\cX}\abs{\cY}(\abs{\cB}-1)\abs{\cB}}{2} + (\abs{\cB} - 1)t \abs{\cX}\abs{\cY}}2^{nt \widehat{R}_1}\sum_{y^n \in \cY^n}P^{n}_Y(y^n)\tr\left[ \bigotimes_{i=1}^{n}\left(\sum_{x \in \cX}P_X(x)\left(\rho^{\cB}_{Q_{s^n},x,y_i}\right)^{1-t}\right)(\widehat{\rho}^{\cB^n}_{y^n})^{t}
     \right]\nn\\
     &\overset{e}{\leq} (n+1)^{\abs{\cS_\varepsilon}-1 + \frac{t\abs{\cX}\abs{\cY}(\abs{\cB}-1)\abs{\cB}}{2} + (\abs{\cB} - 1)t \abs{\cX}\abs{\cY}}2^{nt \widehat{R}_1}\sum_{y^n \in \cY^n}P^{n}_Y(y^n)\left(\tr\left[ \left(\bigotimes_{i=1}^{n}\left(\sum_{x \in \cX}P_X(x)\left(\rho^{\cB}_{Q_{s^n},x,y_i}\right)^{1-t}\right)\right)^{\frac{1}{1 - t}}
     \right]\right)^{1-t}\nn\\
     &=(n+1)^{\abs{\cS_\varepsilon}-1 + \frac{t\abs{\cX}\abs{\cY}(\abs{\cB}-1)\abs{\cB}}{2} + (\abs{\cB} - 1)t \abs{\cX}\abs{\cY}}2^{nt \widehat{R}_1}\sum_{y^n \in \cY^n}P^{n}_Y(y^n)\prod_{i=1}^{n}\left(\tr\left[ \left(\sum_{x \in \cX}P_X(x)\left(\rho^{\cB}_{Q_{s^n},x,y_i}\right)^{1-t}\right)^{\frac{1}{1 - t}}
     \right]\right)^{1-t}\nn\\
     &= (n+1)^{\abs{\cS_\varepsilon}-1 + \frac{t\abs{\cX}\abs{\cY}(\abs{\cB}-1)\abs{\cB}}{2} + (\abs{\cB} - 1)t \abs{\cX}\abs{\cY}}2^{nt \widehat{R}_1}\left(\sum_{y \in \cY}P_Y(y)\left(\tr\left[ \left(\left(\sum_{x \in \cX}P_X(x)\rho^{\cB}_{Q_{s^n},x,y}\right)^{1-t}\right)^{\frac{1}{1 - t}}
     \right]\right)^{1-t}\right)^{n}\nn\\
     &\overset{f}{=} (n+1)^{\abs{\cS_\varepsilon}-1 + \frac{t\abs{\cX}\abs{\cY}(\abs{\cB}-1)\abs{\cB}}{2} + (\abs{\cB} - 1)t \abs{\cX}\abs{\cY}}2^{nt( \widehat{R}_1 - I_{1-t}[X;\cB|Y]_{P_{X},P_{Y},Q_{s^n}})}\nn\\
     &{\leq}(n+1)^{\abs{\cS_\varepsilon}-1 + \frac{t\abs{\cX}\abs{\cY}(\abs{\cB}-1)\abs{\cB}}{2} + (\abs{\cB} - 1)t \abs{\cX}\abs{\cY}}2^{nt( \widehat{R}_1 - \min_{Q \in \cP(\cS_\varepsilon)}I_{1-t}[X;\cB|Y]_{P_{X},P_{Y},Q})}\nn\\
     &{\leq}f(n,\abs{\cX},\abs{\cY}, \abs{\cB}, \abs{\cS_\varepsilon},t)2^{nt( \widehat{R}_1 - \min_{Q \in \cP(\cS_\varepsilon)}I_{1-t}[X;\cB|Y]_{P_{X},P_{Y},Q})}\label{TnY_H0_1},
\end{align}
 where $a$ follows from the definition of $\hat{\bbT}_{n,Y}$ in \eqref{universal_test_Y_TS}, $b$ follows from the fact that since $ \widehat{\rho}^{\cB^n}_{x^n,y^n}$ commutes with $\widehat{\rho}^{\cB^n}_{y^n}$ (which follows from Lemma \ref{lemma_commutativity}), we have $\left\{\widehat{\rho}^{\cB^n}_{x^n,y^n} - 2^{n\widehat{R}_1}\widehat{\rho}^{\cB^n}_{y^n} \prec 0\right\} \leq 2^{nt\widehat{R}_1}\left(\widehat{\rho}^{\cB^n}_{x^n,y^n}\right)^{-t}\left(\widehat{\rho}^{\cB^n}_{y^n}\right)^{t}$ for all $t \in (0,1)$, $c$ follows using arguments similar to that used in the proof of Lemma \ref{claim_f}, $d$ follows from Fact \ref{fact_type_size_ub}, $e$ follows from \eqref{Holder_corllary_eq} of Fact \ref{Holder_corllary} and $f$ follows from \eqref{sibsonqmiY;XB_1}.
 
 Similarly, we can upper-bound the second term in \eqref{combined} as follows
 \begin{align}
     &\hspace{14pt}(n+1)^{\abs{\cS_\varepsilon}-1}\tr\left[\left( \bbI^{X^nY^n\cB^n} - \hat{\bbT}_{n}\right)\left(\rho^{{XY\cB}}_{P_X,P_Y,Q_{s^n}}\right)^{\otimes n}
     \right] \leq f(n,\abs{\cX},\abs{\cY}, \abs{\cB}, \abs{\cS_\varepsilon},t)2^{nt( \widehat{R}_1 + \widehat{R}_2 - \min_{Q \in \cP(\cS_\varepsilon)}I_{1-t}[XY;\cB]_{P_{X},P_{Y},Q})}.\label{TnY_H0_2}
 \end{align}
 Thus, for any $s^n \in \cS_\varepsilon^n$,
 from \cref{TnY_H0_1,TnY_H0_2,combined}, we get 
 \begin{align}
     &\tr\left[\left( \bbI^{X^nY^n\cB^n} - \hat{\bbT}_{n}^{X}\right)\rho^{X^nY^n\cB^n}_{P_{X},P_{Y},s^n}\right] \leq 
f(n,\abs{\cX},\abs{\cY}, \abs{\cB}, \abs{\cS_\varepsilon},t)
\left(2^{nt\left(\widehat{R}_1 - \min _{Q \in \cP(\cS_\varepsilon)} I_{1-t}[X;\cB|Y]_{P_{X},P_{Y},Q}\right)} \right.\nn\\
&\hspace{170pt}+ \left.2^{nt\left(\widehat{R}_1 + \widehat{R}_2 -  \min _{Q \in \cP(\cS_\varepsilon)} I_{1-t}[XY;\cB]_{P_{X},P_{Y},Q}\right)}\right)\label{BNB78}.
\end{align}
 The combination of \eqref{epsilon_net_state_density_eq} and \eqref{BNB78} implies \eqref{lemma_achievability_mac_independent_eq1_TS}
for any $s^n \in \cS$.

 Using similar calculations we can prove \eqref{lemma_achievability_mac_independent_eq2_TS} for the test $\hat{\bbT}_{n}^{Y}$ as well. Now, for any collection $\left\{\sigma_{y^n}^{\cB^n}\right\}_{\substack{y^n \in \cY^n}} \subset \cD(\cH_\cB^{\otimes n})$, we have the following upper-bound:
\begin{align*}
    &\tr\left[\hat{\bbT}_{n}^{X}\left(\sum_{\substack{ x^n\in \cX^n\\y^n \in \cY^n}}P^{n}_{X}(x^n)\ketbrasys{x^n}{X^n}\otimes P^{n}_{Y}(y^n)\ketbrasys{y^n}{Y^n} \otimes \sigma^{\cB^n}_{y^n}\right)\right] \\
    &\leq \tr\left[\hat{\bbT}_{n,Y}\left(\sum_{\substack{ x^n\in \cX^n\\y^n \in \cY^n}}P^{n}_{X}(x^n)\ketbrasys{x^n}{X^n}\otimes P^{n}_{Y}(y^n)\ketbrasys{y^n}{Y^n} \otimes \sigma^{\cB^n}_{y^n}\right)\right]\\
    &\overset{a}{=} \tr\left[\hat{\bbT}_{n,Y}\left(\sum_{\substack{ x^n\in \cX^n\\y^n \in \cY^n}}P^{n}_{X}(x^n)\ketbrasys{x^n}{X^n}\otimes P^{n}_{Y}(y^n)\ketbrasys{y^n}{Y^n} \otimes \frac{1}{\abs{S_{y^n}}}\sum_{\pi \in S_{n,y^n}}V^{\cB^n}(\pi)(\sigma^{\cB^n}_{y^n}) {V^{\cB^n}}^{\dagger}(\pi)\right)\right] \\
    &\overset{b}{\leq} n^{\frac{\abs{\cY}\abs{\cB}(\abs{\cB}-1)}{2}}\abs{\Lambda^{n}_{\abs{\cB}}}^{\abs{\cY}}\tr\left[\hat{\bbT}_{n,Y}\left(\sum_{\substack{ x^n\in \cX^n\\y^n \in \cY^n}}P^{n}_{X}(x^n)\ketbrasys{x^n}{X^n}\otimes P^{n}_{Y}(y^n)\ketbrasys{y^n}{Y^n} \otimes \widehat{\rho}^{\cB^n}_{y^n}\right)\right] \\\
    &\overset{c}{\leq}  n^{\frac{\abs{\cY}\abs{\cB}(\abs{\cB}-1)}{2}}\abs{\Lambda^{n}_{\abs{\cB}}}^{\abs{\cY}}2^{-n\widehat{R}_1}\tr\left[\hat{\bbT}_{n,Y}\widehat{\rho}^{X^nY^n\cB^n}_{U,P_X,P_Y}\right]
    \leq  n^{\frac{\abs{\cX}\abs{\cB}(\abs{\cB}-1)}{2}}\abs{\Lambda^{n}_{\abs{\cB}}}^{\abs{\cX}}2^{-n\widehat{R}_1}.
\end{align*}
Here, to show $a$, we define the subgroup $S_{n,y^n} := \left\{\pi \in S_n : \pi(y^n) = (y^n)\right\}$. 
Then, since $\hat{\bbT}_{n,Y}$ is permutation invariant, 
any $\pi \in S_{n,y^n}$ satisfies the relation
\begin{align*}
    &\tr\left[\hat{\bbT}_{n,Y}\left(\sum_{\substack{ x^n\in \cX^n\\y^n \in \cY^n}}P^{n}_{X}(x^n)\ketbrasys{x^n}{X^n}\otimes P^{n}_{Y}(y^n)\ketbrasys{y^n}{Y^n} \otimes \sigma^{\cB^n}_{y^n}\right)\right]\\               
    &{=} \tr\left[\hat{\bbT}_{n,Y}\left(\sum_{\substack{ x^n\in \cX^n\\y^n \in \cY^n}}P^{n}_{X}(x^n)\ketbrasys{x^n}{X^n}\otimes P^{n}_{Y}(y^n)\ketbrasys{y^n}{Y^n} \otimes V^{\cB^n}(\pi)(\sigma^{\cB^n}_{y^n}) {V^{\cB^n}}^{\dagger}(\pi)\right)\right],
\end{align*}
which implies the equality $a$.
Further, the inequality $b$ above, follows from Lemma \ref{lemma_perm_Inv_universal} and $c$ follows from the definition of $\hat{\bbT}_{n,Y}$, given in \eqref{universal_test_Y_TS}. Using similar approaches one can derive a proof for \cref{lemma_achievability_mac_independent_eq4_TS}.

The proofs for \cref{lemma_achievability_mac_independent_eq5_TS,lemma_achievability_mac_independent_eq6_TS} follow similarly. This completes the proof of Lemma \ref{lemma_cq_mac_avht_independent_ts}.
 \endproof

\subsection{Universal Simultaneous Test for Multiple Independence Testing }\label{cq_mac_avht_null_statement_sim}
In the problem discussed in the above subsection, for each classical subsystem, we designed a test that accepts the states where the quantum sub-system is not correlated with this particular classical subsystem with arbitrarily small probability and accepts the states where the quantum subsystem is correlated with each of the classical subsystems with high probability. A drawback of this test is that it won't be able to make any decision for the states where the quantum sub-system is not correlated with that particular classical subsystem. Therefore, we call these above tests as \textit{non-simultaneous universal tests}.

We now aim to construct a test that accepts the states where the quantum subsystem is independent of any of the classical subsystems with arbitrarily low probability and accepts the state where the quantum subsystem is correlated with all the classical subsystems with high probability.
For simplicity, we assume that
$|\cS|< \infty$.

For each $s^n \in \cS^n$ and $\{\sigma^{\cB^n}_{x^n}\}, \{\sigma^{\cB^n}_{y^n}\}\subset \cD(\cH_\cB^{\otimes n})$ and $\sigma^{\cB^n} \in \cD(\cH_\cB^{\otimes n})$, we first consider the following probabilities
    \begin{align}
        &P_{0,s^n} :=  \tr\left[\bbT^{\star}_{n} \rho^{X^nY^n\cB^n}_{P_{X},P_{Y},s^n}\right]\label{bounde0}\\
        & P_{1,Y,\{\sigma^{\cB^n}_{x^n}\}} := \tr\left[\bbT^{\star}_{n} \left({\sum_{\substack{ x^n\in \cX^n\\y^n \in \cY^n}}P^{n}_{X}(x^n)\ketbrasys{x^n}{X^n}\otimes P^{n}_{Y}(y^n)\ketbrasys{y^n}{Y^n} \otimes \sigma^{\cB^n}_{x^n}}\right)\right] ,\label{bounde1Y}\\
       &P_{1,X,\{\sigma^{\cB^n}_{y^n}\}} := \tr\left[\bbT^{\star}_{n} \left(\sum_{\substack{ x^n\in \cX^n\\y^n \in \cY^n}}P^{n}_{X}(x^n)\ketbrasys{x^n}{X^n}\otimes P^{n}_{Y}(y^n)\ketbrasys{y^n}{Y^n} \otimes \sigma^{\cB^n}_{y^n}\right)\right] ,\label{bounde1X}\\
        & P_{1,\sigma^{\cB^n}} := \tr\left[\bbT^{\star}_{n} \left(\sum_{\substack{ x^n\in \cX^n\\y^n \in \cY^n}}P^{n}_{X}(x^n)\ketbrasys{x^n}{X^n}\otimes P^{n}_{Y}(y^n)\ketbrasys{y^n}{Y^n} \otimes \sigma^{\cB^n}\right)\right].\label{bounde1}
    \end{align}
    Our aim here is to design a measurement-based test $\bbT^{\star}_{n}$ in the composite Hilbert space 
    $\cH_X^{\otimes n} \otimes \cH_Y^{\otimes n} \otimes \cH_\cB^{\otimes n}$ such that for each $s^n \in \cS^n$, 
    $P_{0,s^n}$ is arbitrarily high and for each $\{\sigma^{\cB^n}_{x^n}\}, \{\sigma^{\cB^n}_{y^n}\}\subset \cD(\cH_\cB^{\otimes n})$ and $\sigma^{\cB^n} \in \cD(\cH_\cB^{\otimes n})$, the probabilities
    $P_{1,Y,\{\sigma^{\cB^n}_{x^n}\}},  P_{1,X,\{\sigma^{\cB^n}_{y^n}\}}$ and $P_{1,\sigma^{\cB^n}}$ are arbitrarily small.

    This problem is very similar to the problem discussed by Sen in \cite{Sen2021}, in the context of CQ-MAC (classical-quantum multiple-access channel). However, we note here that there is a subtle difference between the problem discussed above and the problem studied in \cite{Sen2021}. To highlight this difference we make the following remark:
    \begin{remark}\label{remark_sen}
        In \cite{Sen2021}, the author considered the singleton collections $H^{\mbox{Sen}}_0, H^{\mbox{Sen}}_{1,Y},H^{\mbox{Sen}}_{1,X},H^{\mbox{Sen}}_1$ which are defined as follows,
        \begin{align*}
            H^{\mbox{Sen}}_0 &:=  \rho^{X^nY^n\cB^n}_{P_X,P_Y},\\
            H^{\mbox{Sen}}_{1,Y} &:= \rho^{Y^n}_{P_Y} \otimes \rho^{X^n\cB^n}_{P_X},\\
            H^{\mbox{Sen}}_{1,X} &:= \rho^{X^n}_{P_X} \otimes \rho^{Y^n\cB^n}_{P_Y},\\
            H^{\mbox{Sen}}_{1} &:= \rho^{X^n}_{P_X} \otimes \rho^{Y^n}_{P_Y} \otimes \rho^{Y^n\cB^n},
        \end{align*}
    where
    \begin{align*}
        \rho^{X^nY^n\cB^n}_{P_X,P_Y} &:= \sum_{\substack{x^n\in \cX^n\\y^n \in \cY^n}}P^{n}_{X}(x^n)\ketbrasys{x^n}{X^n}\otimes P^{n}_{Y}(y^n)\ketbrasys{y^n}{Y^n}\otimes \rho^{\cB^n}_{x^n,y^n},\\
        \rho^{X^n\cB^n}_{P_Y} &:= \tr_{Y^n}\left[\rho^{X^nY^n\cB^n}_{P_X,P_Y}\right], \rho^{X^n}_{P_X} :=\tr_{Y^n\cB^n}\left[\rho^{X^nY^n\cB^n}_{P_X,P_Y}\right], \\
        \rho^{Y^n\cB^n}_{P_Y} &:= \tr_{X^n}\left[\rho^{X^nY^n\cB^n}_{P_X,P_Y}\right], \rho^{Y^n}_{P_Y} :=\tr_{X^n\cB^n}\left[\rho^{X^nY^n\cB^n}_{P_X,P_Y}\right]. 
    \end{align*}
    Here the aim is to find a hypothesis test-based measurement that accepts $H^{\mbox{Sen}}_0$ with very high probability and accepts $H^{\mbox{Sen}}_{1,X}, H^{\mbox{Sen}}_{1,Y} $ and $H^{\mbox{Sen}}_{1}$ with arbitrarily low probability. Unlike $H^{\mbox{Sen}}_0$, which has only a single state, in $H_0$,  mentioned in \eqref{GIT_hypotheses}, not only do we have more than one state, but also we have a dependency of the index $s^n$ i.e. for each $s^n \in \cS^n$, the state changes. Similarly, each of $H^{\mbox{Sen}}_{1,X}, H^{\mbox{Sen}}_{1,Y} $ and $H^{\mbox{Sen}}_{1}$ has only a single state and that too is one of the correct marginals of $\rho^{X^nY^n\cB^n}_{P_X,P_Y}$. However, the collections $H_{1,X}, H_{1,Y}$ and $ H_{1}$ respectively mentioned in \cref{Omega2states_arbit_X,Omega2states_arbit_Y,Omega2states_arbit_None}, can be arbitrarily large and may not even be the correct marginals of the states in $H_0$, mentioned in \eqref{GIT_hypotheses}. 
\end{remark}   
    Despite these subtle differences as mentioned in Remark \ref{remark_sen}, 
we solve the above problem using the projector designed in the proof of Theorem \ref{theorem_generalised_independence_test_arbitrary_varying}, along with the techniques developed in \cite{Sen2021}, by increasing the dimension of the composite Hilbert space in a controlled way and mapping the original collection of states $\left\{\rho^{X^nY^n\cB^n}_{P_{X},P_{Y},s^n}\right\}_{s^n \in \cS^n}$ and $\left\{\sigma^{X^nY^n\cB^n}_{P_{X},P_{Y}, \{\sigma^{\cB^n}_{x^n,y^n}\}} \right\}_{\substack{\{\sigma^{\cB^n}_{x^n,y^n}\} \subset \cD(\cH_\cB^{\otimes n})}}$ inside the enlarged space using an inclusion isometric map:
    
    $$\cE^{\star} : \cH_X^{\otimes n} \otimes \cH_Y^{\otimes n} \otimes \cH_\cB^{\otimes n} \xhookrightarrow{} \cH_{X'}\otimes \cH_{Y'}\otimes\cH_{\cB'},$$
because $\cH_{X'},\cH_{Y'}$ and $\cH_{\cB'}$ are enlarged Hilbert spaces containing $\cH_X^{\otimes n},\cH_Y^{\otimes n}$ and $\cH_\cB^{\otimes n}$ respectively. 
Using the isometry $\cE^{\star}(\cdot)$
and 
a projector ${\bbT}^{\star}_{n}$ in the enlarged composite Hilbert space 
$\cH_{X'}\otimes \cH_{Y'}\otimes\cH_{\cB'}$,
we define 
the probabilities below 
    \begin{align}
        \widehat{P}_{0,s^n} ({\bbT}^{\star}_{n})&:= \tr\left[ {\bbT}^{\star}_{n}  \cE^{\star}\left(\rho^{X^nY^n\cB^n}_{P_{X},P_{Y},s^n}\right)\right],\label{eh1} \\
        \widehat{P}_{1,Y,\{\sigma^{\cB^n}_{x^n}\}} ({\bbT}^{\star}_{n})&:= 
        \tr\left[{\bbT}^{\star}_{n} \left(\tr_{X'\cB'}\left[\cE^{\star}\left(\rho^{X^nY^n\cB^n}_{P_{X},P_{Y},s^n}\right)\right] \otimes \tr_{Y'}\left[\cE^{\star}\left( \sigma^{X^nY^n\cB^n}_{P_{X},P_{Y}, \{\sigma^{\cB^n}_{x^n,y^n}\}}\right)\right]\right)\right] ,\label{eh2X}\\
        \widehat{P}_{1,X,\{\sigma^{\cB^n}_{y^n}\}} ({\bbT}^{\star}_{n}) &:= \tr \left[{\bbT}^{\star}_{n}  \left(\tr_{Y'\cB'}\left[\cE^{\star}\left(\rho^{X^nY^n\cB^n}_{P_{X},P_{Y},s^n}\right)\right]  \otimes \tr_{X'}\bigg[\cE^{\star}\bigg( \sigma^{X^nY^n\cB^n}_{P_{X},P_{Y}, \{\sigma^{\cB^n}_{x^n,y^n}\}}\bigg)\bigg] \right)\right],\label{eh2Y}\\
        \widehat{P}_{1,\sigma^{\cB^n}} ({\bbT}^{\star}_{n}) &:=\tr \left[{\bbT}^{\star}_{n}  \left(\tr_{Y'\cB'}\left[\cE^{\star}\left(\rho^{X^nY^n\cB^n}_{P_{X},P_{Y},s^n}\right)\right]  \otimes \tr_{X'\cB'}\left[\cE^{\star}\left(\rho^{X^nY^n\cB^n}_{P_{X},P_{Y},s^n}\right)\right]  \otimes \tr_{X'Y'}\left[\cE^{\star}\left( \sigma^{X^nY^n\cB^n}_{P_{X},P_{Y}, \{\sigma^{\cB^n}_{x^n,y^n}\}}\right)\right]\right)\right].\label{eh2noth}
    \end{align}
Then, we have the following lemma, which is useful for generalized independent testing.

\begin{lemma}\label{theorem_sen_MAC_generalised_indep}
For $t \in (0,1)$, let the pair  ($\widehat{R}_1,\widehat{R}_2$) satisfy the following conditions
\begin{align}
    0 &< \widehat{R}_1 < \min _{Q \in \cP(\cS)} I_{1-t}[X;\cB|Y]_{P_{X},P_{Y},Q},\label{hatR_1c}\\
    0 &< \widehat{R}_2 < \min _{Q \in \cP(\cS)} I_{1-t}[Y;\cB|X]_{P_{X},P_{Y},Q},\label{hatR_2c}\\
    0 &< \widehat{R}_1 + \widehat{R}_2 < \min _{Q \in \cP(\cS)} I_{1-t}[XY;\cB]_{P_{X},P_{Y},Q},\label{hatSUMc}
\end{align}
where $I_{1-t}[Y;\cB|X]_{P_{X},P_{Y},Q}, I_{1-t}[X;\cB|Y]_{P_{X},P_{Y},Q}$ and $I_{1-t}[XY;\cB]_{P_{X},P_{Y},Q}$ are defined in \cref{sibsonqmiY;XB_1,sibsonqmiX;YB_1,sibsonqmiXY;B_1} respectively. 
Then, for any positive integer $n$, there exists a POVM ${\bbT}^{\star}_{n}$ to satisfy the following conditions.
(i) The relations
    \begin{align}
       \forall \{\sigma^{\cB^n}_{x^n}\} \subset \cD(\cH_\cB^{\otimes n}) &, 
       \widehat{P}_{1,Y,\{\sigma^{\cB^n}_{x^n}\}} ({\bbT}^{\star}_{n})\leq g_1(n,\abs{\cX},\abs{\cB})2^{-n\widehat{R}_2+1},\label{accept_Omega2X_optimal_theo}\\
        \forall \{\sigma^{\cB^n}_{y^n}\} \subset \cD(\cH_\cB^{\otimes n}) &,  
        \widehat{P}_{1,X,\{\sigma^{\cB^n}_{y^n}\}} ({\bbT}^{\star}_{n}) \leq g_2(n,\abs{\cY},\abs{\cB})2^{-n\widehat{R}_2+1},\label{accept_Omega2Y_optimal_theo}\\
        \forall \sigma^{\cB^n} \in \cD(\cH_\cB^{\otimes n}) &, 
    \widehat{P}_{1,\sigma^{\cB^n}} ({\bbT}^{\star}_{n})\leq g_3(n,\abs{\cB})2^{-n(\widehat{R}_1 + \widehat{R}_2)+1}\label{accept_Omega2XY_optimal_theo}
    \end{align}
hold, where 
  $g_1(n,\abs{\cX},\abs{\cB}), g_2(n,\abs{\cY},\abs{\cB})$ and $g_3(n,\abs{\cB})$ are defined in \cref{g_1_f,g_3_f} respectively.
(ii) In addition, for any arbitrary $\delta \in (0,1)$, there exists a sufficiently 
large $n_0$ such that the relation
    \begin{align}
        \forall s^n \in \cS^n &, \widehat{P}_{0,s^n} ({\bbT}^{\star}_{n}) \geq 1 - 76\delta \label{reject_Omega1_optimal_theo}
    \end{align}
holds for any integer $n \ge n_0$.
\end{lemma}
\begin{proof}
    See Appendix \ref{proof_theorem_sen_MAC_generalised_indep} for the proof. 
\end{proof}

We need the following lemma to prove Lemma \ref{theorem_sen_MAC_generalised_indep}. 
\begin{lemma}\label{lemma_cq_mac_avht_independent}
    There exist projectors $\hat{\bbT}_{n,X},\hat{\bbT}_{n,Y}$ and $\hat{\bbT}_{n}$ which, $\forall s^n \in \cS^n$,$\forall\{\sigma^{\cB^n}_{x^n}\}, \forall\{\sigma^{\cB^n}_{y^n}\}, $ and $\forall\sigma^{\cB^n}$ satisfy the following
\begin{align}
     \tr\left[\left( \bbI^{X^nY^n\cB^n} - \hat{\bbT}_{n,X}\right)\rho^{X^nY^n\cB^n}_{P_{X},P_{Y},s^n}\right] &\leq 
f(n,\abs{\cX},\abs{\cY}, \abs{\cB}, \abs{\cS},t)2^{nt\left(\widehat{R}_2 - \min _{Q \in \cP(\cS)} I_{1-t}[Y;\cB|X]_{P_{X},P_{Y},Q}\right)},\label{lemma_achievability_mac_independent_eq1}\\
    \tr\left[\hat{\bbT}_{n,X}\left({\rho^{\cH_Y^{\otimes n}}_{P_Y}} \otimes \sigma^{X^n\cB^n}_{P_{X}, \{\sigma^{\cB^n}_{x^n}\}}\right)\right]
&\leq g_1(n,\abs{\cX},\abs{\cB})2^{-n\widehat{R}_2},
\label{lemma_achievability_mac_independent_eq2}\\
 \tr\left[\left( \bbI^{X^nY^n\cB^n} - \hat{\bbT}_{n,Y}\right)\rho^{X^nY^n\cB^n}_{P_{X},P_{Y},s^n}\right] &\leq 
f(n,\abs{\cX},\abs{\cY}, \abs{\cB}, \abs{\cS},t) 2^{nt\left(\widehat{R}_1 -  \min _{Q \in \cP(\cS)}I_{1-t}[X;\cB|Y]_{P_{X},P_{Y},Q}\right)},\label{lemma_achievability_mac_independent_eq3}\\
    \tr\left[\hat{\bbT}_{n,Y}\left({\rho^{\cH_X^{\otimes n}}_{P_X}} \otimes\sigma^{Y^n\cB^n}_{P_{Y}, \{\sigma^{\cB^n}_{y^n}\}}\right)\right]
&\leq g_2(n,\abs{\cY},\abs{\cB})2^{-n\widehat{R}_1}, \label{lemma_achievability_mac_independent_eq4}\\
 \tr\left[\left( \bbI^{X^nY^n\cB^n} - \hat{\bbT}_{n}\right)\rho^{X^nY^n\cB^n}_{P_{X},P_{Y},s^n}\right] &\leq 
f(n,\abs{\cX},\abs{\cY}, \abs{\cB}, \abs{\cS},t) 2^{nt\left(\widehat{R}_1 + \widehat{R}_2 -  \min _{Q \in \cP(\cS)} I_{1-t}[XY;\cB]_{P_{X},P_{Y},Q}\right)},\label{lemma_achievability_mac_independent_eq5}\\
    \tr\left[\hat{\bbT}_{n}\left({\rho^{\cH_X^{\otimes n}}_{P_X}} \otimes{\rho^{\cH_Y^{\otimes n}}_{P_Y}} \otimes \sigma^{\cB^n}\right)\right]
&\leq g_3(n,\abs{\cB})2^{-n(\widehat{R}_1 + \widehat{R}_2)},\label{lemma_achievability_mac_independent_eq6}
    \end{align}
where $\widehat{R}_1,\widehat{R}_2 > 0$, $t \in (0,1)$, $f(n,\abs{\cX},\abs{\cY}, \abs{\cB}, \abs{\cS},t)$, $g_1(n,\abs{\cX},\abs{\cB}), g_2(n,\abs{\cY},\abs{\cB})$, $g_3(n,\abs{\cB})$ are defined in \cref{g_1_f,g_3_f} respectively and the states $\sigma^{X^n\cB^n}_{P_{X}, \{\sigma^{\cB^n}_{x^n}\}}, \sigma^{Y^n\cB^n}_{P_{Y}, \{\sigma^{\cB^n}_{y^n}\}}, \sigma^{\cB^n}$ are the respective marginals of the state mentioned in \eqref{Omega2states_arbit}.
\end{lemma}

 \begin{proof}
     See Appendix \ref{proof_lemma_cq_mac_avht_independent} for the proof.
 \end{proof}
\begin{remark}
    Note that for the case of general $\cS$, if $\left\{\rho_{x,y,s}^{\cB} : x \in \cX, y \in \cY, s \in \cS\right\}$ satisfies Assumption \ref{NM9}, Lemma \ref{theorem_sen_MAC_generalised_indep} can be directly extended for the case of general $\cS$, using \eqref{epsilon_net_state_density_eq}. However, we have omitted this calculation for simplicity.
\end{remark}

We now mention a corollary of Lemma \ref{theorem_sen_MAC_generalised_indep}, which will be required to give an alternative proof (see subsection \ref{alternative_proof_lemma_rand_capacity_avmac}) of achievability for deriving the capacity of Classical-Quantum Arbitrarily Varying Channels (CQ-AVMAC) under randomization of codes. 

\begin{corollary}\label{theo_gen_indep_mac_corollary}
    Consider the collection $\left\{\rho^{X^nY^n\cB^n}_{P_{X},P_{Y},s^n}\right\}_{s^n \in \cS^n}$, where for each $s^n \in \cS^n, \rho^{X^nY^n\cB^n}_{P_{X},P_{Y},s^n}$ is defined in \eqref{Omega1states_sn}. Then, for any arbitrary $\delta \in (0,1)$ and large enough $n$, the POVM ${\bbT}^{\star}_{n}$ mentioned in \eqref{tilted_augmented_intersection_test} satisfies the following:
\begin{align}
    \tr\left[\left( \bbI^{X'Y'\cB'} - {\bbT}^{\star}_{n}\right)\rho^{X'Y'\cB'}_{P_{X},P_{Y},s^n}\right] &\leq 76\delta,\label{reject_Omega1_optimal_corollary}\\
    \tr\left[{\bbT}^{\star}_{n}\left(\rho^{Y'}_{P_{Y}} \otimes \rho^{X'\cB'}_{P_{X},s^n} \right)\right] &\leq g_1(n,\abs{\cX},\abs{\cB})2^{-n\widehat{R}_2+1},\label{accept_Omega2X_optimal_corollary}\\
    \tr\left[{\bbT}^{\star}_{n}\left(\rho^{X'}_{P_{X}} \otimes \rho^{Y'\cB'}_{P_{Y},s^n} \right)\right] &\leq g_2(n,\abs{\cY},\abs{\cB})2^{-n\widehat{R}_1+1},\label{accept_Omega2Y_optimal_corollary}\\
\tr\left[{\bbT}^{\star}_{n}\left(\rho^{X'}_{P_{X}} \otimes\rho^{Y'}_{P_{Y}} \otimes \rho^{\cB'}_{s^n} \right)\right] &\leq g_3(n,\abs{\cB})2^{-n(\widehat{R}_1+\widehat{R}_2) +1},\label{accept_Omega2XY_optimal_corollary}
\end{align}
where $(\widehat{R}_1,\widehat{R}_2)$ satisfies \cref{hatR_1c,hatR_2c,hatSUMc}, and for each $s^n \in \cS^n$, we define the states as $\rho^{X'\cB'}_{P_{X},s^n}, \rho^{Y'\cB'}_{P_{Y},s^n}$ and $\rho^{\cB'}_{s^n}$ as follows,
\begin{align*}
    \rho^{X'\cB'}_{P_{X}, s^n} &:= \tr_{Y'} \left[\rho^{X^nY^n\cB^n}_{P_{X},P_{Y},s^n}\right] = \frac{1}{\abs{\cA}}\sum_{\substack{x^n\in \cX^n , r \in [\abs{\cA}]}} P^{n}_{X}(x^n) \ketbrasys{x^n,a_{r}}{X'}\otimes   \rho^{\cB'}_{(x^n,a_{r}),s^n,\delta},\\
    \rho^{Y'\cB'}_{P_{Y}, s^n} &:= \tr_{X'} \left[\rho^{X^nY^n\cB^n}_{P_{X},P_{Y},s^n}\right] = \frac{1}{\abs{\cA}}\sum_{\substack{y^n\in \cY^n ,t \in [\abs{\cA}]}} P^{n}_{Y}(y^n) \ketbrasys{y^n,b_{t}}{Y'}\otimes  \rho^{\cB'}_{(y^n,b_{t}),s^n,\delta},\\
\rho^{\cB'}_{s^n} &:= \tr_{X'Y'} \left[\rho^{X^nY^n\cB^n}_{P_{X},P_{Y},s^n}\right] = \frac{1}{\abs{\cA}^2}\sum_{\substack{x^n \in \cX^n ,y^n\in \cY^n \\ r \in [\abs{\cA}], t \in [\abs{\cA}]}} P^{n}_{X}(x^n)P^{n}_{Y}(y^n) \rho^{\cB'}_{(x^n,a_{r}),(y^n,b_{t}),s^n\delta},
\end{align*}
where for each $x^n \in \cX^n, r \in [\abs{\cA}]$ and $y^n \in \cY^n, t \in [\abs{\cA}]$ we have,
\begin{align*}
    \rho^{\cB'}_{(x^n,a_{r}),s^n,\delta} &:= \frac{1}{\abs{\cA}}\sum_{y^n, t \in [\abs{\cA}]}P^{n}_{Y}(y^n) \rho^{\cB'}_{(x^n,a_{r}),(y^n,b_{t}),s^n,\delta}, \\
    \rho^{\cB'}_{(y^n,b_{t}),\delta} &:= \frac{1}{\abs{\cA}}\sum_{x^n, r \in [\abs{\cA}]}P^{n}_{X}(x^n) \rho^{\cB'}_{(x^n,a_{r}),(y^n,b_{t}),s^n,\delta},
\end{align*}
and $\rho^{\cB'}_{(x^n,a_{r}),(y^n,b_{t}),s^n,\delta}$ is defined in \eqref{tilted_rho_inner}.
\end{corollary}
\begin{proof}
    See Appendix \ref{proof_theo_gen_indep_mac_corollary} for the proof.
\end{proof}

\begin{remark}\label{remark_simultanous}
    
    In the above method of universal testing, we have a single test that accepts the states where the quantum subsystem is correlated with all the classical subsystems with high probability and rejects the states where the quantum sub-system is not correlated with any of the classical subsystems \textbf{simultanously} with arbitrarily high probability. This is the reason this test is called a \textbf{simultaneous universal test}. 
\end{remark}

 \section{Reliable communication over Classical Quantum Arbitrarily Varying Multiple Access Channel (CQ-AVMAC)}\label{section_CQAVMAC}
\subsection{Case with $2$ senders}
First, we discuss classical quantum arbitrarily varying multiple access channels (CQ-AVMAC) for the case with $2$ senders.
This case is modeled as follows.
\begin{definition}\label{CQ_AVMAC} (channel) We model a two-sender classical-quantum arbitrarily varying multiple-access channel (CQ-AVMAC) between parties
Alice (sender $1$), Bob (sender $2$) and Charlie (receiver) as a map
$$\cN^{X^nY^n \to B^n}_{s^n} : (x^n,y^n) \to \rho^{B^n}_{x^n,y^n,s^n},$$
where $x^n \in \cX^n, y^n \in \cY^n$ and $s^n \in \cS^n$ (where $\abs{\cX},\abs{\cY}).$ For each $x^n := (x_1\cdots,x_n) \in \cX^n, y^n := (y_1\cdots,y_n) \in \cY^n$ and $s^n := (s_1\cdots,s_n) \in \cS^n$, 
we define the state $\rho^{B^n}_{x^n,y^n,s^n}$ to be $\bigotimes_{i=1}^{n}\rho^{B}_{x_i,y_i,s_i}$.
Further, Alice, Bob (the senders), and Charlie (the receiver) have no information about the sequence $s^n.$
\end{definition}
We aim to use this channel and enable Alice to transmit a message $m_1 \in [2^{nR_1}]$  and Bob to transmit a message $m_2 \in [2^{nR_2}]$, such that Charlie can recover $(m_1,m_2)$ with high probability for every $s^n$.
\begin{definition}(Deterministic code)\label{definition_cq_avmac_code}
An $(n,2^{nR_1},2^{nR_2})$-code $\cC$ for communication over a CQ-AVMAC consists of 
\begin{itemize}
\item two encoding functions $\cE^{(n)}_1: \cM^{1}_n  \to \cX^n$ and $\cE^{(n)}_2: \cM^{2}_n  \to \cY^n$ where $\cM^{1}_n:= [2^{nR_1}]$ and $\cM^{2}_n := [2^{nR_2}]$
\item a decoding POVM $\{ \cD_{m_1,m_2}: (m_1,m_2) \in \cM^{1}_n \times \cM^{2}_n\}$.   \end{itemize}
An $(n,2^{nR_1},2^{nR_2})$-code $\cC$ for communication over a CQ-AVMAC is called 
an $(n,2^{nR_1},2^{nR_2},\beta)$-code when for any $s^n \in \cS^n$, 
the average probability $\bar{e}(\cC,s^n)$ of error
satisfies 
\begin{equation*}
      \bar{e}(\cC,s^n):= \frac{1}{2^{nR_1}}\frac{1}{2^{nR_2}}\sum_{m_1=1}^{2^{nR_1}}\sum_{m_2=1}^{2^{nR_2}}e(m_1,m_2,\cC,s^n) < \beta,
  \end{equation*}
  where $e(m_1,m_2,\cC,s^n):= 1 - \tr[D_{m_1,m_2} \cN^{X^n \to B^n}_{s^n}(\cE^{(n)}_1(m_1), \cE^{(n)}_2(m_2))].$
  \end{definition}

\begin{definition}\label{def2_avmac}
    A pair $\left(R_1,R_2\right)$ is called an \textit{achievable deterministic rate-pair} for a given CQ-AVMAC, if for any $\beta>0,\delta_1>0 \text{ and } \delta_2>0$ and sufficiently large $n$, there exists a $\left(n,2^{nR_1},2^{nR_2},\beta\right)$ deterministic code $\cC$ such that 
    \begin{equation*}
        \begin{split}
            \frac{1}{n}\log |\cM_n^1| &> R_1 - \delta_1,\\
            \frac{1}{n}\log |\cM_n^2| &> R_2 - \delta_2,
        \end{split}
        \begin{split}
            \hspace{20pt}\sup_{s^n \in \cS^n}\bar{e}(\cC,s^n)<\beta.
        \end{split}
    \end{equation*}
    
    The \textit{achievable deterministic code rate region} $\cR_{d}$ (also known as \textit{deterministic code capacity region}) is defined to be the closure of all possible achievable deterministic rate pairs.
\end{definition}

\begin{definition}(Random code)\label{definition_cq_avmac_rand_code}
An $(n,2^{nR_1},2^{nR_2})$ random code $\cC$ for communication over a CQ-AVMAC consists of 
\begin{itemize}
\item two encoding functions $\cE^{(n),\gamma_1}_1: \cM^{1}_n  \to \cX^n$ and $\cE^{(n),\gamma_2}_2: \cM^{2}_n  \to \cY^n$, parameterized by two random variables $\gamma_1$, $\gamma_2$ respectively, where $\cM^{1}_n:= [2^{nR_1}]$, $\cM^{2}_n := [2^{nR_2}]$,
\item a decoding POVM $\{ \cD^{\gamma_1,\gamma_2}_{m_1,m_2}: (m_1,m_2) \in \cM^{1}_n \times \cM^{2}_n\}$, parameterized by $\gamma_1$, $\gamma_2$.  \end{itemize}
In the above,  $(\gamma_1,\gamma_2) \sim G_1 \times G_2$ are two independent random variables taking values over two finite sets. 

An $(n,2^{nR_1},2^{nR_2})$ random code $\cC$ for communication over a CQ-AVMAC is called 
an $(n,2^{nR_1},2^{nR_2},\beta)$-code when for any $s^n \in \cS^n$, 
the average probability $\bbE_{\cC}\left[\bar{e}(\cC,s^n)\right]$ of error under the expectation of choices of random code (choices of $\gamma_1$ and $\gamma_2$) satisfies
\begin{equation*}
      \bbE_{\cC}\left[\bar{e}(\cC,s^n)\right]:= \frac{1}{2^{nR_1}}\frac{1}{2^{nR_2}}\sum_{m_1=1}^{2^{nR_1}}\sum_{m_2=1}^{2^{nR_2}}\bbE_{\cC}\left[e(m_1,m_2,\cC,s^n)\right] < \beta,
  \end{equation*}
  where $\bbE_{\cC}\left[e(m_1,m_2,\cC,s^n)\right] := 1 - \bbE_{\substack{\gamma_1 \sim G_1\\\gamma_2 \sim G_2}}\left[\tr[D^{\gamma_1,\gamma_2}_{m_1,m_2} \cN^{X^n \to B^n}_{s^n}(\cE^{(n),\gamma_1}_1(m_1), \cE^{(n),\gamma_2}_2(m_2))]\right].$
  \end{definition}

When 
$\gamma_1$ is shared between Alice (sender $1$) and Charlie (the receiver) and $\gamma_2$ is shared between Bob (sender $2$) and Charlie (the receiver),
the above random code $\cC$
can be realized.

\begin{definition}\label{def2_avmac_rand}
    A pair $\left(R_1,R_2\right)$ is called an \textbf{achievable randomized rate-pair} for a given CQ-AVMAC, if for any $\beta>0,\delta_1>0 \text{ and } \delta_2>0$ and sufficiently large $n$, there exists a $\left(n,2^{nR_1},2^{nR_2},\beta\right)$ random code $\cC$ such that 
    \begin{equation*}
        \begin{split}
            \frac{1}{n}\log |\cM_n^1| &> R_1 - \delta_1,\\
            \frac{1}{n}\log |\cM_n^2| &> R_2 - \delta_2,
        \end{split}
        \begin{split}
            \hspace{20pt}\sup_{s^n \in \cS^n}\bbE_{\cC}[\bar{e}(\cC,s^n)]<\beta.
        \end{split}
    \end{equation*}
    
    The \textit{achievable random code rate region} $\cR_{r}$ (also known as \textit{random code capacity region}) is defined to be the closure of all possible achievable randomized rate-pairs.
\end{definition}
 
 Since the input systems $X, Y$ are classical we will consider them as a pair of random variables with support sets $\cX$ and $\cY$. Given $X \sim P_{X}$ and $Y \sim P_{Y}$ with a joint probability distribution $P_{XY} := P_{X}\cdot P_{Y}$, we define the following set
\begin{align}
     \cR_{P_{X},P_{Y}} := \left\{\forall (R_1,R_2) :
        \begin{array}{cc}
              0\leq R_1 &\hspace{-10pt}\leq \inf_{Q \in \cP(\cS)} I[X;\cB|Y]_{P_{X},P_{Y},Q}  \\
              0\leq R_2 &\hspace{-10pt}\leq \inf_{Q \in \cP(\cS)} I[Y;\cB|X]_{P_{X},P_{Y},Q}  \\
              R_1 + R_2 &\hspace{-10pt}\leq \inf_{Q \in \cP(\cS)} I[XY;\cB]_{P_{X},P_{Y},Q}
        \end{array}
    \right\},\label{Rpxpy}
 \end{align}
 where $I[X;\cB|Y]_{P_{X},P_{Y},Q},  I[Y;\cB|X]_{P_{X},P_{Y},Q}$ and $I[XY;\cB]_{P_{X},P_{Y},Q}$ are defined in \cref{ix;b|y,iy;b|x,ixy;b}.
 Further, we define $\cR^{\star}$ as follows,
\begin{equation}
    \cR^{\star} := conv\left(\bigcup_{\substack{P_{X} \in \cP(\cX) \\ P_{Y} \in \cP(\cY)}}\cR_{P_{X},P_{Y}}\right).\label{Rstar}
\end{equation}

In the theorem below we establish the random code capacity region $\cR_r$ by giving two achievability proofs and a matching converse. The first achievability proof is based on the non-simultaneous test designed in Lemma \ref{lemma_cq_mac_avht_independent_ts} in subsection 
\ref{cq_mac_avht_null_statement_non_sim}. The second achievability proof is based on the simultaneous test designed in Lemma \ref{lemma_cq_mac_avht_independent} in subsection 
\ref{cq_mac_avht_null_statement_sim}. The detailed proof for this alternate approach is given in Appendix \ref{alternative_proof_lemma_rand_capacity_avmac}. 
For both these proofs we obtain constraints on the pair $(R_1,R_2)$, by showing that for every $s^n \in \cS^n$, the average probability of error averaged over the random choice of code goes to zero if these constraints are met. 
\begin{theorem}\label{lemma_rand_capacity_avmac} 
Assume that
the subset $\left\{\rho_{x,y,s}^{\cB} : x \in \cX, y \in \cY, s \in \cS\right\}  \subset \cD(\cB)$ satisfies Assumption \ref{NM9}.
The random code capacity region $\cR_{r}$ of a CQ-AVMAC is equal to $\cR^{\star}$ (See \eqref{Rstar} for the definition).
\end{theorem}
 This theorem is shown in subsections \ref{S4-B} and \ref{S4-C}. In subsection \ref{S4-B}, we show the proof of the achievability part of Theorem \ref{lemma_rand_capacity_avmac} by designing a joint decoding strategy based on the non-simultaneous test designed in Lemma \ref{lemma_cq_mac_avht_independent_ts} in subsection 
 \ref{cq_mac_avht_null_statement_non_sim}.

\subsection{Achievability part of Theorem \ref{lemma_rand_capacity_avmac}}\label{S4-BY}
Since 
$ \cR^{\star}$ is given as convex full of 
$\cR_{P_{X},P_{Y}}$ with choices of $P_{X},P_{Y}$,
once 
every rate-pair $(R_1,R_2) \in \cR_{P_{X},P_{Y}}$
is achieved with any pair $P_X,P_Y$, 
every rate-pair $(R_1,R_2) \in \cR^{\star}$ is achieved by applying time sharing.
Hence, given distributions $P_X,P_Y $, 
we here show that for every rate-pair $(R_1,R_2) \in \cR_{P_{X},P_{Y}}$ there exists a $(n,2^{nR_1},2^{nR_2},\eps_n)$ random code $\cC$ for which the average error probability, averaged over the choice of code $\cC$ satisfies the following:
\begin{align}
\lim_{n \to \infty}
\sup_{s^n \in \cS^n} \bbE_{\cC}\left[\bar{e}(\cC,s^n)\right] =0 \label{BVQ}.
\end{align}
    We now start with the randomized construction for $\cC$.
  
    \subsubsection{Randomized Encoder Construction}
We choose the rate pair $(R_1,R_2)$ satisfying
   \begin{align*}
    R_1 &< \inf_{Q \in \cP(\cS)} I[X;\cB|Y]_{P_X,P_Y,Q},\\
    R_2 &< \inf_{Q \in \cP(\cS)} I[Y;\cB|X]_{P_X,P_Y,Q},\\
   R_1+ R_2 &< \inf_{Q \in \cP(\cS)} I[XY;\cB]_{P_X,P_Y,Q}.
\end{align*}
Then, we set the pair $( \widehat{R}_1,\widehat{R}_2)$ as
\begin{align}
R_1<    \widehat{R}_1 &< \inf_{Q \in \cP(\cS)} I[X;\cB|Y]_{P_{X},P_{Y},Q},\label{XBHI1}\\
R_2<    \widehat{R}_2 &< \inf_{Q \in \cP(\cS)} I[Y;\cB|X]_{P_{X},P_{Y},Q},\label{XBHI2}\\
 R_1+R_2<   \widehat{R}_1 + \widehat{R}_2 &< \inf_{Q \in \cP(\cS)} I[XY;\cB]_{P_{X},P_{Y},Q}
 \label{XBHI3}.
\end{align}
Then, we choose sufficiently small real numbers 
$t,\varepsilon \in (0,1)$ such that
    \begin{align}
   \widehat{R}_1 &< \min _{Q \in \cP(\cS_\eps)} I_{1-t}[X;\cB|Y]_{P_{X},P_{Y},Q}
   -\frac{1}{t}\log (1+\varepsilon),
   \label{Xtheo_gen_sen_macR_1}\\
    \widehat{R}_2 &< \min _{Q \in \cP(\cS_\eps)} I_{1-t}[Y;\cB|X]_{P_{X},P_{Y},Q}-\frac{1}{t}\log (1+\varepsilon),\label{Xtheo_gen_sen_macR_2}\\
   \widehat{R}_1 + \widehat{R}_2 &< \min _{Q \in \cP(\cS_\eps)} I_{1-t}[XY;\cB]_{P_{X},P_{Y},Q}-\frac{1}{t}\log (1+\varepsilon)\label{Xtheo_gen_sen_macR_1R_2},
\end{align}
where $\cS_\eps$ is chosen in \eqref{epsilon_net_state_density_eq}.

    Given a distribution $P_{X}$ on $\cX$, 
    Alice randomly generates $2^{nR_1}$ sequences $\{X^n(m_1) \in \cX^n : m_1 \in [2^{nR_1}]\}$ according to a probability distribution $P_{X}^{ n}$.
    Similarly, given a distribution $P_{Y}$ on $\cY$, 
    Bob randomly generates $2^{nR_2}$ sequences $\{Y^n(m_2) \in \cY^n : m_2 \in [2^{nR_2}]\}$ according to a probability distribution $P_{Y}^{ n}$.
    That is, for any element $m_1 \in [2^{nR_1}]$, 
the random variables $X^n(m_1)$ and $T^n(m_2)$ obeys the distributions 
 $\Pr\left\{X^n(m_1)\right\} := \prod_{i=1}^{n}P_{X}(X_i(m_1))$ and $\Pr\left\{Y^n(m_2)\right\} := \prod_{i=1}^{n}P_{Y}(Y_i(m_2))$, where $X^n(m_1):=\left(X_1(m_1),\cdots,X_n(m_1)\right)$ and $Y^n(m_2):=\bigl(Y_1(m_2),\cdots,Y_n(m_2)\bigr)$.

    That is, for any element $\forall m \in \cM_n$,
the random variable $X^n(m)$ obeys the distribution 
$\Pr\bigl\{X^n(m)\bigr\} := \prod_{i=1}^{n}P(X_i(m))$, where $X^n(m):=\left(X_1(m),\cdots,X_n(m)\right)$. 
Then, we randomly choose the encoders $\cE^{(n)}_1$ and $\cE^{(n)}_2$ subject to the following distribution;
\begin{align*}
   & \Pr( \cE^{(n)}_1(m_1)=X^n(m_1), \cE^{(n)}_2(m_2)=Y^n(m_2),  \hbox{ for } m_1=1, \ldots,  2^{nR_1} , m_2=1,\ldots, 2^{nR_2}) \\
 &:= \prod_{m_1=1}^{\cM^{1}_{n}}\Pr\left\{X^n(m_1)\right\}\cdot\prod_{m_2=1}^{\cM^{2}_{n}}\Pr\left\{Y^n(m_2)\right\}.
\end{align*}
    To send a message $m_1 \in \cM^{1}_{n}$, Alice encodes it to input sequence $X^n(m_1)$ and transmits it  over the channel. Similarly, to end a message $m_2 \in \cM^{2}_{n}$, Bob encodes it to input sequence $Y^n(m_2)$ and transmits it over the channel.

    \subsubsection{Decoding Strategy (Universal Decoder)}
    Upon receiving the state $\rho_{\substack{X^n,Y^n, s^n}}^{\cB^n}$, Charlie tries to guess the message pair sent simultaneously by measuring the state. We now define the measurement which will be used by Charlie. From \eqref{tests_hayashi_generalised}, we consider $\hat{\bbT}_{n}^{X}$ and $\hat{\bbT}_{n}^{Y}$ to be in the following form:

    \begin{align*}
        \hat{\bbT}_{n}^{X} &:= \sum_{\substack{x^n \in \cX^n \\ y^n \in \cY^n}} \ketbrasys{x^n}{X^n} \otimes \ketbrasys{y^n}{Y^n} \otimes \hat{\bbT}_{n,x^n,y^n}^{X},\\
        \hat{\bbT}_{n}^{Y} &:= \sum_{\substack{x^n \in \cX^n \\ y^n \in \cY^n}} \ketbrasys{x^n}{X^n} \otimes \ketbrasys{y^n}{Y^n} \otimes \hat{\bbT}_{n,x^n,y^n}^{Y}.
    \end{align*}
    
    For each message pair $(m_1,m_2) \in \cM^{1}_n \times  \cM^{2}_n$, we construct two separate decoders (each for one of messages) as follows, 
    \begin{align*}
        \Lambda_{m_1}^{X} &:= \left(\sum_{\hat{m}_1 \in \cM^{1}_n} \Pi^{X}_{\hat{m}_1}\right)^{-\frac{1}{2}}\Pi^{X}_{m_1} \left(\sum_{\hat{m}_1 \in \cM^{1}_n} \Pi^{X}_{\hat{m}_1}\right)^{-\frac{1}{2}},\\
        \Lambda_{m_2}^{Y} &:= \left(\sum_{\hat{m}_2 \in \cM^{2}_n} \Pi^{Y}_{\hat{m}_2}\right)^{-\frac{1}{2}}\Pi^{Y}_{m_2} \left(\sum_{\hat{m}_2 \in \cM^{2}_n} \Pi^{Y}_{\hat{m}_2}\right)^{-\frac{1}{2}},
    \end{align*}
where $\Pi^{X}_{\hat{m}_1}$ and $ \Pi^{Y}_{\hat{m}_2}$ are defined as follows
\begin{equation*}
    \Pi^{X}_{\hat{m}_1} := \sum_{\hat{m}_2 \in  \cM^{2}_n} \hat{\bbT}_{n,X^n(\hat{m}_1),Y^n(\hat{m}_2)}^{X}, \quad  \Pi^{Y}_{\hat{m}_2} := \sum_{\hat{m}_1 \in  \cM^{1}_n} \hat{\bbT}_{n,X^n(\hat{m}_1),Y^n(\hat{m}_2)}^{Y}.
\end{equation*}
    
    Charlie then uses the POVM $\left\{\Lambda_{m_1,m_2}^{XY}\right\}_{m_1,m_2}$ to decode the message pair \emph{jointly}, where $\Lambda_{m_1,m_2}^{XY}$ is defined as follows
\begin{equation}
    \Lambda_{m_1,m_2}^{XY}:= \left(\Lambda_{m_2}^{Y}\right)^{\frac{1}{2}}\left(\Lambda_{m_1}^{X}\right)\left(\Lambda_{m_2}^{Y}\right)^{\frac{1}{2}},\label{decoder_h}
\end{equation}
where
\begin{equation*}
    \sum_{m_1 \in \cM^{1}_n}\sum_{m_2 \in \cM^{2}_n} \Lambda_{m_1,m_2}^{XY} = \bbI^{\cB}.
\end{equation*}

Observe that the above POVM is universal as $\hat{\bbT}_{n,x^n,y^n}^{X}$ and $\hat{\bbT}_{n,x^n,y^n}^{Y}$ are independent of the channel. In the following, we denote the triplet $(\cE^{(n)}_1,\cE^{(n)}_2,\left\{\Lambda_{m_1,m_2}^{XY}\right\}_{m_1,m_2})$ by $\cC$.


    \subsubsection{Error Analysis}\label{S7-B-3}
    We now do the error analysis as follows. For every $s^n \in \cS^n$, the following holds.
    \begin{align}
        &\hspace{10pt} \bar{e}(\cC,s^n) 
        = \frac{1}{2^{n(R_1+R_2)}}\sum_{\substack{m_1\in \cM^{1}_n \\ m_2\in \cM^{2}_n}} e(m_1,m_2,\cC,s^n)\\
        &= \frac{1}{2^{n(R_1+R_2)}}\sum_{\substack{m_1\in \cM^{1}_n \\ m_2\in \cM^{2}_n}} e(m_1,m_2,\cC,s^n)\nn\\
        &= \frac{1}{2^{n(R_1+R_2)}}\sum_{\substack{m_1\in \cM^{1}_n \\ m_2\in \cM^{2}_n}} \tr\left[\left(\bbI^{\cB}-\Lambda_{m_1,m_2}^{XY}\right)\rho_{\substack{X^n(m_1),Y^n(m_2), s^n}}^{\cB}\right]\nn\nn\\
         &= \frac{1}{2^{n(R_1+R_2)}}\sum_{\substack{m_1\in \cM^{1}_n \\ m_2\in \cM^{2}_n}} \tr\left[\left(\bbI^{\cB}- \left(\Lambda_{m_2}^{Y}\right)^{\frac{1}{2}}\left(\Lambda_{m_1}^{X}\right)\left(\Lambda_{m_2}^{Y}\right)^{\frac{1}{2}}\right)\rho_{\substack{X^n(m_1),Y^n(m_2), s^n}}^{\cB}\right]\nn\\
         &= \frac{1}{2^{n(R_1+R_2)}}\sum_{\substack{m_1\in \cM^{1}_n \\ m_2\in \cM^{2}_n}} \tr\left[\left(\bbI^{\cB}- \left(\Lambda_{m_2}^{Y}\right)\right)\rho_{\substack{X^n(m_1),Y^n(m_2), s^n}}^{\cB}\right] \nn\\
         &\hspace{10pt}+ \frac{1}{2^{n(R_1+R_2)}}\sum_{\substack{m_1\in \cM^{1}_n \\ m_2\in \cM^{2}_n}} \tr\left[\left( \left(\Lambda_{m_2}^{Y}\right) - \left(\Lambda_{m_2}^{Y}\right)^{\frac{1}{2}}\left(\Lambda_{m_1}^{X}\right)\left(\Lambda_{m_2}^{Y}\right)^{\frac{1}{2}}\right)\rho_{\substack{X^n(m_1),Y^n(m_2), s^n}}^{\cB}\right]\nn\\
         &\overset{a}{=} \bar{e}^{Y}(\cC,s^n) + \frac{1}{2^{n(R_1+R_2)}}\sum_{\substack{m_1\in \cM^{1}_n \\ m_2\in \cM^{2}_n}} \tr\left[\left( \left(\Lambda_{m_2}^{Y}\right)^{\frac{1}{2}}\left(\bbI^{\cB} - \Lambda_{m_1}^{X}\right)\left(\Lambda_{m_2}^{Y}\right)^{\frac{1}{2}}\right)\rho_{\substack{X^n(m_1),Y^n(m_2), s^n}}^{\cB}\right]\nn\\
         &\overset{b}{\leq} \bar{e}^{Y}(\cC,s^n) + \frac{1}{2^{n(R_1+R_2)}}\sum_{\substack{m_1\in \cM^{1}_n \\ m_2\in \cM^{2}_n}} \tr\left[\left(\bbI^{\cB} - \Lambda_{m_1}^{X}\right)\rho_{\substack{X^n(m_1),Y^n(m_2), s^n}}^{\cB}\right]\nn\\
         &\hspace{10pt}+ \frac{1}{2^{n(R_1+R_2)}}\sum_{\substack{m_1\in \cM^{1}_n \\ m_2\in \cM^{2}_n}} \norm{ \left(\Lambda_{m_2}^{Y}\right)^{\frac{1}{2}}\rho_{\substack{X^n(m_1),Y^n(m_2), s^n}}^{\cB}\left(\Lambda_{m_2}^{Y}\right)^{\frac{1}{2}} - \rho_{\substack{X^n(m_1),Y^n(m_2), s^n}}^{\cB}}{1} \nn\\
         &\overset{c}{\leq} \bar{e}^{Y}(\cC,s^n) 
         + \bar{e}^{X}(\cC,s^n) + \frac{1}{2^{n(R_1+R_2)}}\sum_{\substack{m_1\in \cM^{1}_n \\ m_2\in \cM^{2}_n}} 2 
         \sqrt{\tr\left[\left(\bbI^{\cB}- \left(\Lambda_{m_2}^{Y}\right)\right)\rho_{\substack{X^n(m_1),Y^n(m_2), s^n}}^{\cB}\right]}\nn\\
         &\leq \bar{e}^{Y}(\cC,s^n) 
         + \bar{e}^{X}(\cC,s^n) + \sqrt{\frac{1}{2^{n(R_1+R_2)}}\sum_{\substack{m_1\in \cM^{1}_n \\ m_2\in \cM^{2}_n}} 2 
         \tr\left[\left(\bbI^{\cB}- \left(\Lambda_{m_2}^{Y}\right)\right)\rho_{\substack{X^n(m_1),Y^n(m_2), s^n}}^{\cB}\right]}\nn\\
         &= \bar{e}^{Y}(\cC,s^n)
         + \bar{e}^{X}(\cC,s^n) + 2\sqrt{\bar{e}^{Y}(\cC,s^n)},    
         \end{align}
    where in equality $a$, we define 
    $\bar{e}^{Y}(\cC,s^n) := \frac{1}{2^{n(R_1+R_2)}}\sum_{\substack{m_1\in \cM^{1}_n \\ m_2\in \cM^{2}_n}} \tr\left[\left(\bbI^{\cB}- \left(\Lambda_{m_2}^{Y}\right)\right)\rho_{\substack{X^n(m_1),Y^n(m_2), s^n}}^{\cB}\right],$ 
    $b$ follows from Fact \ref{trace_norm2} and in $c$ we define 
    $\bar{e}^{X}(\cC,s^n) := \frac{1}{2^{n(R_1+R_2)}}\sum_{\substack{m_1\in \cM^{1}_n \\ m_2\in \cM^{2}_n}} \tr\left[\left(\bbI^{\cB}- \left(\Lambda_{m_1}^{X}\right)\right)\rho_{\substack{X^n(m_1),Y^n(m_2), s^n}}^{\cB}\right],$ and the inequality $c$ follows from Fact \ref{gent_measurement}. 
    Taking the average with respect to $\bbE_{\cC}$, we have
        \begin{align}
        \bbE_{\cC}\left[\bar{e}(\cC,s^n)\right]
         \le &\bbE_{\cC}\left[\bar{e}^{Y}(\cC,s^n)\right] + \bbE_{\cC}\left[\bar{e}^{X}(\cC,s^n)\right] + 2\bbE_{\cC}\left[\sqrt{\bar{e}^{Y}(\cC,s^n)}\right] \nn\\
         \le &\bbE_{\cC}\left[\bar{e}^{Y}(\cC,s^n)\right] + \bbE_{\cC}\left[\bar{e}^{X}(\cC,s^n)\right] + 2\sqrt{\bbE_{\cC}\left[\bar{e}^{Y}(\cC,s^n)\right]}\label{joint_decoding_hayashi},
    \end{align}
    where we employ Jensen inequality.
    
    We now upper-bound the quantity $\bbE_{\cC}\left[\bar{e}^{Y}(\cC,s^n)\right]$ as follows
    \begin{align}
        &\bbE_{\cC}\left[\bar{e}^{Y}(\cC,s^n)\right]\nn\\
        &= \frac{1}{2^{n(R_1+R_2)}}\sum_{\substack{m_1\in \cM^{1}_n \\ m_2\in \cM^{2}_n}}\bbE_{\cC}\left[ \tr\left[\left(\bbI^{\cB}- \left(\Lambda_{m_2}^{Y}\right)\right)\rho_{\substack{X^n(m_1),Y^n(m_2), s^n}}^{\cB}\right]\right]\nn\\
        &\overset{a}{\leq} \frac{1}{2^{n(R_1+R_2)}}\sum_{\substack{m_1\in \cM^{1}_n \\ m_2\in \cM^{2}_n}}\left(2\bbE_{\cC}\left[\tr\left[\left(\bbI^{\cB}- \Pi^{Y}_{m_2}\right)\rho_{\substack{X^n(m_1),Y^n(m_2), s^n}}^{\cB}\right]\right] + 4\bbE_{\cC}\left[\tr\left[\left(\sum_{\substack{\hat{m}_2 \in \cM^{2}_n : \\ \hat{m}_2 \neq m_2}} \Pi^{Y}_{\hat{m}_2}\right)\rho_{\substack{X^n(m_1),Y^n(m_2), s^n}}^{\cB}\right]\right]\right)\nn\\
        &= \frac{1}{2^{n(R_1+R_2)}}\sum_{\substack{m_1\in \cM^{1}_n \\ m_2\in \cM^{2}_n}}\left(2\bbE_{\cC}\left[\tr\left[\left(\bbI^{\cB}- \sum_{\hat{m}_1 \in  \cM^{1}_n} \hat{\bbT}_{n,X^n(\hat{m}_1),Y^n(m_2)}^{Y}\right)\rho_{\substack{X^n(m_1),Y^n(m_2), s^n}}^{\cB}\right]\right]\right.\nn\\
        &\hspace{70pt} \left.+ 4\bbE_{\cC}\left[\tr\left[\left(\sum_{\substack{\hat{m}_2 \in \cM^{2}_n : \\ \hat{m}_2 \neq m_2}} \sum_{\hat{m}_1 \in  \cM^{1}_n} \hat{\bbT}_{n,X^n(\hat{m}_1),Y^n(\hat{m}_2)}^{Y}\right)\rho_{\substack{X^n(m_1),Y^n(m_2), s^n}}^{\cB}\right]\right]\right)\nn\\
        &\overset{b}{\leq} \frac{1}{2^{n(R_1+R_2)}}\sum_{\substack{m_1\in \cM^{1}_n \\ m_2\in \cM^{2}_n}}\left(2\bbE_{\cC}\left[\tr\left[\left(\bbI^{\cB}-  \hat{\bbT}_{n,X^n(m_1),Y^n(m_2)}^{Y}\right)\rho_{\substack{X^n(m_1),Y^n(m_2), s^n}}^{\cB}\right]\right] \right.\nn\\
        & \hspace{70pt}\left.+ 4\sum_{\substack{\hat{m}_2 \in \cM^{2}_n : \\ \hat{m}_2 \neq m_2}} \bbE_{\cC}\left[\tr\left[\hat{\bbT}_{n,X^n(m_1),Y^n(\hat{m}_2)}^{Y}\rho_{\substack{X^n(m_1),Y^n(m_2), s^n}}^{\cB}\right]\right]\right.\nn\\
        &\hspace{70pt} \left.+ 4\sum_{\substack{(\hat{m}_1,\hat{m}_2) \in \cM^{1}_n \times \cM^{2}_n \\ :\hat{m}_1 \neq m_1, \hat{m}_2 \neq m_2}}\bbE_{\cC}\left[\tr\left[ \hat{\bbT}_{n,X^n(\hat{m}_1),Y^n(\hat{m}_2)}^{Y}\rho_{\substack{X^n(m_1),Y^n(m_2), s^n}}^{\cB}\right]\right]\right)\nn\\
        &= \frac{1}{2^{n(R_1+R_2)}}\sum_{\substack{m_1\in \cM^{1}_n \\ m_2\in \cM^{2}_n}}\left(2\bbE_{X^n(m_1),Y^n(m_2)}\left[\tr\left[\left(\bbI^{\cB}-  \hat{\bbT}_{n,X^n(m_1),Y^n(m_2)}^{Y}\right)\rho_{\substack{X^n(m_1),Y^n(m_2), s^n}}^{\cB}\right]\right] \right.\nn\\
        & \hspace{70pt}\left.+ 4\sum_{\substack{\hat{m}_2 \in \cM^{2}_n : \\ \hat{m}_2 \neq m_2}} \bbE_{X^n(m_1),Y^n(\hat{m}_2)}\left[\tr\left[\hat{\bbT}_{n,X^n(m_1),Y^n(\hat{m}_2)}^{Y}\rho_{\substack{X^n(m_1), s^n}}^{\cB}\right]\right]\right.\nn\\
        &\hspace{70pt} \left.+ 4\sum_{\substack{(\hat{m}_1,\hat{m}_2) \in \cM^{1}_n \times \cM^{2}_n \\ : \hat{m}_1 \neq m_1, \hat{m}_2 \neq m_2}}\bbE_{X^n(\hat{m}_1),Y^n(\hat{m}_2)}\left[\tr\left[ \hat{\bbT}_{n,X^n(\hat{m}_1),Y^n(\hat{m}_2)}^{Y}\rho_{\substack{ s^n}}^{\cB}\right]\right]\right)\nn\\
        &{\leq}2\tr\left[\left( \bbI^{X^nY^n\cB^n} - \hat{\bbT}_{n}^{Y}\right)\rho^{X^nY^n\cB^n}_{P_{X},P_{Y},s^n}\right] + 4\cdot2^{nR_2}\tr\left[\hat{\bbT}_{n}^{Y}\left(\sum_{\substack{ x^n\in \cX^n\\y^n \in \cY^n}}P^{n}_{X}(x^n)\ketbrasys{x^n}{X^n}\otimes P^{n}_{Y}(y^n)\ketbrasys{y^n}{Y^n} \otimes \rho^{\cB^n}_{x^n,s^n}\right)\right]\nn\\
        &\hspace{10pt}+ 4\cdot2^{n(R_1 + R_2)}\tr\left[\hat{\bbT}_{n}^{Y}\left(\sum_{\substack{ x^n\in \cX^n\\y^n \in \cY^n}}P^{n}_{X}(x^n)\ketbrasys{x^n}{X^n}\otimes P^{n}_{Y}(y^n)\ketbrasys{y^n}{Y^n} \otimes \rho^{\cB^n}_{s^n}\right)\right]\nn\\
        &\overset{c}{\leq} 2f(n,\abs{\cX},\abs{\cY}, \abs{\cB}, \abs{\cS_\varepsilon},t)(1+\varepsilon)^n
        \left(2^{nt\left(\widehat{R}_2 - \min _{Q \in \cP(\cS_\varepsilon)} I_{1-t}[Y;\cB|X]_{P_{X},P_{Y},Q}\right)} + 2^{nt\left(\widehat{R}_1 + \widehat{R}_2 -  \min _{Q \in \cP(\cS_\varepsilon)} I_{1-t}[XY;\cB]_{P_{X},P_{Y},Q}\right)}\right)\nn \\
        &\hspace{10pt}+ 4\cdot g_1(n,\abs{\cY},\abs{\cB})2^{n(R_2 - \widehat{R}_2)} + 4\cdot g_3(n,\abs{\cB})2^{-n(\widehat{R}_1 + \widehat{R}_2 - R_1 - R_2)},\label{error_term_Y_hayashi}
    \end{align}
    where $a$ follows from Fact \ref{Hayashi_nagaoka}, $b$ follows from the fact that $\sum_{\hat{m}_1 \in  \cM^{1}_n} \hat{\bbT}_{n,X^n(\hat{m}_1),Y^n(m_2)}^{Y} \geq \hat{\bbT}_{n,X^n(m_1),Y^n(m_2)}^{Y}$ and $c$ follows from \cref{corllary_achievability_mac_independent_eq2_TS,corllary_achievability_mac_independent_eq5_TS,corllary_achievability_mac_independent_eq6_TS}.
    
    Similarly, using \cref{corllary_achievability_mac_independent_eq1_TS,corllary_achievability_mac_independent_eq3_TS,corllary_achievability_mac_independent_eq4_TS}, we get the following upper-bound for the quantity $\bbE_{\cC}\left[\bar{e}^{X}(\cC,s^n)\right]$:
    \begin{align}
        &\bbE_{\cC}\left[\bar{e}^{X}(\cC,s^n)\right] \nn\\
        &\leq 2f(n,\abs{\cX},\abs{\cY}, \abs{\cB}, \abs{\cS_\varepsilon},t)(1+\varepsilon)^n
        \left(2^{nt\left(\widehat{R}_1 - \min _{Q \in \cP(\cS_\varepsilon)} I_{1-t}[X;\cB|Y]_{P_{X},P_{Y},Q}\right)} 
        + 2^{nt\left(\widehat{R}_1 + \widehat{R}_2 -  \min _{Q \in \cP(\cS_\varepsilon)} I_{1-t}[XY;\cB]_{P_{X},P_{Y},Q}\right)}\right)\nn\\
        &\hspace{10pt}+ 4\cdot g_2(n,\abs{\cY},\abs{\cB})2^{n(R_1 - \widehat{R}_1)} + 4\cdot g_3(n,\abs{\cB})2^{-n(\widehat{R}_1 + \widehat{R}_2 - R_1 - R_2)}.\label{error_term_X_hayashi}
    \end{align}
        To reduce redundancy in the main script, the proof of \eqref{error_term_X_hayashi} is given in Appendix \ref{proof_error_term_X_hayashi}. 
        Due to the conditions \eqref{XBHI1} -- \eqref{Xtheo_gen_sen_macR_1R_2},
the pair of \eqref{error_term_Y_hayashi} and \eqref{error_term_X_hayashi} implies 
\begin{align}
\lim_{n\to \infty}\sup_{s^n \in \cS^n}\bbE_{\cC}\left[\bar{e}^Y(\cC,s^n)\right]=0, \quad
\lim_{n\to \infty}\sup_{s^n \in \cS^n}\bbE_{\cC}\left[\bar{e}^X(\cC,s^n)\right]=0.\label{FJ9}
\end{align}
Then, the combination of 
\eqref{joint_decoding_hayashi} and \eqref{FJ9}
guarantees \eqref{BVQ}.
This completes the proof for achievability using the non-simultaneous testing approach, mentioned in Lemma \ref{lemma_cq_mac_avht_independent_ts}.
\endproof

\subsection{Converse part of Theorem \ref{lemma_rand_capacity_avmac}}\label{S4-C}
Let $\cC$ be a $(n,2^{nR_1},2^{nR_2},\eps_n)$ random code satisfying
\begin{align}
\lim_{n\to \infty}\sup_{s^n \in \cS^n}\bbE_{\cC}\left[\bar{e}(\cC,s^n)\right]=0.
\end{align}
For any any probability distribution $Q$ over $\cS$, we have
\begin{align}
\sup_{s^n \in \cS^n}\bbE_{\cC}\left[\bar{e}(\cC,s^n)\right]
\ge 
\bbE_{\cC} \left[\sum_{s^n \in \cS^n}Q^n(s^n)\bar{e}(\cC,s^n)\right].
\end{align}
Hence, the RHS of the above also goes to zero.
We consider a classical-quantum multiple access channel (CQ-MAC) $\cW_{mac}^{Q^{n}} := \sum_{s^n \in \cS^n} Q^{n}(s^n) \cN^{X^nY^n \to B^n}_{s^n}$.
Then, the joint state induced on $X^n$, $Y^n$ and $\cB^n$ is represented by the following CQ state,
\begin{equation}\rho^{X^nY^n\cB^n}_{P_X,P_Y, Q^{n}} := \sum_{x^n \in \cX^n , y^n \in \cY^n, s^n \in \cS^n}Q^{n}(s^n) P^{n}_{X}(x^n)\ketbrasys{x^n}{X^n}\otimes P^{n}_{Y}(y^n)\ketbrasys{y^n}{Y^n}\otimes \rho^{\cB^n}_{x^n,y^n,s^n}.\label{cons}
\end{equation}
All the information-theoretic quantities below are calculated with respect to \eqref{cons}.
We will prove a converse under the MAC $\cW_{mac}^{Q^{n}}$. 
This will imply an outer bound on the random code capacity region for the CQ-AVMAC (mentioned in Definition \ref{def2_avmac_rand}).

Let $(M_1,M_2)$ be two independent and uniformly distributed random variables over $[2^{nR_1}]\times[2^{nR_2}]$ and let $(\hat{M}_1,\hat{M_2})$ be an estimation of $(M_1,M_2)$. The diagram below shows the general process of communicating the message pair $(M_1M_2)$ over $\cW_{mac}^{Q^{n}}$.
\begin{equation*}
    (M_1,M_2) \longrightarrow (X^n,Y^n) \underset{\cW_{mac}^{Q^{n}}}{\longrightarrow} \cB^n \longrightarrow (\hat{M}_1,\hat{M_2}).
\end{equation*}
Fano's inequality (Fact \ref{fano}) implies
\begin{align}
    H(M_1 M_2 |\hat{M}_1 \hat{M_2}\cC) &\leq n(R_1+R_2)\Pr\left\{(M_1,M_2) \neq (\hat{M}_1,\hat{M_2})\right\} 
    + H_{b} \left(\Pr\left\{(M_1,M_2) \neq (\hat{M}_1,\hat{M_2})\right\}\right)\nn\\
    &\leq n(R_1+R_2)\eps_n + H_{b}(\eps_n)\nn\\
    &\overset{a}{=} n \delta_n,\label{fano_bound}
\end{align}
 where in $a$, we define $\delta_n$ as follows
\begin{equation}
   \delta_n := (R_1+R_2)\eps_n + \frac{1}{n}H_{b}(\eps_n).\label{epsn}
\end{equation}

Observe that, $\lim_{n \to \infty}\delta_n = 0$. We now upper-bound the rate $R_1$ as follows
\begin{align*}
    nR_1 &= H(M_1|\cC)\\
    &= I[M_1;\hat{M}_1 M_2\hat{M}_2 \cC] + H(M_1|\hat{M}_1 M_2 \hat{M}_2 \cC)\\
    &\leq I[M_1;\hat{M}_1 M_2\hat{M}_2|\cC] 
    + H(M_1,  M_2|\hat{M}_1 \hat{M}_2\cC)\\
    &\overset{b}{\leq} I[X^n ; Y^n\cB^n|\cC] +  n \delta_n\\
    &\overset{c}{=} I[X^n ; \cB^n|Y^n\cC] +  n \delta_n\\
    &= H(\cB^n|Y^n\cC) - H(\cB^n|X^n Y^n\cC) +  n \delta_n\\
    &\overset{d}{=} H(\cB^n|Y^n\cC) - \sum_{i=1}^{n}H(B|X_i,Y_i\cC) +  n \delta_n\\ 
    &\overset{e}{\leq} \sum_{i=1}^{n}H(B|Y_i\cC) - \sum_{i=1}^{n}H(B|X_i,Y_i\cC) +  n \delta_n\\ 
    &=  \sum_{i=1}^{n} I[X_i : \cB|Y_i\cC]_{P_{X_i},P_{Y_i},Q}  +  n \delta_n,
\end{align*}
where $b$ follows from data processing inequality \cite[Theorem 4.3.3]{gallager-68-book} and \eqref{fano_bound}, $c$ follows from the fact that $X^n$ and $Y^n$ are independent, $d$ follows from the fact that $\rho_{x^n,y^n,Q^{n}}^{\cB^n}$ can be written as $\bigotimes_{i=1}^{n}\rho_{x_i,y_i,Q}^{\cB}$ (where for any  
$x \in \cX, y \in \cY, \rho_{x_i,y_i,Q}^{\cB} := \sum_{s \in \cS} Q(s)\rho_{x,y,s}^{\cB}$) and $e$ follows from chain rule for entropy and the fact that removing the conditioning increases the entropy. Since the chosen $Q$ is any distribution over $\cS$ therefore, we get the following upper-bound on $R_1$ by optimizing over the choice of $Q$,
\begin{align}
    R_1- \delta_n\leq& \frac{1}{n} \sum_{i=1}^{n} \inf_{Q \in \cP(\cS)} I[X_i : \cB|Y_i\cC]_{P_{X_i},P_{Y_i},Q}  \nn\\
=& \frac{1}{n} \sum_{i=1}^{n} \sum_{c}P_{\cC}(c)
\inf_{Q \in \cP(\cS)} I[X_i : \cB|Y_i\cC]_{P_{X_i|\cC=c},P_{Y_i|\cC=c},Q} .\label{FO1}
\end{align}

Similarly, the rate $R_2$ is upper-bounded as follows
\begin{align}
    R_2-\delta_n \leq &\frac{1}{n} \sum_{i=1}^{n} \inf_{Q \in \cP(\cS)} I[Y_i : \cB|X_i\cC]_{P_{X_i},P_{Y_i},Q}\nn \\
= &\frac{1}{n} \sum_{i=1}^{n} \sum_{c}P_{\cC}(c)
\inf_{Q \in \cP(\cS)} I[Y_i : \cB|X_i\cC]_{P_{X_i|\cC=c},P_{Y_i|\cC=c},Q}.\label{FO2}
\end{align}
The sum rate $(R_1 + R_2)$ is upper-bounded as follows
\begin{align*}
    n(R_1 + R_2) &= H(M_1 M_2|\cC)\\
    &= I[M_1 M_2;\hat{M}_1 \hat{M}_2|\cC] + H(M_1 M_2|\hat{M}_1 \hat{M}_2\cC)\\
    &\overset{f}{\leq} I[X^nY^n : \cB^n|\cC] +  n \delta_n\\
    &= H(\cB^n) - H(\cB^n | X^nY^n\cC) + n \delta_n\\
    &\overset{g}{=}  H(\cB^n|\cC) 
    - \sum_{i=1}^{n}H(\cB | X_i Y_i\cC) + n \delta_n\\
    &\leq n H(\cB|\cC) - \sum_{i=1}^{n}H(\cB | X_i Y_i\cC) + n \delta_n\\
    &= \sum_{i=1}^{n} I[X_i Y_i :\cB|\cC]_{P_{X_i},P_{Y_i},Q}  +  n \delta_n,
\end{align*}
where $f$ follows from data processing inequality \cite[Theorem 4.3.3]{gallager-68-book} and \eqref{fano_bound}, $g$ from the fact that $\rho_{x^n,y^n,Q^n}^{\cB^n}$ can be written as $\bigotimes_{i=1}^{n}\rho_{x_i,y_i,Q}^{\cB}$ (where for any  $x \in \cX, y \in \cY, \rho_{x_i,y_i,Q}^{\cB} := \sum_{s \in \cS} Q(s)\rho_{x,y,s}^{\cB}$). Since the chosen $Q$ is any distribution over $\cS$ therefore, we get the following upper-bound on $(R_1+R_2)$ by optimizing over the choice of $Q$,
\begin{align}
    (R_1-\delta_n) + (R_2 -\delta_n)
\le    R_1 + R_2 -\delta_n
    \leq &\frac{1}{n} \sum_{i=1}^{n} \inf_{Q \in \cP(\cS)} I[X_i Y_i : \cB|\cC]_{P_{X_i},P_{Y_i},Q}\nn  \\
    =&\frac{1}{n} \sum_{i=1}^{n} \sum_{c}P_{\cC}(c)
    \inf_{Q \in \cP(\cS)} I[X_i Y_i : \cB]_{P_{X_i|\cC=c},P_{Y_i|\cC=c},Q} .\label{FO3}
\end{align}
The combination of \eqref{FO1}, \eqref{FO2}, and \eqref{FO3}
guarantees that $(R_1-\delta_n,R_2-\delta_n)$ belongs to $\cR^{\star}$.
Since $\delta_n\to 0$,
$(R_1,R_2)$ also belongs to $\cR^{\star}$.
This completes the proof of Theorem \ref{lemma_rand_capacity_avmac}. 
\endproof

In the subsection below we discuss the generalization of our achievability techniques to the case of $T (T>2)$ senders.
\subsection{Discussion on the achievability techniques for the case of $T$ $(T > 2)$ senders}
\subsubsection{Formulation}
In this subsection, for simplicity, we assume that
$|\cS|< \infty$.
To address the case of $T$ $(T > 2)$ senders, we formulate CQ-AVMAC with this case as follows.
For the case of CQ-AVMAC with $T$ senders we will denote each of the classical input systems as $X_{[T]}:= \{X_i\}_{i\in [T]}$ and the output quantum system as $\cB$. The channel output is represented by a collection $\{\rho^{\cB}_{x_{[T]},s}\}_{x_{[T]} \in \cX_{[T]}, s \in \cS}$, where $\cX_{[T]} := (\times_{i \in [T]}\cX_i)$ (where $\abs{\cX_i}<\infty$), $x_{[T]} := (x_1,\cdots,x_{T})$, for each state $s \in \cS$, $\rho^{\cB}_{x_{[T]},s}$ is the output of the channel corresponding to state $s$, upon inputting $x_{[T]}$. We can represent the joint quantum state induced on $X_{[T]}$ and $\cB$ using the following CQ-states:
\begin{align*}
    \rho^{X^n_{[T]}\cB^n}_{P_{X_{[T]}},s^n} := \sum_{x^n_{[T]} \in \cX^n_{[T]}} \prod_{i=1}^{[T]}P^{n}_{X_i}(x^n_i) \bigotimes_{i=1}^{T}\ketbrasys{x^n_i}{X^n_i} \otimes \rho^{\cB^n}_{x^n_{[T]},s^n}, 
\end{align*}
where $X^n_{[T]} := \{X^n_i\}_{i\in [T]}$, $P_{X_{[T]}} := \{P_{X_{1}},\cdots,P_{X_{T}}\}$, $\cX^n_{[T]}:= (\times_{i=1}^{T}\cX^n_i)$, $x^n_{[T]} := \{x^n_i\}_{i \in [T]}$, where for each $i \in [T]$, $x^n_i := (x_{i,1},\cdots,x_{i,n})$, $P_{X_{i}}$ is a probability distribution over $\cX_i$ and $P^{n}_{X_{i}}(x^n_i) := \prod_{j =1}^{n}P_{X_{i}}(x_{i,j})$. For any $x^n_{[T]} \in \cX^n_{[T]}$ and $s^n := (s_1,\cdots,s_n) \in \cS^n$, we define $\rho^{\cB^n}_{x^n_{[T]},s^n} := \otimes_{j=1}^{n}\rho^{\cB}_{\{x_{i,j}\}_{i\in [T]},s_j}$ and $\rho^{\cB}_{\{x_{i,j}\}_{i\in [T]},s_j}$ is the output of the channel corresponding to $s_j$ upon inputting $\{x_{i,j}\}_{i\in [T]}$ into it. 

We assume that for each $i \in [T]$, $i$-th sender has a message set $\cM^{i}_{n}:= [2^{nR_i}]$. For a particular message tuple $m_{[T]} := (m_1,\cdots,m_T)$ chosen by the $T$ senders respectively and 
given an encoder $\mathcal{E}^{(n),T}$ and an element $s^n$
and a decoding POVM $\{ \Lambda_{m_{[T]}}^{X_{[T]}}\}_{m_{[T]}}$,
we define the probability of error as 
\begin{align}
    e(m_{[T]},\mathcal{E}^{(n),T},s^n, \{ \Lambda_{m_{[T]}}^{X_{[T]}}\}_{m_{[T]}}) &:= 
\tr\left[\left(\bbI^{\cB^n} - \Lambda_{m_{[T]}}^{X_{[T]}}\right)\rho^{\cB^n}_{X^n_{[T]}(m_{[T]}),s^n}\right].
\end{align}
Then, we define its average error probability as
  \begin{equation*}
      \bar{e}(\mathcal{E}^{(n),T},s^n,\{ \Lambda_{m_{[T]}}^{X_{[T]}}\}_{m_{[T]}})
      := \frac{1}{2^{n(\sum_{i \in [T]} R_i)}} \sum_{m_{[T]} \in (\times_{i \in [T]} \cM^{i}_{n})}e(m_{[T]},\mathcal{E}^{(n),T},s^n,\{ \Lambda_{m_{[T]}}^{X_{[T]}}\}_{m_{[T]}}).
  \end{equation*}

\subsubsection{Generalization of Techniques Based on Lemma \ref{lemma_cq_mac_avht_independent_ts}}
To construct our decoder, 
we next discuss the generalization of the techniques based on Lemma \ref{lemma_cq_mac_avht_independent_ts}. 
For any $\widehat{R}_1,\cdots,\widehat{R}_T > 0$ and each $V\subset [T]$, we define the following tests (projectors):
\begin{align*}
    \widehat{\bbT}_{n,X_{V}} := \sum_{x^n_{[T]} \in \cX^n_{[T]}} \bigotimes_{i=1}^{T}\ketbrasys{x^n_i}{X^n_i} \otimes \left\{\widehat{\rho}^{\cB^n}_{x^n_{[T]}} \succeq 2^{n \sum_{t \in [T] \setminus V} \widehat{R}_t}\widehat{\rho}^{\cB^n}_{x^n_{V}}\right\},
\end{align*}
where for any $V \subset [T] : V \neq \varnothing$, $\widehat{\rho}^{\cB^n}_{x^n_{V}}$ is defined using the same techniques used in \cref{rhohatxbarnybarn1,rhohatxnyn1} and for $V = \varnothing$, $\widehat{\rho}^{\cB^n}_{x^n_{V}} := \rho^{\cB^n}_{U,n}$ (which is defined in \eqref{uniform_n_state}) and $\widehat{\bbT}_{n,X_{V}} \triangleq \widehat{\bbT}_{n}$. For each $i \in [T]$, as a generalization of \eqref{tests_hayashi_generalised}, we define a projector $\widehat{\bbT}^{X_i}_{n}$ as follows,
\begin{align}
    \widehat{\bbT}^{X_i}_{n} := \underset{\substack{\\V \subseteq [T]\setminus \{i\}}}{\bullet} \widehat{\bbT}_{n,X_{V}},\label{tests_hayashi_generalised_n_sender}
\end{align}
where $\bullet$ denotes iterative matrix multiplication (see Definition \ref{def_iter_mul}). Using similar techniques used in the proof of Lemma \ref{lemma_cq_mac_avht_independent_ts}, one can show that 
for each $i \in [T]$ and each $t \in (0,1)$, 
the projector $\widehat{\bbT}^{X_i}_{n}$ satisfies the following:
\begin{align*}
\tr\left[\left(\bbI^{X^n_{[T]}\cB^n} -\widehat{\bbT}^{X_i}_{n}\right)\rho^{X^n_{[T]}\cB^n}_{P_{X_{[T]}},s^n}\right] &\leq \hat{f}(n,\abs{\cX_1},\cdots,\abs{\cX_{T}},\abs{\cB},\abs{\cS}) \sum_{V \subseteq [T] : i \in V} 2^{n\left(\sum_{v \in V} \widehat{R}_v - \min_{Q \in \cP(\cS)}I_{1-t}[X_{V};B|X_{[T]\setminus V}]_{P_{X_{[T]}},Q}\right)},\\
\tr\left[\widehat{\bbT}^{X_i}_{n}\left(\rho^{X^n_{M}}_{P_{X_{M}}} \otimes \rho^{X^n_{[T]\setminus M}\cB^n}_{P_{X_{[T]\setminus M}},s^n}\right)\right] &\leq \hat{g}(n,\prod_{t \in [T]\setminus M}\abs{\cX_t},\abs{\cB}) 2^{-n\left(\sum_{u \in M} \widehat{R}_u\right)}, \quad \forall M \subseteq [T]:  i \in M,
\end{align*}
where for each $V \subseteq [T]$ and probability distributions $Q \in \cP(\cS)$, we define $I_{1-t}[X_{V};B|X_{[T]\setminus V}]_{P_{X_{[T]}},Q}$ as follows,
\begin{align*}
     &I_{1-t}[X_{V};B|X_{[T]\setminus V}]_{P_{X_{[T]}},Q} \\
     &:= \min_{\{\mu_{x_{[T]\setminus V}} \in \cD(\cH_\cB^{\otimes n})\}_{x_{[T]\setminus V}}} D_{1-t}\left(\rho^{X_{[T]}\cB}_{P_{X_{[T]}},Q} \left\| \rho^{X_{V}}_{P_{X_{V}}} \otimes \left(\sum_{x_{[T]\setminus V}\in \cX_{[T]\setminus V}} \prod_{j\in [T]\setminus V} P_{X_j}(x_j) \bigotimes_{j\in [T]\setminus V}\ketbrasys{x_j}{X_j} \otimes \mu_{x_{[T]\setminus V}}\right.\right)\right)\\
     &= -\frac{1-t}{t}\log\left(\sum_{x_{[T]\setminus V}\in \cX_{[T]\setminus V}} \prod_{j\in [T]\setminus V} P_{X_j}(x_j)\tr\left[\left(\sum_{x_{V} \in \cX_{V}} \prod_{k\in V} P_{X_k}(x_k)(\rho^{\cB}_{Q,x_{[T]}})^{1-t}\right)^{\frac{1}{1 - t}}\right]\right),\\
     &\rho^{X_{[T]}\cB}_{P_{X_{[T]}},Q} := \sum_{x_{[T]} \in \cX_{[T]}} \prod_{i=1}^{[T]}P_{X_i}(x_i) \bigotimes_{i=1}^{T}\ketbrasys{x_i}{X_i} \otimes \rho^{\cB}_{x_{[T]},Q},
\end{align*}
where the symbols used here are defined as
$\rho^{\cB}_{x_{[T]},Q} := \sum_{s \in \cS}Q(s)\rho^{\cB^n}_{x_{[T]},s}$, and 
\begin{align*}
    \hat{f}(n,\abs{\cX_1},\cdots,\abs{\cX_{T}},\abs{\cB},\abs{\cS}) &:= (n+1)^{\abs{\cS}-1 + \frac{t(\prod_{i \in [T]}\abs{\cX_i})(\abs{\cB}-1)\abs{\cB}}{2} + (\abs{\cB} - 1)t(\prod_{i \in [T]}\abs{\cX_i})},\\
    \hat{g}(n,\prod_{t \in [T]\setminus M}\abs{\cX_t},\abs{\cB}) &:= n^{\frac{(\prod_{t \in [T]\setminus M}\abs{\cX_t})\abs{\cB}(\abs{\cB}-1)}{2}}\abs{\Lambda^{n}_{\abs{\cB}}}^{(\prod_{t \in [T]\setminus M}\abs{\cX_t})}.
\end{align*}

    \subsubsection{Decoding Strategy}
    Upon receiving the state $\rho_{\substack{X^n_{[T]}, s^n}}^{\cB^n}$ (where $s^n$ is unknown to the receiver), the receiver then tries to guess the message pair sent using certain measurement based decoder. Now from \eqref{tests_hayashi_generalised_n_sender}, for each $i \in [T]$, we assume that $\widehat{\bbT}^{X_i}_{n}$ has the following form:
    \begin{equation*}
       \widehat{\bbT}^{X_i}_{n} := \sum_{x^n_{[T]} \in \cX^n_{[T]}} \bigotimes_{i=1}^{T}\ketbrasys{x^n_i}{X^n_i} \otimes \widehat{\bbT}^{X_i}_{n,x^n_{[T]}}.
    \end{equation*}
    
    For each message $T$-tuple $m_{[T]} \in (\times_{i\in [T]}\cM^{i}_{n})$, we construct $T$ separate decoders (each for one of messages). For each $i \in [T]$, we define a separate decoder for $i$-th sender as follows,
    \begin{align}
        \Lambda_{m_i}^{X_i} &:= \left(\sum_{\hat{m}_i \in \cM^{i}_n} \Pi^{X_i}_{\hat{m}_i}\right)^{-\frac{1}{2}}\Pi^{X_i}_{m_i} \left(\sum_{\hat{m}_i \in \cM^{i}_n} \Pi^{X_i}_{\hat{m}_i}\right)^{-\frac{1}{2}},\label{lambda_m_i}
    \end{align}
where for any $\hat{m}_{i} \in \cM^{i}_{n}$, $\Pi^{X_i}_{\hat{m}_i}$ is defined as follows,
\begin{equation}
    \Pi^{X_i}_{\hat{m}_i} := \sum_{\hat{m}_{[T]\setminus i} \in  (\times_{\substack{j \in [T]\setminus i}}\cM^{j}_{n})} \hat{\bbT}_{n,X^n_{[T]}(\hat{m}_{[T]})},\label{pi_X_i_m_i}
\end{equation}
where $\hat{m}_{[T]} : (\hat{m}_1,\cdots,\hat{m}_T)$ 
 and $\hat{m}_{[T]\setminus i} := (\hat{m}_j)_{j \in [T] : j \neq i}$. The receiver then uses the POVM $\left\{\Lambda_{m_{[T]}}^{X_{[T]}}\right\}_{m_{[T]}}$ to decode the message jointly, where $\Lambda_{m_{[T]}}^{X_{[T]}}$ is defined as follows,
\begin{equation}
    \Lambda_{m_{[T]}}^{X_{[T]}}:= \left(\Lambda_{m_T}^{X_T}\right)^{\frac{1}{2}}\left(\Lambda_{m_{T-1}}^{X_{T-1}}\right)^{\frac{1}{2}}\cdots\left(\Lambda_{m_2}^{X_2}\right)^{\frac{1}{2}}\left(\Lambda_{m_1}^{X_1}\right)\left(\Lambda_{m_2}^{X_2}\right)^{\frac{1}{2}}\cdots\left(\Lambda_{m_{T-1}}^{X_{T-1}}\right)^{\frac{1}{2}}\left(\Lambda_{m_T}^{Y}\right)^{\frac{1}{2}}.\label{lambda_m[T]}
\end{equation}

It is important to note that, for each $m_{[T]} \in (\times_{i\in [T]}\cM^{i}_{n})$, 
the matrices $\Lambda_{m_1}^{X_1}, \cdots, \Lambda_{m_T}^{X_T}$ commute with each other. Thus, the following holds,
\begin{equation*}
    \sum_{m_{[T]} \in (\times_{i\in [T]}\cM^{i}_{n})} \Lambda_{m_{[T]}}^{X_{[T]}} = \bbI^{\cB}.
\end{equation*}
Now, if the receiver gets the measurement outcome corresponding to $\Lambda_{\hat{m}_{[T]}}^{X_{[T]}}$, the receiver then declares the message tuple to be $\hat{m}_{[T]}$.
Under the above choice of the decoder,
we simplify 
$e(m_{[T]},\mathcal{E}^{(n),T},s^n,\{ \Lambda_{m_{[T]}}^{X_{[T]}}\}_{m_{[T]}})$ and 
$\bar{e}(\mathcal{E}^{(n),T},s^n,\{ \Lambda_{m_{[T]}}^{X_{[T]}}\}_{m_{[T]}})$ to
$e(m_{[T]},\mathcal{E}^{(n),T},s^n)$ and 
$\bar{e}(\mathcal{E}^{(n),T},s^n)$.

In the next part, we will show the characterization of rate-tuples $(R_1,\cdots,R_T)$, for which there exists 
an $(n,2^{nR_1},\cdots,2^{nR_T},\eps_n)$ random code with the encoder  $\mathcal{E}^{(n),T}$ for which the average error probability, averaged over the choice of encoder $\mathcal{E}^{(n),T}$ satisfies
$\max_{s^n \in \cS^n}\bbE_{\mathcal{E}^{(n),T}} \left[\bar{e}(\mathcal{E}^{(n),T},s^n)\right]  < \eps_n$, where $\lim_{n \to \infty}\eps_n = 0$. We now start with the construction of the encoder.
    \subsubsection{Randomized Encoder Construction}
    For each $i \in [T]$ and a probability distribution $P_{X_i}$ on $\cX_i$, 
    the $i$-th sender randomly generates $2^{nR_i}$ sequences $\{X^n_i(m_i) \in \cX^n_i : m_i \in [2^{nR_i}]\}$ according to 
    the probability distribution $P_{X_i}^{n}$.
    That is, for any element
    $ m_i \in [2^{nR_i}]$, the random variable $X^n_i(m_i)$ obeys the distribution 
$\Pr\left\{X^n_i(m_i)\right\} := \prod_{j=1}^{n}P_{X_i}(X_{i,j}(m_i))$, where $X^n_i(m_i):=\left(X_{i,1}(m_i),\cdots,X_{i,n}(m_i)\right)$. 
Then, we randomly choose the encoder $\cE^{(n),T}$ subject to the following distribution;
\begin{equation*}
    \Pr( \cE^{(n),T}_i(m_i)=X_i^n(m) \hbox{ for } i=1, \ldots, T,~ m_i=1, \ldots,  2^{nR_i} ) 
    := \prod_{i=1}^{T}\prod_{m_i=1}^{\cM^{i}_{n}}\Pr\left\{X^n_i(m_i)\right\}.
\end{equation*}

    If the $T$ senders have to send a message tuple  $m_{[T]} := (m_1,m_2,\cdots,m_T) \in (\times_{i\in [T]}\cM^{i}_{n})$, for each $i \in [T]$, $i$-th sender encodes $m_i$ to the input sequence $X^n_i(m_i)$ and 
    sends it over the channel.

\subsubsection{Error Analysis}
    Using the fact that  for each $m_{[T]} \in (\times_{i\in [T]}\cM^{i}_{n})$, $\Lambda_{m_1}^{X_1}, \cdots, \Lambda_{m_T}^{X_T}$ commutes with each other, we can show that the average error probability 
    for any particular $s^n \in \cS^n$ is of the following form:
    \begin{align}
        &\hspace{10pt} \bar{e}(\mathcal{E}^{(n),T},s^n,\{ \Lambda_{m_{[T]}}^{X_{[T]}}\}_{m_{[T]}})\nn\\
        &= \frac{1}{2^{n(\sum_{t \in [T]} R_t)}} \sum_{m_{[T]} \in (\times_{j \in [T]} \cM^{j}_{n})}e(m_{[T]},\mathcal{E}^{(n),T},s^n,\{ \Lambda_{m_{[T]}}^{X_{[T]}}\}_{m_{[T]}})\nn\\
        &= \frac{1}{2^{n(\sum_{t \in [T]} R_t)}} \sum_{m_{[T]} \in (\times_{j \in [T]} \cM^{j}_{n})}\tr\left[\left(\bbI^{\cB^n} - \Lambda_{m_{[T]}}^{X_{[T]}}\right)\rho^{\cB^n}_{X^n_{[T]}(m_{[T]}),s^n}\right]\nn\\
        &= \frac{1}{2^{n(\sum_{t \in [T]} R_t)}} \sum_{m_{[T]} \in (\times_{j \in [T]} \cM^{j}_{n})}\tr\left[\left(\bbI^{\cB^n} - \Lambda_{m_{T}}^{X_{T}}\right)\rho^{\cB^n}_{X^n_{[T]}(m_{[T]}),s^n}\right]\nn\\
        &\hspace{10pt}+  \frac{1}{2^{n(\sum_{t \in [T]} R_t)}} \sum_{m_{[T]} \in (\times_{j \in [T]} \cM^{j}_{n})}\tr\left[\left( \left(\Lambda_{m_T}^{X_T}\right)^{\frac{1}{2}}\left(\bbI^{\cB^n} - \left(\Lambda_{m_{T-1}}^{X_{T-1}}\right)^{\frac{1}{2}}\cdots\left(\Lambda_{m_1}^{X_1}\right)\cdots\left(\Lambda_{m_{T-1}}^{X_{T-1}}\right)^{\frac{1}{2}}\right)\left(\Lambda_{m_T}^{Y}\right)^{\frac{1}{2}}\right)\rho^{\cB^n}_{X^n_{[T]}(m_{[T]}),s^n}\right]\nn\\
        &\overset{a}{\leq} \frac{1}{2^{n(\sum_{t \in [T]} R_t)}} \sum_{m_{[T]} \in (\times_{j \in [T]} \cM^{j}_{n})}\tr\left[\left(\bbI^{\cB^n} - \Lambda_{m_{T}}^{X_{T}}\right)\rho^{\cB^n}_{X^n_{[T]}(m_{[T]}),s^n}\right]\nn\\
        &\hspace{10pt}+  \frac{1}{2^{n(\sum_{t \in [T]} R_t)}} \sum_{m_{[T]} \in (\times_{j \in [T]} \cM^{j}_{n})}\tr\left[\left(\bbI^{\cB^n} - \left(\Lambda_{m_{T-1}}^{X_{T-1}}\right)^{\frac{1}{2}}\cdots\left(\Lambda_{m_1}^{X_1}\right)\cdots\left(\Lambda_{m_{T-1}}^{X_{T-1}}\right)^{\frac{1}{2}}\right)\rho^{\cB^n}_{X^n_{[T]}(m_{[T]}),s^n}\right]\nn\\
        &\hspace{10pt}+ \frac{1}{2^{n(\sum_{t \in [T]} R_t)}} \sum_{m_{[T]} \in (\times_{j \in [T]} \cM^{j}_{n})} \norm{\left(\Lambda_{m_T}^{Y}\right)^{\frac{1}{2}}\rho^{\cB^n}_{X^n_{[T]}(m_{[T]}),s^n}\left(\Lambda_{m_T}^{Y}\right)^{\frac{1}{2}} - \rho^{\cB^n}_{X^n_{[T]}(m_{[T]}),s^n}}{1}\nn\\
        &\overset{b}{\leq} \frac{1}{2^{n(\sum_{t \in [T]} R_t)}} \sum_{m_{[T]} \in (\times_{j \in [T]} \cM^{j}_{n})}\tr\left[\left(\bbI^{\cB^n} - \Lambda_{m_{T}}^{X_{T}}\right)\rho^{\cB^n}_{X^n_{[T]}(m_{[T]}),s^n}\right]\nn\\
        &\hspace{10pt}+  \frac{1}{2^{n(\sum_{t \in [T]} R_t)}} \sum_{m_{[T]} \in (\times_{j \in [T]} \cM^{j}_{n})}\tr\left[\left(\bbI^{\cB^n} - \left(\Lambda_{m_{T-1}}^{X_{T-1}}\right)^{\frac{1}{2}}\cdots\left(\Lambda_{m_1}^{X_1}\right)\cdots\left(\Lambda_{m_{T-1}}^{X_{T-1}}\right)^{\frac{1}{2}}\right)\rho^{\cB^n}_{X^n_{[T]}(m_{[T]}),s^n}\right]\nn\\
        &\hspace{10pt}+ 2\sqrt{\frac{1}{2^{n(\sum_{t \in [T]} R_t)}} \sum_{m_{[T]} \in (\times_{j \in [T]} \cM^{j}_{n})} \tr\left[\left(\bbI^{\cB^n} - \Lambda_{m_{T}}^{X_{T}}\right)\rho^{\cB^n}_{X^n_{[T]}(m_{[T]}),s^n}\right]}\nn\\
        &\overset{c}{=} \bar{e}^{X_T}(\mathcal{E}^{(n),T},s^n,\{ \Lambda_{m_{[T]}}^{X_{[T]}}\}_{m_{[T]}}) + 2\sqrt{\bar{e}^{X_T}(\mathcal{E}^{(n),T},s^n,\{ \Lambda_{m_{[T]}}^{X_{[T]}}\}_{m_{[T]}})} \nn\\
        &\hspace{10pt}+  \frac{1}{2^{n(\sum_{t \in [T]} R_t)}} \sum_{m_{[T]} \in (\times_{j \in [T]} \cM^{j}_{n})}\tr\left[\left(\bbI^{\cB^n} - \left(\Lambda_{m_{T-1}}^{X_{T-1}}\right)^{\frac{1}{2}}\cdots\left(\Lambda_{m_1}^{X_1}\right)\cdots\left(\Lambda_{m_{T-1}}^{X_{T-1}}\right)^{\frac{1}{2}}\right)\rho^{\cB^n}_{X^n_{[T]}(m_{[T]}),s^n}\right]\nn\\
        &\overset{d}{\leq} \bar{e}^{X_T}(\mathcal{E}^{(n),T},s^n,\{ \Lambda_{m_{[T]}}^{X_{[T]}}\}_{m_{[T]}})+ \bar{e}^{X_{T-1}}(\mathcal{E}^{(n),T},s^n,\{ \Lambda_{m_{[T]}}^{X_{[T]}}\}_{m_{[T]}}) + 2\sqrt{\bar{e}^{X_T}(\mathcal{E}^{(n),T},s^n,\{ \Lambda_{m_{[T]}}^{X_{[T]}}\}_{m_{[T]}})}\nn\\
        &\hspace{10pt} + 2\sqrt{\bar{e}^{X_{T-1}}(\mathcal{E}^{(n),T},s^n,\{ \Lambda_{m_{[T]}}^{X_{[T]}}\}_{m_{[T]}})} \nn\\
        &\hspace{10pt}+  \frac{1}{2^{n(\sum_{t \in [T]} R_t)}} \sum_{m_{[T]} \in (\times_{j \in [T]} \cM^{j}_{n})}\tr\left[\left(\bbI^{\cB^n} - \left(\Lambda_{m_{T-2}}^{X_{T-2}}\right)^{\frac{1}{2}}\cdots\left(\Lambda_{m_1}^{X_1}\right)\cdots\left(\Lambda_{m_{T-2}}^{X_{T-2}}\right)^{\frac{1}{2}}\right)\rho^{\cB^n}_{X^n_{[T]}(m_{[T]}),s^n}\right]\nn\\
        &\leq \sum_{\substack{i \in [T]:\\i \neq 1}} \bar{e}^{X_i}(\mathcal{E}^{(n),T},s^n,\{ \Lambda_{m_{[T]}}^{X_{[T]}}\}_{m_{[T]}}) + 2\sum_{\substack{i \in [T]:\\i \neq 1}} \sqrt{\bar{e}^{X_i}(\mathcal{E}^{(n),T},s^n,\{ \Lambda_{m_{[T]}}^{X_{[T]}}\}_{m_{[T]}})}\nn\\
        &\hspace{10pt}+
        \frac{1}{2^{n(\sum_{t \in [T]} R_t)}} \sum_{m_{[T]} \in (\times_{j \in [T]} \cM^{j}_{n})}\tr\left[\left(\bbI^{\cB^n} - \Lambda_{m_1}^{X_1}\right)\rho^{\cB^n}_{X^n_{[T]}(m_{[T]}),s^n}\right]\nn\\
        &=\sum_{\substack{i \in [T]}} \bar{e}^{X_i}(\mathcal{E}^{(n),T},s^n,\{ \Lambda_{m_{[T]}}^{X_{[T]}}\}_{m_{[T]}}) + 2\sum_{\substack{i \neq 1}} \sqrt{\bar{e}^{X_i}(\mathcal{E}^{(n),T},s^n,\{ \Lambda_{m_{[T]}}^{X_{[T]}}\}_{m_{[T]}})},\nn
    \end{align}
    where $a$ follows from Fact \ref{trace_norm2}, $b$ follows from Fact \ref{gent_measurement}, in equality $c$, for any $i \in [T]$, we define $$\bar{e}^{X_i}(\mathcal{E}^{(n),T},s^n,\{ \Lambda_{m_{[T]}}^{X_{[T]}}\}_{m_{[T]}}):= \frac{1}{2^{n(\sum_{t \in [T]} R_t)}} \sum_{m_{[T]} \in (\times_{j \in [T]} \cM^{j}_{n})}\tr\left[\left(\bbI^{\cB^n} - \Lambda_{m_{i}}^{X_{i}}\right)\rho^{\cB^n}_{X^n_{[T]}(m_{[T]}),s^n}\right],$$
    $d$ follows from repertitve calculations performed in earlier steps.    Taking the expectation with respect to $\mathcal{E}^{(n),T}$, for any particular $s^n \in \cS^n$ we have,
\begin{align}
    \bbE_{\mathcal{E}^{(n),T}}\left[\bar{e}(\mathcal{E}^{(n),T},s^n,\{ \Lambda_{m_{[T]}}^{X_{[T]}}\}_{m_{[T]}})\right] \leq \sum_{i \in [T]}\bbE_{\mathcal{E}^{(n),T}}\left[\bar{e}^{X_i}(\mathcal{E}^{(n),T},s^n,\{ \Lambda_{m_{[T]}}^{X_{[T]}}\}_{m_{[T]}})\right] + 2\sum_{\substack{i \neq 1}} \sqrt{\bbE_{\mathcal{E}^{(n),T}}\left[\bar{e}^{X_i}(\mathcal{E}^{(n),T},s^n,\{ \Lambda_{m_{[T]}}^{X_{[T]}}\}_{m_{[T]}})\right]},\label{error_analysis_n_sender}
\end{align}
where for each $i \in [T]$, we define the individual error term $\bar{e}^{X_i}(\mathcal{E}^{(n),T},s^n)$ corresponding to $i$-th sender as follows,
\begin{equation*}
\bar{e}^{X_i}(\mathcal{E}^{(n),T},s^n):=  \frac{1}{2^{n(\sum_{j \in [T]} R_j)}} \sum_{m_{[T]} \in (\times_{j \in [T]} \cM^{j}_{n})}\tr\left[\left(\bbI^{\cB^n} - \Lambda_{m_{i}}^{X_{i}}\right)\rho^{\cB^n}_{X^n_{[T]}(m_{[T]}),s^n}\right].
\end{equation*}

The derivation of \eqref{error_analysis_n_sender} has not been included in the manuscript since it can be directly shown by iterating the derivations of \eqref{joint_decoding_hayashi} for  $T-1$ times. Now using similar arguments that were used to prove \eqref{error_term_Y_hayashi}, for each $i \in [T]$, we show that
\begin{align*}
    \bbE_{\mathcal{E}^{(n),T}}\left[\bar{e}^{X_i}(\mathcal{E}^{(n),T},s^n)\right] \leq & 2\hat{f}(n,\abs{\cX_1},\cdots,\abs{\cX_{T}},\abs{\cB},\abs{\cS}) \sum_{V \subset [T] : i \in V} 2^{n\left(\sum_{v \in V} \widehat{R}_v - \min_{Q \in \cP(\cS)}I_{1-t}[X_{V};B|X_{[T]\setminus V}]_{P_{X_{[T]}},Q}\right)}\\
    &+ 4\sum_{M \subset [T]:  i \in M}\hat{g}(n,\prod_{t \in [T]\setminus M}\abs{\cX_t},\abs{\cB}) 2^{-n\left(\sum_{u \in M} \widehat{R}_u\right)}.
\end{align*}

Thus, for sufficiently large $n$, $\max_{s^n \in \cS^n}\bbE_{\mathcal{E}^{(n),T}}\left[\bar{e}(\mathcal{E}^{(n),T},s^n)\right]$ goes to zero, if we choose a rate tuple $(R_1,\cdots,R_T)$, which satisfies the following:
\begin{align*}
    \sum_{v \in V}R_v < \min_{Q \in \cP(\cS)}I[X_{V};B|X_{[T]\setminus V}]_{P_{X_{[T]}},Q}, \quad\quad \forall V \subseteq [T] : V \neq \varnothing.
\end{align*}
This gives a proof of achievability for the capacity region of $T$-sender CQ-AVMAC under random codes using an extension of the non-simultaneous testing approach for generalized independence testing discussed in subsection \ref{cq_mac_avht_null_statement_non_sim}.
\textit{The achievability technique based on Lemma \ref{lemma_cq_mac_avht_independent} (also see Appendix \ref{alternative_proof_lemma_rand_capacity_avmac}) can be easily generalized using similar techniques required to prove \cite[Corollary 4]{Sen2021}.}
The code obtained through these techniques can be derandomized by defining $2^{T} - 1$ symmetrizablity conditions mentioned in \cref{x_sym,y_sym,xy_sym}.

\section{Deterministic Code Capacity Region of CQ-AVMAC}\label{sec:determinsitic_cap_section}
\subsection{Characterization of Deterministic Code Capacity Region}
In the previous sections, we derived 
the random code capacity of CQ-AVC in Theorem \ref{lemma_rand_cap_cqavc} and
the random code capacity region $\cR_{r}$ of CQ-AVMAC in Theorem \ref{lemma_rand_capacity_avmac}.  
However, unlike the random code capacity and the random code capacity region, 
the deterministic code capacity and 
the deterministic code capacity region of CQ-AVMAC may be empty sometimes. 
In this section, we only consider the 
deterministic code capacity with two senders, where
Alice (sender 1) and Bob (sender 2) have input variables 
$X$ and $Y$, respectively. 
The case with a single sender can be easily recovered 
by the following modification;
We set $Y$ to be a singleton variable.
We define the topological interior $\mathbf{Int}\{A\}$ of 
a set $A\subset \mathbb{R}\times \{0\}$
as a subset $\mathbb{R}\times \{0\}$.

To discuss this more formally, we first introduce a notion of symmetrizability in Definition \ref{defsym} 
for a discrete set $\cS$
(a similar notion was also discussed in the classical case in \cite{gubner1990deterministic}). We then use this definition to establish a necessary and sufficient condition for $\mathbf{Int}\{\cR_{d}\} = \varnothing$ in Lemma \ref{symm_avmac_cq}. 

Finally we use Lemma \ref{symm_avmac_cq} to derandomize the code obtained in Theorem \ref{lemma_rand_capacity_avmac} and show that either $\mathbf{Int}\{\cR_{d}\} = \varnothing$ or $\cR_{d} = \cR_{r}$. 

\begin{definition}(Symmetrizability for a discrete set $\cS$)\label{defsym}
Assume that $\cS$ is a discrete set. 
  A CQ-AVMAC 
  $\left\{\rho_{x,y,s}^{\cB} : x \in \cX, y \in \cY, s \in \cS\right\}$
  is called \textit{symmetrizable} if at least one of the following conditions holds. 
\begin{itemize}
    \item ($\cX-$symmetrizable) There exists a set of probability distributions $\{U_1(. | x)\}_{x \in \cX}$ over $\cS$ such that 
    any two elements $x,x' \in \cX$ satisfy the following condition;
    \begin{equation}
    \sum_{s \in \cS}U_1(s|x')\rho_{x,y,s}^{\cB} = \sum_{s \in \cS}U_1(s|x)\rho_{x',y,s}^{\cB}.\label{x_sym}
\end{equation}
\item ($\cY-$symmetrizable) There exists a set of probability distributions $\{U_2(. | y)\}_{y \in \cY}$ over $\cS$ such that 
any two elements $y,y' \in \cY$ satisfy the following condition;
\begin{equation}
    \sum_{s \in \cS}U_2(s|y')\rho_{x,y,s}^{\cB} = \sum_{s \in \cS}U_2(s|y)\rho_{x,y',s}^{\cB}.\label{y_sym}
\end{equation}
\item ($\cX\cY-$symmetrizable) There exists a set of probability distributions $\{U_3(. | x, y)\}_{x \in \cX, y \in \cY}$ over $\cS$ such that 
any four elements $x,x' \in \cY$ and $y,y' \in \cY$ satisfy the following condition;
\begin{equation}
    \sum_{s \in \cS}U_3(s|x',y')\rho_{x,y,s}^{\cB} = \sum_{s \in \cS}U_3(s|x,y)\rho_{x',y',s}^{\cB}.\label{xy_sym}
\end{equation}
\end{itemize}
 \end{definition}

This definition is generalized to a general measurable space $\cS$
as follows,

\begin{definition}(Symmetrizability for a general measurable space $\cS$)\label{2defsym}
Assume that $\cS$ is a general measurable space. 
  A CQ-AVMAC 
  $\left\{\rho_{x,y,s}^{\cB} : x \in \cX, y \in \cY, s \in \cS\right\}$
  is called \textit{symmetrizable} if at least one of the following conditions holds. 
\begin{itemize}
    \item ($\cX-$symmetrizable) 
    The relation
\begin{equation}
\theta_1:= \inf_{U_1(.|x)}
\max_{x,x'\in \cX,y \in \cY, z \in \cZ}
\left\|    \int_{\cS}U_1(ds|x')\rho_{x,y,s}^{\cB} - \int_{\cS} U_1(ds|x)\rho_{x',y,s}^{\cB} \right\|_1 =0
\label{2x_sym}
\end{equation}
holds, where
$\{U_1(. | x)\}_{x \in \cX}$ is
a set of probability measures over $\cS$.
\item ($\cY-$symmetrizable) 
    The relation
\begin{equation}
\theta_2:= \inf_{U_2(.|y)}
\max_{x\in \cX,y,y'\in \cY, z \in \cZ}
\left\|    \int_{\cS}U_2(ds|y')\rho_{x,y,s}^{\cB} - \int_{\cS} U_2(ds|y)\rho_{x,y',s}^{\cB}\right\|_1 =0
\label{2y_sym}
\end{equation}
holds, where
$\{U_2(. | y)\}_{y \in \cY}$ is
a set of probability measures over $\cS$.
\item ($\cX\cY-$symmetrizable) 
The relation
\begin{equation}
\theta_3 := \inf_{U_3(.|x,y)}
\max_{x,x'\in \cX,y,y'\in \cY, z \in \cZ}
\left\|    \int_{\cS}U_3(ds|x',y')\rho_{x,y,s}^{\cB} - \int_{\cS} U_3(ds|x,y)\rho_{x',y',s}^{\cB}\right\|_1 =0
\label{2xy_sym}
\end{equation}
holds, where
$\{U_3(. |x, y)\}_{x \in \cX,y \in \cY}$ is
a set of probability measures over $\cS$.
\end{itemize}
 \end{definition}

Using the above definition, we obtain the following lemma.
 \begin{lemma}
 \label{symm_avmac_cq}
Assume that $\cS$ is a general measurable space
and satisfies Assumption \ref{NM9}.
The relation 
$\mathbf{Int}\{\cR_{d}\} = \varnothing$ holds if and only if 
CQ-AVMAC is symmetrizable.
\end{lemma}

This lemma is shown in subsections \ref{S7-B1} and \ref{S7-B2}.
We also have the following theorem when $|\cS|<\infty$.

\begin{theorem} (\textbf{Deterministic Code Capacity Region}) \label{theorem_avmac_capacity_restated}
Assume that $|\cS|<\infty$.
     If $\mathbf{Int}\{\cR_{d}\} \neq \varnothing$, then $\cR_{d} = \cR_{r}$.
 \end{theorem}
Therefore, when CQ-AVMAC is not symmetrizable,
the capacity region of deterministic codes equals
the capacity region of randomized codes.
     The proof of Theorem \ref{theorem_avmac_capacity_restated} directly follows from the following two lemmas mentioned below.

\begin{lemma}\label{lemma_converse_det}
    $\cR_{d} \subseteq \cR_{r}.$
\end{lemma}
\begin{proof}
    Since any deterministic code $\cC$ is a special case of a random code, it directly follows that $\cR_{d} \subseteq \cR_{r}$.
\end{proof}

\begin{lemma}\label{lemma_derand}
Assume that $|\cS|<\infty$.
    If $\mathbf{Int}\left\{\cR_{d}\right\}\neq \varnothing$ then $\cR_{d} \supseteq \cR_{r}.$ 
\end{lemma}
We prove this lemma in subsection \ref{S7-C}
 and it allows us to derandomize the code obtained in Theorem \ref{lemma_rand_capacity_avmac}.
 That is, when $|\cS|<\infty$,
the deterministic code capacity region $\cR_{d}$ is given by 
$\cR^{\star}$ defined in \eqref{Rstar}.

Next, we extend Lemma \ref{lemma_derand} to the case with a continuous set $\cS$.
For this aim, we assume the following.
\begin{assumption}\label{ASS2}
The set $\cS$ is assumed to be a compact subset of 
$\mathbb{R}^{d}$. The map $s (\in \cS)\mapsto \rho_{x,y,s}^\cB$
is assumed to be a $C^1$ continuous map for $x\in \cX,y\in\cY$.
Then, the symmetric logarithmic derivative (SLD) $L_{x,y,s|j}$ is defined as \cite[Section 6.2]{H2017QIT}
\begin{align}
\frac{1}{2}(L_{x,y,s|j} \rho_{x,y,s}^\cB+ \rho_{x,y,s}^\cB L_{x,y,s|j} )
=\frac{\partial}{\partial s^j}\rho_{x,y,s}^\cB.
\end{align}
The SLD Fisher information matrix $J_{x,y,s|j,j'}$ is defined as
\begin{align}
J_{x,y,s|j,j'} :=\Tr (\frac{\partial}{\partial s^j}\rho_{x,y,s}^\cB )L_{x,y,s|j'}.
\end{align}
It is also assumed that the map $s (\in \cS)\mapsto
J_{x,y,s|j,j'}$ is continuous.
\end{assumption}

We extend Lemma \ref{lemma_derand} as follows.
\begin{lemma}\label{lemma_derand2}
Assume Assumption \ref{ASS2}.
    If $\mathbf{Int}\left\{\cR_{d}\right\}\neq \varnothing$ then $\cR_{d} \supseteq \cR_{r}.$ 
\end{lemma}
This lemma is shown in subsection \ref{S7D}.
Combining Lemmas \ref{lemma_converse_det} and \ref{lemma_derand2}, we have
\begin{theorem} (\textbf{Deterministic Code Capacity Region}) \label{theorem_avmac_capacity_restated2}
Assume Assumption \ref{ASS2}.
     If $\mathbf{Int}\{\cR_{d}\} \neq \varnothing$, then $\cR_{d} = \cR_{r}$.
 \end{theorem}

Therefore, when Assumptions \ref{NM9} and \ref{ASS2} hold,
the deterministic code capacity region $\cR_{d}$ is given by 
$\cR^{\star}$ defined in \eqref{Rstar}.

\subsection{Proof of sufficiency in Lemma \ref{symm_avmac_cq}}\label{S7-B1} 
We will show that if a CQ-AVMAC is either $\cX-$symmetrizable or $\cY-$symmetrizable or $\cX\cY-$symmetrizable then $\mathbf{Int}\{\cR_{d}\} = \varnothing$. 
In this part, we will not use Assumption \ref{NM9}.

\subsubsection{$\cX\cY-$symmetrizable case 1}  \label{S7-B1-1}  
We will first analyze the term \eqref{error_term_combined1} by assuming that the CQ-AVMAC is $\cX\cY-$symmetrizable.
For simplicity, we assume that there exists a collection of distribution $\{U_3(. | x, y)\}_{x \in \cX, y \in \cY}$ over $\cS$ that 
achieves $0$ in \eqref{2xy_sym}.
    The error probability $e(m_1,m_2,\cC,s^n)$ can be written as follows,
    \begin{align}
       &e(m_1,m_2,\cC,s^n)  = \sum_{(\hat{m}_1,\hat{m}_2) \neq (m_1,m_2)}\tr[D_{\hat{m}_1,\hat{m}_2}\rho_{x^n(m_1),y^n(m_2),s^n}^{\cB^n}] \label{error_term_combined1}.
    \end{align}
We choose a random variable $S^n_{\bar{m}_1,\bar{m}_2}$ 
subject to the probability measure $ U_3^n(d s^n| x^n(\bar{m}_1),y^n(\bar{m}_2))$.
The expectation \par
\noindent
$\bbE_{S^n_{\bar{m}_1,\bar{m}_2}}\left[e(m_1,m_2,\cC,S^n_{\bar{m}_1,\bar{m}_2})\right]$ 
with respect to the random variable $S^n_{\bar{m}_1,\bar{m}_2}$
is calculated as follows.
    \begin{align}
    \bbE_{S^n_{\bar{m}_1,\bar{m}_2}}\left[e(m_1,m_2,\cC,S^n_{\bar{m}_1,\bar{m}_2})\right] 
    &{:=} 
    \sum_{(\hat{m}_1,\hat{m}_2) \neq (m_1,m_2)}
    \int_{\cS^n}
    U_3^n(d s^n| x^n(\bar{m}_1),y^n(\bar{m}_2))\tr[D_{\hat{m}_1,\hat{m}_2}\rho_{x^n(m_1),y^n(m_2),s^n}^{\cB^n}]\nn\\
    &= \sum_{(\hat{m}_1,\hat{m}_2) \neq (m_1,m_2)}
    \tr\left[D_{\hat{m}_1,\hat{m}_2}\bigotimes_{i=1}^{n}
    \int_{\cS}U_3(d s_i| x_i(\bar{m}_1),y_i(\bar{m}_2))\rho_{x_i(m_1),y_i(m_2),s_i}^{\cB}\right]\nn\\
    &\overset{a}{=} 
    \sum_{(\hat{m}_1,\hat{m}_2) \neq (m_1,m_2)}
     \tr\left[D_{\hat{m}_1,\hat{m}_2}\bigotimes_{i=1}^{n}\int_{\cS} U_3(ds_i| x_i(m_1),y_i(m_2))\rho_{x_i(\bar{m}_1),y_i(\bar{m}_2),s_i}^{\cB}\right]\nn\\
    &=\sum_{(\hat{m}_1,\hat{m}_2) \neq (m_1,m_2)}
    \int_{\cS^n} U_3^n(ds^n| x^n(m_1),y^n(m_2))\tr[D_{\hat{m}_1,\hat{m}_2}\rho_{x^n(\bar{m}_1),y^n(\bar{m}_2),s^n}^{\cB^n}],\label{BNT1}
    \end{align}    
 where $a$ follows from $\cX\cY-$symmetrizabilty condition mentioned in \eqref{2xy_sym}. Similarly, we have 
 \begin{equation}
\bbE_{S^n_{\bar{m}_1,\bar{m}_2}}\left[e(\bar{m}_1,\bar{m}_2,\cC,S^n_{m_1,m_2})\right] = 
\sum_{(\hat{m}_1,\hat{m}_2) \neq (\bar{m}_1,\bar{m}_2)}
\int_{\cS^n} U_3^n(ds^n| x^n(m_1),y^n(m_2))\tr[D_{\hat{m}_1,\hat{m}_2}\rho_{x^n(\bar{m}_1),y^n(\bar{m}_2),s^n}^{\cB^n}].
\label{BNT2}
 \end{equation}

We set $(\bar{m}_1,\bar{m}_2) \neq (m_1,m_2)$.
Since
any element $(\hat{m}_1,\hat{m}_2)\in \cM_n^1\times\cM_n^2
$ satisfies 
the condition $(\hat{m}_1,\hat{m}_2) \neq ({m}_1,{m}_2)$
or $(\hat{m}_1,\hat{m}_2) \neq (\bar{m}_1,\bar{m}_2)$,
\eqref{BNT1} and \eqref{BNT2} imply
 the relation 
 \begin{equation}
     \bbE_{S^n_{m_1,m_2}}\left[e(\bar{m}_1,\bar{m}_2,\cC,S^n_{m_1,m_2})\right] + \bbE_{S^n_{\bar{m}_1,\bar{m}_2}}\left[e(m_1,m_2,\cC,S^n_{\bar{m}_1,\bar{m}_2})\right] \geq 1. \label{err_lb_xy_sym}
 \end{equation}
Thus, we have
\begin{align*}
    \frac{1}{|\cM_n^1| |\cM_n^2|}
    \sum_{\hat{m}_1 =1}^{|\cM_n^1|}
    \sum_{\hat{m}_2 =1}^{|\cM_n^2|}
    \bbE_{S^n_{\hat{m}_1,\hat{m}_2}}\left[\bar{e}(\cC,S_{\hat{m}_1, \hat{m}_2}^{n})\right] 
    &\geq \frac{1}{|\cM_n^1|^2|\cM_n^2|^2}
    \sum_{\hat{m}_1 =1}^{|\cM_n^1|}\sum_{\hat{m}_2 =1}^{|\cM_n^2|}\sum_{m_1 =1}^{|\cM_n^1|}\sum_{m_2 =1}^{|\cM_n^2|}\bbE_{S^n_{\hat{m}_1,\hat{m}_2}}\left[e(m_1,m_2,\cC,S^n_{\hat{m}_1,\hat{m}_2})\right]\\
    &\overset{b}{\geq} \frac{1}{|\cM_n^1|^2|\cM_n^2|^2}\left(\begin{array}{c}
        |\cM_n^1| |\cM_n^2| \\
        2
    \end{array}\right)
    = \frac{(|\cM_n^1| |\cM_n^2|- 1)}{2\cdot |\cM_n^1||\cM_n^2|}\\
    &\overset{c}{\geq} \frac{1}{4},
\end{align*}
where $b$ follows from the fact that $\left(\begin{array}{c}
        |\cM_n^1| |\cM_n^2| \\
        2
    \end{array}\right)$
    pairs $((\bar{m}_1,\bar{m}_2) , (m_1,m_2))$
    satisfy the condition $(\bar{m}_1,\bar{m}_2) \neq (m_1,m_2)$, which implies
    \eqref{err_lb_xy_sym}. 
    $c$ follows from the fact that $|\cM_n^1||\cM_n^2| \geq 2$ since the cardinality of the smallest message set would be $2$ (when the length of message is only $1$). Thus, it directly follows that there exists a message pair $({m}_1,{m}_2) \in [|\cM_n^1|] \times [|\cM_n^2|]$ such that 
    \begin{equation*}
        \bbE_{S^n_{{m}_1,{m}_2}}\left[\bar{e}(\cC,S_{{m}_1, {m}_2}^{n})\right] \geq \frac{1}{4}.
    \end{equation*}
        Thus, the following holds
    \begin{equation}
        \sup_{s^n \in \cS^n}\bar{e}(\cC,s^{n}) \geq \frac{1}{4}.\label{BNJ7}
    \end{equation}
         Thus, if a CQ-AVMAC is $\cX\cY$-symmetrizable then $\forall R_1,R_2>0$, there does not exist an 
     $(n,|\cM_n^1|,|\cM_n^2|,\eps_n)$ deterministic code $\cC$ such that $\lim_{n \to \infty} \eps_n = 0.$ This implies that $\mathbf{Int}\{\cR_{d}\} = \varnothing$. 

\subsubsection{$\cX\cY-$symmetrizable case 2}\label{VB86}    
Next, we consider the case when 
the CQ-AVMAC is $\cX\cY-$symmetrizable and 
no collection of distribution $\{U_3(. | x, y)\}_{x \in \cX, y \in \cY}$ over $\cS$
achieves $0$ in \eqref{2xy_sym}.
In this case, there exists a sequence of 
collections of distribution $\{U_{3,k}(. | x, y)\}_{x \in \cX, y \in \cY}$ such that
\begin{equation}
\theta_{3,k} := 
\max_{x,x'\in \cX,y,y'\in \cY, z \in \cZ}
\left\|    \int_{\cS}U_{3,k}(ds|x',y')\rho_{x,y,s}^{\cB} - \int_{\cS} U_{3,k}(ds|x,y)\rho_{x',y',s}^{\cB}\right\|_1 \to 0 \hbox{ as }k \to \infty.
\label{3xy_sym}
\end{equation}
When we choose $\{U_{3,k}(. | x, y)\}_{x \in \cX, y \in \cY}$ 
as $\{U_3(. | x, y)\}_{x \in \cX, y \in \cY}$
in the discussion of Subsubsection \ref{S7-B1-1},
the difference between LHS and RHS in \eqref{BNT1} is upper bounded by 
$n \theta_{3,k}$. 
We choose a random variable $S^n_{\bar{m}_1,\bar{m}_2|k}$ 
subject to the probability measure $ U_{3,k}^n(d s^n| x^n(\bar{m}_1),y^n(\bar{m}_2))$.
Hence, instead of \eqref{err_lb_xy_sym}, we have
 \begin{equation}
     \bbE_{S^n_{m_1,m_2|k}}\left[e(\bar{m}_1,\bar{m}_2,\cC,S^n_{m_1,m_2})\right] + \bbE_{S^n_{\bar{m}_1,\bar{m}_2|k}}\left[e(m_1,m_2,\cC,S^n_{\bar{m}_1,\bar{m}_2})\right] \geq 1-n \theta_{3,k}. \label{2Terr_lb_xy_sym}
 \end{equation}
        Thus, 
        there exists a message pair $({m}_1,{m}_2) \in [|\cM_n^1|] \times [|\cM_n^2|]$ such that 
    \begin{equation}
        \bbE_{S^n_{{m}_1, {m}_2|k}}\left[\bar{e}(\cC,S_{ {m}_1,  {m}_2}^{n})\right] \geq \frac{1}{4}(1-n \theta_{3,k}).\label{VOG}
    \end{equation}
Combining \eqref{3xy_sym} and \eqref{VOG}, we obtain \eqref{BNJ7}.

\subsubsection{$\cX-$symmetrizable case}    \label{MF4}
          To discuss the case when 
a CQ-AVMAC is $\cX-$symmetrizable,
       we rewrite the error probability $e(m_1,m_2,\cC,s^n)$ as follows
    \begin{align}
       &e(m_1,m_2,\cC,s^n)  \nn\\
       &= \underbrace{\sum_{\hat{m}_1 \neq m_1}\tr[D_{\hat{m}_1,m_2}\rho_{x^n(m_1),y^n(m_2),s^n}^{\cB^n}]}_{e_1(m_1,m_2,\cC,s^n)} + \underbrace{\sum_{\hat{m}_2 \neq m_2}\tr[D_{m_1,\hat{m}_2}\rho_{x^n(m_1),y^n(m_2),s^n}^{\cB^n}]}_{e_2(m_1,m_2,\cC,s^n)} + \underbrace{\sum_{\substack{\hat{m}_1 \neq m_1 \\ \hat{m}_2 \neq m_2}}\tr[D_{\hat{m}_1,\hat{m}_2}\rho_{x^n(m_1),y^n(m_2),s^n}^{\cB^n}]}_{e_3(m_1,m_2,\cC,s^n)}.\label{error_term_combined}
    \end{align}
For simplicity, we assume that there exists a collection of distribution $\{U_1(. | x)\}_{x \in \cX}$ over $\cS$ that 
achieves $0$ in \eqref{2x_sym}.
We choose a random variable $S^n_{\bar{m}_1}$ 
subject to the probability measure $ U_1^n(d s^n| x^n(\bar{m}_1))$.
We employ the first term $e_1(m_1,m_2,\cC,s^n) $ instead of $e(m_1,m_2,\cC,s^n) $. 
The expectation 
$\bbE_{S^n_{\bar{m}_1}}\left[e_1(m_1,m_2,\cC,S^n_{\bar{m}_1})\right]$ 
with respect to the random variable $S^n_{\bar{m}_1}$
is calculated as follows.
    \begin{align}
    \bbE_{S^n_{\bar{m}_1}}\left[e_1(m_1,m_2,\cC,S^n_{\bar{m}_1})\right] 
    &{:=} 
    \sum_{\hat{m}_1 \neq m_1}
    \int_{\cS^n}
    U_1^n(d s^n| x^n(\bar{m}_1))\tr[D_{\hat{m}_1,{m}_2}\rho_{x^n(m_1),y^n(m_2),s^n}^{\cB^n}]\nn\\
    &= \sum_{\hat{m}_1 \neq m_1}
    \tr\left[D_{\hat{m}_1,{m}_2}\bigotimes_{i=1}^{n}
    \int_{\cS}U_1(d s_i| x_i(\bar{m}_1))
    \rho_{x_i(m_1),y_i(m_2),s_i}^{\cB}\right]\nn\\
    &\overset{a}{=} 
    \sum_{\hat{m}_1 \neq m_1 }
     \tr\left[D_{\hat{m}_1,{m}_2}\bigotimes_{i=1}^{n}
     \int_{\cS} U_1(ds_i| x_i(m_1))\rho_{x_i(\bar{m}_1),y_i({m}_2),s_i}^{\cB}\right]\nn\\
    &=\sum_{\hat{m}_1 \neq m_1}
    \int_{\cS^n} U_1^n(ds^n| x^n(m_1))
    \tr[D_{\hat{m}_1, {m}_2} \rho_{x^n(\bar{m}_1),y^n({m}_2),s^n}^{\cB^n}],\label{2BNT1}
    \end{align}    
 where $a$ follows from $\cX-$symmetrizabilty condition mentioned in \eqref{2x_sym}. Similarly, we have 
 \begin{equation}
\bbE_{S^n_{\bar{m}_1}}\left[e_1(\bar{m}_1, {m}_2,\cC,S^n_{m_1,m_2})\right] 
= 
\sum_{\hat{m}_1  \neq \bar{m}_1 }
\int_{\cS^n} U_1^n(ds^n| x^n(m_1) )\tr[D_{\hat{m}_1, {m}_2}\rho_{x^n(\bar{m}_1),y^n( {m}_2),s^n}^{\cB^n}].
\label{2BNT2}
 \end{equation}

We choose $\bar{m}_1  \neq m_1 $.
Since any element
$\hat{m}_1 \in \cM_n^1$ satisfies 
the condition $\hat{m}_1 \neq {m}_1$
or $\hat{m}_1 \neq \bar{m}_1$,
\eqref{2BNT1} and \eqref{2BNT2} imply
 the relation 
 \begin{equation}
     \bbE_{S^n_{m_1}}\left[e_1(\bar{m}_1,{m}_2,\cC,S^n_{m_1})\right] 
     + \bbE_{S^n_{\bar{m}_1}}\left[e_1(m_1,m_2,\cC,S^n_{\bar{m}_1})\right] \geq 1. \label{2err_lb_xy_sym}
 \end{equation}
Thus, we have, 
\begin{align*}
    \frac{1}{|\cM_n^1| }
    \sum_{\hat{m}_1 =1}^{|\cM_n^1|}
    \bbE_{S^n_{\hat{m}_1}}
    \left[\bar{e}(\cC,S_{\hat{m}_1}^{n})\right] 
    &= \frac{1}{|\cM_n^1|^2|\cM_n^2|}\sum_{\hat{m}_1 =1}^{|\cM_n^1|}    \sum_{m_1 =1}^{|\cM_n^1|}\sum_{m_2 =1}^{|\cM_n^2|}\bbE_{S^n_{\hat{m}_1}}\left[e(m_1,m_2,\cC,S^n_{\hat{m}_1})\right]\\
    &\geq \frac{1}{|\cM_n^1|^2|\cM_n^2|}
    \sum_{\hat{m}_1 =1}^{|\cM_n^1|}
    \sum_{m_1 =1}^{|\cM_n^1|}\sum_{m_2 =1}^{|\cM_n^2|}
    \bbE_{S^n_{\hat{m}_1}}\left[e_1(m_1,m_2,\cC,S^n_{\hat{m}_1})\right]\\
    &\overset{b}{\geq} \frac{1}{|\cM_n^1|^2|\cM_n^2|}
    \left(\begin{array}{c}
        |\cM_n^1| \\
        2
    \end{array}\right) \cdot |\cM_n^2|
    = \frac{(|\cM_n^1| - 1)}{2\cdot |\cM_n^1|}
    \overset{c}{\geq} \frac{1}{4},
\end{align*}
where $b$ follows from the fact that $\left(\begin{array}{c}
        |\cM_n^1| \\
        2
    \end{array}\right)$ pairs $(\bar{m}_1 , m_1)$
    satisfy $\bar{m}_1  \neq m_1  $, which yields
    \eqref{2err_lb_xy_sym}.
    $c$ follows from the fact that $|\cM_n^1| \geq 2$ since the cardinality of the smallest message set would be $2$. In the same way as the $\cX\cY-$symmetrizable case, we have
    \begin{equation}
        \sup_{s^n \in \cS^n}\bar{e}(\cC,s^{n}) \geq \frac{1}{4}.\label{BNI}
    \end{equation}
Further, when 
no collection of distribution $\{U_1(. | x)\}_{x \in \cX}$ over $\cS$ achieves $0$ in \eqref{2x_sym},
we can show \eqref{BNI} in the same way as subsubsection \ref{VB86}.

Using a similar discussion, we can derive 
     the property $\mathbf{Int}\{\cR_{d}\} = \varnothing$. 
When a CQ-AVMAC is $\cY-$symmetrizable, 
using the second term $e_2(m_1,m_2,\cC,s^n) $,
we can derive the same property.
This completes the proof of sufficiency.

\subsubsection{$\cY-$symmetrizable case}    
When a CQ-AVMAC is $\cY-$symmetrizable,
we can show \eqref{BNI} 
in the same way as subsubsection \ref{MF4}
by using $e_2(m_1,m_2,\cC,s^n) $ instead of $e_1(m_1,m_2,\cC,s^n) $.

\subsection{Proof of necessity in Lemma \ref{symm_avmac_cq}}\label{S7-B2} 
\subsubsection{Classical channel} 
We will use a contrapositive argument, i.e., we will show that if a CQ-AVMAC is non-symmetrizable (i.e. none of the \cref{x_sym,y_sym,xy_sym} satisfies), then $\mathbf{Int}\{\cR_{d}\} \neq \varnothing$. 
For this proof, we employ the results in \cite{Ahlswede1999} for the classical case, which composed 
their Lemma 1, Lemma 2, and Theorem 1.
Although the proof of Theorem 1 in \cite{Ahlswede1999}
uses the finiteness of $\cS$,
the proof of Lemma  1 in \cite{Ahlswede1999}
does not use the finiteness of $\cS$.
Lemma  1 in \cite{Ahlswede1999} is rewritten as follows.

We consider a classical arbitrary varying MAC (C-AVMAC)
$\{W_{x,y,s}\}_{x\in \cX,y\in \cY,s \in \cS}\subset \cP(\cZ)$
with $|\cX|,|\cY|,|\cZ|<\infty$.
We denote the set of possible empirical distributions on $\Omega$
with $n$ outcomes by $\cP(n,\Omega)$.
We denote the typical subspace of $P \in \cP(n,\Omega)$
by $\cT_n(P)$.

\begin{definition}[Collection of subsets 
$\{  D_{u,v|\kappa}\}_{(u,v)\in \cU\times \cV}$]
We fix $P_X\in \cP(\cX)$, $P_Y\in \cP(\cY)$,
Given small real numbers $\xi,\zeta,\zeta_1,\zeta_2>0$
and $\cU\subset \cT_n(P_X),\cV\subset \cT_n(P_Y)$,
we simplify $\kappa:=(\xi,\zeta,\zeta_1,\zeta_2)$ and
define the subsets $\{  D_{u,v|\kappa}\subset \cZ^n
\}_{(u,v)\in \cU\times \cV}$ as follows.
The subset $D_{u,v|\kappa}$ is defined as the set of elements $z^n \in \cZ^n$ that has a pair 
$(P_{XYSZ},s^n)\in 
\cP(n,\cX\times \cY\times \cS\times \cZ)
\times \cS^n$ 
satisfying the following four conditions.
\begin{description}
\item[C0)]
$(u,v,s^n,z^n)\in \cT_n(P_{XYSZ})$ and 
$D(P_{XYSZ}\|P_X\times P_Y\times P_S\times P_Z )\le \xi$.
\item[C1)]
When a joint distribution $P_{XX'YY'SS'Z}$
and $u'(\neq u)\in \cU,v'(\neq v)\in \cV,s'\in \cS$
satisfies the conditions
\begin{align}
(u,u',v,v',s,s',z^n)\in \cT_n(P_{XX'YY'SS'Z}) \\
D(P_{X'Y'S'Z}\|P_{X'}\times P_{Y'}\times P_{S'}\times P_Z )\le \xi,
\end{align}
we have
\begin{align}
I(XYZ;X'Y'|S)< \zeta.
\end{align}
\item[C2)]
When a joint distribution $P_{XX'YY'SS'Z}$
and $u'(\neq u)\in \cU,s'\in \cS$
satisfies the conditions
\begin{align}
(u,u',v,s,s',z^n)\in \cT_n(P_{XX'YSS'Z}) \\
D(P_{X'YS'Z}\|P_{X'}\times P_{Y}\times P_{S'}\times P_Z )\le \xi,
\end{align}
we have
\begin{align}
I(XYZ;X'Y|S)< \zeta_1.
\end{align}
\item[C3)]
When a joint distribution $P_{XYY'SS'Z}$
and $v'(\neq v)\in \cV,s'\in \cS$
satisfies the conditions
\begin{align}
(u,v,v',s,s',z^n)\in \cT_n(P_{XYY'SS'Z}) \\
D(P_{XY'S'Z}\|P_{X}\times P_{Y'}\times P_{S'}\times P_Z )\le \xi,
\end{align}
we have
\begin{align}
I(XYZ;Y'|S)< \zeta_2.
\end{align}
\end{description}
\end{definition}

Then Reference \cite{Ahlswede1999} showed the following.
\begin{proposition}[\protect{\cite[Lemma 1]{Ahlswede1999}}]\label{OO2}
Assume that an AVMAC
$\{W_{x,y,s}\}_{x\in \cX,y\in \cY,s \in \cS}$ with $|\cX|,|\cY|,|\cZ|<\infty$
is non-symmetrizable.
We fix $\alpha,\beta>0$.
We choose positive numbers $\kappa=(\xi,\zeta,\zeta_1,\zeta_2)$ 
to satisfy the following conditions\footnote{
Lemma 1 of \cite{Ahlswede1999} gives the concrete choice of 
$\kappa=(\xi,\zeta,\zeta_1,\zeta_2)$ gives in its proof.
In their choice of $\kappa$,
a constant $c$ is used to express
the relation between the total variation distance and Kullback–Leibler divergence in the application of Pinsker's inequality.
Since this coefficient is known to $\sqrt{2}$, we used their choice with 
$c=\sqrt{2}$ in \eqref{NMT7}.}
\begin{align}
\xi+\zeta_1< \frac{1}{8}\alpha^2\beta^2 \theta_1^2,\quad
\xi+\zeta_2< \frac{1}{8}\alpha^2\beta^2 \theta_2^2,\quad
\xi+\zeta < \frac{1}{64}|\cX|^{-2}|\cY|^{-2}|\cZ|^{-2}\alpha^4\beta^4 \theta^2,\label{NMT7}
\end{align}
where $\theta_1$, $\theta_2,$ $\theta$ are defined in
Definition \ref{2defsym}.
When $P_X \in \cP(n,\cX)$, $P_Y \in \cP(n,\cY)$ 
satisfy 
$\min_{x \in \cX}P_X(x) \ge \alpha$
and $\min_{y \in \cY}P_Y(y) \ge \beta$,
for any two subsets $\cU\subset \cT_n(P_X),\cV\subset \cT_n(P_Y)$,
the collection of subsets 
$\{  D_{u,v|\kappa}\}_{(u,v)\in \cU\times \cV}$ are disjoint.
\end{proposition}

Lemma 1 of \cite{Ahlswede1999} showed 
the contraposition of the above proposition.
That is, it showed the following; 
When the collection of subsets 
$\{  D_{u,v|\kappa}\}_{(u,v)\in \cU\times \cV}$ are not disjoint,
one of the relation in \eqref{NMT7} does not hold.

Hence, when $\kappa$ satisfies the condition of Proposition \ref{OO2}
and two subsets $\cU,\cV$ are given,
we have the tuple $\cU,\cV,\{  D_{u,v|\kappa}\}_{(u,v)\in \cU\times \cV}$ forms a code. 
This code is written as
$\cC[\kappa, \cU,\cV]$.
We also define
\begin{align}
 \gamma(n,|\cX|,|\cY|,|\cS|,|\cZ|)
:=(n+1)^{|\cX|^2|\cY|^2|\cS||\cZ|}.
\end{align}
Then, 
the combination of Lemma 2 and Theorem 1 of \cite{Ahlswede1999}
is rewritten as follows.

\begin{proposition}[\protect{\cite[Lemma 2 \& Theorem 1]{Ahlswede1999}}]
Assume that an AVMAC
$\{W_{x,y,s}\}_{x\in \cX,y\in \cY,s \in \cS}$ with $|\cX|,|\cY|,|\cZ|<\infty$
is non-symmetrizable.
Given $\alpha,\beta>0$,
we choose positive numbers $\kappa=(\xi,\zeta,\zeta_1,\zeta_2)$ 
according to Proposition \ref{OO2}.
We choose $\varepsilon',\delta,r>0$ as
\begin{align}
0<\varepsilon'<\frac{\xi}{2}, 
0<\varepsilon'<\delta<r<\frac{\zeta^*}{11},\label{BP1}
\end{align}
where $\zeta^*:= \min(\zeta,\zeta_1,\zeta_2)$.
We choose a finite subset $\cS_{f} \subset \cS$ and two distributions 
$P_X \in \cP(n,\cX)$ and $P_Y \in \cP(n,\cY)$ as
\begin{align}
\min_{x \in \cX}P_X(x) \ge \alpha,\quad
\min_{y \in \cY}P_Y(y) \ge \beta.\label{BNO}
\end{align}
Then, there exist 
two subsets $\cU\subset \cT_n(P_X),\cV\subset \cT_n(P_Y)$
such that
$|\cU|=|\cV|=M_r:=e^{nr}$ and
    \begin{align}
&\frac{1}{M_r^2}\sum_{m_1=1}^{M_r}\sum_{m_2=1}^{M_r}
e(m_1,m_2,\cC[\kappa, \cU,\cV],s^n) \nn\\
\le & 4 \gamma(n,|\cX|,|\cY|,|\cS_{f}|,|\cZ|) \exp(n 
\max(-\xi/2,8r+2\eps'-\zeta,4r+2\eps'-\zeta_1,4r+2\eps'-\zeta_2))
    \end{align}
    for $s^n \in \cS_{f}^n$.
\end{proposition}
 
When we choose $\eps $ such that 
\begin{align}
\log(1+\eps) < -\max(-\xi/2,8r+2\eps'-\zeta,4r+2\eps'-\zeta_1,4r+2\eps'-\zeta_2),\label{BP2}
\end{align}
using the property \eqref{epsilon_net_state_density_eq}, we can rewrite 
the above proposition as follows.

\begin{lemma}\label{BFT}
Assume that an AVMAC
$\{W_{x,y,s}\}_{x\in \cX,y\in \cY,s \in \cS}$ with $|\cX|,|\cY|,|\cZ|<\infty$
is non-symmetrizable and satisfie Assumption \ref{NM9}.
Given $\alpha,\beta>0$,
we choose positive numbers $\kappa=(\xi,\zeta,\zeta_1,\zeta_2)$ 
according to Proposition \ref{OO2}.
We choose 
$\varepsilon',\eps,\delta,r>0$, 
a subset $\cS_\eps\subset \cS$,
and two distributions 
$P_X \in \cP(n,\cX)$ and $P_Y \in \cP(n,\cY)$ 
as \eqref{BP1} and \eqref{BP2}, \eqref{epsilon_net_state_density_eq}, and
\eqref{BNO}, respectively.
Then, there exist 
two subsets $\cU\subset \cT_n(P_X),\cV\subset \cT_n(P_Y)$
such that
$|\cU|=|\cV|=m_r:=e^{nr}$ and
    \begin{align}
&\frac{1}{M_r^2}\sum_{m_1=1}^{M_r}\sum_{m_2=1}^{M_r}
e(m_1,m_2,\cC[\kappa, \cU,\cV],s^n) \nn\\
\le & 4 \gamma(n,|\cX|,|\cY|,|\cS_\eps|,|\cZ|)(1+\eps)^n
 \exp(n 
\max(-\xi/2,8r+2\eps'-\zeta,4r+2\eps'-\zeta_1,4r+2\eps'-\zeta_2))
\label{BF8}
    \end{align}
    for $s^n \in \cS^n$.
\end{lemma}
Since \eqref{BF8} goes to zero exponentially,
the above lemma is simplified to the following lemma.

\begin{lemma}\label{NMT}
Assume that an AVMAC
$\{W_{x,y,s}\}_{x\in \cX,y\in \cY,s \in \cS}$
with $|\cX|,|\cY|,|\cZ|<\infty$ satisfies 
Assumption \ref{NM9} and
is non-symmetrizable.
Then, the relation 
$\mathbf{Int}\{\cR_{d}\} \neq \varnothing$ holds for
the AVMAC $\{W_{x,y,s}\}_{x\in \cX,y\in \cY,s \in \cS}$.
\end{lemma}

\subsubsection{CQ-channel} 
Now, we start with the proof of necessity part of Lemma \ref{symm_avmac_cq}.
Similar to \cite[Lemma 2]{Ahlswede07},
we reduce the problem of Lemma \ref{symm_avmac_cq}
to the classical case.
We choose a POVM 
$\left\{\cO_z^{\cB}\right\}_{z=1}^{\abs{\cB}^2}$
such that the operators $\cO_1^{\cB}, \ldots, \cO_{\abs{\cB}^2}^{\cB}$
are linearly independent, i.e., they linearly span 
the set of Hermitian matrices on $\cH_{\cB}$.
Then, we define the channel 
$P_{x,y,s}(z):= \Tr \cO_z^{\cB}\rho_{x,y,s}^{\cB}$.
Since the linear map mapping a Hermitian matrix $X$ to 
the vector $ (\Tr \cO_z^{\cB} X)_{z=1}^{\abs{\cB}^2}$ is invertible,
the non-symmetrizable property of 
the CQ-AVMAC $\left\{\rho_{x,y,s}^{\cB}\right\}$
is equivalent to 
the non-symmetrizable property of C-AVMAC $\{P_{x,y,s}(z)\}$.
Therefore, 
$P_{x,y,s}(z)$ is non-symmetrizable.
Also, since 
the CQ-AVMAC $\left\{\rho_{x,y,s}^{\cB}\right\}$ satisfies 
Assumption \ref{NM9},
the C-AVMAC $\{P_{x,y,s}(z)\}$ also 
satisfies Assumption \ref{NM9}.
Applying Lemma \ref{NMT} to 
the C-AVMAC $\{P_{x,y,s}(z)\}$,
we can achieve the positive rate pair.
Since the capacity region of 
the CQ-AVMAC $\left\{\rho_{x,y,s}^{\cB}\right\}$
contains 
the capacity region of 
the C-AVMAC $\{P_{x,y,s}(z)\}$,
this proves Lemma \ref{symm_avmac_cq}.

\endproof

\subsection{Proof of Lemma \ref{lemma_derand}}\label{S7-C}
From Lemma \ref{symm_avmac_cq}, it follows that $\mathbf{Int}\{\cR_{d}\} \neq \varnothing$ is equivalent to CQ-AVMAC being non-symmetrizable (none of the conditions mentioned in \cref{x_sym,y_sym,xy_sym} are satisfied). We will use the achievability result when the code is randomized (as discussed in Theorem \ref{lemma_rand_capacity_avmac}) and we will require to show that for any random code achieving a rate-pair $(R_1,R_2)\in \mathbf{Int}\{\cR_{d}\}$, 
there always exists a fixed code $\cC$ which achieves the rate-pair $(R_1,R_2)$ as well i.e., 
$\cR_{d} \supseteq \mathbf{Int}\{\cR_{d}\}$, which implies
the relation $\cR_{d} \supseteq \cR_{r}$. 
We will accomplish this using the technique of elimination of randomness which is a standard tool used in computer science and has also been used in \cite{Ahlswede1978,Ahlswede1980,Jahn1981,Ericcon1985,gubner1990deterministic,Ahlswede1999,Ahlswede07}. 

\subsubsection{Structure of our proof}
In the following proof, we fix a rate-pair $(R_1,R_2)\in \mathbf{Int}\{\cR_{d}\}$, and will show the existence of a fixed code $\cC$ which achieves the rate-pair $(R_1,R_2)$.
For this aim, we concatenate two codes by setting the real number $t$ as $t > 1/2$.
The first one has $n^{2t}$ codewords and length $v(n):\omega(\log n) = o(n)$ (where $v(n) \overset{n \to \infty}{\longrightarrow} \infty$, $o(\cdot)$ and $\omega(\cdot)$ are defined in Definition \ref{smallo} and \ref{smallomega} respectively).
The second one has $2^{n(R_1+R_2)}$ codewords and 
length $n$. 
Then, concatenating them, we
construct a code $\hat{C}_{\text{concat}}$ with $n^{2t}\cdot 2^{n(R_1+R_2)}$ codewords of length $(v(n) + n)$, which will be shown to achieve the rate-pair $(R_1,R_2)$.

\subsubsection{First code construction}
Since $\mathbf{Int}\left\{\cR_{d}\right\}\neq \varnothing$, 
for an arbitrary small real number $\mu >0$, 
there exists a deterministic code with the form;
\begin{align}
\mathcal{C}':\left\{\left((\bar{x}^{v(n)}(\hat{m}_1),\bar{y}^{v(n)}(\hat{m}_2)), \hat{D}_{\hat{m}_1,\hat{m}_2}\right) : (\hat{m}_1,\hat{m}_2) \in [n^t] \times [n^t]\right\}
\end{align} such that 
the probability of decoding error satisfies 
\begin{equation}
    \max_{s^{v(n)} \in \cS^{v(n)}} \bar{e}(\mathcal{C}',s^{v(n)}) {\leq} \mu.\label{errvfiix} 
\end{equation}
The existence of the above code can be shown as follows.
We choose an element $(R_1',R_2') \in \mathbf{Int}\left\{\cR_{d}\right\}$. 
Since $\lim_{n \to \infty} \frac{\log n}{v(n)} = 0$, 
the rates of the code $\mathcal{C}'$ for the first and second messages
are smaller than $R_1'$ and $R_2'$.
Hence, the above code exists.

\subsubsection{Second code construction}
From the achievability of Theorem \ref{lemma_rand_capacity_avmac}, 
there exists a random code  $\cC^{\gamma_{1},\gamma_{2}}$ 
characterized by two random variables $\gamma_{1},\gamma_{1}$
taking values in finite non-empty support sets 
$\Gamma_{1},\Gamma_{2}$ that satisfies the following conditions.
\begin{description}
\item[(1)] 
The two random variables $\gamma_{1},\gamma_{1}$ are subject to two probability distributions 
$G_1 \in \cP(\Gamma_{1}),G_2 \in \cP(\Gamma_{2})$
over the support sets $\Gamma_{1},\Gamma_{2}$, respectively. 
\item[(2)] 
The random code $\cC^{\gamma_{1},\gamma_{2}}$ has the following form;
    \begin{equation}
        \cC^{\gamma_{1},\gamma_{2}}= \left\{\left((x^{n}_{\gamma_{1}}(m_1),y^{n}_{\gamma_{2}}(m_2)), D_{m_1,m_2}^{\gamma_{1},\gamma_{2}}\right):(m_1,m_2) \in [2^{nR_1}] \times [2^{nR_2}]\right\}.\label{NNM2}
    \end{equation}  
\item[(3)] 
The code $\cC^{\gamma_{1},\gamma_{2}}$ satisfies 
the condition
    \begin{equation}
    \max_{s^n \in \cS^n}\bbE_{\substack{\gamma_{1} \sim G_{1},\\ \gamma_{2}\sim G_{2}}}\left[e(\cC^{\gamma_{1},\gamma_{2}},s^n)\right] < 
    \frac{e^{\mu} - 1}{e^{2}}
.\label{NNM}
    \end{equation}  
\end{description}
When we choose the two random variables $\gamma_{1}$ and $\gamma_{2}$
to be
$\{X^n(m_1) \in \cX^n : m_1 \in [2^{nR_1}]\}$
and $\{Y^n(m_2) \in \cY^n : m_2 \in [2^{nR_2}]\}$,
our proof of Theorem \ref{lemma_rand_capacity_avmac}
guarantees the existence of the above code.
The choice of the two random variables $\gamma_{1}$ and $\gamma_{2}$ depends on $n$, but this dependency causes no problem in the latter discussion.

We now consider $n^{2t}$ i.i.d random variable pairs $\left\{Z_{i,j}\right\}_{\substack{i \in [n^t], j \in [n^t]}}$, 
which take values in $\Gamma_{1}\times\Gamma_{2}$ such that 
for any $(i,j) \in [n^t]\times[n^t],$ 
$\Pr\{Z_{i,j}=(\hat{\gamma}_1,\hat{\gamma}_2)\} = G_{1}(\hat{\gamma}_1) \cdot G_{2}(\hat{\gamma}_2).$ 
For any particular $s^n \in \cS^n$, 
since the random variables $\{e(\cC^{Z_{i,j}},s^n)\}$ 
are independent and identically distributed,
due to \eqref{NNM}, any $(i,j) \in [n^t]\times[n^t]$ satisfies
    \begin{equation}
        \bbE_{\substack{Z_{i,j}}}\left[e(\cC^{Z_{i,j}},s^n)\right] <     \frac{e^{\mu} - 1}{e^{2}}. \label{XMNO}
    \end{equation}
Since $e(\cC^{Z_{i,j}},s^n) \in (0,1)$, we have:
\begin{align*}
    \bbE\left[e^{2 e(\cC^{Z_{i,j}},s^n)
    }\right] 
    &= 1 + \sum_{k=1}^\infty\frac{2^k}{k!}\bbE_{\substack{Z_{i,j}}}\left[e(\cC^{Z_{i,j}},s^n)^k \right] \\
    &\leq 1 + \sum_{k=1}^\infty\frac{2^k}{k!}\bbE_{\substack{Z_{i,j}}}\left[e(\cC^{Z_{i,j}},s^n) \right] \\
    &\overset{a}{\leq} 1 + \frac{e^{\mu} - 1}{e^{2}}\sum_{k=0}^\infty\frac{2^k}{k!}= 1 + \frac{e^{\mu} - 1}{e^{2}} e^{2}=e^\mu,
\end{align*}
where $a$ follows from \eqref{XMNO}.
Therefore, we have
\begin{align}
    \text{err}(n,s^n) &:= \Pr_{G_1\times G_2}
    \left\{\sum_{\substack{i \in [n^t] \\ j \in [n^t]}}e(\cC^{Z_{i,j}},s^n) > n^{2t} \mu\right\}\label{errorlemma9}\\
    &\overset{b}{\leq} e^{-2 \mu n^{2t}}\cdot
    \prod_{\substack{i \in [n^t] \\ j \in [n^t]}}\bbE\left[e^{2 \cdot
    e(\cC^{Z_{i,j}},s^n)}\right] \le e^{-2\mu n^{2t}}e^{\mu n^{2t}} 
    = e^{-\mu n^{2t}},
\end{align}
where $b$ follows from Fact \ref{chernoff} where $\alpha > 0$ is a constant. 
Thus, we have
\begin{align}
\Pr_{G_1\times G_2}\left\{\sum_{\substack{i \in [n^t] \\ j \in [n^t]}}e(\cC^{Z_{i,j}},s^n) > n^{2t} \mu, ~\exists s^n \in \cS^n
\right\}
\le \sum_{s^n \in \cS^n} \text{err}(n,s^n)
\le \sum_{s^n \in \cS^n} e^{-\mu n^{2t}}
=|\cS|^n e^{-\mu n^{2t}} \to 0.\label{BZO}
\end{align}
Hence, there exists a subset $\left\{(\gamma_{1,i},\gamma_{2,j})
\right\}_{(i,j) \in [n^t]\times[n^t]}\subset \Gamma_{1}\times\Gamma_{2}$ such that 
\begin{equation}
    \frac{1}{n^{2t}}\sum_{(i,j) \in [n^t]\times[n^t]}
    e(\cC^{(\gamma_{1,i},\gamma_{2,j})},s^n)
 \leq \mu, ~\forall s^n \in \cS^n .\label{err_n4}
\end{equation}

\subsubsection{Concatenated code}
Concatenating the above two codes,
we now construct a new code $\hat{C}_{\text{concat}}$ with 
$n^{2t}\cdot 2^{n(R_1+R_2)}$ codewords of length $(v(n) + n)$ given as follows
\begin{equation*}
    \hat{C}_{\text{concat}} := \left\{\left((\bar{x}^{v(n)}(\hat{m}_1), x^{n}_{\gamma_{1,\hat{m}_1}}(m_1), \bar{y}^{v(n)}(\hat{m}_2), y^{n}_{\gamma_{2,\hat{m}_2}}(m_2)), \hat{D}_{\hat{m}_1,\hat{m}_2} \otimes D_{m_1,m_2}^{\gamma_{1,\hat{m}_1},\gamma_{2,\hat{m}_2}}\right) : \substack{(\hat{m}_1,\hat{m}_2) \in [n]\times[n], \\ (m_1,m_2)\in [2^{nR_1}]\times[2^{nR_2}]}\right\}.
\end{equation*}
Here, the symbol $\bar{x}^{v(n)}(\hat{m}_1), x^{n}_{\gamma_{1,\hat{m}_1}}(m_1)$ expresses the 
concatenation of
$\bar{x}^{v(n)}(\hat{m}_1)$ and $x^{n}_{\gamma_{1,\hat{m}_1}}(m_1)$.

For each $s^{v(n)+n} := s_1^{v(n)}\cdot s_2^n  \in \cS^{v(n) + n}$ we have
\begin{align*}
    &\frac{1}{n^{2t}\cdot 2^{n(R_1+R_2)}}
    \sum_{\substack{\hat{m}_1 \in [n^t] \\ \hat{m}_1 \in [n^t] \\ m_1 \in [2^{nR_1}] \\ m_2 \in [2^{nR_1}]}}
    \tr\left[\left(\hat{D}_{\hat{m}_1,\hat{m}_2} \otimes D_{m_1,m_2}^{\gamma_{1,\hat{m}_1},\gamma_{2,\hat{m}_2}}\right)
    \left(\rho_{\bar{x}^{v(n)}(\hat{m}_1),\bar{y}^{v(n)}(\hat{m}_2),s_1^{v(n)}} \otimes \rho_{x^{n}_{\gamma_{1,\hat{m}_1}}(m_1),y^{n}_{\gamma_{2,\hat{m}_2}}(m_2),s_2^n}\right)\right]\\
    &= \frac{1}{n^{2t}\cdot 2^{n(R_1+R_2)}}
    \sum_{\substack{\hat{m}_1 \in [n^t] \\ \hat{m}_1 \in [n^t] \\ m_1 \in [2^{nR_1}] \\ m_2 \in [2^{nR_1}]}}
    \tr\left[\hat{D}_{\hat{m}_1,\hat{m}_2} \cdot\rho_{\bar{x}^{v(n)}(\hat{m}_1),\bar{y}^{v(n)}(\hat{m}_2),s_1^{v(n)}}\right] 
    \tr\left[D_{m_1,m_2}^{\gamma_{1,\hat{m}_1},\gamma_{2,\hat{m}_2}}\cdot\rho_{x^{n}_{\gamma_{1,\hat{m}_1}}(m_1),y^{n}_{\gamma_{2,\hat{m}_2}}(m_2),s_2^n}\right].
\end{align*}

Now, for any $(\hat{m}_1,\hat{m}_2) \in [n^t]\times[n^t]$, if we consider,
\begin{align*}
    \alpha_{\hat{m}_1,\hat{m}_2} &:= \tr\left[\hat{D}_{\hat{m}_1,\hat{m}_2} \cdot\rho_{\bar{x}^{v(n)}(\hat{m}_1),\bar{y}^{v(n)}(\hat{m}_2),s_1^{v(n)}}\right],\\
    \beta_{\hat{m}_1,\hat{m}_2} &:= \frac{1}{2^{n(R_1+R_2)}}\sum_{\substack{m_1 \in [2^{nR_1}] \\ m_2 \in [2^{nR_1}]}}\tr\left[D_{m_1,m_2}^{\gamma_{1,\hat{m}_1},\gamma_{2,\hat{m}_2}}\cdot\rho_{x^{n}_{\gamma_{1,\hat{m}_1}}(m_1),y^{n}_{\gamma_{2,\hat{m}_2}}(m_2),s_2^n}\right],
\end{align*}
then, from Fact \ref{fact_robust}, the following holds,
\begin{align*}
    &\frac{1}{n^{2t}\cdot 2^{n(R_1+R_2)}}\sum_{\substack{\hat{m}_1 \in [n^t] \\ \hat{m}_1 \in [n^t] \\ m_1 \in [2^{nR_1}] \\ m_2 \in [2^{nR_1}]}}\tr\left[\left(\hat{D}_{\hat{m}_1,\hat{m}_2} 
    \otimes D_{m_1,m_2}^{\gamma_{1,\hat{m}_1},\gamma_{2,\hat{m}_2}}\right)
    \left(\rho_{\bar{x}^{v(n)}(\hat{m}_1),\bar{y}^{v(n)}(\hat{m}_2),s_1^{v(n)}} \otimes \rho_{x^{n}_{\gamma_{1,\hat{m}_1}}(m_1),y^{n}_{\gamma_{2,\hat{m}_2}}(m_2),s_2^n}\right)\right]\\
    &= \frac{1}{n^{2t}}\sum_{\substack{\hat{m}_1 \in [n^t] \\ 
    \hat{m}_1 \in [n^t]}}\alpha_{\hat{m}_1,\hat{m}_2}\beta_{\hat{m}_1,\hat{m}_2}
    \overset{c}{\geq} 1 - 2\mu,
\end{align*}
where $c$ follows from the following series of inequalities,
\begin{align*}
    \frac{1}{n^{2t}}\sum_{\substack{\hat{m}_1 \in [n^t] \\ \hat{m}_1 \in [n^t]}}\alpha_{\hat{m}_1,\hat{m}_2} \overset{d}{\geq} 1 - \mu,\quad
    \frac{1}{n^{2t}}\sum_{\substack{\hat{m}_1 \in [n^t] \\ \hat{m}_1 \in [n^t]}}\beta_{\hat{m}_1,\hat{m}_2} \overset{e}{\geq} 1 - \mu,
\end{align*}
where $d$ and $e$ follows from \cref{err_n4,errvfiix} respectively. Thus, for the concatenated code $\hat{C}_{\text{concat}}$ we have,
\begin{equation}
    \max_{s^{v(n)+n} \in \cS^{v(n) + n}}\bar{e}(\hat{C}_{\text{concat}}, s^{v(n)+n}) < 2\mu. \label{ZXU}
\end{equation}
Since $\mu$ is an arbitrarily small number and
\begin{align}
\lim_{n\to \infty} \frac{\log (2^{nR_1}n)}{v(n)+n} = R_1,\quad
\lim_{n\to \infty} \frac{\log (2^{nR_2}n)}{v(n)+n} = R_2,
\end{align}
the rate pair $(R_1,R_2)$ is achievable.
This completes the proof of Lemma \ref{lemma_derand}. 
\endproof

\begin{remark}
In the single sender case,
we can show the lemma corresponding to Lemma \ref{lemma_derand}
by modifying the proof of Lemma \ref{lemma_derand}
so that 
the size of Alice sends in the first code is $n^{2t}$
and the size of Bob sends in the first code is $0$.
\end{remark}

\subsection{Proof of Lemma \ref{lemma_derand2}}\label{S7D}
\subsubsection{Use of Proof of Lemma \ref{lemma_derand}}
We choose the set 
$\cS_{n}:= \cS \cap \left\{\left(\frac{k_1}{n},\ldots, \frac{k_d}{n}\right):
k_j \in \mathbb{Z}\right\}$.
Since $\cS$ is a compact subset,
$|\cS_{n}|=O(n^d)$.
Given an arbitrary small number $\mu>0$,
substituting $\cS_n$ into $\cS$,
we apply the same discussion as the proof of Lemma \ref{lemma_derand} given in Subsection \ref{S7-C}.
In this case, only \eqref{BZO} requires a different discussion.
The upper bound given in \eqref{BZO} is calculated as
\begin{align}
|\cS_{n}|^n e^{-\mu n^{2t}} =O(n^{nd}) e^{-\mu n^{2t}}=O(e^{ nd \log n -\mu n^{2t}})\to 0.
\end{align}
Therefore, due to \eqref{ZXU}, we have
\begin{equation}
 \max_{s^{v(n)+n} \in \cS_n^{v(n) + n}}
\bar{e}(\hat{C}_{\text{concat}}, s^{v(n)+n}) < 2\mu \label{ZXU2}
\end{equation}
for sufficiently large $n$.
\subsubsection{Use of Assumption \ref{ASS2}}
We employ 
Fidelity $F(\rho,\sigma):=\Tr |\sqrt{\rho}\sqrt{\sigma}|$,
Bure's distance 
$d_B(\rho,\sigma):= \sqrt{2(1-F(\rho,\sigma))}$,
and sandwich Renyi divergence of order $1/2$
$D_{1/2}(\rho\|\sigma):= - \log F(\rho,\sigma)^2$.
We have
\begin{align}
\frac{1}{2} \|\rho-\sigma\|_1 \le \sqrt{1-F(\rho,\sigma)^2}
=\sqrt{1-e^{-D_{1/2}(\rho\|\sigma)}}.\label{VBU}
\end{align}
Since $d_B(\rho_{x,y,s}^\cB,\rho_{x,y,\bar{s}}^\cB)^2$ is approximated to be 
Fisher information times $\|s-\bar{s}\|^2$ \cite[Eq.(6.33)]{H2017QIT},
Assumption \ref{ASS2} guarantees 
\begin{align}
\max_{s\in \cS}\min_{\bar{s}\in \cS_n} 
\max_{x\in \cX,y\in \cY}
d_B(\rho_{x,y,s}^\cB,\rho_{x,y,\bar{s}}^\cB)^2=O(\frac{1}{n^2}).
\end{align}
Remember that Assumption \ref{ASS2} contains 
the compactness of $\cS$, which implies the existence of the maximum.
Thus,
\begin{align}
\max_{s\in \cS}\min_{\bar{s}\in \cS_n} 
\max_{x\in \cX,y\in \cY}
D_{1/2}(\rho_{x,y,s}^\cB\|\rho_{x,y,\bar{s}}^\cB)=O(\frac{1}{n^2}).
\end{align}
Hence,
\begin{align}
\max_{s^n\in \cS^n}\min_{{\bar{s}}^n\in \cS_n^n} 
\max_{x^n \in \cX^n,y^n\in \cY^n}
D_{1/2}(\rho_{x^n,y^n,s^n}^{\cB^n}\|\rho_{x^n,y^n,\bar{s}^n}^{\cB^n})
=O(\frac{1}{n}).\label{VBU2}
\end{align}
The combination of \eqref{VBU} and \eqref{VBU2} implies 
\begin{align}
\max_{s^n\in \cS^n}\min_{{\bar{s}}^n\in \cS_n^n} 
\max_{x^n \in \cX^n,y^n\in \cY^n}
\|\rho_{x^n,y^n,s^n}^{\cB^n}-\rho_{x^n,y^n,\bar{s}^n}^{\cB^n}\|_1
\le O(\frac{1}{\sqrt{n}}).\label{VBU3}
\end{align}
The combination of \eqref{ZXU2} and \eqref{VBU3} implies 
\begin{equation}
 \max_{s^{v(n)+n} \in \cS^{v(n) + n}}
\bar{e}(\hat{C}_{\text{concat}}, s^{v(n)+n}) < 2\mu+ O(\frac{1}{\sqrt{n}})\label{ZXU3}
\end{equation}
for sufficiently large $n$.
Since $\mu$ is an arbitrarily small number,
the rate pair $(R_1,R_2)$ is achievable.
This completes the proof of Lemma \ref{lemma_derand2}. 
\endproof

\section{Conclusion} 

We have studied the problem of independence testing by building on the techniques developed in \cite{Hayashi2009,HT}. 
The aim here is to test whether the quantum system is correlated with the classical subsystem or is independent of it. 
We have solved this problem by obtaining a universal test that does not depend on the form of the classical-quantum state. 
Further, our test has the property that it is permutation invariance. 
This property plays an important role in studying this problem.
Using this test, we have studied the problem of reliable communication over a classical-quantum arbitrarily varying point-to-point channel (CQ-AVC) and 
have obtained an optimal universal decoder for this problem.

Next, we have generalized the problem of ``independence testing'' where we have more than $2$ hypothesis based on different types of independence between the quantum subsystem and classical subsystems. 
We have studied this problem by designing two different types of hypothesis testing techniques. 
In the first technique, we have designed a collection of non-simultaneous universal tests building on the techniques discussed in \cite{Hayashi2009,HT,HC}, where each test classifies different types of independence between the quantum subsystem and classical subsystems. 
In the second technique, we have constructed a simultaneous universal test building on the techniques discussed in \cite{HT} and \cite{Sen2021}. 
The simultaneous tests obtained in the second technique works like a union of the non-simultaneous universal tests obtained in the first technique. 
However, both these techniques has their advantages and disadvantages as summarized later.

We have studied the problem of reliable communication over classical-quantum multiple-access-arbitrarily varying channels (CQ-AVMAC). 
We have given a proof of achievability for the capacity region of CQ-AVMAC under randomization of codes by designing a simultaneous decoding POVM. 
We have constructed two kinds of simultaneous decoders with the help of the universal tests (both simultaneous and non-simultaneous tests) 
while studying the problem of generalized independence testing. 
We have observed that both of these methods can be extended to the case of more than $2$ senders. 
We have established the optimality of the capacity region (under randomization) by proving the converse. 
While the capacity region of CQ-AVMAC for fixed codes may be empty sometimes,
we have shown that it will always be non-empty under a certain necessary and sufficient condition, which we denote as ``non-symmetrizability'' condition. 
We have then demonstrated a derandomization technique to derandomize any random code under the non-symmetrizability condition and finally establish a capacity region if a CQ-AVMAC satisfies this condition.

\begin{table}[ht]
    \centering
    \caption{Comparision between Technique $1$ and Technique $2$ with respect to the  decoder for CQ-AVMAC}
    \label{tab:table1}
    \begin{tabular}{ || m{3cm} || m{7cm} || m{7cm} ||}
    \hline
    \hline
   Comparison & Technique $1$ & Technique $2$\\
   [0.5ex]
    \hline
         Extension for $T$ ($T>2$) senders. & It can be extended for $T$ senders. & It can be extended for $T$ senders using similar techniques required to prove \cite[Corollary 4]{Sen2021}. \\
         \hline
         Enlargement of the Hilbert spaces of subsystems 
         & Enlargement of the Hilbert spaces of subsystems is not required. & Requires enlargement of the dimension of the Hilbert spaces of subsystems. This enlarged dimension grows exponentially with the number of senders. \\
         \hline
         Sandwiching of Projectors in the Decoding POVM  & Sandwiching of $O(2^{T-1})$ projectors required to construct the decoding POVM for the case of $T$ senders. & The construction of decoder POVM does not require any sandwiching of projectors.\\
         \hline
          Tools involved &  This technique is based on Schur-Weyl duality and the method of types. & This technique is based on Schur-Weyl duality, the method of types and geometric ideas developed in \cite{Sen2021}, \\
          \hline
         Hypothesis Testing  & The decoding POVM is obtained from multiple non-simultaneous tests. & The decoding POVM is obtained from a simultaneous test.\\
         \hline Universal & It is universal and therefore does not depend on the channel. & It is universal and therefore does not depend on the channel.\\
         \hline
         Error analysis & The error analysis is simple. & The error analysis is arguably mathematically tedious.\\
         \hline
         Time sharing & This is a simultaneous decoder and hence, does not require time sharing. & This is a simultaneous decoder and hence, does not require time sharing.\\
         
          \hline
          \hline
    \end{tabular}
\end{table} 

In table \ref{tab:table1}, we have compared the two different techniques used to prove the achievability of Theorem \ref{lemma_rand_capacity_avmac}. In the table below, we use the term ``Technique $1$'' to denote the technique of non-simultaneous hypothesis test mentioned in Theorem \ref{lemma_rand_capacity_avmac} and the term ``Technique $2$'' to denote the technique of simultaneous hypothesis test mentioned in Appendix \ref{alternative_proof_lemma_rand_capacity_avmac}. 

Finally, we remark on the relation with the compound channel.
In the single-sender case, 
the randomized code setting in a compound channel
can be done in the same way as Section \ref{sec:section_CQAVC}
by restricting distributions $Q$ to delta distributions.
However, in the deterministic code case, 
we need derandomization.
Since symmetrizable case does not yield zero transmission rate
in a compound channel, 
The discussion in Section \ref{sec:determinsitic_cap_section} does not work for 
in a compound channel.
Since the channel has a permutation invariant form 
in a compound channel,
the papers \cite{Hayashi2009} showed that
the random permutation of a constant composition code 
simulates a randomized code.
In the multiple-sender case, 
even under randomized codes,
the capacity region of compound channels 
cannot be obtained by 
the capacity region of arbitrarily varying channel
with restricting distributions $Q$ to delta distributions \cite{HC}.
Hence, a simple modification of our approach in Section \ref{section_CQAVMAC}
does not yield
the capacity region of a compound channel in this case.


\bibliographystyle{IEEEtran}
\bibliography{master}
\appendices

\section{Converse part of Theorem \ref{theorem_generalised_independence_test_arbitrary_varying}
}\label{NMA8}
We will show the converse part of Theorem \ref{theorem_generalised_independence_test_arbitrary_varying}.
We show the inequality $\le$ in \eqref{Stein} by contradiction.
We assume that
    \begin{equation}
        \limsup_{n \to \infty}-\frac{1}{n}\log\beta_{n}(\eps) > 
        \inf_{Q} I[X;\cB]_{P,Q}. \label{Stein5}
    \end{equation}
Using 
    \begin{align}
\max_{\sigma^{\cB^n} \in \cD(\cH_\cB^{\otimes n})}{\tr\left[
\bbT_{n}
    \left({\rho^{X}_{P}}^{\otimes n} \otimes \sigma^{\cB^n}\right)
    \right]}
&    \ge
    \max_{\sigma^{\cB} \in \cD(\cH_\cB)}
    {\tr\left[
\bbT_{n}
    \left({\rho^{X}_{P}}^{\otimes n} \otimes     
(\sigma^{\cB})^{\otimes n}
\right)
    \right]},\label{NZIO} \\
\sup_{s^n \in \cS^n} \tr[(\bbI^{X^n\cB^n} - \bbT_{n})\rho^{X^n\cB^n}_{P,s^n}]
&=\sup_{Q_n \in \cP(\cS^n)} \tr[(\bbI^{X^n\cB^n} - \bbT_{n})\rho^{X^n\cB^n}_{P,Q_n}]\notag\\
&\ge 
\sup_{Q\in \cP(\cS)} \tr[(\bbI^{X^n\cB^n} - \bbT_{n})\rho^{X^n\cB^n}_{P,Q^n}]  ,
    \end{align}
    we have     
    \begin{equation}
\beta_{n}(\eps)\ge
\beta_{n}'(\eps) := 
\min_{\substack{0 \preceq \bbT_{n} \preceq \bbI}} \left\{\max_{\sigma^{\cB} \in \cD(\cH_\cB)}{\tr\left[
\bbT_{n}
    \left({\rho^{X}_{P}}^{\otimes n} \otimes (\sigma^{\cB})^{\otimes n}
\right)
    \right]} \quad \bigg| \quad 
    \sup_{Q} \tr[(\bbI^{X^n\cB^n} - \bbT_{n})\rho^{X^n\cB^n}_{P,Q^n}] \leq \eps\right\}.
    \end{equation}    
As
\begin{align}
\frac{1}{n}
D_{\alpha} \Big(\rho^{X^n\cB^n}_{P,Q^n} \Big\| \Big({\rho^{X}_{P}}^{\otimes n} \otimes (\sigma^{\cB})^{\otimes n}\Big)\Big) 
=
D_{\alpha} \Big(\rho^{X\cB}_{P,Q} \Big\| 
\Big({\rho^{X}_{P}} \otimes \sigma^{\cB}\Big)\Big) ,
\end{align}

Fact 4 guarantees that
\begin{align}
\inf_{Q} I[X;\cB]_{P,Q}
=\lim_{\alpha \to 1+0} 
\inf_{
(\sigma^{\cB},Q) \in \cD(\cH_\cB)\times \cP(\cS)}
D_{\alpha} \Big(\rho^{X\cB}_{P,Q} \Big\| 
\Big({\rho^{X}_{P}} \otimes \sigma^{\cB}\Big)\Big) .
\end{align}
Hence, we can choose $\alpha>1$ such that
\begin{align}
r:=\limsup_{n \to \infty}-\frac{1}{n}\log\beta_{n}'(\eps)
\ge \limsup_{n \to \infty}-\frac{1}{n}\log\beta_{n}(\eps)
> r_{\alpha}:=
\inf_{
(\sigma^{\cB},Q) \in \cD(\cH_\cB)\times \cP(\cS)}
D_{\alpha} \Big(\rho^{X\cB}_{P,Q} \Big\| 
\Big({\rho^{X}_{P}} \otimes \sigma^{\cB}\Big)\Big) .
\label{BNS}
\end{align}
We choose $(\sigma^{\cB}_*,Q_*) \in \cD(\cH_\cB)\times \cP(\cS)$
such that
\begin{align}
\limsup_{n \to \infty}-\frac{1}{n}\log\beta_{n}(\eps) > r_{\alpha}':=
D_{\alpha} \Big(\rho^{X\cB}_{P,Q_*} \Big\| 
\Big({\rho^{X}_{P}} \otimes \sigma^{\cB}_*\Big)\Big) .
\end{align}
Assume that a test $\bbT_{n}$ satisfies
\begin{align}
\tr[(\bbI^{X^n\cB^n} - \bbT_{n})\rho^{X^n\cB^n}_{P,Q_*^n}] \leq \eps.\label{BMCV}
\end{align}
Using the information processing inequality of Petz Quantum 
R\'{e}nyi divergence in a way as Eq. (3.137) of \cite{H2017QIT},
we have
\begin{align}
(\alpha-1) n r_{\alpha} &=
(\alpha-1) n 
D_{\alpha} \Big(\rho^{X \cB}_{P,Q_*} \Big\| \Big({\rho^{X}_{P}}^{\otimes n} \otimes \sigma_*^{\cB}\Big)\Big)
=
(\alpha-1) D_{\alpha} \Big(\rho^{X^n\cB^n}_{P,Q_*^n} \Big\| \Big({\rho^{X}_{P}}^{\otimes n} \otimes (\sigma_*^{\cB})^{\otimes n}\Big)\Big) \notag\\
& \ge 
\alpha \log \Tr \Big[\bbT_{n} \rho^{X^n\cB^n}_{P,Q_*^n}\Big]
-(\alpha-1)\log \Tr \Big[\bbT_{n} \Big({\rho^{X}_{P}}^{\otimes n} \otimes (\sigma_*^{\cB})^{\otimes n}\Big)\Big].
\end{align}
Hence,
\begin{align}
- \frac{1}{n}\log \Tr \Big[\bbT_{n} \rho^{X^n\cB^n}_{P,Q_*^n}\Big]
\ge \frac{ -(\alpha-1) r_{\alpha}'
+\frac{(\alpha-1)}{n}\log 
\Tr \Big[\bbT_{n} \Big({\rho^{X}_{P}}^{\otimes n} \otimes (\sigma_*^{\cB})^{\otimes n}\Big)\Big]
}{\alpha}.
\end{align}
Since 
the condition \eqref{BMCV} and the definition of $r$ in \eqref{BNS}
imply 
$\limsup_{n\to \infty}
\frac{-1}{n}\log \Tr \Big[\bbT_{n} \Big({\rho^{X}_{P}}^{\otimes n} \otimes (\sigma_*^{\cB})^{\otimes n}\Big)\Big]
\ge r$, 
we have
\begin{align}
\limsup_{n\to \infty}
- \frac{1}{n}\log 
\Tr \Big[\bbT_{n} \rho^{X^n\cB^n}_{P,Q_*^n}\Big]
\ge \frac{ (\alpha-1) (r-r_{\alpha}')
}{\alpha},
\end{align}
which implies
\begin{align}
\liminf_{n\to \infty}
\Tr \Big[\bbT_{n} \rho^{X^n\cB^n}_{P,Q_*^n}\Big]
=0 \label{XADF}.
\end{align}
Eq. \eqref{XADF} contradicts \eqref{BMCV}.

\section{Proof of Lemma \ref{claim_f}:}\label{proof_claim_f}
Given an $x^n := (x_1,x_2,\cdots,x_n) \in \cX^n$, there always exists an an $\bar{x}^n := (\bar{x}_1,\bar{x}_2,\cdots,\bar{x}_n)\in \cX^n$, which has a form $\bar{x}^n := \left(\underbrace{1,\cdots,1}_{m_1},\underbrace{2,\cdots,2}_{m_2},\cdots,\underbrace{\abs{\cX},\cdots,\abs{\cX}}_{m_{\abs{\cX}}}\right)$, where $\forall i \in [\abs{\cX}], m_i \geq 0$ and $\sum_{i \in [\abs{\cX}]}m_i = n$, such that  $x^n := \pi(\bar{x}^n)$, for some $\pi \in S_n$. Now for $\bar{x}^n$, we consider the following state:
\begin{align}
\left(\bigotimes_{i=1}^{n}\rho^{\cB}_{Q_{s^n},\bar{x}_i}\right) &= \left(\rho^{\cB}_{Q_{s^n},1}\right)^{\otimes m_1} \otimes \left(\rho^{\cB}_{Q_{s^n},2}\right)^{\otimes m_2} \otimes \cdots \otimes \left(\rho^{\cB}_{Q_{s^n},\abs{\cX}}\right)^{\otimes m_{\abs{\cX}}}\nn\\
&\overset{a}{\leq} \prod_{i=1}^{\abs{\cX}} m_{i}^{\frac{\abs{\cB}(\abs{\cB}-1)}{2}}\abs{\Lambda^{m_i}_{\abs{\cB}}} \bigotimes_{i=1}^{\abs{\cX}} \rho^{\cB^{m_i}}_{U,m_i}\nn\\
&\overset{b}{\leq} n^{\frac{\abs{\cX}(\abs{\cB}-1)\abs{\cB}}{2}}\abs{\Lambda^{n}_{\abs\cB}}^{\abs{\cX}}  \widehat{\rho}^{\cB^n}_{\bar{x}^n},\label{claim_f_proof_eq1}
\end{align}
where $a$ follows from Lemma \ref{lemma_perm_Inv_universal}, $b$ follows from the fact $\forall i \in [\abs{\cX}] , m_i \leq n ,  \abs{\Lambda^{m_i}_{\abs{\cB}}} \leq \abs{\Lambda^{n}_{\abs\cB}}^{\abs{\cX}} $.
Now from \eqref{perm_op_eq} of Definition \ref{fact_perm_op}, it directly follows that
\begin{align*}
    \left(\bigotimes_{i=1}^{n}\rho^{\cB}_{Q_{s^n},x_i}\right) &= V^{\cB^n}(\pi) \left(\bigotimes_{i=1}^{n}\rho^{\cB}_{Q_{s^n},\bar{x}_i}\right){V^{\cB^n}}^{\dagger}(\pi)\\
    &\overset{a}{\leq} n^{\frac{\abs{\cX}(\abs{\cB}-1)\abs{\cB}}{2}}\abs{\Lambda^{n}_{\abs\cB}}^{\abs{\cX}}V^{\cB^n}(\pi) \left(\widehat{\rho}^{\cB^n}_{\bar{x}^n}\right){V^{\cB^n}}^{\dagger}(\pi)\\
    &\overset{b}{=}  n^{\frac{\abs{\cX}(\abs{\cB}-1)\abs{\cB}}{2}}\abs{\Lambda^{n}_{\abs\cB}}^{\abs{\cX}}\widehat{\rho}^{\cB^n}_{x^n},
\end{align*}
where $a$ follows from \eqref{claim_f_proof_eq1} and $b$ follows from \eqref{rhohatxn}.
Now from the above inequality for any $t \in (0,1)$, we can write the following:
\begin{align*}
    \left(\bigotimes_{i=1}^{n}\rho^{\cB}_{Q_{s^n},x_i}\right)^{t} &\leq n^{\frac{t\abs{\cX}(\abs{\cB}-1)\abs{\cB}}{2}}\abs{\Lambda^{n}_{\abs\cB}}^{t\abs{\cX}}\left(\widehat{\rho}^{\cB^n}_{x^n}\right)^{t}\\
     \Rightarrow\left(\bigotimes_{i=1}^{n}\rho^{\cB}_{Q_{s^n},x_i}\right)^{t}\left(\widehat{\rho}^{\cB^n}_{x^n}\right)^{-t} &\overset{a}{\leq} n^{\frac{t\abs{\cX}(\abs{\cB}-1)\abs{\cB}}{2}}\abs{\Lambda^{n}_{\abs\cB}}^{t\abs{\cX}},
\end{align*}
where $a$ follows from the fact that for any $t \in (0,1)$, $\left(\widehat{\rho}^{\cB^n}_{x^n}\right)^{t}$ is invertible since $\widehat{\rho}^{\cB^n}_{x^n}$ is a positive operator. This completes the proof of Lemma \ref{claim_f}.\hfill$\blacksquare$

\section{Proof of Lemma \ref{theorem_sen_MAC_generalised_indep}}\label{proof_theorem_sen_MAC_generalised_indep}
\subsection{Organization of our proof}
From Lemma \ref{lemma_cq_mac_avht_independent}, we have three projectors (tests) $\hat{\bbT}_{n,X},\hat{\bbT}_{n,Y}$ and $\hat{\bbT}_{n}$, all accepting $H_0$ with very high probabilities and accepting $H_{1,Y},H_{1,X}$ and $H_1$ respectively with arbitrarily low probabilities. We design a projector (test) ${\bbT}^{\star}_{n}$, which satisfies the properties of $\hat{\bbT}_{n,X},\hat{\bbT}_{n,Y}$ and $\hat{\bbT}_{n}$ i.e. $\forall s^n \in \cS^n$, for all $\{\sigma^{\cB^n}_{x^n}\}, \{\sigma^{\cB^n}_{y^n}\}\subset \cD(\cH_\cB^{\otimes n})$ and $\sigma^{\cB^n} \in \cD(\cH_\cB^{\otimes n})$, for any $(\widehat{R}_1,\widehat{R}_2)$ satisfying \cref{hatR_1c,hatR_2c,hatSUMc}, for large enough $n$ and arbitrary $\delta \in (0,1)$,   ${\bbT}^{\star}_{n}$ satisfies the following:
\begin{align}
     &\tr\left[ {\bbT}^{\star}_{n}\left(\rho^{X^nY^n\cB^n}_{P_{X},P_{Y},s^n}\right)\right] \geq 1 - Poly(\delta),\label{Tstarprop1}\\
    &\tr\left[{\bbT}^{\star}_{n}\left({\rho^{\cH_Y^{\otimes n}}_{P_Y}} \otimes \sigma^{X^n\cB^n}_{P_{X}, \{\sigma^{\cB^n}_{x^n}\}}\right)\right]
\leq g_1(n,\abs{\cX},\abs{\cB})2^{-n\widehat{R}_2},
\label{Tstarprop2}\\
    &\tr\left[{\bbT}^{\star}_{n}\left({\rho^{\cH_X^{\otimes n}}_{P_X}} \otimes\sigma^{Y^n\cB^n}_{P_{Y}, \{\sigma^{\cB^n}_{y^n}\}}\right)\right]
\leq g_2(n,\abs{\cY},\abs{\cB})2^{-n\widehat{R}_1}, \label{Tstarprop3}\\
    &\tr\left[{\bbT}^{\star}_{n}\left({\rho^{\cH_X^{\otimes n}}_{P_X}} \otimes{\rho^{\cH_Y^{\otimes n}}_{P_Y}} \otimes \sigma^{\cB^n}\right)\right]
\leq g_3(n,\abs{\cB})2^{-n(\widehat{R}_1 + \widehat{R}_2)}.\label{Tstarprop4}
\end{align}

Our proof for $\bbT^{\star}_{n}$ is motivated by the work of Sen \cite{Sen2021}. Before giving this proof we first discuss some observations made in \cite{Sen2021}.
 One possible construction for ${\bbT}^{\star}_{n}$ can be
${\bbT}^{\star}_{n}:= \hat{\bbT}_{n,X}\hat{\bbT}_{n,Y}\hat{\bbT}_{n}.$ However, for any $(x^n,y^n) \in (\cX^n \times \cY^n)$, the quantum states $\widehat{\rho}^{\cB^n}_{x^n},\widehat{\rho}^{\cB^n}_{y^n}$ may not commute with each other and therefore, $\hat{\bbT}_{n,X}$ and $\hat{\bbT}_{n,Y}$ may not commute with each other. Hence, this construction will not work.

Another possible construction for ${\bbT}^{\star}_{n}$ one may think of, is to first construct a projector over the span of the complements of the corresponding support spaces of the projectors $\hat{\bbT}_{n,X}, \hat{\bbT}_{n,Y}$, $\hat{\bbT}_{n}$ and then, set ${\bbT}^{\star}_{n}$ to be the complement of that projector. This idea is motivated by finding the intersection of two sets by finding the complement of the union of the complements of corresponding sets. Although this is a well-established concept in classical scenarios, in the quantum setting, the span (analog of ``union'' operation in classical setting) of the above-mentioned complements can be the whole Hilbert space and therefore ${\bbT}^{\star}_{n}$ may not even exist. We discuss below, how to get around with this issue.

Suppose, we have three subspaces $\cW_1, \cW_2, \cW_3$ in $\cH$, whose span is the whole Hilbert space $\cH$ they are residing in. We keep $\cW_3$ as it is, but tilt $\cW_1$ and $\cW_2$ very minutely in two mutually orthogonal directions which are also orthogonal to the parent Hilbert space $\cH$. 
Here by 'tilt'-ing towards an orthogonal direction, we mean adding a small component of this particular direction with the original vector and then normalizing the resultant vector to keep the length of the resultant vector the same as the original vector. This phenomenon can be thought of as tilting the vector minutely toward that particular direction. In the upcoming discussion, we will formally explain this idea of tilting. 
Let us now denote the tilted version of subspaces $\cW_1$ and $\cW_2$ as $\cW_1'$ and $\cW_2'$ respectively. 
This increases the dimension of the Hilbert space in which three subspaces 
($\cW_1'$, $\cW_2'$, and $\cW_3$) reside and now if we calculate the span it does not cover the enlarged Hilbert space. 
It is important to note that since we have minutely tilted $\cW_1$ and $\cW_2$, for any vector $\ket{h}$, the individual projection lengths of $\ket{h}$ on $\cW_1'$ and $\cW_2'$ do not vary much from the individual projection lengths on $\cW_1$ and $\cW_2$ respectively. 

In our proof, tilting maps play a key role.
A tilting map is nothing but a type of isometry which Sen termed in \cite{Sen2021}, 
it projects any statevector into a bigger Hilbert space while keeping its length unchanged and changing the direction of the vector minutely (defined in terms of $\delta$ in the following discussion) towards some pre-defined direction. For example, tilting a vector in 2D (say having $X$ and $Y$ plane) towards $Z$ plane (which is orthogonal to $X$ and $Y$) in a very small angle.

\subsection{Isometric tilting maps}
We now proceed with the approach of Sen \cite{Sen2021}, for the construction of ${\bbT}^{\star}_{n}$ which behaves like an ``intersection'' of $\hat{\bbT}_{n,X},\hat{\bbT}_{n,Y}$ and $\hat{\bbT}_{n}$ and we will show that this ${\bbT}^{\star}_{n}$ satisfies \cref{Tstarprop1,Tstarprop2,Tstarprop3,Tstarprop4} \textit{simultaneously}. 

    Using Sen's approach in \cite[Section $4$]{Sen2021} while defining an intersection projector of a collection of projectors under a classical-quantum setting, we first consider a Hilbert space $\cH_{\cA}^{\otimes 2}$ with dimension $\abs{\cA}^2$ which can be decomposed as follows:
        \begin{equation*}
        \cH_{\cA}^{\otimes 2} := \cH_{{\cA}^{X}} \otimes \cH_{{\cA}^{Y}},
    \end{equation*}
   where $\cH_{{\cA}^{X}},\cH_{{\cA}^{Y}}$ are two mutually orthogonal finite-dimensional isomorphic Hilbert spaces with dimension $\abs{\cA}$ (where $\abs{\cA} < \infty$), and $\left\{\ket{a_{r}}\right\}_{r \in [\abs{\cA}]}$ and $\left\{\ket{b_{t}}\right\}_{t \in [\abs{\cA}]}$ are orthonormal basis for $\cH_{{\cA}^{X}}$ and $\cH_{{\cA}^{Y}}$ respectively.

    For each $r \in [\abs{\cA}], t \in [\abs{\cA}]$ and $\delta \in (0,1)$, we define two isometric maps $\cM_{X,a_{r},\delta} : \cH_\cB^{\otimes n} \to \cH_\cB^{\otimes n}  \oplus \cH_\cB^{\otimes n}  \otimes \cH_{\cA^{X}}$ and $\cM_{Y,b_{t},\delta} : \cH_\cB^{\otimes n}  \to \cH_\cB^{\otimes n}  \oplus (\cH_\cB^{\otimes n}  \otimes \cH_{\cA^{Y}})$ as follows
    \begin{align*}
      \cM_{X,a_{r},\delta} &: \ket{h} \mapsto \frac{1}{\sqrt{1 + \delta^2}} (\ket{h} \oplus \delta\ket{h}\ket{a_{r}}),\\
      \cM_{Y,b_{t},\delta} &: \ket{h} \mapsto \frac{1}{\sqrt{1 + \delta^2}} (\ket{h} \oplus \delta\ket{h}\ket{b_{t}})
    \end{align*}
    for $\ket{h}$ in $\cH_\cB^{\otimes n} $.
    The two maps
    $\cM_{X,a_{r},\delta}$ and $\cM_{Y,b_{t},\delta}$ 
        are two ''tilting'' maps which 
       tilt any vector $\ket{h}$ in $\cH_\cB^{\otimes n} $ minutely 
       (parametrized by $\delta$) towards the directions of $a_{r}$ and $b_{t}$ respectively in significantly larger Hilbert spaces $\cH_\cB^{\otimes n}  \oplus \cH_\cB^{\otimes n}  \otimes \cH_{\cA^{X}}$ and $\cH_\cB^{\otimes n}  \oplus (\cH_\cB^{\otimes n}  \otimes \cH_{\cA^{Y}})$, respectively. 
     
       For each $(r,t) \in [\abs{\cA}]^2$ and for $\delta \in (0,1)$, we define another tilting map $\cM_{XY,a_{r},b_{t},\delta} : \cH_\cB^{\otimes n}  \to \cH_\cB^{\otimes n}  \oplus (\cH_\cB^{\otimes n}  \otimes \cH_{\cA^{X}}) \oplus (\cH_\cB^{\otimes n}  \otimes \cH_{\cA^{Y}})$ as
    \begin{align*}
        \cM_{XY,a_{r},b_{t},\delta} &: \ket{h} \mapsto \frac{1}{\sqrt{1 + 2\delta^2}} (\ket{h} \oplus \delta\ket{h}\ket{a_{r}} \oplus \delta\ket{h}\ket{b_{t}})
        \hbox{ for } \forall \ket{h} \in \cH_\cB^{\otimes n}.
        \end{align*}
    The above map tilts any vector $\ket{h}$ in $\cH_\cB^{\otimes n} $ minutely (parametrized by $\delta$) towards both directions of $a_{r}$ and $b_{t}$ in a significantly larger Hilbert space 
    $\cH_{\cB'} := \cH_\cB^{\otimes n}  \oplus \cH_\cB^{\otimes n}  \otimes \cH_{\cA^{X}} \oplus \cH_\cB^{\otimes n}  \otimes \cH_{\cA^{Y}}$. 
    
The following lemma evaluates a small perturbation of quantum states under the tilting map $\cM_{XY,a_{r},b_{t},\delta}$.

\begin{lemma}\label{claim_small_perturbation}
    For any state $\rho \in \cD(\cH_{\cB}^{\otimes n})$ and any $(r,t) \in [\abs{\cA}]^2$, 
we choose two states   
    $\widehat{\rho}_{a_{r},b_{t},\delta}^{\cB'}:=\cM_{XY,a_{r},b_{t},\delta}\left(\rho\right)$
    and
    \begin{equation}
        \Tilde{\rho}^{\cB'} := \left[\begin{array}{cc}
      \left[\rho\right]_{\abs{\cB}^n \times \abs{\cB}}   &  \left[\mathbf{0}\right]_{\abs{\cB}^n \times 2\abs{\cA}\abs{\cB}^n}\\
      \left[\mathbf{0}\right]_{ 2\abs{\cA}\abs{\cB}^n \times \abs{\cB}^n}   & \left[\mathbf{0}\right]_{ 2\abs{\cA}\abs{\cB}^n \times  2\abs{\cA}\abs{\cB}^n}
    \end{array}\right]_{\abs{\cB'} \times \abs{\cB'}}, \label{proof_theo_gen_identity_1}
    \end{equation}
     where $\delta \in (0,1)$,  $\left[\mathbf{0}\right]_{A \times B}$ is a $0$-matrix of order $A \times B$ with $0$ in each entry. Note that we get $\Tilde{\rho}^{\cB'}$ by padding $0$-matrices with $\rho$.
Then, the following relation holds;
    \begin{equation*}
        \norm{\widehat{\rho}_{a_{r},b_{t},\delta}^{\cB'} - \Tilde{\rho}^{\cB'}}{1} \leq 8\delta.
    \end{equation*}
\end{lemma}
This lemma will be shown in Subsection \ref{subsec_small_perturbation}.

\subsection{Tilting and Augmenting the classical-quantum states}\label{tilt_augment}
For each $s^n \in \cS^n$, for a particular $(x^n,y^n)\in \cX^n$ and $(r,t) \in [\abs{\cA}]^2$, 
we tilt the state $\rho^{\cB^n}_{x^n,y^n,s^n}$ using the tilting map $\cM_{XY,a_{r},b_{t},\delta}$ and we obtain the following state:
\begin{equation}
    \rho^{\cB'}_{(x^n,a_{r}),(y^n,b_{t}),s^n,\delta} := \cM_{XY,a_{r},b_{t},\delta}(\rho^{\cB^n}_{x^n,y^n,s^n}).\label{tilted_rho_inner}
\end{equation}
Note  that for each $s^n \in \cS^n$, for a particular $(x^n,y^n)\in \cX^n\times \cY^n$, we get a collection of tilted states $\left\{ \rho^{\cB'}_{(x^n,a_{r}),(y^n,b_{t}),s^n,\delta}\right\}_{(r,t) \in [\abs{\cA}]^2}$. For each $(x^n,y^n)\in \cX^n\times \cY^n$, to incorporate the collection of states $\left\{ \rho^{\cB'}_{(x^n,a_{r}),(y^n,b_{t}),s^n,\delta}\right\}_{(r,t) \in [\abs{\cA}]^2}$ into a single CQ state, we need to augment (or, enlarge) the classical subsystems 
$\cH_{X^n}, \cH_{Y^n}$  in the following way:
        \begin{equation}
        \rho^{X'Y'\cB'}_{P_{X},P_{Y},s^n} := \frac{1}{\abs{\cA}^2} \sum_{\substack{x^n\in \cX^n ,y^n \in \cY^n\\ r \in [\abs{\cA}], t \in [\abs{\cA}]}}P^{n}_{X}(x^n)\ketbrasys{x^n,a_{r}}{X'}\otimes P^{n}_{Y}(y^n)\ketbrasys{y^n,b_{t}}{Y'}\otimes \rho^{\cB'}_{(x^n,a_{r}),(y^n,b_{t}),s^n,\delta}, \label{tilted_augmented_omega_1}
    \end{equation}
    where $\cH_{X'} := \cH_{X^n} \otimes \cH_{\cA^{X}}, 
    \cH_{Y'} :=\cH_{Y^n} \otimes \cH_{\cA^{Y}}$.
    Similarly, for $\{\sigma^{\cB^n}_{x^n,y^n}\} \subset \cD(\cH_\cB^{\otimes n})$, 
    we define the following tilted and augmented CQ state:
    \begin{equation}
       \sigma^{X'Y'\cB'}_{P_{X},P_{Y}, \{\sigma^{\cB^n}_{x^n,y^n}\}} := \frac{1}{\abs{\cA}^2}\sum_{\substack{x^n\in \cX^n , y^n \in \cY^n \\ r \in [\abs{\cA}], t \in [\abs{\cA}]}} P^{n}_{X}(x^n) \ketbrasys{x^n,a_{r}}{X'}\otimes P^{n}_{Y}(y^n)\ketbrasys{y^n,b_{t}}{Y'} \otimes \sigma^{\cB'}_{(x^n,a_{r}),(y^n,b_{t}),\delta},  \label{tilted_augmented_omega_1B}
    \end{equation}
    where $\sigma^{\cB'}_{(x^n,a_{r}),(y^n,b_{t}),\delta} := \cM_{XY,a_{r},b_{t},\delta}\left( \sigma^{\cB^n}_{x^n,y^n}\right)$. 
    
    Lemma \ref{claim_small_perturbation} shown in Subsection \ref{subsec_small_perturbation} guarantees 
that the above-mentioned method of tilting does not perturb the quantum state much. 
That is, Lemma \ref{claim_small_perturbation} guarantees the following relations
\begin{align}
    \norm{\rho^{\cB'}_{(x^n,a_{r}),(y^n,b_{t}),s^n,\delta} - \rho^{\cB^n}_{x^n,y^n,s^n}}{1} &\leq 8\delta, \forall s^n \in \cS^n, \label{perturbed_tilted_state_eq1}\\
    \norm{ \sigma^{\cB'}_{(x^n,a_{r}),(y^n,b_{t}),\delta} - \sigma^{\cB^n}_{x^n,y^n}}{1} &\leq 8\delta.\label{perturbed_tilted_state_eq3}
\end{align}
for each $(x^n,y^n) \in \cX^n \times \cY^n, (r,t)  \in [\abs{\cA}]^2$.
    Note that, in the LHS of the \cref{perturbed_tilted_state_eq1,perturbed_tilted_state_eq3}, the dimensions of the systems $\cB^n$ and $\cB'$ do not match and $\cB'$ contains $\cB^n$. Here, we pad the operators (that reside in $\cB^n$) with $0$ matrices in the similar way mentioned in \eqref{proof_theo_gen_identity_1}.
    Thus, it directly follows that
    \begin{align}
        &\hspace{10pt}\norm{\rho^{X'Y'\cB'}_{P_{X},P_{Y},s^n} - \rho^{X^nY^n\cB^n}_{P_{X},P_{Y},s^n} \otimes \frac{\bbI^{\cA^{\otimes 2}}}{\abs{\cA^2}}}{1}\nn \\
        &= \norm{\frac{1}{\abs{\cA}^2} \sum_{\substack{x^n\in \cX^n ,y^n \in \cY^n\\ r \in [\abs{\cA}], t \in [\abs{\cA}]}}P^{n}_{X}(x^n)\ketbrasys{x^n,a_{r}}{X'}\otimes P^{n}_{Y}(y^n)\ketbrasys{y^n,b_{t}}{Y'}\otimes \left(\rho^{\cB'}_{(x^n,a_{r}),(y^n,b_{t}),s^n,\delta} - \rho^{\cB^n}_{x^n,y^n,s^n}\right)}{1}\nn\\
        &\overset{a}{=} \frac{1}{\abs{\cA}^2} \sum_{\substack{x^n\in \cX^n ,y^n \in \cY^n\\ r \in [\abs{\cA}], t \in [\abs{\cA}]}}P^{n}_{X}(x^n)\ketbrasys{x^n,a_{r}}{X'}\otimes P^{n}_{Y}(y^n)\ketbrasys{y^n,b_{t}}{Y'}\otimes \norm{\rho^{\cB'}_{(x^n,a_{r}),(y^n,b_{t}),s^n,\delta} - \rho^{\cB^n}_{x^n,y^n,s^n}}{1}
        \overset{b}{\leq} 8\delta,\label{perturbed_tilted_state_eq2}
     \end{align}
where $a$ follows from direct sum property \cite[eq. ($2.2.97$)]{khatri2024principlesquantumcommunicationtheory} of Schatten-$l_1$ norm and $b$ follows from \eqref{perturbed_tilted_state_eq1}. Similarly using \eqref{perturbed_tilted_state_eq3}, we have the following relation;
        \begin{align}
\norm{\sigma^{X'Y'\cB'}_{P_{X},P_{Y}, \{\sigma^{\cB^n}_{x^n,y^n}\}} - \sigma^{X^nY^n\cB^n}_{P_{X},P_{Y}, \{\sigma^{\cB^n}_{x^n,y^n}\}}\otimes\frac{\bbI^{\cA^{\otimes 2}}}{\abs{\cA^2}}}{1} &\leq 8\delta\label{perturbed_tilted_state_eq4}.
    \end{align}
    \subsection{Construction of the ``Intersection Projector''}
    It follows from \cref{universal_test_X,universal_test_Y,universal_test_None} that $\hat{\bbT}_{n,X},\hat{\bbT}_{n,Y} $ and $\hat{\bbT}_{n}$ have the following form:
    \begin{align*}
        \hat{\bbT}_{n,X} &:= \sum_{\substack{x^n \in \cX^n \\ y^n \in \cY^n}} \ketbrasys{x^n}{X^n} \otimes \ketbrasys{y^n}{Y^n} \otimes \hat{\bbT}_{n,X, x^n,y^n},\\
        \hat{\bbT}_{n,Y} &:= \sum_{\substack{x^n \in \cX^n \\ y^n \in \cY^n}} \ketbrasys{x^n}{X^n} \otimes \ketbrasys{y^n}{Y^n} \otimes \hat{\bbT}_{n,Y, x^n,y^n},\\
        \hat{\bbT}_{n} &:= \sum_{\substack{x^n \in \cX^n \\ y^n \in \cY^n}} \ketbrasys{x^n}{X^n} \otimes \ketbrasys{y^n}{Y^n} \otimes \hat{\bbT}_{n, x^n,y^n}.
    \end{align*}
    
    For each $(x^n,y^n) \in \cX^n \times \cY^n$, we denote $\cW_{X, x^n,y^n}, \cW_{Y, x^n,y^n}$ and $\cW_{x^n,y^n}$ to be the support spaces of orthogonal complements of  $\hat{\bbT}_{n,X,x^n,y^n},\hat{\bbT}_{n,Y,x^n,y^n} $ and $\hat{\bbT}_{n,x^n,y^n}$ respectively. 
    For each $(x^n,y^n) \in \cX^n \times \cY^n$, our objective is to construct a measurement that behaves like an intersection of $\hat{\bbT}_{n,X,x^n,y^n},\hat{\bbT}_{n,Y,x^n,y^n} $ and $\hat{\bbT}_{n,x^n,y^n}$. The sole idea is to find the complement of the span of $\cW_{X, x^n,y^n}, \cW_{Y, x^n,y^n}$ and $\cW_{x^n,y^n}$. 
    However, there is a high possibility that the span of $\cW_{X, x^n,y^n}, \cW_{Y, x^n,y^n}$ and $\cW_{x^n,y^n}$ can cover the whole Hilbert space $\cH_\cB^{\otimes n}$. 
    Thus, to resolve this issue, not only do we need to enlarge the Hilbert space but also tilt at least two of them in all possible pairs of orthogonal directions i.e. each from the collection $\{a_{r}, b_{t}\}_{(r,t) \in [\abs{\cA}]^2}$. For each $(x^n,y^n) \in \cX^n \times \cY^n$ and $(r,t) \in [\abs{\cA}]^2$, we define the following tilted subspaces:
    \begin{align}
        \cW'_{X, (x^n,a_{r}), y^n, \delta} &:= \cM_{X,a_{r},\delta}(\cW_{X, x^n,y^n}),\label{tilted_span_hayashi_X}\\
         \cW'_{Y, x^n, (y^n,b_{t}), \delta} &:= \cM_{Y,b_{t},\delta}(\cW_{Y, x^n,y^n}),\label{tilted_span_hayashi_Y}
    \end{align}
        where $\delta \in(0,1)$ is a free parameter, which can be finetuned. 
Note that for each $x^n \in \cX^n,y^n \in \cY^n, (r,t) \in [\abs{\cA}]^2$, $\cW'_{X, (x^n,a_{r}), y^n\delta}$, 
$\cW'_{Y, x^n, (y^n,b_{t}), \delta}$ and $\cW_{x^n,y^n}$ reside inside 
$\cH_\cB^{\otimes n}  \oplus \cH_\cB^{\otimes n}  \otimes\cH_{\cA^{X}}$,  $\cH_\cB^{\otimes n}  \oplus \cH_\cB^{\otimes n}  \otimes\cH_{\cA^{Y}}$ and 
$\cH_\cB^{\otimes n}$ respectively and all three of them reside inside $\cH_{\cB'}$. For each $x^n \in \cX^n,y^n \in \cY^n, (r,t) \in [\abs{\cA}]^2$, we define the following span of subspaces:
    \begin{equation}
        \cW'_{(x^n,a_{r}),(y^n,b_{t}), \delta} := \cW'_{X, (x^n,a_{r}), y^n\delta} \bigplus \cW'_{Y, x^n, (y^n,b_{t}), \delta} \bigplus \cW_{x^n,y^n},\label{tilted_span_hayashi}
    \end{equation}
which is considered to be the span of $\cW'_{X, (x^n,a_{r}), y^n\delta}, \cW'_{Y, x^n, (y^n,b_{t}), \delta} $ and 
$ \cW_{x^n,y^n}$ (we denote $\bigplus$ as the span of subspaces). 
We consider $\Pi_{\cW'_{(x^n,a_{X_r}),(y^n,a_{Y_t}), \delta}}$ to be the orthogonal projection in $\cB'$ onto the spanned subspace 
$ \cW'_{(x^n,a_{X.r}),(y^n,a_{Y.t}), \delta}$. Thus, for each possible $(x^n,y^n) \in  \cX^n \times \cY^n, (r,t) \in [\abs{\cA}]^2$, we define the following POVM element as follows
    \begin{equation}
        \bbT^{\star}_{n,(x^n,a_{r}),(y^n,b_{t}), \delta} := \left(\bbI^{\cB'} - \Pi_{\cW'_{(x^n,a_{r}),(y^n,b_{t}), \delta}}\right) \mathbf{P}_{\cH_\cB^{\otimes n}}
        \left(\bbI^{\cB'} - \Pi_{\cW'_{(x^n,a_{r}),(y^n,b_{t}), \delta}}\right),\label{intersection_inner_povms}
    \end{equation}
    where $\mathbf{P}_{\cH_\cB^{\otimes n}}$ 
    is a orthogonal projector in $\cH_{\cB'}$ onto $\cH_\cB^{\otimes n}$. 
    We have to incorporate the collection of the POVMs $\left\{\bbT^{\star}_{n,(x^n,a_{r}),(y^n,b_{t}), \delta}\right\}$ into a single POVM in the following manner:
    \begin{equation}
        {\bbT}^{\star}_{n} := \sum_{\substack{x^n\in \cX^n, y^n \in \cY^n\\
        r \in [\abs{\cA}] ,t \in [\abs{\cA}]}}\ketbrasys{x^n,a_{r}}{X'}\otimes \ketbrasys{y^n,b_{t}}{Y'}\otimes \bbT^{\star}_{n,(x^n,a_{r}),(y^n,b_{t}),\delta}.\label{tilted_augmented_intersection_test}
    \end{equation}
\subsection{Smoothness of the tilting maps}
    For an arbitrary collection $\{\sigma^{\cB^n}_{x^n,y^n}\}_{(x^n,y^n) \in  \cX^n \times \cY^n} \subset \cD(\cH_\cB^{\otimes n})$, we consider the following marginals of the state $\sigma^{X'Y'\cB'}_{P_{X},P_{Y}, \{\sigma^{\cB^n}_{x^n,y^n}\}}$:
\begin{align}
    \sigma^{X'\cB'}_{P_{X}, \{\sigma^{\cB^n}_{x^n}\}} &:= \tr_{Y'} \left[\sigma^{X'Y'\cB'}_{P_{X},P_{Y}, \{\sigma^{\cB^n}_{x^n,y^n}\}}\right] = \frac{1}{\abs{\cA}}\sum_{\substack{x^n\in \cX^n , r \in [\abs{\cA}]}} P^{n}_{X}(x^n) \ketbrasys{x^n,a_{r}}{X'}\otimes   \sigma^{\cB'}_{(x^n,a_{r}),\delta},\label{traced_tilted_sigma_X}\\
    \sigma^{Y'\cB'}_{P_{Y}, \{\sigma^{\cB^n}_{y^n}\}} &:= \tr_{X'} \left[\sigma^{X'Y'\cB'}_{P_{X},P_{Y}, \{\sigma^{\cB^n}_{x^n,y^n}\}}\right] = \frac{1}{\abs{\cA}}\sum_{\substack{y^n\in \cY^n ,t \in [\abs{\cA}]}} P^{n}_{Y}(y^n) \ketbrasys{y^n,b_{t}}{Y'}\otimes  \sigma^{\cB'}_{(y^n,b_{t}),\delta},\label{traced_tilted_sigma_Y}\\
\sigma^{\cB'} &:= \tr_{X'Y'} \left[\sigma^{X'Y'\cB'}_{P_{X},P_{Y}, \{\sigma^{\cB^n}_{x^n,y^n}\}}\right] = \frac{1}{\abs{\cA}^2}\sum_{\substack{x^n \in \cX^n ,y^n\in \cY^n \\ r \in [\abs{\cA}], t \in [\abs{\cA}]}} P^{n}_{X}(x^n)P^{n}_{Y}(y^n) \sigma^{\cB'}_{(x^n,a_{r}),(y^n,b_{t}),\delta}\label{traced_tilted_sigma_None},
\end{align}
where for each $x^n \in \cX^n, r \in [\abs{\cA}]$ and $y^n \in \cY^n, t \in [\abs{\cA}]$ we have
\begin{align*}
    \sigma^{\cB'}_{(x^n,a_{r}),\delta} &:= \frac{1}{\abs{\cA}}\sum_{y^n, t \in [\abs{\cA}]}P^{n}_{Y}(y^n) \sigma^{\cB'}_{(x^n,a_{r}),(y^n,b_{t}),\delta}, \\
    \sigma^{\cB'}_{(y^n,b_{t}),\delta} &:= \frac{1}{\abs{\cA}}\sum_{x^n, r \in [\abs{\cA}]}P^{n}_{X}(x^n) \sigma^{\cB'}_{(x^n,a_{r}),(y^n,b_{t}),\delta},
\end{align*}
and we define the following marginals of $\rho^{X'Y'\cB'}_{P_{X},P_{Y},s^n}$:
\begin{align*}
 \rho^{X'}_{P_{X}} &:=  \tr_{Y'\cB'}\left[\rho^{X'Y'\cB'}_{P_{X},P_{Y},s^n}\right] = \frac{1}{\abs{\cA}} \sum_{\substack{x^n \in \cX^n , a_{X}}} P^{n}_{X}(x^n) \ketbrasys{x^n, a_{X}}{X'} = \rho^{\cH_X^{\otimes n}}_{P_{X}} \otimes \frac{\bbI^{\cA^{X}}}{\abs{\cA}},\\
 \rho^{Y'}_{P_{Y}} &:= \tr_{X'\cB'}\left[\rho^{X'Y'\cB'}_{P_{X},P_{Y},s^n}\right] = \frac{1}{\abs{\cA}} \sum_{\substack{y^n \in \cY^n , a_{Y}}} P^{n}_{Y}(y^n) \ketbrasys{y^n, a_{Y}}{Y'} = \rho^{\cH_Y^{\otimes n}}_{P_{Y}} \otimes \frac{\bbI^{\cA^{Y}}}{\abs{\cA}}.
\end{align*}

Now, for a statevector $\ket{h}$ in 
 $\cH_{\cB}^{\otimes n}$ and for any $(r,t) \in [\abs{\cA}]^2$, we apply the tilting map $\cM_{XY,a_{X_r},a_{Y_t},\delta}$. Then, we have the following series of equalities:

\begin{align}
    &\hspace{10pt}\frac{1}{\abs{\cA}} \sum_{t \in [\abs{\cA}]} \cM_{XY,a_{X_r},a_{Y_t},\delta}(\ketbra{h})\nn \\
    & = \frac{1}{\abs{\cA}(1 + 2\delta^2)} \sum_{t \in [\abs{\cA}]}\Big((\ket{h} \oplus \delta\ket{h}\ket{a_{r}})(\bra{h} \oplus \delta\bra{h}\bra{a_{r}}) + \delta (\ket{h} \oplus \delta\ket{h}\ket{a_{r}}) \bra{h}\bra{b_{t}} \oplus \delta\ket{h}\ket{b_{t}}(\bra{h} \oplus \delta\bra{h}\bra{a_{r}}) \nn\\
    &\hspace{100pt}\oplus\delta^2\ketbra{h}\ketbra{b_{t}} \Big)\nn\\
    &= \frac{1 + \delta^2}{1 + 2\delta^2}\cM_{X,a_{r},\delta}(\ketbra{h}) \oplus \cN_{X,a_{r},\delta}(\ketbra{h}),\label{Smooth_X}
\end{align}
where for each $r \in [\abs{\cA}]$, we define
\begin{equation}
    \cN_{X,a_{r},\delta}(\ketbra{h}) := \frac{1}{\abs{\cA}(1 + 2\delta^2)}\sum_{t \in [\abs{\cA}]}\left(\delta (\ket{h} \oplus \delta\ket{h}\ket{a_{r}}) \bra{h}\bra{b_{t}} \oplus \delta\ket{h}\ket{b_{t}}(\bra{h} \oplus \delta\bra{h}\bra{a_{r}}) \oplus \delta^2\ketbra{h}\ketbra{b_{t}} \right).\label{N_Tilt_X}
\end{equation}

Similarly, for each $t \in [\abs{\cA}]$, we define
\begin{align}
    &\cN_{Y,b_{t},\delta}(\ketbra{h}) := \frac{1}{\abs{\cA}(1 + 2\delta^2)}\sum_{r \in [\abs{\cA}]}\left(\delta (\ket{h} \oplus \delta\ket{h}\ket{b_{t}}) \bra{h}\bra{a_{r}} \oplus \delta\ket{h}\ket{a_{r}}(\bra{h} \oplus \delta\bra{h}\bra{b_{t}}) \oplus \delta^2\ketbra{h}\ketbra{a_{r}} \right),\label{N_Tilt_Y}
    \end{align}
        and
  \begin{align}
    &\cN_{\delta}(\ketbra{h}) := \frac{1}{\abs{\cA}^2(1 + 2\delta^2)}\sum_{r \in [\abs{\cA}], t \in [\abs{\cA}]}\big(\delta\ket{h}(\bra{h}\bra{a_{r}} \oplus \bra{h}\bra{b_{t}}) \oplus \delta(\ket{h}\ket{a_{r}} \oplus \ket{h}\ket{b_{t}})\bra{h}\nn\\
    &\hspace{170pt}\oplus \delta^2\ketbra{h}(\ket{a_{r}} \oplus \ket{b_{t}})(\bra{a_{r}} \oplus \bra{b_{t}})\big).\label{N_no_Tilt}
\end{align}

For each $(r,t) \in [\abs{\cA}]^2, \delta \in (0,1)$ and any $\ket{h}$ in $\cH_{\cB}^{\otimes n}$, we have
\begin{align}
    \norm{\cN_{X,a_{r},\delta}(\ketbra{h})}{\infty} &\leq \frac{3\delta}{\sqrt{\abs{\cA}}}\label{small_N_X},\\
    \norm{\cN_{Y,b_{t},\delta}(\ketbra{h})}{\infty} &\leq \frac{3\delta}{\sqrt{\abs{\cA}}}\label{small_N_Y},\\
    \norm{\cN_{\delta}(\ketbra{h})}{\infty} &\leq
\frac{7\delta}{\sqrt{A}},\label{small_N_no_Tilt}
\end{align}
where we give a proof for \cref{small_N_X,small_N_no_Tilt} in Appendices \ref{proof_small_N_X} and \ref{proof_small_N_no_tilt} respectively. The proof of \eqref{small_N_Y} follows directly from the proof of \eqref{small_N_X}. Now, from \cref{Smooth_X,N_Tilt_X,N_Tilt_Y,N_no_Tilt}, for each $r \in [\abs{\cA}]$ and $t \in [\abs{\cA}]$, and for any collection $\{\sigma^{\cB^n}_{x^n,y^n}\} \subset \cD(\cH_\cB^{\otimes n})$, we  have
\begin{align}
    \sigma^{\cB'}_{(x^n,a_{r}),\delta} &= \frac{1 + \delta^2}{1 + 2\delta^2}\cM_{X,a_{r},\delta}(\sum_{y^n \in \cY^n}P^{n}_{Y}(y^n) \sigma^{\cB^n}_{x^n,y^n}) \oplus \cN_{X,a_{r},\delta}(\sum_{y^n \in \cY^n}P^{n}_{Y}(y^n) \sigma^{\cB^n}_{x^n,y^n}),\label{Theo_Gen_smooth_X}\\
    \sigma^{\cB'}_{(y^n,b_{t}),\delta} &= \frac{1 + \delta^2}{1 + 2\delta^2}\cM_{Y,b_{t},\delta}(\sum_{x^n \in \cX^n}P^{n}_{X}(x^n) \sigma^{\cB^n}_{x^n,y^n}) \oplus \cN_{Y,b_{t},\delta}(\sum_{x^n \in \cX^n}P^{n}_{X}(x^n) \sigma^{\cB^n}_{x^n,y^n}),\label{Theo_Gen_smooth_Y}\\
    \sigma^{\cB'} &= \frac{1}{1 + 2\delta^2}(\sum_{\substack{x^n \in \cX^n\\ y^n \in \cY^n}}P^{n}_{X}(x^n)P^{n}_{Y}(y^n) \sigma^{\cB^n}_{x^n,y^n}) \oplus \cN_{\delta}(\sum_{\substack{x^n \in \cX^n\\ y^n \in \cY^n}}P^{n}_{X}(x^n)P^{n}_{Y}(y^n) \sigma^{\cB^n}_{x^n,y^n}).\label{Theo_Gen_smooth_none}
\end{align}

    Hence, from \cref{Theo_Gen_smooth_X,Theo_Gen_smooth_Y,Theo_Gen_smooth_none} and \cref{small_N_X,small_N_Y,small_N_no_Tilt}, we conclude that a quantum state $\rho \in \cD(\cB)$ can be tilted along $2$ mutually orthogonal directions and has $\abs{\cA}^2$ such choice of directions each of $\{(a_{r}, b_{t})\}_{(r,t) \in [\abs{\cA}]^2}$. Now consider
$\{a_{r}\}_{r \in [\abs{\cA}]}$ and tilt the state $\rho$ towards each elements of $\{a_{r},b_{\bar{t}}\}_{r \in [\abs{\cA}]}$ for a fixed $\bar{t} \in [\abs{\cA}]$. 
Then, if we take a uniform mixture of these tilted states, 
the resultant state will only be primarily tilted towards the direction of $b_{\bar{t}}$ in $\cH_{\cA^{Y}}$. The above phenomenon  holds for $\{b_{t}\}_{t \in [\abs{\cA}]}$ and both $\{(a_{r}, b_{t})\}_{(r,t) \in [\abs{\cA}]^2}$ as well. 
In the following section, we prove \cref{reject_Omega1_optimal_theo,accept_Omega2X_optimal_theo,accept_Omega2Y_optimal_theo,accept_Omega2XY_optimal_theo}, for the intersection projector $\bbT^{\star}_{n}$ with respect to the collections of tilted states 
$\left\{\rho^{X'Y'\cB'}_{P_{X},P_{Y},s^n}\right\}_{s^n \in \cS^n}$, $\{\rho^{Y'}_{P_{Y}} \otimes \sigma^{X'\cB'}_{P_{X}, \{\sigma^{\cB^n}_{x^n}\}}\}_{\{\sigma^{\cB^n}_{x^n}\} \subset \cD(\cH_\cB^{\otimes n})}$, $\{\rho^{X'}_{P_{X}} \otimes \sigma^{Y'\cB'}_{P_{Y}, \{\sigma^{\cB^n}_{y^n}\}}\}_{\{\sigma^{\cB^n}_{y^n}\} \subset \cD(\cH_\cB^{\otimes n})}$ and $\{\rho^{X'}_{P_{X}} \otimes\rho^{Y'}_{P_{Y}} \otimes \sigma^{\cB'}\}_{\sigma^{\cB^n} \in \cD(\cH_\cB^{\otimes n})}$, 
which are basically the tilted and augmented versions of the collections $H_0, H_{1,Y},H_{1,X},H_{1}$ mentioned in section \cref{GIT_hypotheses,Omega2states_arbit_X,Omega2states_arbit_Y,Omega2states_arbit_None}.

    \subsection{Performance analysis of ${\bbT}^{\star}_{n}$}\label{subsub4theo_gen_indep}
    Now, for any $s^n \in \cS^n$, we directly upper-bound the error probability $\tr\big[\left(\bbI^{X'Y'\cB'} - {\bbT}^{\star}_{n}\right)$ $\rho^{X'Y'\cB'}_{P_{X},P_{Y},s^n}\big]$ as follows
    \begin{align}
        &\hspace{10pt}\tr\left[\left( \bbI^{X'Y'\cB'} - {\bbT}^{\star}_{n}\right)\rho^{X'Y'\cB'}_{P_{X},P_{Y},s^n}\right]\nn\\
        &= \frac{1}{\abs{\cA}^2} \sum_{\substack{x^n\in \cX^n , y^n \in \cY^n \\ r \in [\abs{\cA}], t \in [\abs{\cA}]}}P^{n}_{X}(x^n) P^{n}_{Y}(y^n)\left(\tr\left[(\bbI^{\cB'} - \bbT^{\star}_{n,(x^n,a_{r}),(y^n,b_{t}), \delta})\rho^{\cB'}_{(x^n,a_{r}),(y^n,b_{t}),s^n,\delta}\right]\right)\nn\\
        &\overset{a}{\leq} \frac{1}{\abs{\cA}^2} \sum_{\substack{x^n\in \cX^n , y^n \in \cY^n \\ r \in [\abs{\cA}], t \in [\abs{\cA}]}}P^{n}_{X}(x^n) P^{n}_{Y}(y^n)\left(\tr\left[(\bbI^{\cB'} - \bbT^{\star}_{n,(x^n,a_{r}),(y^n,b_{t}), \delta})\rho^{\cB^n}_{x^n,y^n,s^n}\right]\right) + 4\delta\nn\\
        &\overset{b}{\leq} \frac{4}{\abs{\cA}^2} \sum_{\substack{x^n\in \cX^n , y^n \in \cY^n \\ r \in [\abs{\cA}], t \in [\abs{\cA}]}}P^{n}_{X}(x^n) P^{n}_{Y}(y^n)\left(\tr\left[(\bbI^{\cB'} - \mathbf{P}_{\cH_\cB^{\otimes n}})\rho^{\cB^n}_{x^n,y^n,s^n}\right] +\tr\left[\Pi_{W'_{(x^n,a_{r}),(y^n,b_{t}), \delta}}\rho^{\cB^n}_{x^n,y^n,s^n}\right]\right)+ 4\delta\nn\\
        &= 4\delta + \frac{4}{\abs{\cA}^2} \sum_{\substack{x^n\in \cX^n , y^n \in \cY^n \\ r \in [\abs{\cA}], t \in [\abs{\cA}]}}P^{n}_{X}(x^n) P^{n}_{Y}(y^n)\left(\tr\left[\Pi_{W'_{(x^n,a_{r}),(y^n,b_{t}), \delta}}\rho^{\cB^n}_{x^n,y^n,s^n}\right]\right)\nn\\
        &\overset{c}{=} 4\delta + \frac{4}{\abs{\cA}^2} \sum_{\substack{x^n\in \cX^n , y^n \in \cY^n \\ r \in [\abs{\cA}], t \in [\abs{\cA}]}}P^{n}_{X}(x^n) P^{n}_{Y}(y^n)\left(\tr\left[\Pi_{W'_{(x^n,a_{r}),(y^n,b_{t}), \delta}}\sum_{j=1}^{\abs{\cB}^n}\lambda_{x^n,y^n,s^n,j}\ketbrasys{h_{{x^n,y^n,s^n,j}}}{\cB^n}\right]\right)\nn\\
        &= 4\delta + \frac{4}{\abs{\cA}^2} \sum_{\substack{x^n\in \cX^n , y^n \in \cY^n \\ r \in [\abs{\cA}], t \in [\abs{\cA}]}}P^{n}_{X}(x^n) P^{n}_{Y}(y^n)\left(\sum_{j=1}^{\abs{\cB}^n}\lambda_{x^n,y^n,s^n,j}\norm{\Pi_{W'_{(x^n,a_{r}),(y^n,b_{t}), \delta}}\ket{h_{{x^n,y^n,s^n,j}}}}{2}^{2}\right)\nn\\
        &\overset{d}{\leq} 4\delta +\frac{24}{\abs{\cA}^2}\frac{1}{\delta^2}\sum_{\substack{x^n\in \cX^n , y^n \in \cY^n \\ r \in [\abs{\cA}], t \in [\abs{\cA}]}}P^{n}_{X}(x^n) P^{n}_{Y}(y^n)\sum_{j=1}^{\abs{\cB}^n}\lambda_{x^n,y^n,s^n,j}\Bigg(\norm{\Pi_{W'_{X,x^n,y^n}}\ket{h_{{x^n,y^n,s^n,j}}}}{2}^{2} + \norm{\Pi_{W'_{Y,x^n,y^n}}\ket{h_{{x^n,y^n,s^n,j}}}}{2}^{2}\nn\\
        &\hspace{80pt}+ \norm{\Pi_{W'_
        {x^n,y^n}}\ket{h_{{x^n,y^n,s^n,j}}}}{2}^{2}\Bigg)\nn\\
        &= 4\delta +\frac{24}{\abs{\cA}^2}\frac{1}{\delta^2}\sum_{\substack{x^n\in \cX^n , y^n \in \cY^n \\ r \in [\abs{\cA}], t \in [\abs{\cA}]}}P^{n}_{X}(x^n) P^{n}_{Y}(y^n)\left(\tr\left[\Pi_{W'_{X,x^n,y^n}}\sum_{j=1}^{\abs{\cB}^n}\lambda_{x^n,y^n,s^n,j}\ketbrasys{h_{{x^n,y^n,s^n,j}}}{\cB^n}\right] \right.\nn\\
        &+\left. \tr\left[\Pi_{W'_{Y,x^n,y^n}}\sum_{j=1}^{\abs{\cB}^n}\lambda_{x^n,y^n,s^n,j}\ketbrasys{h_{{x^n,y^n,s^n,j}}}{\cB^n}\right] + \tr\left[\Pi_{W'_
        {x^n,y^n}}\lambda_{x^n,y^n,s^n,j}\sum_{j=1}^{\abs{\cB}^n}\lambda_{x^n,y^n,s^n,j}\ketbrasys{h_{{x^n,y^n,s^n,j}}}{\cB^n}\right]\right)\nn\\
        &= 4\delta + \frac{24}{\abs{\cA}^2}\frac{1}{\delta^2}\sum_{\substack{x^n\in \cX^n , y^n \in \cY^n \\ r \in [\abs{\cA}], t \in [\abs{\cA}]}}P^{n}_{X}(x^n) P^{n}_{Y}(y^n)\Bigg(\left( 1 - \tr[\hat{\bbT}_{n,X,x^n,y^n}\rho^{\cB^n}_{x^n,y^n,s^n}]\right) + \left( 1 - \tr[\hat{\bbT}_{n,Y,x^n,y^n}\rho^{\cB^n}_{x^n,y^n,s^n}]\right) \nn\\
        &\hspace{60pt}+ \left( 1 - \tr[\hat{\bbT}_{n,x^n,y^n}\rho^{\cB^n}_{x^n,y^n,s^n}]\right)\Bigg)\nn\\
        &\leq 4\delta + \frac{24}{\delta^2} \Bigg(\tr\left[\left( \bbI^{X^nY^n\cB^n} - \hat{\bbT}_{n,X}\right)\rho^{X^nY^n\cB^n}_{P_{X},P_{Y},s^n}\right] + \tr\left[\left( \bbI^{X^nY^n\cB^n} - \hat{\bbT}_{n,Y}\right)\rho^{X^nY^n\cB^n}_{P_{X},P_{Y},s^n}\right] + \tr\left[\left( \bbI^{X^nY^n\cB^n} - \hat{\bbT}_{n}\right)\rho^{X^nY^n\cB^n}_{P_{X},P_{Y},s^n}\right]\Bigg)\nn\\
        &\overset{e}{\leq} 4\delta + \frac{72\gamma_n}{\delta^2},\label{VIM}
    \end{align}
    where $a$ follows from \eqref{perturbed_tilted_state_eq1} and Fact \ref{trace_norm2}, $b$ follows from Fact \ref{Gao}, in equality $c$, we assume $\rho^{\cB^n}_{x^n,y^n,s^n} := \sum_{j=1}^{\abs{\cB}^n}\lambda_{x^n,y^n,s^n,j}$ $\ketbrasys{h_{x^n,y^n,s^n,j}}{\cB^n}$.
    In inequality $e$, we define 
    
    \begin{align}
        &\gamma_n := f(n,\abs{\cX},\abs{\cY}, \abs{\cB}, \abs{\cS},t)\max\Big\{2^{nt\left(\widehat{R}_2 - \min _{Q \in \cP(\cS)} I_{1-t}[Y;\cB|X]_{P_{X},P_{Y},Q}\right)}, 2^{nt\left(\widehat{R}_1 - \min _{Q \in \cP(\cS)} I_{1-t}[X;\cB|Y]_{P_{X},P_{Y},Q}\right)},\nn\\
        &\hspace{120pt}2^{nt\left( \widehat{R}_1 + \widehat{R}_2 - \min _{Q \in \cP(\cS)} I_{1-t}[XY;\cB]_{P_{X},P_{Y},Q}\right)}
        \Big\},\label{gamma_theo_sen_mac}
    \end{align}

    and $e$ follows from \cref{lemma_achievability_mac_independent_eq1,lemma_achievability_mac_independent_eq3,lemma_achievability_mac_independent_eq5}. Inequality $d$ follows from the following discussion. Eqs.  \eqref{tilted_span_hayashi_X} to \eqref{tilted_span_hayashi} imply
    $$\cW'_{(x^n,a_{r}),(y^n,b_{t}), \delta} = \cM_{X,a_{r},\delta}(\cW_{X, x^n,y^n}) \bigplus \cM_{Y,b_{t},\delta}(\cW_{Y, x^n,y^n}) \bigplus \cW_{x^n,y^n}.$$
Thus, for any statevector $\ket{h}$ over $\cH_{\cB}^{\otimes n}$, fixing $l=2$ and $\alpha = \delta$ in \eqref{corol_sen2_eq1} of Fact \ref{corol_sen2}, we have 
\begin{equation}
    \norm{\Pi_{W'_{(x^n,a_{r}),(y^n,b_{t}), \delta}}\ket{h}}{2}^{2} \leq \frac{6}{\delta^2} \left(\norm{\Pi_{W'_{X,x^n,y^n}}\ket{h}}{2}^{2} + \norm{\Pi_{W'_{Y,x^n,y^n}}\ket{h}}{2}^{2} + \norm{\Pi_{W'_
        {x^n,y^n}}\ket{h}}{2}^{2}\right),\label{NBL}
\end{equation}
where $\Pi_{W'_{(x^n,a_{r}),(y^n,b_{t}), \delta}}, \Pi_{\cW'_{X,x^n,y^n}},\Pi_{\cW'_{Y,x^n,y^n}}$ and $\Pi_{\cW'_{x^n,y^n}}$ 
    are the orthogonal projectors on subspaces $\cW'_{(x^n,a_{r}),(y^n,b_{t}), \delta}, \cW'_{X,x^n,y^n},$ $\cW'_{Y,x^n,y^n}$, and $\cW'_{x^n,y^n}$, respectively. 
Eq. \eqref{NBL} implies Step $d$.
        
Now, for any arbitrary $\delta \in (0,1)$, for ($\widehat{R}_1,\widehat{R}_2$) satisfying \cref{hatR_1c,hatR_2c,hatSUMc}, we choose $n$ large enough  such that $\gamma_n \leq \delta^3$. 
Using \eqref{VIM}, we have
\begin{equation}
   \tr\left[ {\bbT}^{\star}_{n}\left(\rho^{X'Y'\cB'}_{P_{X},P_{Y},s^n}\right)\right] \geq 1 - 76\delta.\label{reject_Omega1_optimal}
\end{equation}
For each $(r,t) \in [\abs{\cA}]^2$ and for any collection $\{\sigma^{\cB^n}_{x^n,y^n}\} \subset \cD(\cH_\cB^{\otimes n})$, we upper-bound the quantity $\tr\left[\bbT^{\star}_{n,(x^n,a_{r}),(y^n,b_{t}), \delta}\right.$ $\left.\big(\sigma^{\cB'}_{(x^n,a_{X,r
}),\delta}\big)\right]$ as follows
\begin{align}
    &\hspace{10pt}\tr\left[\bbT^{\star}_{n,(x^n,a_{r}),(y^n,b_{t}), \delta}\big(\sigma^{\cB'}_{(x^n,a_{r}),\delta}\big)\right] \nn\\
    &\overset{a}{=} \tr\left[\bbT^{\star}_{n,(x^n,a_{r}),(y^n,b_{t}), \delta}\left(\frac{1 + \delta^2}{1 + 2\delta^2}\cM_{X,a_{r},\delta}(\sum_{y^n \in \cY^n}P^{n}_{Y}(y^n) \sigma^{\cB^n}_{x^n,y^n}) \oplus \cN_{X,a_{r},\delta}(\sum_{y^n \in \cY^n}P^{n}_{Y}(y^n) \sigma^{\cB^n}_{x^n,y^n})\right)\right]\nn\\
    &\overset{b}{\leq} \tr\left[\bbT^{\star}_{n,(x^n,a_{r}),(y^n,b_{t}), \delta}\left(\cM_{X,a_{r},\delta}(\sum_{y^n \in \cY^n}P^{n}_{Y}(y^n) \sigma^{\cB^n}_{x^n,y^n})\right)\right]\nn\\
    &\hspace{55pt}+ \norm{\bbT^{\star}_{n,(x^n,a_{r}),(y^n,b_{t}), \delta}}{1}.\norm{\cN_{X,a_{r},\delta}(\sum_{y^n \in \cY^n}P^{n}_{Y}(y^n) \sigma^{\cB^n}_{x^n,y^n})}{\infty} \nn\\
     &\overset{c}{\leq} \tr\left[\bbT^{\star}_{n,(x^n,a_{r}),(y^n,b_{t}), \delta}\left(\cM_{X,a_{r},\delta}(\sum_{y^n \in \cY^n}P^{n}_{Y}(y^n) \sigma^{\cB^n}_{x^n,y^n})\right)\right] + \abs{\cB}^n \norm{\cN_{X,a_{r},\delta}(\sum_{y^n \in \cY^n}P^{n}_{Y}(y^n) \sigma^{\cB^n}_{x^n,y^n})}{\infty} \nn\\
    &\overset{d}{\leq}\tr\left[\left(\bbI^{\cB'} - \Pi_{\cW'_{(x^n,a_{r}),(y^n,b_{t}), \delta}}\right)\left(\cM_{X,a_{r},\delta}(\sum_{y^n \in \cY^n}P^{n}_{Y}(y^n) \sigma^{\cB^n}_{x^n,y^n})\right)\right] + \frac{3\delta\abs{\cB}^n}{\sqrt{\abs{\cA}}}\nn\\
    &\overset{e}{\leq} \tr\left[\left(\bbI^{\cB'} - \bbT^{\cB'}_{\cW'_{X,(x^n,a_{r}), y^n,\delta}}\right)\left(\cM_{X,a_{r},\delta}(\sum_{y^n \in \cY^n}P^{n}_{Y}(y^n) \sigma^{\cB^n}_{x^n,y^n})\right)\right] + \frac{3\delta\abs{\cB}^n}{\sqrt{\abs{\cA}}}\nn\\
    &\overset{f}{=} \tr\left[\left(\bbI^{\cB'} - \bbT^{\cB'}_{\cW'_{X,(x^n,a_{r}), y^n,\delta}}\right)\mathbf{P}_{\cM_{X,a_{r},\delta}(\cH_{\cB}^{\otimes n})}\left(\cM_{X,a_{r},\delta}(\sum_{y^n \in \cY^n}P^{n}_{Y}(y^n) \sigma^{\cB^n}_{x^n,y^n})\right)\right] + \frac{3\delta\abs{\cB}^n}{\sqrt{\abs{\cA}}}\nn\\
    &\overset{g}{=} \tr\left[\left(\bbI^{\cM_{X,a_{r},\delta}(\cB')} - \bbT^{\cM_{X,a_{r},\delta}(\cB^n)}_{\cM_{X,a_{r},\delta}(\cW_{X, x^n,y^n})}\right)\left(\cM_{X,a_{r},\delta}(\sum_{y^n \in \cY^n}P^{n}_{Y}(y^n) \sigma^{\cB^n}_{x^n,y^n})\right)\right] + \frac{3\delta\abs{\cB}^n}{\sqrt{\abs{\cA}}}\nn\\
    &\overset{h}{=} \tr\left[\left(\bbI^{\cB^n} - \bbT^{\cB^n}_{\cW_{X, x^n,y^n}}\right)\left(\sum_{y^n \in \cY^n}P^{n}_{Y}(y^n) \sigma^{\cB^n}_{x^n,y^n}\right)\right] + \frac{3\delta\abs{\cB}^n}{\sqrt{\abs{\cA}}}\nn\\
    &= \tr\left[\hat{\bbT}_{n,X,x^n,y^n}\left(\sum_{y^n \in \cY^n}P^{n}_{Y}(y^n) \sigma^{\cB^n}_{x^n,y^n}\right)\right] + \frac{3\delta\abs{\cB}^n}{\sqrt{\abs{\cA}}}\label{accept_Omega2X_inner},
\end{align}
where $a$ follows from \eqref{Theo_Gen_smooth_X}, 
$b$ follows from Fact \ref{trace_ineq}, 
$c$ follows from the fact that 
$$\norm{\bbT^{\star}_{n,(x^n,a_{r}),(y^n,b_{t}), \delta}}{1} \leq  
\norm{\left(\bbI^{\cB'} - \Pi_{\cW'_{(x^n,a_{r}),(y^n,b_{t}), \delta}}\right)}{\infty}^{2} \norm{\mathbf{P}_{\cH_\cB^{\otimes n}}}{1} \leq \abs{\cB}^n,$$
where the first inequality follows from Fact \ref{trace_ineq}.
 Inequality $d$ follows from \eqref{small_N_X}, 
 $e$ follows from the fact that 
 $\cW'_{X,(x^n,a_{r}), y^n,\delta}$ resides inside 
 $\cW'_{(x^n,a_{r}),(y^n,b_{t}), \delta}$ which follows from \eqref{tilted_span_hayashi}, 
 in equality $f$, $\mathbf{P}_{\cM_{X,a_{r},\delta}(\cH_{\cB}^{\otimes n})}$ is the orthogonal projector in $\cH_{\cB'}$ onto $\cM_{X,a_{r},\delta}(\cH_\cB^{\otimes n})$, 
 $g$ follows from $\cW'_{X,(x^n,a_{r}), y^n,\delta} = \cM_{X,a_{r},\delta}(\cW_{X,x^n,y^n})$ $ \preceq \cM_{X,a_{r},\delta}(\cH_\cB^{\otimes n})$, 
 and $h$ follows from the fact that $\cM_{X,a_{r},\delta}$ is an isometry.

Now, from \eqref{accept_Omega2X_inner}, for any collection $\{\sigma^{\cB^n}_{x^n}\} \subset \cD(\cH_\cB^{\otimes n})$,  we upper-bound the quantity $\tr\left[{\bbT}^{\star}_{n}\left(\rho^{Y'}_{P_{Y}} \otimes \sigma^{X'\cB'}_{P_{X}, \{\sigma^{\cB^n}_{x^n}\}} \right)\right]$ as follows
\begin{align}
    \hspace{10pt}\tr\left[{\bbT}^{\star}_{n}\left(\rho^{Y'}_{P_{Y}} \otimes \sigma^{X'\cB'}_{P_{X}, \{\sigma^{\cB^n}_{x^n}\}} \right)\right]
    &= \frac{1}{\abs{\cA}^2}\sum_{\substack{x^n \in \cX^n , y^n\in \cY^n \\ r \in [\abs{\cA}] , t \in  [\abs{\cA}]}}P^{n}_{X}(x^n)P^{n}_{Y}(y^n)\tr\left[\bbT^{\star}_{n,(x^n,a_{r}),(y^n,b_{t}), \delta}\big(\sigma^{\cB'}_{(x^n,a_{r}),\delta}\big)\right]\nn\\
    &\leq  \sum_{\substack{x^n \in \cX^n \\ y^n\in \cY^n }}P^{n}_{X}(x^n)P^{n}_{Y}(y^n)\tr\left[\hat{\bbT}_{n,X,x^n,y^n}\left(\sum_{y^n \in \cY^n}P^{n}_{Y}(y^n) \sigma^{\cB^n}_{x^n,y^n}\right)\right] + \frac{3\delta\abs{\cB}^n}{\sqrt{\abs{\cA}}}\nn\\
    &= \tr\left[\hat{\bbT}_{n,X}\left({\rho^{\cH_Y^{\otimes n}}_{P_Y}} \otimes \sigma^{X^n\cB^n}_{P_{X}, \{\sigma^{\cB^n}_{x^n}\}}\right)\right] + \frac{3\delta\abs{\cB}^n}{\sqrt{\abs{\cA}}}\nn\\
&\overset{a}{\leq} g_1(n,\abs{\cX},\abs{\cB})2^{-n\widehat{R}_2}+\frac{3\delta\abs{\cB}^n}{\sqrt{\abs{\cA}}},\label{accept_Omega2X}
\end{align}
where $a$ follows from \eqref{lemma_achievability_mac_independent_eq2}. Similar to \eqref{accept_Omega2X}, we get the following:
\begin{align}
\tr\left[{\bbT}^{\star}_{n}\left(\rho^{X'}_{P_{X}} \otimes \sigma^{Y'\cB'}_{P_{Y}, \{\sigma^{\cB^n}_{y^n}\}} \right)\right] \leq g_2(n,\abs{\cY},\abs{\cB})2^{-n\widehat{R}_1}+\frac{3\delta\abs{\cB}^n}{\sqrt{\abs{\cA}}},\label{accept_Omega2Y}\\
\tr\left[{\bbT}^{\star}_{n}\left(\rho^{X'}_{P_{X}} \otimes\rho^{Y'}_{P_{Y}} \otimes \sigma^{\cB'} \right)\right] \leq g_3(n,\abs{\cB})2^{-n(\widehat{R}_1+\widehat{R}_2)}+\frac{7\delta\abs{\cB}^n}{\sqrt{\abs{\cA}}}.\label{accept_Omega2XY}
\end{align}

We now choose $\abs{\cA}$ large enough such that 
\begin{equation*}
    \frac{7\delta\abs{\cB}^n}{\sqrt{\abs{\cA}}} \leq \min\big\{g_1(n,\abs{\cX},\abs{\cB})2^{-n\widehat{R}_2}, g_2(n,\abs{\cY},\abs{\cB})2^{-n\widehat{R}_1}, g_3(n,\abs{\cB})2^{-n(\widehat{R}_1+\widehat{R}_2)}\big\}.
\end{equation*}
Then, we have the following inequalities:
\begin{align}
    \tr\left[{\bbT}^{\star}_{n}\left(\rho^{Y'}_{P_{Y}} \otimes \sigma^{X'\cB'}_{P_{X}, \{\sigma^{\cB^n}_{x^n}\}} \right)\right] \leq g_1(n,\abs{\cX},\abs{\cB})2^{-n\widehat{R}_2+1},\label{accept_Omega2X_optimal}\\
    \tr\left[{\bbT}^{\star}_{n}\left(\rho^{X'}_{P_{X}} \otimes \sigma^{Y'\cB'}_{P_{Y}, \{\sigma^{\cB^n}_{y^n}\}} \right)\right] \leq g_2(n,\abs{\cY},\abs{\cB})2^{-n\widehat{R}_1+1},\label{accept_Omega2Y_optimal}\\
\tr\left[{\bbT}^{\star}_{n}\left(\rho^{X'}_{P_{X}} \otimes\rho^{Y'}_{P_{Y}} \otimes \sigma^{\cB'} \right)\right] \leq g_3(n,\abs{\cB})2^{-n(\widehat{R}_1+\widehat{R}_2) +1}.\label{accept_Omega2XY_optimal}
\end{align}
Hence, inequality \cref{reject_Omega1_optimal,accept_Omega2X_optimal,accept_Omega2Y_optimal,accept_Omega2XY_optimal} prove \cref{reject_Omega1_optimal_theo,accept_Omega2X_optimal_theo,accept_Omega2Y_optimal_theo,accept_Omega2XY_optimal_theo} of Lemma \ref{theorem_sen_MAC_generalised_indep}.

\subsection{Proof of Lemma \ref{claim_small_perturbation}}\label{subsec_small_perturbation}
    Let us consider $\rho := \sum_{j \in [\abs{\cB}^n]} \lambda_j \ketbra{h_j}$, where $\{\ket{h_j}\}_{j\in\abs{\cB}^n}$ is an orthonormal eigen basis of $\cH_{\cB}^{\otimes n}$ and $\{ \lambda_j\}_{j\in\abs{\cB}^n}$ is a distribution. 
Applying the map $\cM_{XY,a_{r},b_{t},\delta}$ on $\rho$ with any $(r,t) \in [\abs{\cA}]^2$, we get the following operator:
    \begin{align*}
        \cM_{XY,a_{r},b_{t},\delta}\left(\rho\right) &= \frac{1}{1 + 2\delta^2}\sum_{j \in [\abs{\cB}^n]}\lambda_j \left(\ket{h_j} \oplus \delta\ket{h_j}\ket{a_{r}} \oplus \delta\ket{h_j}\ket{b_{t}}\right)(\left\bra{h_j} \oplus \delta\bra{h_j}\bra{a_{r}} \oplus \delta\bra{h_j}\bra{b_{t}}\right)\\
        &= \frac{1}{1 + 2\delta^2}\sum_{j \in [\abs{\cB}^n]}\lambda_j\left[\begin{array}{c c c }
        \ketbra{h_j} & \delta\ket{h_j}\bra{h_j}\bra{a_{r}} &  \delta\ket{h_j}\bra{h_j}\bra{b_{t}}\\
        \delta\ket{h_j}\ket{a_{r}}\bra{h_j} &  \delta^2\ket{h_j}\ket{a_{r}}\bra{h_j}\bra{a_{r}} & \delta^2\ket{h_j}\ket{a_{r}}\bra{h_j}\bra{b_{t}}\\ 
        \delta\ket{h_j}\ket{b_{t}}\bra{h_j} & \delta^2\ket{h_j}\ket{b_{t}}\bra{h_j}\bra{a_{r}}  & \delta^2\ket{h_j}\ket{b_{t}}\bra{h_j}\bra{b_{t}}
    \end{array}\right].
    \end{align*} 
    
    Now for each $j\in [\abs{\cB}^n]$, we have    
    \begin{align*}
        &\hspace{12pt}\norm{\cM_{XY,a_{r},b_{t},\delta}\left(\ketbra{h_j}\right) - \left[\begin{array}{cc}
      \ketbra{h_j}   &  \left[\mathbf{0}\right]_{\abs{\cB}^n \times 2\abs{\cA}\abs{\cB}^n}\\
      \left[\mathbf{0}\right]_{ 2\abs{\cA}\abs{\cB}^n \times \abs{\cB}^n}   & \left[\mathbf{0}\right]_{ 2\abs{\cA}\abs{\cB}^n \times  2\abs{\cA}\abs{\cB}^n}
    \end{array}\right]_{\abs{\cB'} \times \abs{\cB'}}}{1}\\
    &= \norm{\left[\begin{array}{c c c }
        \left[ \mathbf{0} \right]_{\abs{\cB} \times \abs{\cB}} & \delta\ket{h_j}\bra{h_j}\bra{a_{r}} &  \delta\ket{h_j}\bra{h_j}\bra{b_{t}}\\
        \delta\ket{h_j}\ket{a_{r}}\bra{h_j} &  \delta^2\ket{h_j}\ket{a_{r}}\bra{h_j}\bra{a_{r}} & \delta^2\ket{h_j}\ket{a_{r}}\bra{h_j}\bra{b_{t}}\\ 
        \delta\ket{h_j}\ket{b_{t}}\bra{h_j} & \delta^2\ket{h_j}\ket{b_{t}}\bra{h_j}\bra{a_{r}}  & \delta^2\ket{h_j}\ket{b_{t}}\bra{h_j}\bra{b_{t}}
    \end{array}\right]}{1}\\
    &\leq 2\delta\norm{\ket{h_j}\bra{h_j}\bra{a_{r}}} 
 {1} +  2\delta\norm{\ket{h_j}\bra{h_j}\bra{b_{t}}}{1} + \delta^2(\norm{\ket{h_j}\ket{a_{r}}\bra{h_j}\bra{a_{r}}}{1} + \norm{\ket{h_j}\ket{b_{t}}\bra{h_j}\bra{b_{t}}}{1} + 2\norm{\ket{h_j}\ket{b_{t}}\bra{h_j}\bra{a_{r}}}{1}) \\
    &= 4\delta + 4\delta^2
    \leq 8\delta.
    \end{align*}
    Thus, we have    
    \begin{align*}
    \norm{\widehat{\rho}_{a_{r},b_{t},\delta}^{\cB'} - \Tilde{\rho}^{\cB'}}{1} &\leq \sum_{j=1}^{\abs{\cB}^n} \lambda_{j} \norm{\cM_{XY,a_{r},b_{t},\delta}\left(\ketbra{h_j}\right) - \left[\begin{array}{cc}
      \ketbra{h_j}   &  \left[\mathbf{0}\right]_{\abs{\cB}^n \times 2\abs{\cA}\abs{\cB}^n}\\
      \left[\mathbf{0}\right]_{ 2\abs{\cA}\abs{\cB}^n \times \abs{\cB}^n}   & \left[\mathbf{0}\right]_{ 2\abs{\cA}\abs{\cB}^n \times  2\abs{\cA}\abs{\cB}^n}
    \end{array}\right]_{\abs{\cB'} \times \abs{\cB'}}}{1}\\
    &\leq 8\delta.
    \end{align*}
\endproof

\section{Proof of Lemma \ref{lemma_cq_mac_avht_independent}}\label{proof_lemma_cq_mac_avht_independent}
For an $\left(\bar{x}^n,\bar{y}^n\right) := ((\bar{x}_1,\bar{y}_1),(\bar{x}_2, \bar{y}_2),\cdots,(\bar{x}_n,\bar{y}_n))\in \cX^n\times\cY^n$ which has a form $\left(\bar{x}^n,\bar{y}^n\right) := \Big(\underbrace{(1,1),\cdots,(1,1)}_{m_{1,1}},\underbrace{(1,2),\cdots,(1,2)}_{m_{1,2}},$\\$\cdots,\underbrace{(\abs{\cX},\abs{\cY}),\cdots,(\abs{\cX},\abs{\cY})}_{m_{\abs{\cX},\abs{\cY}}}\Big)$, where $\forall i \in [\abs{\cX}], j \in [\abs{\cY}], m_{i,j} \geq 0$ and $\sum_{\substack{i \in [\abs{\cX}] \\ j \in[ \abs{\cY}]}}m_{i,j} = n$. Given $\left(\bar{x}^n,\bar{y}^n\right)$, we consider the following state:
\begin{equation}
    \widehat{\rho}^{\cB^n}_{\bar{x}^n,\bar{y}^n} := \rho^{\cB^{m_{1,1}}}_{U,m_{1,1}} \otimes \rho^{\cB^{m_{1,2}}}_{U,m_{1,2}} \otimes \cdots \otimes \rho^{\cB^{m_{\abs{\cX},\abs{\cY}}}}
    _{U,m_{\abs{\cX},\abs{\cY}}},\label{rhohatxbarnybarn}
\end{equation}
where $\forall i \in [\abs{\cX}], j \in [\abs{\cY}], \rho^{\cB^{m_{i,j}}}
    _{U,m_{i,j}}$ is defined in \eqref{uniform_n_state}. Now for any general pair $(x^n,y^n) \in \cX^n\times\cY^n$, which is a permutation of $(\bar{x}^n,\bar{y}^n)$ i.e. $x^n  = \pi(\bar{x}^n) = (\bar{x}_{\pi^{-1}(1)},\bar{x}_{\pi^{-1}(2)},\cdots,\bar{x}_{\pi^{-1}(n)})$ and $y^n  = \pi(\bar{y}^n) = (\bar{y}_{\pi^{-1}(1)},\bar{y}_{\pi^{-1}(2)},\cdots,\bar{y}_{\pi^{-1}(n)})$, for some $\pi \in S_n$, we define $\hat{\rho}^{\cB^n}_{x^n,y^n}$ as follows
\begin{equation}
    \widehat{\rho}^{\cB^n}_{x^n,y^n} := V^{\cB^n}(\pi) \left(\widehat{\rho}^{\cB^n}_{\bar{x}^n,\bar{y}^n}\right){V^{\cB^n}}^{\dagger}(\pi).\label{rhohatxnyn}
\end{equation}

We define $\widehat{\rho}^{\cB^n}_{y^n}$ for any $y^n \in \cY^n$ in an analogous way to how  $\widehat{\rho}^{\cB^n}_{x^n}$ is defined for any $x^n \in \cX^n$ in \eqref{rhohatxn}. We now define the following states:
\begin{align*}
    \widehat{\rho}^{X^nY^n\cB^n}_{U,P_X,P_Y} &:= \sum_{\substack{x^n \in \cX^n \\ y^n \in \cY^n}}P^{n}_X(x^n) \ketbrasys{x^n}{X^n} \otimes P^{n}_Y(y^n) \ketbrasys{y^n}{Y^n} \otimes \widehat{\rho}^{\cB^n}_{x^n,y^n},\\
    \widehat{\rho}^{X^n\cB^n}_{U,P_X} &:= \sum_{\substack{x^n \in \cX^n}}P^{n}_X(x^n) \ketbrasys{x^n}{X^n} \otimes \widehat{\rho}^{\cB^n}_{x^n},\\
    \widehat{\rho}^{Y^n\cB^n}_{U,P_Y} &:= \sum_{\substack{y^n \in \cY^n}}P^{n}_Y(y^n) \ketbrasys{y^n}{Y^n} \otimes \widehat{\rho}^{\cB^n}_{y^n}.
\end{align*}
For $\widehat{R}_1,\widehat{R}_2>0$, we define three projective measurement-based tests $\hat{\bbT}_{n,X},\hat{\bbT}_{n,Y}$ and $\hat{\bbT}_{n}$ as follows
\begin{align}
    \hat{\bbT}_{n,X} &:= \sum_{\substack{x^n \in \cX^n \\ y^n \in \cY^n}} \ketbrasys{x^n}{X^n} \otimes \ketbrasys{y^n}{Y^n} \otimes \left\{\widehat{\rho}^{\cB^n}_{x^n,y^n} \succeq 2^{n\widehat{R}_2}\widehat{\rho}^{\cB^n}_{x^n}\right\} \triangleq \left\{\widehat{\rho}^{X^nY^n\cB^n}_{U,P_X,P_Y} \succeq 2^{n\widehat{R}_2} \left({\rho^{\cH_Y^{\otimes n}}_{P_Y}} \otimes \rho^{X^n\cB^n}_{U,P_X}\right)\right\}\label{universal_test_X},\\
    \hat{\bbT}_{n,Y} &:= \sum_{\substack{x^n \in \cX^n \\ y^n \in \cY^n}} \ketbrasys{x^n}{X^n} \otimes \ketbrasys{y^n}{Y^n} \otimes \left\{\widehat{\rho}^{\cB^n}_{x^n,y^n} \succeq 2^{n\widehat{R}_1}\widehat{\rho}^{\cB^n}_{y^n}\right\} \triangleq \left\{\widehat{\rho}^{X^nY^n\cB^n}_{U,P_X,P_Y} \succeq 2^{n\widehat{R}_1} \left({\rho^{\cH_X^{\otimes n}}_{P_X}} \otimes \rho^{Y^n\cB^n}_{U,P_Y}\right)\right\}\label{universal_test_Y},\\
    \hat{\bbT}_{n} &:= \sum_{\substack{x^n \in \cX^n \\ y^n \in \cY^n}} \ketbrasys{x^n}{X^n} \otimes \ketbrasys{y^n}{Y^n} \otimes \left\{\widehat{\rho}^{\cB^n}_{x^n,y^n} \succeq 2^{n(\widehat{R}_1+\widehat{R}_2)}\rho^{\cB^n}_{U,n}\right\} \triangleq \left\{\widehat{\rho}^{X^nY^n\cB^n}_{U,P_X,P_Y} \succeq 2^{n(\widehat{R}_1+\widehat{R}_2)} \left({\rho^{\cH_X^{\otimes n}}_{P_X}} \otimes {\rho^{\cH_Y^{\otimes n}}_{P_Y}} \otimes \rho^{\cB^n}_{U,n}\right)\right\}\label{universal_test_None}.
\end{align}

    Since the states $\widehat{\rho}^{X^nY^n\cB^n}_{U,P_X,P_Y}, {\rho^{\cH_Y^{\otimes n}}_{P_Y}} \otimes \rho^{X^n\cB^n}_{U,P_X}, {\rho^{\cH_X^{\otimes n}}_{P_X}} \otimes \rho^{Y^n\cB^n}_{U,P_Y}$, 
    and $\rho^{\cH_X^{\otimes n}}_{P_X} \otimes \rho^{\cH_Y^{\otimes n}}_{P_Y} \otimes \rho^{\cB^n}_{U,n}$ are permutation invariant, 
    the matrices $\hat{\bbT}_{n,X}, \hat{\bbT}_{n,Y}$, and $\hat{\bbT}_{n}$ are permutation invariant. It is also important to note that like $\bbT_{n}$ defined in \eqref{universal_test}, the expressions of $\hat{\bbT}_{n,X}, \hat{\bbT}_{n,Y}$ and $\hat{\bbT}_{n}$ are independent of the distribution $P_X,P_Y$. Now, for any $s^n \in \cS^n$, we upper-bound the quantity $ \tr\left[\left( \bbI^{X^nY^n\cB^n} - \hat{\bbT}_{n,X}\right)\rho^{X^nY^n\cB^n}_{P_X,P_Y,s^n}\right]$ as follows
\begin{align}
     &\tr\left[\left( \bbI^{X^nY^n\cB^n} - \hat{\bbT}_{n,X}\right)\rho^{X^nY^n\cB^n}_{P_X,P_Y,s^n}\right]\nn\\ &\overset{a}{=}  \tr\left[\left( \bbI^{X^nY^n\cB^n} - \hat{\bbT}_{n,X}\right)\frac{1}{\abs{T_{Q_{s^n}}}} \sum_{\pi \in S_n}\rho^{X^nY^n\cB^n}_{P_X,P_Y,\pi(s^n)}\right]\nn\\
     &= \tr\left[\left( \bbI^{X^nY^n\cB^n} - \hat{\bbT}_{n,X}\right)\frac{1}{\abs{T_{Q_{s^n}}}}\sum_{\widehat{s}^{n} \in T_{Q_{s^n}}} \rho^{X^nY^n\cB^n}_{P_X,P_Y,\widehat{s}^n}\right]\nn\\
     &\overset{b}{\leq} (n+1)^{\abs{\cS}-1}\tr\left[\left( \bbI^{X^nY^n\cB^n} - \hat{\bbT}_{n,X}\right)\sum_{\widehat{s}^{n} \in T_{Q_{s^n}}} Q^{n}_{s^n}(\widehat{s}^n)\rho^{X^nY^n\cB^n}_{P_X,P_Y,\widehat{s}^n}\right]\nn\\
     &\leq (n+1)^{\abs{\cS}-1}\tr\left[\left( \bbI^{X^nY^n\cB^n} - \hat{\bbT}_{n,X}\right)\sum_{\widehat{s}^{n} \in \cS^n} Q^{n}_{s^n}(\widehat{s}^n)\rho^{X^nY^n\cB^n}_{P_X,P_Y,\widehat{s}^n}\right]\nn\\
     &\overset{c}{=} (n+1)^{\abs{\cS}-1}\tr\left[\left( \bbI^{X^nY^n\cB^n} - \hat{\bbT}_{n,X}\right)\left(\rho^{{XY\cB}}_{P_X,P_Y,Q_{s^n}}\right)^{\otimes n}
     \right]\nn\\
     &= (n+1)^{\abs{\cS}-1}\tr\left[\left( \bbI^{X^nY^n\cB^n} - \hat{\bbT}_{n,X}\right)\sum_{\substack{x^n \in \cX^n \\ y^n \in \cY^n}}P^{n}_X(x^n)\ketbrasys{x^n}{X^n} \otimes P^{n}_Y(y^n)\ketbrasys{y^n}{Y^n} \otimes \left(\bigotimes_{i=1}^{n}\rho^{\cB}_{Q_{s^n},x_i,y_i}\right)
     \right]\nn\\
     &\overset{d}{=}(n+1)^{\abs{\cS}-1}\sum_{x^n \in \cX^n}P^{n}_X(x^n)\tr\left[ \sum_{y^n \in \cY^n}P^{n}_Y(y^n)\left(\bigotimes_{i=1}^{n}\rho^{\cB}_{Q_{s^n},x_i,y_i}\right)\left\{\widehat{\rho}^{\cB^n}_{x^n,y^n} \prec 2^{n\widehat{R}_2}\widehat{\rho}^{\cB^n}_{x^n}\right\}
     \right]\nn\\
     &\overset{e}{\leq}(n+1)^{\abs{\cS}-1}2^{nt \widehat{R}_2 }\sum_{x^n \in \cX^n}P^{n}_X(x^n)\tr\left[ \sum_{y^n \in \cY^n}P^{n}_Y(y^n)\left(\bigotimes_{i=1}^{n}\rho^{\cB}_{Q_{s^n},x_i,y_i}\right)(\widehat{\rho}^{\cB^n}_{x^n,y^n})^{-t}(\widehat{\rho}^{\cB^n}_{x^n})^{t}
     \right]\nn\\
     &\overset{f}{\leq}(n+1)^{\abs{\cS}-1}2^{nt \widehat{R}_2}n^{\frac{t\abs{\cX}\abs{\cY}(\abs{\cB}-1)\abs{\cB}}{2}}\abs{\Lambda^{n}_{\abs\cB}}^{t \abs{\cX}\abs{\cY}}\sum_{x^n \in \cX^n}P^{n}_X(x^n)\tr\left[ \sum_{y^n \in \cY^n}P^{n}_Y(y^n) \left(\bigotimes_{i=1}^{n}\rho^{\cB}_{Q_{s^n},x_i,y_i}\right)^{1-t}(\widehat{\rho}^{\cB^n}_{x^n})^{t}
     \right]\nn\\
     &\overset{g}{\leq} (n+1)^{\abs{\cS}-1 + \frac{t\abs{\cX}\abs{\cY}(\abs{\cB}-1)\abs{\cB}}{2} + (\abs{\cB} - 1)t \abs{\cX}\abs{\cY}}2^{nt \widehat{R}_2}\sum_{x^n \in \cX^n}P^{n}_X(x^n)\tr\left[ \bigotimes_{i=1}^{n}\left(\sum_{y \in \cY}P_Y(y)\left(\rho^{\cB}_{Q_{s^n},x_i,y}\right)^{1-t}\right)(\widehat{\rho}^{\cB^n}_{x^n})^{t}
     \right]\nn\\
     &\overset{h}{\leq} (n+1)^{\abs{\cS}-1 + \frac{t\abs{\cX}\abs{\cY}(\abs{\cB}-1)\abs{\cB}}{2} + (\abs{\cB} - 1)t \abs{\cX}\abs{\cY}}2^{nt \widehat{R}_2}\sum_{x^n \in \cX^n}P^{n}_X(x^n)\left(\tr\left[ \left(\bigotimes_{i=1}^{n}\left(\sum_{y \in \cY}P_Y(y)\left(\rho^{\cB}_{Q_{s^n},x_i,y}\right)^{1-t}\right)\right)^{\frac{1}{1 - t}}
     \right]\right)^{1-t}\nn\\
     &=(n+1)^{\abs{\cS}-1 + \frac{t\abs{\cX}\abs{\cY}(\abs{\cB}-1)\abs{\cB}}{2} + (\abs{\cB} - 1)t \abs{\cX}\abs{\cY}}2^{nt \widehat{R}_2}\sum_{x^n \in \cX^n}P^{n}_X(x^n)\prod_{i=1}^{n}\left(\tr\left[ \left(\sum_{y \in \cY}P_Y(y)\left(\rho^{\cB}_{Q_{s^n},x_i,y}\right)^{1-t}\right)^{\frac{1}{1 - t}}
     \right]\right)^{1-t}\nn\\
     &= (n+1)^{\abs{\cS}-1 + \frac{t\abs{\cX}\abs{\cY}(\abs{\cB}-1)\abs{\cB}}{2} + (\abs{\cB} - 1)t \abs{\cX}\abs{\cY}}2^{nt \widehat{R}_2}\left(\sum_{x \in \cX}P_X(x)\left(\tr\left[ \left(\sum_{y \in \cY}P_Y(y)\left(\rho^{\cB}_{Q_{s^n},x,y}\right)^{1-t}\right)^{\frac{1}{1 - t}}
     \right]\right)^{1-t}\right)^{n}\nn\\
     &\overset{i}{=} (n+1)^{\abs{\cS}-1 + \frac{t\abs{\cX}\abs{\cY}(\abs{\cB}-1)\abs{\cB}}{2} + (\abs{\cB} - 1)t \abs{\cX}\abs{\cY}}2^{nt( \widehat{R}_2 - I_{1-t}[Y;\cB|X]_{P_{X},P_{Y},Q_{s^n}})}\nn\\
     &{\leq}(n+1)^{\abs{\cS}-1 + \frac{t\abs{\cX}\abs{\cY}(\abs{\cB}-1)\abs{\cB}}{2} + (\abs{\cB} - 1)t \abs{\cX}\abs{\cY}}2^{nt( \widehat{R}_2 - \min_{Q \in \cP(\cS)}I_{1-t}[Y;\cB|X]_{P_{X},P_{Y},Q})}\nn\\
     &{\leq}f(n,\abs{\cX},\abs{\cY}, \abs{\cB}, \abs{\cS},t)2^{nt( \widehat{R}_2 - \min_{Q \in \cP(\cS)}I_{1-t}[Y;\cB|X]_{P_{X},P_{Y},Q})},\label{CABT2}
\end{align}
where $a$ follows from the fact that $\hat{\bbT}_{n,X}$ is permutation invariant, $b$ follows from \eqref{inv_type_set_size_ub}, $c$ follows from \eqref{rhoPXPYQ_1}, $d$ follows from the definition of $\hat{\bbT}_{n,X}$ in \eqref{universal_test_X}, 
$f$ follows from arguments similar to prove Lemma \ref{claim_f}, 
$g$ follows from Fact \ref{fact_type_size_ub}, $h$ follows from \eqref{Holder_corllary_eq} of Fact \ref{Holder_corllary}, and $i$ follows from \eqref{sibsonqmiY;XB_1}.
Here, the derivation of $e$ is given as follows.
Since Lemma \ref{lemma_commutativity} guarantees that 
$ \widehat{\rho}^{\cB^n}_{x^n,y^n}$ commutes with $\widehat{\rho}^{\cB^n}_{x^n}$, 
we have $\left\{\widehat{\rho}^{\cB^n}_{x^n,y^n} - 2^{n\widehat{R}_2}\widehat{\rho}^{\cB^n}_{x^n} \prec 0\right\} \leq 2^{nt\widehat{R}_2}\left(\widehat{\rho}^{\cB^n}_{x^n,y^n}\right)^{-t}\left(\widehat{\rho}^{\cB^n}_{x^n}\right)^{t}$ for all $t \in (0,1)$, 
which implies $e$. 
Inequality \eqref{CABT2} proves \eqref{lemma_achievability_mac_independent_eq1}.

For any collection $\left\{\sigma_{x^n}^{\cB^n}\right\}_{\substack{x^n \in \cX^n}} \subset \cD(\cH_\cB^{\otimes n})$, we now upper-bound the quantity $\tr\left[\hat{\bbT}_{n,X}\left({\rho^{\cH_Y^{\otimes n}}_{P_Y}} \otimes \sigma^{X^n\cB^n}_{P_{X}, \{\sigma^{\cB^n}_{x^n}\}}\right)\right]$ in the following discussions.
 We define the subgroup $S_{n,x^n} := \left\{\pi_n \in S_n : \pi(x^n) = x^n\right\}$. Since $\hat{\bbT}_{n,X}$ is permutation invariant, any element $\pi \in S_{n,x^n}$ satisfies 
 the relation 
 \begin{align}
    &\hspace{13pt}\tr\left[\hat{\bbT}_{n,X}\left(\rho^{\cH_Y^{\otimes n}}_{P_Y}\otimes \left(\sum_{x^n \in \cX^n} P^{n}_X(x^n) \ketbrasys{x^n}{X^n}\otimes  \sigma^{\cB^n}_{x^n}\right)\right)\right]\nn\\ 
    &= \tr\left[\hat{\bbT}_{n,X}\left(V^{X^nY^n\cB^n}(\pi)\left(\rho^{\cH_Y^{\otimes n}}_{P_Y}\otimes \left(\sum_{x^n \in \cX^n} P^{n}_X(x^n) \ketbrasys{x^n}{X^n}\otimes  \sigma^{\cB^n}_{x^n}\right)\right){V^{X^nY^n\cB^n}}^{\dagger}(\pi)\right)\right]\nn\\ 
    &\overset{a}{=}\tr\left[\hat{\bbT}_{n,X}\left(\rho^{\cH_Y^{\otimes n}}_{P_Y}\otimes \left(\sum_{x^n \in \cX^n} P^{n}_X(x^n) \ketbrasys{\pi(x^n)}{X^n}\otimes  V^{\cB^n}(\pi)\left(\sigma^{\cB^n}_{x^n}\right){V^{\cB^n}}^{\dagger}(\pi)\right)\right)\right]\nn\\ 
    &\overset{b}{=}\tr\left[\hat{\bbT}_{n,X}\left(\rho^{\cH_Y^{\otimes n}}_{P_Y}\otimes\left(\sum_{x^n \in \cX^n} P^{n}_X(x^n) \ketbrasys{x^n}{X^n}\otimes V^{\cB^n}(\pi)\left(\sigma^{\cB^n}_{x^n}\right){V^{\cB^n}}^{\dagger}(\pi)\right)\right)\right],\label{prop_perm_inv_subgroup}
\end{align}
where $a$ follows from the fact that $\rho^{\cH_Y^{\otimes n}}_{P_Y}$ is permutation invariant and $b$ follows from the property of $S_{n,x^n}$. Thus, we have the following
\begin{align}
    &\tr\left[\hat{\bbT}_{n,X}\left({\rho^{\cH_Y^{\otimes n}}_{P_Y}} \otimes \sigma^{X^n\cB^n}_{P_{X}, \{\sigma^{\cB^n}_{x^n}\}}\right)\right] \notag\\
    &\overset{a}{=} \frac{1}{\abs{S_{n,x^n}}}\sum_{\pi \in S_{n,x^n}}\tr\left[\hat{\bbT}_{n,X}\left(\rho^{\cH_Y^{\otimes n}}_{P_Y}\otimes\left(\sum_{x^n \in \cX^n} P^{n}_X(x^n) \ketbrasys{x^n}{X^n}\otimes  V^{\cB^n}(\pi)\left(\sigma^{\cB^n}_{x^n}\right){V^{\cB^n}}^{\dagger}(\pi)\right)\right)\right] \notag\\
    &=\tr\left[\hat{\bbT}_{n,X}\left(\rho^{\cH_Y^{\otimes n}}_{P_Y}\otimes\left(\sum_{x^n \in \cX^n} P^{n}_X(x^n) \ketbrasys{x^n}{X^n}\otimes\frac{1}{\abs{S_{n,x^n}}}\sum_{\pi \in S_{n,x^n}}  V^{\cB^n}(\pi)\left(\sigma^{\cB^n}_{x^n}\right){V^{\cB^n}}^{\dagger}(\pi)\right)\right)\right] \notag\\
    &\overset{b}{\leq} n^{\frac{\abs{\cX}\abs{\cB}(\abs{\cB}-1)}{2}}\abs{\Lambda^{n}_{\abs{\cB}}}^{\abs{\cX}}\tr\left[\hat{\bbT}_{n,X}\left(\rho^{\cH_Y^{\otimes n}}_{P_Y}\otimes\left(\sum_{x^n \in \cX^n} P^{n}_X(x^n) \ketbrasys{x^n}{X^n} \otimes \widehat{\rho}^{\cB^n}_{x^n}\right)\right)\right] \notag\\
    &\overset{c}{\leq} n^{\frac{\abs{\cX}\abs{\cB}(\abs{\cB}-1)}{2}}\abs{\Lambda^{n}_{\abs{\cB}}}^{\abs{\cX}}2^{-n\widehat{R}_2}\tr\left[\hat{\bbT}_{n,X}\widehat{\rho}^{X^nY^n\cB^n}_{U,P_X,P_Y}\right]
    \leq n^{\frac{\abs{\cX}\abs{\cB}(\abs{\cB}-1)}{2}}\abs{\Lambda^{n}_{\abs{\cB}}}^{\abs{\cX}}2^{-n\widehat{R}_2} = g_1(n,\abs{\cX},\abs{\cB})2^{-n\widehat{R}_2},\label{CABT}
\end{align}
where $a$ follows from \eqref{prop_perm_inv_subgroup}, $b$ follows from Lemma \ref{lemma_perm_Inv_universal}, and $c$ follows from the definition of $\hat{\bbT}_{n,X}$, given in \eqref{universal_test_X}. 
Inequality \eqref{CABT} proves \eqref{lemma_achievability_mac_independent_eq2}. 

For $\hat{\bbT}_{n,Y}$, using approaches similar to the approaches discussed above to prove \cref{lemma_achievability_mac_independent_eq1,lemma_achievability_mac_independent_eq2}, one can prove \cref{lemma_achievability_mac_independent_eq3,lemma_achievability_mac_independent_eq4}. Further, for $\hat{\bbT}_{n}$, the proofs of \cref{lemma_achievability_mac_independent_eq5,lemma_achievability_mac_independent_eq6} follow directly from the proofs of \cref{lemma_achievability_theor_gen_eq1,lemma_achievability_theor_gen_eq2} in Section \ref{sec:IT}. 
\endproof

\section{Proof of Corollary \ref{theo_gen_indep_mac_corollary}}\label{proof_theo_gen_indep_mac_corollary}
We choose the collection $\{\sigma^{\cB^n}_{x^n,y^n}\}_{\substack{x^n \in \cX^n \\ y^n \in \cY^n}}$ in $\sigma^{X^nY^n\cB^n}_{P_{X},P_{Y}, \{\sigma^{\cB^n}_{x^n,y^n}\}}$ (defined in \eqref{Omega2states_arbit}) is to be $\{\rho^{\cB^n}_{x^n,y^n,s^n}\}$ for an element $s^n \in \cS^n$.
Using the discussion of the proof of Lemma \ref{theorem_sen_MAC_generalised_indep}, 
we have the following identities:
\begin{align*}
    \sigma^{X^nY^n\cB^n}_{P_{X},P_{Y}, \{\sigma^{\cB^n}_{x^n,y^n}\}} &\triangleq \rho^{X^nY^n\cB^n}_{P_{X},P_{Y},s^n},\\
    \sigma^{X'\cB'}_{P_{X}, \{\sigma^{\cB^n}_{x^n}\}} &\triangleq \rho^{X'\cB'}_{P_{X}, s^n},\\
    \sigma^{Y'\cB'}_{P_{Y}, \{\sigma^{\cB^n}_{y^n}\}} &\triangleq  \rho^{Y'\cB'}_{P_{Y}, s^n},\\
\sigma^{\cB'} &\triangleq \rho^{\cB'}_{s^n},
\end{align*}
where $\sigma^{X'\cB'}_{P_{X}, \{\sigma^{\cB^n}_{x^n}\}},\sigma^{Y'\cB'}_{P_{Y}, \{\sigma^{\cB^n}_{y^n}\}}$ and $\sigma^{\cB'}$ are defined in \cref{traced_tilted_sigma_X,traced_tilted_sigma_Y,traced_tilted_sigma_None}. Now for any arbitrary $\delta \in (0,1)$, 
\cref{reject_Omega1_optimal,accept_Omega2X_optimal,accept_Omega2Y_optimal,accept_Omega2XY_optimal} imply that
\begin{align*}
    \tr\left[\left( \bbI^{X'Y'\cB'} - {\bbT}^{\star}_{n}\right)\rho^{X'Y'\cB'}_{P_{X},P_{Y},s^n}\right] &\leq 76\delta,\\
    \tr\left[{\bbT}^{\star}_{n}\left(\rho^{Y'}_{P_{Y}} \otimes \rho^{X'\cB'}_{P_{X},s^n} \right)\right] &\leq g_1(n,\abs{\cX},\abs{\cB})2^{-n\widehat{R}_2+1},\\
    \tr\left[{\bbT}^{\star}_{n}\left(\rho^{X'}_{P_{X}} \otimes \rho^{Y'\cB'}_{P_{Y},s^n} \right)\right] &\leq g_2(n,\abs{\cY},\abs{\cB})2^{-n\widehat{R}_1+1},\\
\tr\left[{\bbT}^{\star}_{n}\left(\rho^{X'}_{P_{X}} \otimes\rho^{Y'}_{P_{Y}} \otimes \rho^{\cB'}_{s^n} \right)\right] &\leq g_3(n,\abs{\cB})2^{-n(\widehat{R}_1+\widehat{R}_2) +1}.
\end{align*}
The above calculations hold for every $s^n \in \cS^n$. This completes the proof of Corollary \ref{theo_gen_indep_mac_corollary}.
\endproof

\section{Alternative Proof for Achievability of Theorem \ref{lemma_rand_capacity_avmac} (Random Code Capacity Region $\cR_{r}$ of a CQ-AVMAC)}\label{alternative_proof_lemma_rand_capacity_avmac}

(\textit{This proof of achievability is established by designing a joint decoding strategy based on the simultaneous test designed in  Lemma \ref{theorem_sen_MAC_generalised_indep} in subsection 
 \ref{cq_mac_avht_null_statement_sim}.})
    We here show that for every rate-pair $(R_1,R_2) \in \cR^{\star}$, there exists an $(n,2^{nR_1},2^{nR_2},\eps_n)$ random code $\cC$ for which the average error probability, 
    averaged over the choice of code $\cC$ satisfies the following: 
    $$\max_{s^n \in \cS^n} \bbE_{\cC}\left[\bar{e}(\cC,s^n)\right] < \eps_n,$$
    where $\lim_{n \to \infty} \eps_n = 0$. We now start with the construction of $\cC$.
    \subsection{Randomized Encoder Construction}
    Alice and Bob randomly generate $2^{nR_1}$ and $2^{nR_2}$ sequences $\{X^n(m_1) \in \cX^n : m_1 \in [2^{nR_1}]\}$ 
and $\{Y^n(m_2) \in \cY^n : m_2 \in [2^{nR_2}]\}$ according to the following distributions;
\begin{align}
\Pr\left\{X^n(m_1)\right\} &:= \prod_{i}P_{X}(X_i(m_1)) \hbox{ with } X^n(m_1):=\left(X_1(m_1),\cdots,X_n(m_1)\right),
\\
\Pr\left\{Y^n(m_2)\right\} &:= \prod_{i}P_{Y}(Y_i(m_2))\hbox{ with }
Y^n(m_2):=\bigl(Y_1(m_2),\cdots,Y_n(m_2)\bigr).
\end{align}
        For the code $\cC$, we define
\begin{equation*}
    \Pr\{\cC\} := \prod_{m_1=1}^{\cM^{1}_{n}}\Pr\left\{X^n(m_1)\right\}\cdot\prod_{m_2=1}^{\cM^{2}_{n}}\Pr\left\{Y^n(m_2)\right\}.
\end{equation*}
    To send a message $m_1 \in \cM^{1}_{n}$, Alice encodes it to input sequence $X^n(m_1)$ and transmits it over the channel. Similarly, to end a message $m_2 \in \cM^{2}_{n}$, Bob encodes it to input sequence $Y^n(m_2)$ and transmits it over the channel.
    \subsection{Decoding Strategy}
    To construct our decoder, we 
   employ Lemma \ref{theorem_sen_MAC_generalised_indep}, several relations appearing 
   in its proof, and Corollary \ref{theo_gen_indep_mac_corollary}, which depend on the choice of  
$\widehat{R}_1,\widehat{R}_2 > 0$ and $t \in (0,1)$.
For this aim, we set the pair $( \widehat{R}_1,\widehat{R}_2)$ as
\begin{align}
R_1<    \widehat{R}_1 &< \min _{Q \in \cP(\cS)} I[X;\cB|Y]_{P_{X},P_{Y},Q},\label{BHI1}\\
R_2<    \widehat{R}_2 &< \min _{Q \in \cP(\cS)} I[Y;\cB|X]_{P_{X},P_{Y},Q},\label{BHI2}\\
 R_1+R_2<   \widehat{R}_1 + \widehat{R}_2 &< \min _{Q \in \cP(\cS)} I[XY;\cB]_{P_{X},P_{Y},Q}
 \label{BHI3}.
\end{align}
Then, we choose a sufficiently small real number $t \in (0,1)$ such that
    \begin{align}
   \widehat{R}_1 &< \min _{Q \in \cP(\cS)} I_{1-t}[X;\cB|Y]_{P_{X},P_{Y},Q},\label{theo_gen_sen_macR_1}\\
    \widehat{R}_2 &< \min _{Q \in \cP(\cS)} I_{1-t}[Y;\cB|X]_{P_{X},P_{Y},Q},\label{theo_gen_sen_macR_2}\\
   \widehat{R}_1 + \widehat{R}_2 &< \min _{Q \in \cP(\cS)} I_{1-t}[XY;\cB]_{P_{X},P_{Y},Q}\label{theo_gen_sen_macR_1R_2}.
\end{align}

    Upon receiving the state $\rho_{\substack{X^n,Y^n, s^n}}^{\cB^n}$, Charlie first embeds the state $\rho^{\cB^n}_{X^n,Y^n,s^n}$ into $\cB'$ using the way mentioned in \eqref{proof_theo_gen_identity_1}.
    We will represent this embedded state as $\rho^{\cB'}_{X^n,Y^n,s^n}$. Charlie then generates two collections of i.i.d random variables $\Tilde{R} := \{R(i)\}_{i \in \cM^{1}_n}$ and $\Tilde{T} := \{T(j)\}_{j \in \cM^{2}_n}$ over the support set $[\abs{\cA}]$, where 
    $R(m_1)$ and $T(m_2)$ follow a uniform distribution over $[\abs{\cA}]$
    for each $(m_1,m_2) \in \cM^{1}_{n} \times \cM^{2}_{n}$.  
Using the notations $(X^n,a_{R})(\hat{m}_1) := (X^n(\hat{m}_1),a_{R(\hat{m}_1)})$ and $(Y^n,b_{T})(\hat{m}_2) := (Y^n(\hat{m}_2),b_{T(\hat{m}_2)})$ 
with $(\hat{m}_1,\hat{m}_2) \in \cM^{1}_n \times \cM^{2}_n$, we choose 
the matrix $\bbT^{\star}_{n,\substack{(X^n,a_{R})(m_1),(Y^n,b_{T})(m_2),\delta}}$ as defined in \eqref{intersection_inner_povms}. 
Then, we define the POVM operator $\Lambda_{m_1,m_2}^{\Tilde{R},\Tilde{T}}$ as
{\footnotesize\begin{equation*}
    \Lambda_{m_1,m_2}^{\Tilde{R},\Tilde{T}}:= \left(\sum_{\substack{\hat{m}_1 \in \cM^{1}_n \\ \hat{m}_2 \in \cM^{2}_n}}\bbT^{\star}_{n,\substack{(X^n,a_{R})(\hat{m}_1),(Y^n,b_{T})(\hat{m}_2),\delta}}\right)^{-\frac{1}{2}}\bbT^{\star}_{n,\substack{(X^n,a_{R})(m_1),(Y^n,b_{T})(m_2),\delta}}\left(\sum_{\substack{\hat{m}_1 \in \cM^{1}_n \\ \hat{m}_2 \in \cM^{2}_n}}\bbT^{\star}_{n,\substack{(X^n,a_{R})(\hat{m}_1),(Y^n,b_{T})(\hat{m}_2),\delta}}\right)^{-\frac{1}{2}}.
\end{equation*}}
Charlie then uses the POVM $\left\{\Lambda_{m_1,m_2}^{\Tilde{R},\Tilde{T}}\right\}$ to decode the message.
    
    \subsection{Error Analysis}
    We now do the error analysis as follows. For every $s^n \in \cS^n$, the following holds.
    \begin{align}
        &\hspace{10pt}\bbE_{\Tilde{R},\Tilde{T}}\bbE_{\cC}\left[\bar{e}(\cC,s^n)\right] = \frac{1}{2^{n(R_1+R_2)}}\sum_{\substack{m_1\in \cM^{1}_n \\ m_2\in \cM^{2}_n}}\bbE_{\Tilde{R},\Tilde{T}}\bbE_{\cC}\left[e(m_1,m_2,\cC,s^n)\right]\nn\\
        &= \frac{1}{2^{n(R_1+R_2)}}\sum_{\substack{m_1\in \cM^{1}_n \\ m_2\in \cM^{2}_n}}\bbE_{\Tilde{R},\Tilde{T}}\bbE_{\cC}\left[e(m_1,m_2,\cC,s^n)\right]\nn\\
        &= \frac{1}{2^{n(R_1+R_2)}}\sum_{\substack{m_1\in \cM^{1}_n \\ m_2\in \cM^{2}_n}}\bbE_{\Tilde{R},\Tilde{T}}\bbE_{\cC}\left[ \tr\left[\left(\bbI^{\cB'}-\Lambda_{m_1,m_2}^{\Tilde{R},\Tilde{T}}\right)\rho_{\substack{X^n(m_1),Y^n(m_2), s^n}}^{\cB'}\right]\right]\nn\\
        &\overset{a}{\leq }\frac{1}{2^{n(R_1+R_2)}}\sum_{\substack{m_1\in \cM^{1}_n \\ m_2\in \cM^{2}_n}}\bbE_{\Tilde{R},\Tilde{T}}\bbE_{\cC}\left[ \tr\left[\left(\bbI^{\cB'}-\Lambda_{m_1,m_2}^{\Tilde{R},\Tilde{T}}\right)\rho_{\substack{(X^n,a_{R})(m_1),(Y^n,b_{T})(m_2), s^n,\delta}}^{\cB'}\right]\right] + 4\delta\nn\\
        &\overset{b}{\leq}\frac{1}{2^{n(R_1+R_2)}}\sum_{\substack{m_1\in \cM^{1}_n \\ m_2\in \cM^{2}_n}}\Bigg( 2\bbE_{\Tilde{R},\Tilde{T}}\bbE_{\cC}\left[\tr\left[\left(\bbI^{\cB'} - \bbT^{\star}_{n,\substack{(X^n,a_{R})(m_1),(Y^n,b_{T})(m_2),\delta}}\right)\rho_{\substack{(X^n,a_{R})(m_1),(Y^n,b_{T})(m_2), s^n,\delta}}^{\cB'}\right]\right]\nn\\
        &\hspace{10pt}+ 4\cdot \sum_{\substack{\hat{m}_1 \in \cM^{1}_n:\\\hat{m}_1\neq m_1}}\bbE_{\Tilde{R},\Tilde{T}}\bbE_{\cC}\left[\tr\left[\bbT^{\star}_{n,\substack{(X^n,a_{R})(\hat{m}_1),(Y^n,b_{T})(m_2),\delta}}.\rho_{\substack{(X^n,a_{R})(m_1),(Y^n,b_{T})(m_2), s^n,\delta}}^{\cB'}\right]\right]\nn\\
        &\hspace{10pt}+ 4\cdot \sum_{\substack{\hat{m}_2 \in \cM^{2}_n:\\\hat{m}_2\neq m_2}}\bbE_{\Tilde{R},\Tilde{T}}\bbE_{\cC}\left[\tr\left[\bbT^{\star}_{n,\substack{(X^n,a_{R})(m_1),(Y^n,b_{T})(\hat{m}_2),\delta}}\cdot\rho_{\substack{(X^n,a_{R})(m_1),(Y^n,b_{T})(m_2), s^n,\delta}}^{\cB'}\right]\right]\nn\\
        &\hspace{10pt}+ 4\cdot \sum_{\substack{(\hat{m}_1,\hat{m}_2) \in \cM^{1}_n \times \cM^{2}_n:\\\hat{m}_1\neq m_1, \hat{m}_2\neq m_2}}\bbE_{\Tilde{R},\Tilde{T}}\bbE_{\cC}\left[\tr\left[\bbT^{\star}_{n,\substack{(X^n,a_{R})(\hat{m}_1),(Y^n,b_{T})(\hat{m}_2),\delta}}\cdot\rho_{\substack{(X^n,a_{R})(m_1),(Y^n,b_{T})(m_2), s^n,\delta}}^{\cB'}\right]\right]\Bigg) + 4\delta\nn\\
        &=\frac{1}{2^{n(R_1+R_2)}}\sum_{\substack{m_1\in \cM^{1}_n \\ m_2\in \cM^{2}_n}}\Bigg( 2\bbE_{\substack{(X^n,a_{R})(m_1)\\(Y^n,b_{T})(m_2)}}\left[\tr\left[\left(\bbI^{\cB'} - \bbT^{\star}_{n,\substack{(X^n,a_{R})(m_1),(Y^n,b_{T})(m_2),\delta}}\right)\rho_{\substack{(X^n,a_{R})(m_1),(Y^n,b_{T})(m_2), s^n,\delta}}^{\cB'}\right]\right]\nn\\
        &\hspace{10pt}+ 4\cdot \sum_{\substack{\hat{m}_1 \in \cM^{1}_n:\\\hat{m}_1\neq m_1}}\bbE_{\substack{(X^n,a_{R})(\hat{m}_1)\\(Y^n,b_{T})(m_2)}}\left[\tr\left[\bbT^{\star}_{n,\substack{(X^n,a_{R})(\hat{m}_1),(Y^n,b_{T})(m_2),\delta}}\cdot\rho_{\substack{(Y^n,b_{T})(m_2), s^n,\delta}}^{\cB'}\right]\right]\nn\\
        &\hspace{10pt}+ 4\cdot \sum_{\substack{\hat{m}_2 \in \cM^{2}_n:\\\hat{m}_2\neq m_2}}\bbE_{\substack{(X^n,a_{R})(m_1)\\(Y^n,b_{T})(\hat{m}_2)}}\left[\tr\left[\bbT^{\star}_{n,\substack{(X^n,a_{R})(m_1),(Y^n,b_{T})(\hat{m}_2),\delta}}\cdot\rho_{\substack{(X^n,a_{R})(m_1), s^n,\delta}}^{\cB'}\right]\right]\nn\\
        &\hspace{10pt}+ 4\cdot \sum_{\substack{(\hat{m}_1,\hat{m}_2) \in \cM^{1}_n \times \cM^{2}_n:\\\hat{m}_1\neq m_1, \hat{m}_2\neq m_2}}\bbE_{\substack{(X^n,a_{R})(\hat{m}_1)\\(Y^n,b_{T})(\hat{m}_2)}}\left[\tr\left[\bbT^{\star}_{n,\substack{(X^n,a_{R})(\hat{m}_1),(Y^n,b_{T})(\hat{m}_2),\delta}}\cdot\rho_{\substack{s^n,\delta}}^{\cB'}\right]\right]\Bigg) + 4\delta \nn\\
        &= 4\delta + \frac{1}{2^{n(R_1+R_2)}}\sum_{\substack{m_1\in \cM^{1}_n \\ m_2\in \cM^{2}_n}}\Bigl(2\tr\left[\left(\bbI^{X'Y'\cB'} - {\bbT}^{\star}_{n}\right)\rho^{X'Y'\cB'}_{P_{X},P_{Y},s^n}\right] 
        + 4\cdot 2^{nR_1} \tr\left[{\bbT}^{\star}_{n}\left(\rho^{X'}_{P_{X}} \otimes \rho^{Y'\cB'}_{P_{Y},s^n} \right)\right]\nn\\
        &\hspace{95pt}+ 4\cdot 2^{nR_2} \tr\left[{\bbT}^{\star}_{n}\left(\rho^{Y'}_{P_{Y}} \otimes \rho^{X'\cB'}_{P_{X},s^n} \right)\right] + 4\cdot 2^{n(R_1 + R_2)} \tr\left[{\bbT}^{\star}_{n}\left(\rho^{X'}_{P_{X}} \otimes \rho^{Y'\cB'}_{P_{Y},s^n} \right)\right]\Bigr)\nn\\
        &\overset{c}{\leq} 156\delta  + 4 \cdot 2^{-n \left\{\widehat{R}_1-R_1 -\frac{1}{n}\right\}}g_2(n,\abs{\cY},\abs{\cB}) + 4\cdot 2^{-n \left\{\widehat{R}_2 - R_2 - \frac{1}{n}\right\}}g_1(n,\abs{\cX},\abs{\cB})\nn\\
        &\hspace{10pt} + 4\cdot 2^{-n \left\{(\widehat{R}_1 + \widehat{R}_2) - (R_1+R_2) - \frac{1}{n}\right\}}g_3(n,\abs{\cB}) \nn\\
        \label{error_bound_alter}
        & \overset{d} \leq 157 \delta,
    \end{align}
where 
$a$ follows from 
Fact \ref{trace_norm2}, and \eqref{perturbed_tilted_state_eq1}, $b$ follows from Fact \ref{Hayashi_nagaoka}, and $c$ follows from Corollary \ref{theo_gen_indep_mac_corollary} and $d$ follows because we have chosen $n$ large enough such that $4 \cdot 2^{-n \left\{\widehat{R}_1-R_1 -\frac{1}{n}\right\}}g_2(n,\abs{\cY},\abs{\cB}) + 4\cdot 2^{-n \left\{\widehat{R}_2 - R_2 - \frac{1}{n}\right\}}g_1(n,\abs{\cX},\abs{\cB}) + 4\cdot 2^{-n \left\{(\widehat{R}_1 + \widehat{R}_2) - (R_1+R_2) - \frac{1}{n}\right\}}g_3(n,\abs{\cB}) \leq \delta.$

Since $\delta \in (0,1)$ is arbitrary, therefore,
\begin{align}
\lim_{n\to \infty}\max_{s^n \in \cS^n}\bbE_{\Tilde{R},\Tilde{T}}\bbE_{\cC}\left[\bar{e}(\cC,s^n)\right]
=0.
\end{align}
This completes the proof for achievability of Theorem \ref{lemma_rand_capacity_avmac} using simultaneous test mentioned in Lemma \ref{theorem_sen_MAC_generalised_indep}.
\endproof

\section{Proof of \eqref{error_term_X_hayashi}}\label{proof_error_term_X_hayashi}
The relation \eqref{error_term_X_hayashi} can be shown as follows;
    \begin{align}
        &\bbE_{\cC}\left[\bar{e}^{X}(\cC,s^n)\right]\nn\\
        &= \frac{1}{2^{n(R_1+R_2)}}\sum_{\substack{m_1\in \cM^{1}_n \\ m_2\in \cM^{2}_n}}\bbE_{\cC}\left[ \tr\left[\left(\bbI^{\cB}- \left(\Lambda_{m_1}^{X}\right)\right)\rho_{\substack{X^n(m_1),Y^n(m_2), s^n}}^{\cB}\right]\right]\nn\\
        &\overset{a}{\leq} \frac{1}{2^{n(R_1+R_2)}}\sum_{\substack{m_1\in \cM^{1}_n \\ m_2\in \cM^{2}_n}}\left(2\bbE_{\cC}\left[\tr\left[\left(\bbI^{\cB}- \Pi^{X}_{m_1}\right)\rho_{\substack{X^n(m_1),Y^n(m_2), s^n}}^{\cB}\right]\right] + 4\bbE_{\cC}\left[\tr\left[\left(\sum_{\substack{\hat{m}_1 \in \cM^{1}_n : \\ \hat{m}_1 \neq m_1}} \Pi^{X}_{\hat{m}_1}\right)\rho_{\substack{X^n(m_1),Y^n(m_2), s^n}}^{\cB}\right]\right]\right)\nn\\
        &= \frac{1}{2^{n(R_1+R_2)}}\sum_{\substack{m_1\in \cM^{1}_n \\ m_2\in \cM^{2}_n}}\left(2\bbE_{\cC}\left[\tr\left[\left(\bbI^{\cB}- \sum_{\hat{m}_2 \in  \cM^{2}_n} \hat{\bbT}_{n,X^n(m_1),Y^n(\hat{m}_2)}^{X}\right)\rho_{\substack{X^n(m_1),Y^n(m_2), s^n}}^{\cB}\right]\right]\right.\nn\\
        &\hspace{70pt} \left.+ 4\bbE_{\cC}\left[\tr\left[\left(\sum_{\substack{\hat{m}_1 \in \cM^{1}_n : \\ \hat{m}_1 \neq m_1}} \sum_{\hat{m}_2 \in  \cM^{2}_n} \hat{\bbT}_{n,X^n(\hat{m}_1),Y^n(\hat{m}_2)}^{X}\right)\rho_{\substack{X^n(m_1),Y^n(m_2), s^n}}^{\cB}\right]\right]\right)\nn\\
        &\overset{b}{\leq} \frac{1}{2^{n(R_1+R_2)}}\sum_{\substack{m_1\in \cM^{1}_n \\ m_2\in \cM^{2}_n}}\left(2\bbE_{\cC}\left[\tr\left[\left(\bbI^{\cB}-  \hat{\bbT}_{n,X^n(m_1),Y^n(m_2)}^{X}\right)\rho_{\substack{X^n(m_1),Y^n(m_2), s^n}}^{\cB}\right]\right] \right.\nn\\
        & \hspace{70pt}\left.+ 4\sum_{\substack{\hat{m}_1 \in \cM^{1}_n : \\ \hat{m}_1 \neq m_1}} \bbE_{\cC}\left[\tr\left[\hat{\bbT}_{n,X^n(\hat{m}_1),Y^n(m_2)}^{x}\rho_{\substack{X^n(m_1),Y^n(m_2), s^n}}^{\cB}\right]\right]\right.\nn\\
        &\hspace{70pt} \left.+ 4\sum_{\substack{(\hat{m}_1,\hat{m}_2) \in \cM^{1}_n \times \cM^{2}_n \\ :\hat{m}_1 \neq m_1, \hat{m}_2 \neq m_2 }}\bbE_{\cC}\left[\tr\left[ \hat{\bbT}_{n,X^n(\hat{m}_1),Y^n(\hat{m}_2)}^{X}\rho_{\substack{X^n(m_1),Y^n(m_2), s^n}}^{\cB}\right]\right]\right)\nn\\
        &= \frac{1}{2^{n(R_1+R_2)}}\sum_{\substack{m_1\in \cM^{1}_n \\ m_2\in \cM^{2}_n}}\left(2\bbE_{X^n(m_1),Y^n(m_2)}\left[\tr\left[\left(\bbI^{\cB}-  \hat{\bbT}_{n,X^n(m_1),Y^n(m_2)}^{X}\right)\rho_{\substack{X^n(m_1),Y^n(m_2), s^n}}^{\cB}\right]\right] \right.\nn\\
        & \hspace{70pt}\left.+ 4\sum_{\substack{\hat{m}_1 \in \cM^{1}_n : \\ \hat{m}_1 \neq m_1}} \bbE_{X^n(\hat{m}_1),Y^n(m_2)}\left[\tr\left[\hat{\bbT}_{n,X^n(\hat{m}_1),Y^n(m_2)}^{X}\rho_{\substack{Y^n(m_2), s^n}}^{\cB}\right]\right]\right.\nn\\
        &\hspace{70pt} \left.+ 4\sum_{\substack{(\hat{m}_1,\hat{m}_2) \in \cM^{1}_n \times \cM^{2}_n \\ : \hat{m}_1 \neq m_1, \hat{m}_2 \neq m_2}}\bbE_{X^n(\hat{m}_1),Y^n(\hat{m}_2)}\left[\tr\left[ \hat{\bbT}_{n,X^n(\hat{m}_1),Y^n(\hat{m}_2)}^{X}\rho_{\substack{ s^n}}^{\cB}\right]\right]\right)\nn\\
        &{\leq}2\tr\left[\left( \bbI^{X^nY^n\cB^n} - \hat{\bbT}_{n}^{X}\right)\rho^{X^nY^n\cB^n}_{P_{X},P_{Y},s^n}\right] + 4\cdot2^{nR_1}\tr\left[\hat{\bbT}_{n}^{X}\left(\sum_{\substack{ x^n\in \cX^n\\y^n \in \cY^n}}P^{n}_{X}(x^n)\ketbrasys{x^n}{X^n}\otimes P^{n}_{Y}(y^n)\ketbrasys{y^n}{Y^n} \otimes \rho^{\cB^n}_{y^n,s^n}\right)\right]\nn\\
        &\hspace{10pt}+ 4\cdot2^{n(R_1 + R_2)}\tr\left[\hat{\bbT}_{n}^{X}\left(\sum_{\substack{ x^n\in \cX^n\\y^n \in \cY^n}}P^{n}_{X}(x^n)\ketbrasys{x^n}{X^n}\otimes P^{n}_{Y}(y^n)\ketbrasys{y^n}{Y^n} \otimes \rho^{\cB^n}_{s^n}\right)\right]\nn\\
        &\overset{c}{\leq} 2f(n,\abs{\cX},\abs{\cY}, \abs{\cB}, \abs{\cS_{\eps}},t)(1 +\eps)^{n}\left(2^{nt\left(\widehat{R}_1 - \min _{Q \in \cP(\cS_\eps)} I_{1-t}[X;\cB|Y]_{P_{X},P_{Y},Q}\right)} + 2^{nt\left(\widehat{R}_1 + \widehat{R}_2 -  \min _{Q \in \cP(\cS_\eps)} I_{1-t}[XY;\cB]_{P_{X},P_{Y},Q}\right)}\right)\nn \\
        &\hspace{10pt}+ 4\cdot g_2(n,\abs{\cY},\abs{\cB})2^{n(R_1 - \widehat{R}_1)} + 4\cdot g_3(n,\abs{\cB})2^{-n(\widehat{R}_1 + \widehat{R}_2 - R_1 - R_2)},\nn
    \end{align}
    where $a$ follows from Fact \ref{Hayashi_nagaoka}, $b$ follows from the fact that $\sum_{\hat{m}_2\in  \cM^{2}_n} \hat{\bbT}_{n,X^n(m_1),Y^n(\hat{m}_2)}^{X} \geq \hat{\bbT}_{n,X^n(m_1),Y^n(m_2)}^{X}$ and $c$ follows from \cref{corllary_achievability_mac_independent_eq1_TS,corllary_achievability_mac_independent_eq3_TS,corllary_achievability_mac_independent_eq4_TS}.
\section{Proof of Lemma \ref{lemma_commutativity}}\label{proof_lemma_commutativity}
Consider the following series of equalities:
\begin{align}
    &\rho^{\cB^n}_{U,n}  \widehat{\rho}^{\cB^n}_{x^n}
    = \rho^{\cB^n}_{U,n} V^{\cB^n}(\pi) \left(\widehat{\rho}^{\cB^n}_{\bar{x}^n}\right){V^{\cB^n}}^{\dagger}(\pi)\nn\\
    &\overset{a}{=} V^{\cB^n}(\pi) \rho^{\cB^n}_{U,n}{V^{\cB^n}}^{\dagger}(\pi) V^{\cB^n}(\pi) \left(\rho^{\cB^{m_1}}_{U,m_1} \otimes \rho^{\cB^{m_2}}_{U,m_2} \otimes \cdots \otimes \rho^{\cB^{m_{\abs{\cX}}}}
    _{U,m_{\abs{\cX}}}\right){V^{\cB^n}}^{\dagger}(\pi)\nn\\
    &\overset{b}{=} V^{\cB^n}(\pi) \rho^{\cB^n}_{U,n}\left(\rho^{\cB^{m_1}}_{U,m_1} \otimes \rho^{\cB^{m_2}}_{U,m_2} \otimes \cdots \otimes \rho^{\cB^{m_{\abs{\cX}}}}
    _{U,m_{\abs{\cX}}}\right){V^{\cB^n}}^{\dagger}(\pi)\nn\\
    &\overset{c}{=} V^{\cB^n}(\pi) \left(\sum_{\lambda \in \Lambda^{n}_{\abs{\cB}}}\frac{1}{\abs{\Lambda^{n}_{\abs{\cB}}}}\frac{1}{\abs{\cU_{\lambda} \otimes \cV_{\lambda}}} \Pi^{\cB^n}_{\lambda}\right)\left(\sum_{\substack{(\lambda_1,\cdots,\lambda_{\abs{\cX}}) \\ \in (\times_{i \in [\abs{\cX}]}\Lambda^{m_i}_{\abs{\cB}})}}\frac{1}{\prod_{i \in [\abs{\cX}]}\abs{\Lambda^{m_i}_{\abs{\cB}}}}\frac{1}{\prod_{i\in [\abs{\cX}]}\abs{\cU_{\lambda_i} \otimes \cV_{\lambda_i}}} (\otimes_{i \in [\abs{\cA}]}\Pi^{\cB^{m_i}}_{\lambda_i})\right){V^{\cB^n}}^{\dagger}(\pi)\nn\\
    &= V^{\cB^n}(\pi) \left(\sum_{\lambda \in \Lambda^{n}_{\abs{\cB}}}\sum_{\substack{(\lambda_1,\cdots,\lambda_{\abs{\cX}}) \\ \in (\times_{i \in [\abs{\cX}]}\Lambda^{m_i}_{\abs{\cB}})}}\frac{1}{\abs{\Lambda^{n}_{\abs{\cB}}}}\frac{1}{\abs{\cU_{\lambda} \otimes \cV_{\lambda}}} \frac{1}{\prod_{i \in [\abs{\cX}]}\abs{\Lambda^{m_i}_{\abs{\cB}}}}\frac{1}{\prod_{i\in [\abs{\cX}]}\abs{\cU_{\lambda_i} \otimes \cV_{\lambda_i}}}\Pi^{\cB^n}_{\lambda} (\otimes_{i \in [\abs{\cA}]}\Pi^{\cB^{m_i}}_{\lambda_i})\right){V^{\cB^n}}^{\dagger}(\pi)\nn\\
    &\overset{d}{=} V^{\cB^n}(\pi) \left(\sum_{\lambda \in \Lambda^{n}_{\abs{\cB}}}\sum_{\substack{(\lambda_1,\cdots,\lambda_{\abs{\cX}}) \\ \in (\times_{i \in [\abs{\cX}]}\Lambda^{m_i}_{\abs{\cB}})}}\frac{1}{\abs{\Lambda^{n}_{\abs{\cB}}}}\frac{1}{\abs{\cU_{\lambda} \otimes \cV_{\lambda}}} \frac{1}{\prod_{i \in [\abs{\cX}]}\abs{\Lambda^{m_i}_{\abs{\cB}}}}\frac{1}{\prod_{i\in [\abs{\cX}]}\abs{\cU_{\lambda_i} \otimes \cV_{\lambda_i}}}(\otimes_{i \in [\abs{\cA}]}\Pi^{\cB^{m_i}}_{\lambda_i})\Pi^{\cB^n}_{\lambda} \right){V^{\cB^n}}^{\dagger}(\pi)\nn\\
    &= V^{\cB^n}(\pi) \left(\rho^{\cB^{m_1}}_{U,m_1} \otimes \rho^{\cB^{m_2}}_{U,m_2} \otimes \cdots \otimes \rho^{\cB^{m_{\abs{\cX}}}}
    _{U,m_{\abs{\cX}}}\right)\rho^{\cB^n}_{U,n}{V^{\cB^n}}^{\dagger}(\pi)\nn\\
    &= V^{\cB^n}(\pi) \left(\rho^{\cB^{m_1}}_{U,m_1} \otimes \rho^{\cB^{m_2}}_{U,m_2} \otimes \cdots \otimes \rho^{\cB^{m_{\abs{\cX}}}}
    _{U,m_{\abs{\cX}}}\right){V^{\cB^n}}^{\dagger}(\pi) V^{\cB^n}(\pi) \rho^{\cB^n}_{U,n}{V^{\cB^n}}^{\dagger}(\pi)
    =  \widehat{\rho}^{\cB^n}_{x^n} \rho^{\cB^n}_{U,n},\nn 
\end{align}
where $a$ follows from the fact that $\rho^{\cB^n}_{U,n}$ is permutation invariant, $b$ follows from. \eqref{uniform_n_state}, $c$ follows from \eqref{perm_op_prop_1} of Definition \ref{fact_perm_op}.

Equality $d$ follows from the following discussions. 
We have two cases for the relation between
the elements $(\lambda_1,\cdots,\lambda_{\abs{\cX}}) \in (\times_{i \in [\abs{\cX}]}\Lambda^{m_i}_{\abs{\cB}})$ 
and the young tableux $\lambda \in \Lambda^{n}_{\abs{\cB}}$.
In the first case, 
the row-wise concatenation of $(\lambda_1,\cdots,\lambda_{\abs{\cX}})$
forms the young tableux $\lambda$.
In the second case, 
the row-wise concatenation of $(\lambda_1,\cdots,\lambda_{\abs{\cX}})$
forms another young tableux $\lambda'(\neq \lambda)$.
In the first case, 
the support of $ (\otimes_{i \in [\abs{\cA}]}\Pi^{\cB^{m_i}}_{\lambda_i}) $ is a subspace of support of $ \Pi^{\cB^n}_{\lambda}$ i.e. $ (\otimes_{i \in [\abs{\cA}]}\Pi^{\cB^{m_i}}_{\lambda_i}) \preceq \Pi^{\cB^n}_{\lambda}$.   
Hence, we have $\Pi^{\cB^n}_{\lambda}(\otimes_{i \in [\abs{\cA}]}\Pi^{\cB^n}_{\lambda_i}) =(\otimes_{i \in [\abs{\cA}]}\Pi^{\cB^n}_{\lambda_i})\Pi^{\cB^n}_{\lambda}=(\otimes_{i \in [\abs{\cA}]}\Pi^{\cB^{m_i}}_{\lambda_i})$.
In the second case, 
the support of $ (\otimes_{i \in [\abs{\cA}]}\Pi^{\cB^{m_i}}_{\lambda_i}) $ is orthogonal to the support of $ \Pi^{\cB^n}_{\lambda}$, 
because the support of $\Pi^{\cB^n}_{\lambda}$ is orthogonal to $\Pi^{\cB^n}_{\lambda'}$. 
Hence, we have $\Pi^{\cB^n}_{\lambda}(\otimes_{i \in [\abs{\cA}]}\Pi^{\cB^n}_{\lambda_i}) = 0$.
For both cases, the equality holds for the inside of the summations before and after $d$, which implies $d$.
\endproof
 
\section{Proof of \eqref{small_N_X}}\label{proof_small_N_X}
Consider the following series of inequalities:
\begin{align}
&\hspace{10pt}\norm{\cN_{X,a_{r},\delta}(\ketbra{h})}{\infty}\nn\\
    &\leq \frac{1}{\abs{\cA}(1 + 2\delta^2)}\left( \norm{\sum_{t \in [\abs{\cA}]}\delta (\ket{h} \oplus \delta\ket{h}\ket{a_{r}}) \bra{h}\bra{b_{t}}}{\infty} + \norm{\sum_{t \in [\abs{\cA}]}\delta\ket{h}\ket{b_{t}}(\bra{h} \oplus \delta\bra{h}\bra{a_{r}}) }{\infty} + \norm{\sum_{t \in [\abs{\cA}]}\delta^2\ketbra{h}\ketbra{b_{t}} }{\infty}\right)\nn\\
    &\overset{a}{\leq} \frac{1}{\abs{\cA}(1 + 2\delta^2)}\left( \norm{\sum_{t \in [\abs{\cA}]}\delta (\ket{h} \oplus \delta\ket{h}\ket{a_{r}}) \bra{h}\bra{b_{t}}}{1} + \norm{\sum_{t \in [\abs{\cA}]}\delta\ket{h}\ket{b_{t}}(\bra{h} \oplus \delta\bra{h}\bra{a_{r}}) }{1} + \norm{\sum_{t \in [\abs{\cA}]}\delta^2\ketbra{h}\ketbra{b_{t}} }{\infty}\right)\nn\\
    &= \frac{1}{\abs{\cA}(1 + 2\delta^2)}\left( 2\norm{\sum_{t \in [\abs{\cA}]}\delta (\ket{h} \oplus \delta\ket{h}\ket{a_{r}}) \bra{h}\bra{b_{t}}}{1} + \norm{\sum_{t \in [\abs{\cA}]}\delta^2\ketbra{h}\ketbra{b_{t}} }{\infty}\right)\nn\\
    &= \frac{1}{\abs{\cA}(1 + 2\delta^2)} \left(2\delta \tr\left[\sqrt{(\ket{h} \oplus \delta\ket{h}\ket{a_{r}})(\bra{h} \oplus \delta\bra{h}\bra{a_{r}})}\right] \sqrt{\sum_{t \in [\abs{\cA}]}\norm{\ket{b_{t}}}{2}^{2}} + \delta^2\norm{\sum_{t \in [\abs{\cA}]}\ketbra{b_{t}}}{\infty}\right)\nn\\
    &= \frac{1}{\abs{\cA}(1 + 2\delta^2)}\left(2 \delta \norm{\ket{h} \oplus \delta\ket{h}\ket{a_{r}}}{2}^{2}\sqrt{\sum_{t \in [\abs{\cA}]}\norm{\ket{b_{t}}}{2}^{2}} + \delta^2\right) \nn\\
    &\leq \frac{1}{\abs{\cA}(1 + 2\delta^2)}\left(3\delta \norm{\ket{h} \oplus \delta\ket{h}\ket{a_{r}}}{2}^{2}\sqrt{\sum_{t \in [\abs{\cA}]}\norm{{\ket{b_{t}}}}{2}^{2}}\right)\nn\\
    &= \frac{3\delta(1 + \delta^2)\sqrt{\abs{\cA}}}{(1 + 2\delta^2\abs{\cA}}
    \leq \frac{3\delta\sqrt{\abs{\cA}}}{\abs{\cA}}
    = \frac{3\delta}{\sqrt{\abs{\cA}}},\nn
\end{align}
where $a$ follows from the inequality 
$\norm{\cO}{\infty} \leq \norm{\cO}{1}$ for 
an operator $\cO$.
\endproof
\section{Proof of \eqref{small_N_no_Tilt}}\label{proof_small_N_no_tilt}
Consider the following series of inequalities: 

\begin{align}
    &\hspace{10pt}\norm{\cN_{\delta}(\ketbra{h})}{\infty}\nn\\
    &\leq \frac{1}{\abs{\cA}^2(1 + 2\delta^2)}\left(\norm{\sum_{r \in [\abs{\cA}], t \in [\abs{\cA}]}\delta\ket{h}(\bra{h}\bra{a_{r}} \oplus \bra{h}\bra{b_{t}})}{\infty} + \norm{\sum_{r \in [\abs{\cA}], t \in [\abs{\cA}]}\delta(\ket{h}\ket{a_{r}} \oplus \ket{h}\ket{b_{t}})\bra{h}}{\infty}\right.\nn\\
    & \hspace{25pt}\left. + \norm{\sum_{r \in [\abs{\cA}], t \in [\abs{\cA}]}\delta^2\ketbra{h}(\ket{a_{r}} \oplus \ket{b_{t}})(\bra{a_{r}} \oplus \bra{b_{t}})}{\infty}\right)\nn\\
    &= \frac{1}{\abs{\cA}^2(1 + 2\delta^2)}\left(2\norm{\sum_{r \in [\abs{\cA}], t \in [\abs{\cA}]}\delta\ket{h}(\bra{h}\bra{a_{r}} \oplus \bra{h}\bra{b_{t}})}{\infty} + \norm{\sum_{r \in [\abs{\cA}], t \in [\abs{\cA}]}\delta^2\ketbra{h}(\ket{a_{r}} \oplus \ket{b_{t}})(\bra{a_{r}} \oplus \bra{b_{t}})}{\infty}\right)\nn\\
    &\overset{a}{\leq} \frac{1}{\abs{\cA}^2(1 + 2\delta^2)}\left(2\norm{\sum_{r \in [\abs{\cA}], t \in [\abs{\cA}]}\delta\ket{h}(\bra{h}\bra{a_{r}} \oplus \bra{h}\bra{b_{t}})}{1} + \norm{\sum_{r \in [\abs{\cA}], t \in [\abs{\cA}]}\delta^2\ketbra{h}(\ket{a_{r}} \oplus \ket{b_{t}})(\bra{a_{r}} \oplus \bra{b_{t}})}{\infty}\right)\nn\\
    &= \frac{1}{\abs{\cA}^2(1 + 2\delta^2)}\left(2\delta\tr\left[\sqrt{\ketbra{h}}\right]\sqrt{\sum_{\substack{r \in [\abs{\cA}], t \in [\abs{\cA}] \\ r' \in [\abs{\cA}], t' \in [\abs{\cA}]}}(\bra{h}\bra{a_{r}} \oplus \bra{h}\bra{b_{t}})(\ket{h}\ket{a_{r'}} \oplus \ket{h}\ket{b_{t'}})} \right.\nn\\
    &\hspace{25pt}\left.+ \delta^2\norm{\sum_{r \in [\abs{\cA}], t \in [\abs{\cA}]}(\ket{a_{r}} \oplus \ket{b_{t}})(\bra{a_{r}} \oplus \bra{b_{t}})}{\infty}\right)\nn\\
    &\leq \frac{1}{\abs{\cA}^2(1 + 2\delta^2)}\left(2\delta\sqrt{2\abs{\cA}^3} + \delta^2\left(\norm{\abs{\cA}\sum_{r \in [\abs{\cA}]}\ketbra{a_{r}}}{\infty} + \norm{\abs{\cA}\sum_{t \in [\abs{\cA}]}\ketbra{b_{t}}}{\infty} + \norm{\sum_{r \in [\abs{\cA}], t \in [\abs{\cA}]}\ket{a_{r}}\bra{b_{t}}}{\infty} \right.\right.\nn\\
    &\hspace{30pt}\left.\left.+\norm{\sum_{r \in [\abs{\cA}], t \in [\abs{\cA}]}\ket{b_{t}}\bra{a_{r}}}{\infty}\right)\right)\nn\\
    &= \frac{1}{\abs{\cA}^2(1 + 2\delta^2)}\left(2\sqrt{2}\delta\abs{\cA}\sqrt{\abs{\cA}} +  4\delta^2\abs{\cA} \right)
    \leq \frac{7\delta}{\sqrt{\cA}}\nn,
\end{align}
where $a$ follows from the relation $\norm{\cO}{\infty} \leq \norm{\cO}{1}$ for any operator $\cO$.
\endproof
\end{document}